\documentclass{amsart}
\usepackage[margin=1in,dvips]{geometry}
\allowdisplaybreaks
\usepackage{slashed}
\usepackage{enumerate}
\usepackage{tikz}
\usetikzlibrary{patterns}
\usepackage{scalerel}
\usepackage{stackengine,wasysym}
\usepackage{mathrsfs}
\usepackage{amssymb}
\usepackage{upgreek}
\usepackage{hyperref}
\usepackage{comment}
\def\f12{\frac 1 2}

\newcommand\reallywidetilde[1]{\ThisStyle{%
  \setbox0=\hbox{$\SavedStyle#1$}%
  \stackengine{-.1\LMpt}{$\SavedStyle#1$}{%
    \stretchto{\scaleto{\SavedStyle\mkern.2mu\AC}{.5150\wd0}}{.6\ht0}%
  }{O}{c}{F}{T}{S}%
}}

\makeindex

\newtheorem{definition}{Definition}[section]

\newtheorem{remark}{Remark}[section]
\newtheorem{convention}{Convention}[section]
\newtheorem{lemma}{Lemma}[section]
\newtheorem{theorem}{Theorem}[section]
\newtheorem*{corollary*}{Corollary}
\newtheorem{proposition}{Proposition}[section]

\title[The Asymptotically Self-Similar Regime for the  Einstein Vacuum Equations]{The Asymptotically Self-Similar Regime \\ for the Einstein Vacuum Equations}
\author{Igor Rodnianski}
\author{Yakov Shlapentokh--Rothman}
\address{\small Princeton University, Department of Mathematics, 
Fine~Hall, ~Washington~Road, ~Princeton, ~NJ~08544 }
\email{irod@math.princeton.edu}
\email{yshlapen@math.princeton.edu}
\date\today

\begin{document}

\maketitle
\begin{abstract}We develop a local theory for the construction of singular spacetimes in all spacetime dimensions which become asymptotically self-similar as the singularity is approached. The techniques developed also allow us to construct and classify exact self-similar solutions which correspond to the formal asymptotic expansions of Fefferman--Graham's ambient metric.
\end{abstract}
\section{Introduction}

As is well-known, understanding the dynamics of even initially regular solutions to the Einstein vacuum equations
\[{\rm Ric}\left(g\right) = 0\]
inevitably leads to the study of various types of singular solutions (Schwarzschild singularity $\{r = 0\}$, naked singularities, etc.). A special class of such singular spacetimes are so-called \emph{self-similar} solutions. These possess a conformally Killing vector field $K$ satisfying
\[\mathcal{L}_Kg = 2g,\]
which also, suitably interpreted, vanishes at the singularity and generates a natural dilation symmetry:
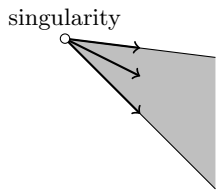
\begin{center}
\begin{tikzpicture}[scale = 1]
\fill[lightgray] (0,0)--(2,-.25)--(2,-2)  -- (0,0);
\draw (0,0) -- (2,-.25);
\draw (0,0) -- (2,-2);
\draw [thick,->] (0,0) -- (1,-1);
\draw [thick,->] (0,0) -- (1,-.5);
\draw [thick,->] (0,0) -- (1,-.125);
\path [draw=black,fill=white] (0,0) circle (1/16);
\draw node[above]{\small singularity}(0,0); 
\end{tikzpicture}
\end{center}

Such spacetimes arise in various guises throughout General Relativity. We specifically draw attention to the following:
\begin{enumerate}
	\item (Approximate) self-similar solutions have been heuristically connected to ``Type-II critical phenomena'' in gravitational collapse~\cite{chop,critical,oripiran} and to the endpoint of the Gregory--Laflamme instability~\cite{greglafl,lehpre}.
	\item Christodoulou's proof of cosmic censorship for the  spherically symmetric Einstein-scalar field system~\cite{notnaked} fundamentally relies on a well-posedness result for the Eisntein equations within the class of (potentially) singular ``solutions of bounded variation''~\cite{BV}.  A fundamental property of these solutions is that successive rescalings converge to a self-similar solution.
	\item Fefferman--Graham's classification of conformal invariants and corresponding asymptotic expansions for the ambient metric rely on a self-similar ansatz for the metric~\cite{FG1,FG2}. (We note also that these expansions play an important role in the analysis of the AdS-CFT~\cite{mald,adscft} and dS-CFT~\cite{dscft} correspondence in high energy physics.)
\end{enumerate}

The primary goal of this paper will to be to develop the local theory underyling the dynamical construction in all spacetime dimensions of singular solutions  which are asymptotically self-similar. Here the word ``dynamic'' is used to emphasize that we are primarily interested in solutions arising from explicitly given suitable characteristic Cauchy data.
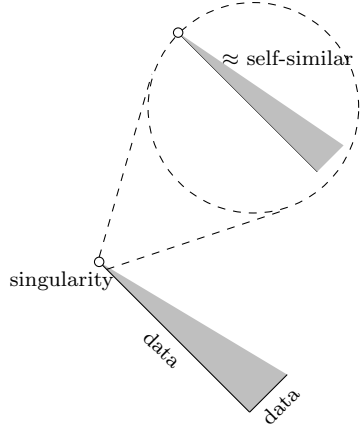
\begin{center}
\begin{tikzpicture}[scale = 1]
\fill[lightgray] (0,0)--(2,-2)--(2.5,-1.5)  -- (0,0);
\draw (0,0) -- (2,-2) node[sloped,below,midway]{\footnotesize data};
\draw (2,-2) -- (2.5,-1.5) node[sloped,below,midway]{\footnotesize data};
\draw (-.5,0) node[below]{\footnotesize singularity}; 

\path [draw=black,fill=white] (0,0) circle (1/16); 
\draw [dashed] (0,.0625) -- (.7,2.5);
\draw [dashed] (.1,-.1) -- (2.5,.7);
\draw[dashed] (2.05,2.05) circle (1.4);

\draw (2.5,2.5) node[above]{\footnotesize $\approx$ self-similar};
\draw (1.05,3.05)--(2.9,1.2);
\fill[lightgray] (1.05,3.05)--(2.9,1.2) -- (3.25,1.55) -- (1.05,3.05);

\path [draw=black,fill=white] (1.05,3.05) circle (1/16);

\end{tikzpicture}
\end{center}

It turns out that in the course of studying the above problem we will develop techniques which in fact allow us to to study \emph{exact} self-similar solutions. In particular, we will construct true self-similar solutions corresponding to all of the formal power series expansions of Fefferman--Graham.

Before we enter into a further discussion of motivation and background, we take the opportunity to explicitly state our main theorems.

\subsection{Basic Definitions and Statement of Results}
We start by introducing some notation and basic definitions. Let $\mathcal{S}$ denote an arbitrary closed and oriented differentiable manifold of dimension $n \geq 2$. We will use $\{\theta^A\}$ to denote  coordinates associated to an arbitrary chart on $\mathcal{S}$. Then we say that an $n+2$ dimensional Lorentzian manifold $\left(\mathcal{M},g\right)$ admits an $\mathcal{S}$-double null foliation if there exists $\mathcal{U}\subset \mathbb{R}^2$ such that $\mathcal{M}$ is diffeomorphic to $\left\{(u,v,\theta^A) \in\mathcal{U}\times \mathcal{S}\right\}$, and the metric $g$ takes the form
\[g = -2\Omega^2\left(du\otimes dv + dv\otimes du\right) + \slashed{g}_{AB}\left(d\theta^A - b^Adu\right)\otimes\left(d\theta^B - b^Bdu\right).\]
The Latin indices $A$ and $B$ run over directions tangent to $\mathcal{S}$, so that, in particular, for every $(u,v)$, $\slashed{g}$ yields a Riemannian metric on $\mathcal{S}_{u,v} \doteq \{u,v\} \times \mathcal{S}$. 

The null coordinates $(u,v)$ of the solutions we construct will always run over a region 
\[\mathcal{D}_{(c,d)} \doteq \{(u,v) : u \in [-d,0)\text{ and }v \in [0,-cu)\},\]
where $d \in \mathbb{R}_{> 0}\cup \{\infty\}$ and $c \in \mathbb{R}_{> 0}$, and one thinks of $c$ as being a small constant. 
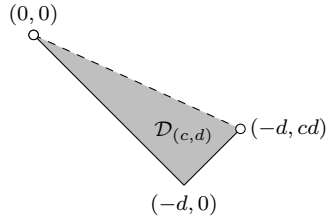
\begin{center}
\begin{tikzpicture}[scale = 1]
\fill[lightgray] (0,0)--(2,-2)--(2.75,-1.25)  -- (0,0);
\draw (0,0) -- (2,-2);
\draw (2,-2) -- (2.75,-1.25);
\draw (2,-2) node[below]{\footnotesize $(-d,0)$};
\draw (2.75,-1.25) node[right]{\footnotesize $(-d,cd)$};
\draw (2,-1.3) node{\footnotesize $\mathcal{D}_{(c,d)}$};
\path [draw=black,fill=white] (0,0) circle (1/16); 
\draw (0,0) node[above]{\footnotesize $(0,0)$};
\draw [dashed] (0,0)--(2.75,-1.25);
\path [draw=black,fill=white] (0,0) circle (1/16); 
\path [draw=black,fill=white] (2.75,-1.25) circle (1/16); 



\end{tikzpicture}
\end{center}

From now on we will always assume that every $\mathcal{S}$-double null foliation has $\mathcal{U} = \mathcal{D}_{(c,d)}$ for some choices of $c$ and $d$.
For any $\lambda > 0$ (thought of as a small constant) we define the associated ``scaling diffeomorphism'' $\Phi_{\lambda} : \mathcal{D}_{(c,\lambda^{-1}d)} \times \mathcal{S} \to \mathcal{D}_{(c,d)}\times \mathcal{S}$ by
\begin{equation}\label{rescalingmap}
\Phi_{\lambda}\left(u,v,\theta^A,\theta^B\right) \doteq \left(\lambda u,\lambda v,\theta^A,\theta^B\right).
\end{equation}
Note that $\Phi^{-1}_{\lambda}$ ``blows-up'' a small neighborhood of $(u,v) = (0,0)$.

Since the Einstein vacuum equations are invariant under diffeomorphism and multiplication by scalars, for any solution $\left(\mathcal{D}_{(c,d)}\times \mathcal{S},g\right)$ and $\lambda > 0$, we obtain a ``rescaled'' solution $\left(\mathcal{D}_{(c,\lambda^{-1}d)}\times\mathcal{S},g_{\lambda}\right)$, where we define
\[g_{\lambda} \doteq \lambda^{-2}\Phi_{\lambda}^*g.\]
(The division by $\lambda^{-2}$ makes the process dimensionless.) We say that a solution $\left(\mathcal{D}_{(c,\infty)},g\right)$ is \emph{self-similar} if $g_{\lambda} = g$. The reader can easily check that if $\left(\mathcal{M},g\right)$ is self-similar, then the vector field
\[K \doteq u\partial_u + v\partial_v,\]
satisfies
\[\mathcal{L}_Kg = 2g,\]
i.e., $K$ is conformally Killing. In the region under consideration, it will in fact either be null (when $v =0 $) or spacelike (when $v > 0$) . We call $K$ the \emph{scaling vector field}. Furthermore, we observe that Minkowski space given in standard null coordinates $(u,v,\theta^A)$ yields a self-similar solution. 

For future reference we note that a straightforward calculation yields that $g$ is self-similar if and only if for all $\lambda > 0$:
\begin{enumerate}
\item $\Omega\left(u,v,\theta\right) = \Omega\left(\lambda u,\lambda v,\theta\right)$.
\item $\slashed{g}_{AB}\left(u,v,\theta\right) = \lambda^{-2}\slashed{g}_{AB}\left(\lambda u,\lambda v,\theta\right)$.
\item $b^A\left(u,v,\theta\right) = \lambda b^A\left(\lambda u,\lambda v,\theta\right)$.
\end{enumerate}

Next, we turn to a discussion of the characteristic initial data we shall impose for our construction of asymptotically self-similar solutions. In what follows we take $d = 1$ so that we are considering a region $\mathcal{D}_{(c,1)}$. Then the choice of this initial data may be broken up into three main parts:
\begin{enumerate}
	\item Along the incoming cone $\{v = 0\}$ we must specify the lapse $\Omega$, the shift $b$, and the conformal class of the induced metrics on $\mathcal{S}_{u,0}$, which we denote by $\hat{\slashed{g}}(u)$.
	\item Along the outgoing cone $\{u=-1\}$ we must specify the lapse $\Omega$ and the conformal class of the induced metrics on $\mathcal{S}_{-1,v}$, which we similarly denote by $\hat{\slashed{g}}(v)$. (These choices must also be compatible with the incoming data in that the lapse $\Omega$ and conformal class $\hat{\slashed{g}}$ must be continuous at $(u,v) = (-1,0)$.)
	\item On $\mathcal{S}_{-1,0}$ we must specify the induced metric $\slashed{g}$, the torsion $\zeta$, ${\rm tr}\chi$, and ${\rm tr}\underline{\chi}$. (See Section~\ref{riccicoeff} for the definitions of $\zeta$, ${\rm tr}\chi$, and ${\rm tr}\underline{\chi}$.)
\end{enumerate}
The remaining parts of the initial data are then determined by the null constraint equations. 

 Let's start with the data along $\{v =0\}$. Since the action of the scaling diffeomorphism $\Phi_{\lambda}$ leaves the hypersurface $\{v = 0\}$ invariant, one immediately obtains a definition of a metric $g$ being self-similar along $\{v=0\}$. In particular, it is not hard to see that such a metric must satisfy
\begin{equation}\label{metricselfsimilarincoming}
\partial_u\Omega^2|_{v=0} = 0,\qquad \partial_u\left(u^{-1}b_A\right)|_{v=0} = 0,\qquad \hat{\slashed{g}}_{AB}|_{v=0} = \left(\hat{\slashed{g}}_0\right)_{AB},
\end{equation}
where the $AB$ refer to a Lie-propagated frame and $\slashed{g}_0$ is some Riemannian metric on $\mathcal{S}$ which is extended to all of $\{v=0\}$ by Lie-propagation. All of the self-similar solutions we will study in this paper will in fact satisfy the following normalization conditions:
\begin{equation}\label{normalized}
\Omega^2|_{v=0} = 1,\qquad b|_{v=0} = 0.
\end{equation}
The assumption that $\Omega = 1$ can be relaxed without any essential change to the arguments. However, the assumption that $b$ vanishes turns out to be necessary in order to guarantee regularity of the lapse $\Omega$. We note that given these normalizations,  Raychaudhuri's equation in the $u$-direction implies that the behavior of $\hat{\slashed{g}}$ prescribed in~\eqref{metricselfsimilarincoming}  is equivalent to choosing $\slashed{g}_{AB}(u)|_{v=0} = u^2\left(\slashed{g}_0\right)_{AB}$. We now group these choices into the following definition:
\begin{definition}\label{indef}We say that the incoming initial data along $\{v = 0\}$ is exactly self-similar if 
\[\Omega|_{v=0} = 1,\qquad b|_{v=0} = 0,\qquad \hat{\slashed{g}}(u) = \hat{\slashed{g}}_0,\]
for some Riemannian metric $\hat{\slashed{g}}_0$ on $\mathcal{S}$.

We say that the incoming initial data along $\{v = 0\}$ is asymptotically self-similar if there exists $\delta > 0$ such that along $\{v = 0\}$ in the coordinate frame
\begin{equation*}
\left|\Omega(u) - 1\right| \lesssim |u|^{2\delta},\qquad \left|b_A(u)\right| \lesssim |u|^{1+2\delta},\qquad \lim_{u\to 0}\hat{\slashed{g}}(u) = \hat{\slashed{g}}_0,\qquad \left(\hat{\slashed{g}}\right)^{AC}\left(\hat{\slashed{g}}\right)^{BC}\mathcal{L}_u\hat{\slashed{g}}_{AB}\mathcal{L}_u\hat{\slashed{g}}_{CD}  \lesssim |u|^{-2+2\delta}.
\end{equation*}
\end{definition}
\begin{remark}The bounds for the asymptotically self-similar data are derived by requiring that they are consistent with the rest of the solution along $\{v = 0\}$ (which is obtained by solving the null constraint equations) eventually satisfying bounds consistent with self-similarity. 
\end{remark}

Next, let's turn to the data we impose along the conjugate null hypersurface $\{u = -1\}$. Since the scaling diffeomorphism does not leave any conjugate null hypersurface invariant, it is not a priori clear what data corresponds to exact self-similarity and, even less so, what data corresponds to asymptotic self-similarity. Fortunately, it turns out that our methods allow us to construct solutions for quite flexible choices of conjugate data. Furthermore, the blow-ups $\{g_{\lambda}\}_{\lambda > 0}$ of the corresponding solutions exhibit a strong universality property in that the behavior as $\lambda \to 0$ turns out to only depend very weakly on the exact form of the conjugate data. In particular, we will find that the limit as $\lambda \to 0$ will be self-similar even without excessive fine-tuning of the data along $\{u=-1\}$.

More concretely, depending on the dimension $n$, our energy estimate scheme will require us to work with varying number of $\partial_v$ derivatives of the metric.\footnote{The specific number of $\partial_v$ derivatives commuted with is as follows: When $n=2$, we use $2$ derivatives, when $n \geq 3$ and odd, we use $\frac{n+1}{2}$ derivatives, and when $n \geq 4$ and even, we use $\frac{n}{2}$ derivatives. See Section~\ref{curvexplained}.} This immediately leads to the requirement that for a suitable function $F(n)$, we have that $\{\mathcal{L}_v^i\hat{\slashed{g}}\}_{i=0}^{F(n)}$ (and an appropriate number of angular derivatives thereof) lie in $L^2$ along $\{u=-1\}$.\footnote{In reality, when $n > 2$ we renormalize out certain of the most singular parts of $\hat{\slashed{g}}$.}  Next, as is well-known, once the values of $\{\mathcal{L}^i_v\hat{\slashed{g}}\}_{i=1}^{F(n)}$ along $(u,v) = (-1,0)$ are known, the null constraint equations determine  $\{\mathcal{L}^i_v\hat{\slashed{g}}\}_{i=1}^{F(n)}$ along all of $\{v = 0\}$. We must require that these satisfy bounds consistent with being asymptotically self-similar as $u\to 0$. Integration of the null constraint equations immediately shows that this requirement is equivalent to requiring that $\{\mathcal{L}_v^i\hat{\slashed{g}}\}_{i=1}^{F(n)}$ takes specific values in terms of $\slashed{g}|_{\mathcal{S}_{-1,0}}$. Lastly, when $n \geq 3$, an analysis of the constraint equations show that it is in fact natural to allow the conjugate data to be mildly singular as $v\to 0$ (see the discussion in Section~\ref{consexpl}). 

Before we present explicitly the class of conjugate data we consider, it is useful to introduce the following notation. For any $1$-parameter family of tensors $\Theta(v)$ defined on $\mathcal{S}$ which have a limit as $v\to 0$, we define
\[\reallywidetilde{\Theta} \doteq \Theta(v) - \Theta(0).\]
Now we are ready for the following definition.
\begin{definition}\label{admissibleconjugatedata}For a $1$-parameter family of tensors $\phi_{AB}(v)$ given in a coordinate frame, we let $\partial^N_{\theta}\phi_{AB}$ denote an arbitrary combination of up to $N$ $\mathcal{S}$-coordinate derivatives of $\phi_{AB}$; then $\left|\partial_{\theta}^N\phi_{AB}\right|$ denotes the supremum of the absolute value of all such derivatives.

Suppose we are given a metric $\slashed{g}_0$ which will eventually be the induced metric on $\mathcal{S}_{-1,0}$. Then we say that a $1$-parameter family $\hat{\slashed{g}}(v)$, for $0 \leq v \leq v_0 \ll 1$, of conformal classes of metrics on $\mathcal{S}$ is \textbf{admissible relative to $\slashed{g}_0$} if one of the following conditions are satisfied, depending on the dimension.
\begin{enumerate}
	\item When $n = 2$, there exists $\delta > 0$ so that:
	\begin{enumerate}
	\item 
	\[\lim_{v\to 0}\partial_{\theta}^N\partial_v^i\hat{\slashed{g}} \text{ exists and is uniformly bounded,}\]
	for some finite but sufficiently large $N$ and $0 \leq i \leq 2$.
	\item 
	\[\int_0^{v_0}\int_{\mathcal{S}}\left|\partial_{\theta}^N\reallywidetilde{\partial_v^i\hat{\slashed{g}}}\right|^2v^{-1+2\delta} < \infty, \]
	for some finite but sufficiently large $N$ and $0 \leq i \leq 2$.
	\end{enumerate}
	\item When $n \geq 3$ and odd, there exists $\delta > 0$ and a tensor $h_{AB}$ so that:
	
	\begin{enumerate}
		\item 
		\[\lim_{v\to 0}\partial_{\theta}^N\partial_v^i\hat{\slashed{g}}\text{ exists and is uniformly bounded,}\]
		for some finite but sufficiently large $N$ and $0 \leq i \leq \frac{n-1}{2}$, and that when $i > 0$ the limits have certain prescribed values in terms of $\slashed{g}_0$. (These values are determined in Proposition~\ref{incomingdatanodd} of Section~\ref{secinitialdata}.)
		\item 
		\[\int_0^{v_0}\int_{\mathcal{S}}\left|\partial_{\theta}^N\reallywidetilde{\partial_v^i\hat{\slashed{g}}}\right|^2v^{-1+2\delta} < \infty,\]
		 for some finite but sufficiently large $N$ and $0 \leq i \leq \frac{n-1}{2}$.
		\item  
	
		\[\lim_{v\to 0}\partial_{\theta}^N\left(\partial_v^{\frac{n+1}{2}}\hat{\slashed{g}}-v^{-1/2}h_{AB}\right)\text{ exists and is uniformly bounded,}\]
		for some finite but sufficiently large $N$. Let's denote this limit by $\mathscr{K}_N$.
		\item 
		\[\int_0^{v_0}\int_{\mathcal{S}}\left|\partial_{\theta}^N\left(\partial_v^{\frac{n+1}{2}}\hat{\slashed{g}}-v^{-1/2}h_{AB}\right) - \mathscr{K}_N\right|^2v^{-1+2\delta} < \infty,\]
		for some finite but sufficiently large $N$.
		
	\end{enumerate}

	\item When $n \geq 4$ and even, there exists $\delta,\iota > 0$ so that:
	\begin{enumerate}
		\item 
		\[\lim_{v\to 0}\partial_{\theta}^N\partial_v^i\hat{\slashed{g}}\text{ exists and is uniformly bounded,}\]
		for some finite but sufficiently large $N$ and $0 \leq i \leq \frac{n-2}{2}$, and that when $i > 0$ the limits have certain prescribed values in terms of $\slashed{g}_0$.  (These values are determined in Proposition~\ref{incomingdataneven} of Section~\ref{secinitialdata}.)
		\item 
		\[\sup_{ 0 < \tilde v \leq v_0}\tilde v^{-2\delta}\int_0^{\tilde v}\int_{\mathcal{S}}\left|\partial_{\theta}^N\reallywidetilde{\partial_v^i\hat{\slashed{g}}}\right|^2v^{-1-2\iota+2\delta} < \infty,\]
		 for some finite but sufficiently large $N$ and $0 \leq i \leq \frac{n-2}{2}$.
		\item In Proposition~\ref{incomingdataneven} of Section~\ref{secinitialdata} we will define a certain tensor $\mathcal{O}_{AB}$ on $\mathcal{S}$ in terms of $\slashed{g}_0$. Then we require that 
		\[\lim_{v\to 0}\partial_{\theta}^N\left(\partial_v^{\frac{n}{2}}\hat{\slashed{g}}-\log(v)\mathcal{O}_{AB}\right)\text{ exists and is uniformly bounded,}\]
		for some finite but sufficiently large $N$. Let's denote this limit by $\mathscr{L}_N$.
		\item 
		\begin{equation}\label{woohoothisisawiseassumptiontomake}
		\sup_{ 0 < \tilde v \leq v_0}\tilde v^{-2\delta}\int_0^{\tilde v}\int_{\mathcal{S}}\left|\partial_{\theta}^N\left(\partial_v^{\frac{n}{2}}\hat{\slashed{g}}-\log(v)\mathcal{O}_{AB}\right) - \mathscr{L}_N\right|^2v^{-1-2\iota+2\delta} < \infty,
		\end{equation}
		for some finite but sufficiently large $N$.
		
	\end{enumerate}

\end{enumerate}

\end{definition}
\begin{remark}Informally, when $n > 2$, the reader can interpret these conditions as saying that the first $\lfloor \frac{n-1}{2}\rfloor$ $v$-derivatives of $\hat{\slashed{g}}$ at $\mathcal{S}_{-1,0}$ are determined in terms of $\hat{\slashed{g}}$, but the $\frac{n}{2}$ $v$-derivative is free in the sense that there is freedom to add an arbitrary term of the form $\hat{\slashed{g}}^{(\frac{n}{2})}_{AB}v^{\frac{n}{2}}$ to $\hat{\slashed{g}}$. Furthermore, Theorem~\ref{classifyself} shows that derivatives at a higher order than $\frac{n}{2}$ do not affect the behavior of the self-similar blow-up of the solution.
\end{remark}
\begin{remark}We note that from the proofs of our main results it is possible to, in principle, extract explicitly the formulas relating the $\partial_v^i\hat{\slashed{g}}(0)$ and $\mathcal{O}_{AB}$ to $\slashed{g}_0$. (These formulas simply correspond to an appropriate finite part of the Fefferman--Graham expansion~\cite{FG1,FG2} expressed in terms of a double null coordinate system.)
\end{remark}
\begin{remark}The number $N$, which depends on the dimension $n$, is determined by how many times we will need to commute with angular derivatives in the proof of our argument, and this is, in turn, primarily tied to the number of derivatives required for applying $L^{\infty}$-Sobolev inequalities on $\mathcal{S}$.
\end{remark}

We are now ready to state our main results. The first is a local existence result for solutions with exactly self-similar incoming data and admissible conjugate data.
\begin{theorem}\label{localexistenceproto}(Local Existence for ``Proto-Ambient'' Metrics) Let $\left(\mathcal{S},\slashed{g}_0\right)$ be a closed and oriented Riemannian manifold of dimension $n$, and $\hat{\slashed{g}}(v)$ be an admissible $1$-parameter family of metrics on $\mathcal{S}$ relative to  $\slashed{g}_0$.

Then, for $\epsilon$ sufficiently small, and $\mathcal{M} \doteq \{(u,v,\theta^A) \in \mathcal{D}_{(\epsilon,-1)} \times \mathcal{S}\}$, there exists a unique metric $g$ on $\mathcal{M}$ which lies in a suitable scale invariant Sobolev space, weakly solves the Einstein equations, 
\[{\rm Ric}\left(g\right) = 0,\]
and such that in the corresponding $\mathcal{S}$-double null gauge we have 
\[\slashed{g}|_{v = 0} = u^2\slashed{g}_0,\qquad \zeta|_{(u,v) = (-1,0)} = 0,\qquad \Omega^2|_{\{v = 0\} \cup \{u=-1\}} = 1,\]
\[b|_{\{v = 0\}} = 0,\qquad {\rm tr}\chi|_{(u,v) = (-1,0)} = \frac{\slashed{R}_0}{n-1},\]

and there exists a function $\Phi\left(v,\theta^A\right)$ with 
\[\slashed{g}|_{u=-1} = \Phi^2\hat{\slashed{g}}.\]

Here $\slashed{R}_0$ denotes the scalar curvature of the metric $\slashed{g}_0$. The torsion $\zeta$ and trace of the second fundamental form ${\rm tr}\chi$ are defined in Section~\ref{riccicoeff}. 

The precise regularity result for the metric $g$ is that $\sup_{\frac{\tilde v}{|\tilde u|} \leq \epsilon}\left\vert\left\vert g\right\vert\right\vert_{\mathfrak{E}_{\tilde u,\tilde v}} < \infty$, where $\left\vert\left\vert \cdot\right\vert\right\vert_{\mathfrak{E}_{\tilde u,\tilde v}}$ is a certain scale-invariant weighted-energy norm defined in each characteristic rectangle $(u,v) \in [-1,\tilde u] \times [0,\tilde v]$. See Definition~\ref{totalnorm}.
\end{theorem}
\begin{remark}When $n=2$ the prescribed value of ${\rm tr}\chi$ is such that the Hawking mass along the incoming cone $\{v = 0\}$ is identically $0$. We note that the spherically symmetric self-similar solutions to the Einstein-Scalar-Field system found by Christodoulou~\cite{BV} also have a vanishing Hawking mass along the corresponding incoming cone.
\end{remark}
\begin{remark} The assumption that $\mathcal{S}$ is oriented is not necessary; but we make it for convenience.
\end{remark}

Note that (at least when $n = 2$) existence of a solution in an open set around $\{v = 0\}$ and $\{u=-1\}$ would follow immediately from the local theory for the characteristic initial value problem~\cite{lukchar,ren}. The key point of the theorem is that we have existence in the region $\mathcal{M} = \mathcal{D}_{(\epsilon,-1)}\times \mathcal{S}$ which is large enough to allow us to apply the rescaling diffeomorphism. In particular, we emphasize that there is no a priori reason to expect that the implicit singularity ``at'' $\mathcal{S}_{0,0}$ will propagate into a region with $\frac{v}{|u|} \gtrsim 1$.

We call the metrics produced by Theorem~\ref{localexistenceproto} ``proto-ambient'' metrics. Of course, given the general class of conjugate data we consider, these will generally not be exactly self-similar. However, the next theorem shows that successive rescalings of any ``proto-ambient'' metrics converge to a unique self-similar solution, i.e., in the nomenclature of~\cite{FG2}, an ``ambient metric''. 
\begin{theorem}\label{theoextractselsimilar}(Self-Similar Extraction) Let $\left(\mathcal{M},g\right)$ be a ``proto-ambient'' metric produced by Theorem~\ref{localexistenceproto}. Let $\{\lambda_i\}_{i=1}^{\infty}$ be any monotonically decreasing sequence with $\lambda_i \to 0$. Then there exists a \underline{unique} self-similar metric $g_{\rm sim}$ such that $g_{\lambda_i} \to g_{\rm sim}|_{\mathcal{M}}$ as $i \to \infty$. The convergence is with respect to the a certain ``supercritical'' norm $\left\vert\left\vert\cdot\right\vert\right\vert_{\overline{\mathfrak{E}}}$ defined in~\eqref{bignorm} of Section~\ref{extractsec}.
\end{theorem}

\begin{remark}\label{whytortrchi}It follows from the proof that the prescribed values of $\zeta$ and ${\rm tr}\chi$ in Theorem~\ref{localexistenceproto} and the conditions imposed on the conjugate data by Definition~\ref{admissibleconjugatedata} are necessary for the conclusions of Theorem~\ref{theoextractselsimilar} to hold. 
\end{remark}

Finally, as a corollary of our understanding of the uniqueness of the rescaling limits in Theorem~\ref{theoextractselsimilar}, we are able to obtain a full classification of self-similar solutions.
\begin{theorem}\label{classifyself}(Existence and Uniqueness of Self-Similar Solutions) Let $\slashed{g}_0$ be a Riemannian metric on $\mathcal{S}$ and $h$ be a symmetric traceless $2$-tensor of $\mathcal{S}$. Then there exists a \underline{unique} self-similar solution $g$ such that
	\[\slashed{g}|_{\mathcal{S}_{-1,0}} = \slashed{g}_0,\qquad {\rm tf}\left(\mathcal{L}^{n/2}_v\slashed{g}\right)|_{\mathcal{S}_{-1,0}} = h,\]
where ${\rm tf}$ denotes the ``trace-free'' part and the specification of $\left(\mathcal{L}^{n/2}_v\slashed{g}\right)$ refers to the specification of  the $v^{\frac{n}{2}}$ term in a power series of $\slashed{g}$ around $\{v = 0\}$.

Here the uniqueness is to be understood within the class of self-similar solutions equipped with an $\mathcal{S}$-double null foliation and satisfying the normalization condition~\eqref{normalized}.
\end{theorem}
The reader should compare this with Theorems 3.7, 3.9, and 3.10 of~\cite{FG2} where the analogue of Theorem~\ref{classifyself} is proven for self-similar solutions given by formal power series.
\subsection{Singular Nature of the Solutions}
We now briefly discuss the senses in which our metrics are singular. First of all, for geometric reasons it is immediately clear that for most choices of initial data there is no extension of our proto-ambient metrics in which $\{(u,v) = (0,0)\}$ corresponds to a smooth point.

On a more analytical level, for generic initial data and in an orthonormal frame, the null curvature components $\alpha_{AB}$ and $\beta_{AB}$ will blow up like $u^{-2}$ along $\{v = 0\}$ as the ``point'' $\{u=0\}$ is approached. In fact, though we will not prove this statement here, one expects that generically along any spacelike hypersurface $\{v = -cu\}$ for $c > 0$, all curvature components in an orthonormal frame will blow-up like $r^{-2}$ as $r \to 0$, where $r$ is distance to the origin $(u,v) = (0,0)$. In particular, the Riemann curvature tensor will have an infinite norm in the scale invariant space $L^{\frac{n+1}{2}}$. 

We thus see that Theorems~\ref{localexistenceproto} and~\ref{theoextractselsimilar} produce a large class of singular asymptotically scale-invariant solutions to the Einstein vacuum equations. Of course, the initial data considered along the incoming cone $\{v = 0\}$ is highly constrained. For the Einstein vacuum equations in $3+1$ dimensions, the $L^2$-based scale-invariant well-posedness result for non-constrained data would be for initial metrics lying in the scale-invariant Sobolev space $H^{\frac{3}{2}}$. Currently, with the resolution of the $L^2$-curvature conjecture~\cite{L21,L22,L23,L24,L25,L26} the best results available hold for $H^2$.

Finally, it is instructive to draw a comparison with the power series expansions from~\cite{FG2}.  In Section~\ref{ambientregsec} we will  show that our proto-ambient metrics are essentially as regular as the initial data allows. In particular, if we assume the initial data is suitably regular (but not that the tensors $h_{AB}$ and $\mathcal{O}_{AB}$ necessarily vanish!) the analysis of Section~\ref{ambientregsec} shows that our metrics are ``regular'' in the sense of the following definition.
\begin{definition}\label{ambientregular}We say that a metric $g$ is ``regular'' if
\begin{enumerate}
\item When $n=2$,  $g \in C^{\infty}\left(\mathcal{M}\right)$.
\item When $n \geq 3$ and odd, we have $g \in C^{\infty}\left(\mathcal{M}\setminus\{v=0\}\right)$ and  there exists smooth tensors
\[\{g^{(i)}_{\alpha\beta}\left(u,\theta^A\right)\}_{i=0}^{\frac{n-1}{2}},\qquad \tilde g_{\alpha\beta}\left(u,\theta^A\right),\]
 such that for every fixed $(u,\theta^A)$,
\begin{equation}\label{expandmetricodd}
g_{\alpha\beta}\left(u,v,\theta^A\right) = \sum_{i=0}^{\frac{n-1}{2}}g^{(i)}_{\alpha\beta}\left(u,\theta^A\right)v^i + v^{\frac{n}{2}}\tilde g_{\alpha\beta}\left(u,\theta^A\right) + O\left(v^{\frac{n+1}{2}}\right)\text{ as }v\to 0.
\end{equation}
\item When $n \geq 4$ and even, we have $g \in C^{\infty}\left(\mathcal{M}\setminus\{v=0\}\right)$ and  there exists smooth tensors
\[\{g^{(i)}_{\alpha\beta}\left(u,\theta^A\right)\}_{i=0}^{\frac{n}{2}},\qquad \tilde g_{\alpha\beta}\left(u,\theta^A\right),\]
 such that for every  fixed $(u,\theta^A)$,
\begin{equation}\label{expandmetriceven}
g_{\alpha\beta}\left(u,v,\theta^A\right) = \sum_{i=0}^{\frac{n}{2}-1}g^{(i)}_{\alpha\beta}\left(u,\theta^A\right)v^i +g^{(\frac{n}{2})}_{\alpha\beta}\left(u,\theta^A\right)v^{\frac{n}{2}}+\tilde g_{\alpha\beta}\left(u,\theta^A\right)v^{\frac{n}{2}}\log\left(v\right) + O\left(v^{\frac{n+2}{2}}\log(v)\right)\text{ as }v\to 0.
\end{equation}

We emphasize that this definition does not require $g$ to be self-similar and that all of the implied constants in the expansions may depend on $u$.  
\end{enumerate}
\end{definition}

This regularity is in complete agreement with the formal power series expansions for self-similar solutions from~\cite{FG1,FG2}. Exactly as in~\cite{FG1,FG2} the tracefree part of the coefficients of $v^{\frac{n}{2}}$ in these expansions is independent of $\slashed{g}^{(0)}$ (see Theorem~\ref{classifyself}) while all of the terms earlier in the expansion are determined by $\slashed{g}^{(0)}$. When $n \geq 4$ and even, the tensor in front of the $\log(v)v^{\frac{n}{2}}$ corresponds to the ``obstruction tensor'' from~\cite{FG1,FG2} (and is also determined by $\slashed{g}^{(0)}$). Finally, though we will not pursue this here, we note that one can establish full asymptotic expansions as $v\to 0$ which are analogous to those of Theorems 3.7, 3.9, and 3.10 from~\cite{FG2}.

\subsection{Regular Extensions to the Past}\label{extendtothepast}

It is of significant interest, obvious in the context of the evolution problem, to determine when our solutions can be extended to the past of the incoming null hypersurface $\{v = 0\}$ in such a way that the resulting spacetime solves the Einstein vacuum equations and is past complete. 

One particularly simple case occurs in exact self-similarity and when $\slashed{g}_0$ is the round metric on $\mathbb{S}^n$. Then one may glue in a copy of Minkowski space in the region $\{v < 0\} \cap \{v-u \geq 0\}$ to {\bf any} of the solutions produced by Theorem~\ref{classifyself} and obtain a solution to the Einstein vacuum equations:
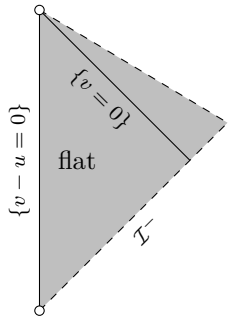
\begin{center}
\begin{tikzpicture}
\fill[lightgray] (0,0)--(0,-4)--(2.5,-1.5)  -- (0,0);

\draw (0,-4)--(0,0) node[sloped,above,midway]{\footnotesize $\{v-u = 0\}$};
\draw (0,0) -- (2,-2) node[sloped,below,midway]{\footnotesize $\{v = 0\}$};
\draw [dashed] (0,0) -- (2.5,-1.5);
\draw [dashed] (0,-4) -- (2.5,-1.5) node[sloped,below,midway]{\footnotesize $\mathcal{I}^-$};
\path [draw=black,fill=white] (0,0) circle (1/16); 
\path [draw = black,fill = white] (0,-4) circle (1/16);
\draw (.5,-2) node{$\text{flat}$};

\end{tikzpicture}
\end{center}
One can show that the resulting spacetime solves the Einstein vacuum equations weakly. When $n$ is even, the solution will lie in $C^{\frac{n}{2}-1,j}$ for $j < 1$, but generically will not lie $C^{\frac{n}{2}}$, and when $n$ is odd, the solution will lie in $C^{\lfloor\frac{n}{2}\rfloor,\frac{1}{2}}$, but generically will not lie in $C^{\lfloor\frac{n}{2}\rfloor,j}$ for $j > \frac{1}{2}$. In both cases, the singularity will be supported along $\{v = 0\}$. When $n = 2$ the singularities are examples of the impulsive gravitational wave singularities locally studied near $\{u=-1\}$ in \cite{impulse1,impulse2}. 

More generally and concretely, for any choice of ``Dirichlet data'' $\slashed{g}_0$ and ``Neumann data'' ${\rm tf}\left(\mathcal{L}_v^{\frac{n}{2}}\slashed{g}_0\right)$, there exists a formal Fefferman--Graham expansion solving the Einstein vacuum equations to the past of $\{v = 0\}$. However, when $\slashed{g}_0$ is the round metric on $\mathbb{S}^n$, under suitable assumptions one can show the requirement of past completeness \emph{uniquely} picks out Minkowski space. Furthermore, even though the left and right hand limits of ${\rm tf}\left(\mathcal{L}_v^{\frac{n}{2}}\slashed{g}_0\right)$ at $\{v = 0\}$ will generically not agree, the spacetime turns out to still be a weak solution of the Einstein vacuum equations.

We take the opportunity to note that when $n = 2$, even if one only desired to extend a single conjugate null hypersurface $\{u = c\}$ to the past of $\{v = 0\}$ as a usual $\{u=c\}$ hypersurface in Minkowski space; then the requirement that the null constraint equations are satisfied weakly along $\{u=c\}$ would force the normalizations for ${\rm tr}\chi$ and $\zeta$ that we have imposed in Theorem~\ref{localexistenceproto}.

Finally, we note that there are many other situations where one expects such extensions to exist, for example, the work of Graham and Lee~\cite{grlee} shows that these extensions exist for exactly self-similar solutions whenever $\slashed{g}_0$ is a small perturbation of the round metric on $\mathbb{S}^n$. In general, for exactly self-similar solutions one expects that after the imposition of a suitable gauge (and assumed topological type), the problem of finding a past complete extension such that $\slashed{g}$ is continuous yields an \emph{elliptic} problem which may have no solutions, a unique solution, or many solutions. See~\cite{andsurv,caigall,GursQing,adscft,wiyau}.

\subsection{Previous Work on Self-Similar Solutions}
In this section we will discuss context for and previous work done on self-similar solutions.
\subsubsection{Christodoulou's Scale-Invariant Solutions}

In the work~\cite{notnaked} Christodoulou studied the spherically symmetric Einstein-Scalar-Field system and showed that cosmic censorship holds. A fundamental role in the analysis is played by a well-posedness theorem for the spherically symmetric Einstein-Scalar-Field system within the class of so-called ``solutions of bounded variation''~\cite{BV}. A key property of this class of solutions is that the corresponding norms are left invariant by the analogue of the scaling transformations $\Phi_{\lambda}$.

With this in mind, our main results are motivated by the following: 
\begin{enumerate}
\item In the work~\cite{BV} Christodoulou studied exact self-similar solutions to the spherically symmetric Einstein-Scalar-Field system. (He calls them ``scale-invariant''.) The rigidities of spherical symmetry are sufficiently strong so that these solutions can be written down explicitly. In particular,  they are completely determined by the $v$-derivative of the scalar field at $(u,v) = (-1,0)$, and this $v$-derivative can take any value. Furthermore, after gluing a copy of Minkowski space, these solutions always admit a regular extension to the past.  Keeping the spherical symmetry in mind, we note that this is analogous to Theorem~\ref{classifyself} and to the discussion in Section~\ref{extendtothepast}. 
\item In the same work~\cite{BV} Christodoulou also established a result in the spirit of Theorem~\ref{theoextractselsimilar}. More specifically, he showed that, after passing to a subsequence, successive rescalings of a solution of bounded variation always converge to a self-similar solution. However, in contrast to Theorem~\ref{theoextractselsimilar}, the limit is not guaranteed to be unique. (We emphasize, however, that Christodoulou's solutions of bounded variation are much more general than the spherically symmetric Einstein-Scalar-Field analogues of the solutions considered in  Theorem~\ref{theoextractselsimilar}.)
\end{enumerate}

\subsubsection{Fefferman--Graham Expansions}\label{fgdiscuss}
In~\cite{FG2,FG1} Fefferman and Graham revolutionized the study of local conformal invariants. A fundamental role in their work was played by the following theorem.

\begin{theorem}[Fefferman--Graham \cite{FG2,FG1}]\label{analyFG} Let $n \geq 2$. To every \underline{analytic} Riemannian $n$-dimensional manifold $\left(\mathcal{S},\slashed{g}\right)$ and symmetric traceless $2$-tensor $h_{AB}$ on $\mathcal{S}$, there is a canonically associated $n+2$ dimensional self-similar Lorentzian manifold $\left(\mathcal{M},g\right)$ solving the Einstein vacuum equations such that $\left(\mathcal{S},\slashed{g}\right)$ embeds isometrically into $\left(\mathcal{M},g\right)$, and such that $h_{AB}$ is equal to the tracefree part of an appropriately normalized outgoing null $n/2$-derivative of $\slashed{g}$. 

The correspondence is conformally invariant in the sense that for a conformally related metric $\tilde{\slashed{g}} = \phi^2\slashed{g}$, there exists a explicitly computable choice of $\tilde h_{AB}$ such that the solution associated to $\left(\tilde{\slashed{g}},\tilde h\right)$ is isometric to the solution associated to $\left(\slashed{g},h\right)$. Furthermore, there is a privileged null hypersurface $\{v = 0\}$ covered by coordinates $(u,\theta^A) \in (0,\infty) \times \mathcal{S}$ where the induced metric is degenerate and takes the form
\[u^2\slashed{g}_{AB}d\theta^Ad\theta^B.\]
\end{theorem}
(Fefferman--Graham did not phrase their results in terms of double null foliations. See Appendix~\ref{somecoordinates} for the coordinates they used.)

In fact, their construction works when the original metric $\slashed{g}$ has any signature, not just Riemannian. The Pseudo-Riemannian geometry $\left(\mathcal{M},g\right)$ one obtains is called the ``ambient metric''. In the smooth category, they obtained an analogous result for formal metrics defined by a power series expansion around a fixed light cone, $\{v = 0\}$. The diagram below depicts the formal domain of the expansions they construct:
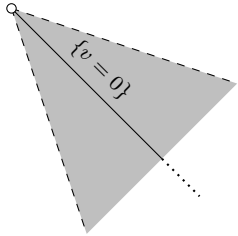
\begin{center}
\begin{tikzpicture}[scale = 1]
\fill[lightgray] (0,0)--(1,-3)--(2,-2)--(3,-1)  -- (0,0);
\draw (0,0) -- (2,-2) node[sloped,midway,above]{\footnotesize $\{v = 0\}$};

\draw [dashed] (0,0)--(3,-1);
\draw [dashed] (0,0)--(1,-3);
\draw [dotted,thick] (2,-2) -- (2.5,-2.5);
\path [draw=black,fill=white] (0,0) circle (1/16); 


\end{tikzpicture}
\end{center}

Keeping the case of Minkowski space in mind where the hypersurface $\{v = 0\}$ is an incoming null cone, we refer to the region $\{v < 0\}$  as the interior of the light cone, and the region $\{v > 0\}$ as the exterior of the light cone. For all Fefferman--Graham formal solutions, the self-similar vector field is null along $\{v = 0\}$, is timelike in the interior of the light cone, and is spacelike in the exterior of the light cone. A fundamental consequence is that one expects the self-similar Einstein equations to be elliptic in the interior and hyperbolic in the exterior. In this paper we are interested in the hyperbolic regime of self-similarity and this is why our results concern the exterior region.

Using this correspondence as a starting point, Fefferman and Graham were able to construct conformal invariants of $\left(\mathcal{S},[\slashed{g}]\right)$ by exploiting the well-known classification of local Pseudo-Riemannian invariants. In fact, they were even able to carry out their program in the smooth category even though in this case $(\mathcal{M},g)$ was only shown to exist as a formal power series. 

Despite the progress made in the construction of conformal invariants, it remained an open question to determine whether or not any analogue of Theorem~\ref{analyFG} holds in the smooth category. Our result Theorem~\ref{classifyself} provides an affirmative answer for the region exterior to the light cone. (Of course, since the problem in the interior of the light cone is elliptic, one does not expect a positive answer there.)  

 We now turn to a brief discussion of previous results in this direction. It turns out that the existing results concern a restricted class of the solutions considered in Theorem~\ref{analyFG}; we now give a brief explanation. 

Fefferman and Graham showed that if the divergence of the tensor $h_{AB}$ takes a certain value\footnote{For odd-dimensional $\mathcal{S}$ one takes $h$ divergence free while in even dimensions one takes the divergence of $h$ to be a certain explicit tensor determined by $\slashed{g}$.} then the resulting metric exhibits various simplifications and is called ``straight.'' In particular, in terms of a $\mathcal{S}$-double null foliation, to infinite order near $\{v = 0\}$ one has a trivial lapse and shift.\footnote{In Appendix~\ref{appstraight} we verify that under the same ``straightness'' assumption on initial data, the solutions produced by Theorem~\ref{classifyself} will have a trivial lapse and shift.} In this case it then possible to quotient out by the dilation invariance of the solutions in a very clean way and, in the interior of the light cone produce a formal expansion for an $n+1$ dimensional Riemannian manifold $\left(\mathcal{N}_1,g_1\right)$ which will be an asymptotically hyperbolic Poincar\'{e}-Einstein manifold, and in the exterior of the light cone produce a formal expansion for an $n+1$ dimensional Lorentzian manifold $(\mathcal{N}_2,g_2)$ which will solve the Einstein equations with a positive cosmological constant. See Appendix~\ref{somecoordinates} for the explicit coordinate calculations. In both cases, these manifolds are non-compact, but after a suitable conformal compactification, the original manifold $(\mathcal{S},\slashed{g})$ lies on the boundary. For the sub-class of straight ambient metrics, essentially all questions of interest descend to questions about these $n+1$ dimensional geometries. (It is worth noting that when $n=2$ then straight ambient metrics must in fact be flat.)

An important model case occurs when $(\mathcal{S},\slashed{g})$ is the round sphere $\left(\mathbb{S}^n,\slashed{g}_{\mathbb{S}^n}\right)$. Then, for vanishing $h_{AB}$, the resulting ambient metric is simply Mikowski space. The associated Poincar\'{e}-Einstein metric is hyperbolic space (this exactly corresponds to the well-known hyperboloid model), and the associated cosmological solution is simply de Sitter space. It is very natural to then ask the corresponding stability question. (We must take into account the expectations regarding ellipticity and hyperbolicity  when phrasing these questions!) Do small (smooth) perturbations of $(\mathbb{S}^n,[\slashed{g}_{\mathbb{S}^n}])$ along with small $h_{AB}$ (consistent with a straight ambient metric) have corresponding cosmological solutions which are qualitatively similar to de Sitter space? Do small (smooth) perturbations of $\left(\mathbb{S}^n,[\slashed{g}_{\mathbb{S}^n}]\right)$ correspond to \emph{complete} asymptotically hyperbolic Poincar\`{e}-Einstein metrics which are qualitatively similar to hyperbolic space?

The first question was answered in the affirmative in the case $n=3$ by the work of Friedrich on the stability of de Sitter space~\cite{fried} (see the later generalization~\cite{and} which extended the analysis to all odd dimensional $n$) and the second questions was answered in the affirmative in dimensions $n \geq 3$ by the work of Graham and Lee~\cite{grlee}. This latter work eventually played an important role in the AdS-CFT correspondence of high energy physics (see~\cite{mald,adscft}).
\subsection{Further Results}
Though we will not provide the proof in this paper in order to simplify the exposition, it is in fact possible to prove a version of Theorems~\ref{localexistenceproto} and~\ref{theoextractselsimilar} where the incoming data is not required to be exactly self-similar. We have

\begin{theorem}The analogues of Theorems~\ref{localexistenceproto} and~\ref{theoextractselsimilar} hold when the incoming data is assumed to asymptotically self-similar in the sense of Definition~\ref{indef}. 
\end{theorem}

Next, we recall that in the work~\cite{BV} Christodoulou produced explicit formulas for spherically symmetric self-similar solutions to the Einstein-Scalar-Field system. In particular, if we let $\phi$ denote the scalar field and set 
\[\epsilon^{-1/2} \doteq \partial_v\left(r\phi\right)|_{(u,v) = (-1,0)},\]
then when $0 < \epsilon \ll 1$, a trapped surface forms when $\frac{v}{|u|} \sim \epsilon$ and the solution exists in the region $\frac{v}{|u|} \lesssim \epsilon^{1/2}$. (All of the implied constants can be made precise, but here we are only interested in qualitatively tracking the dependence on $\epsilon$.)

In a forthcoming work we prove a result which can interpreted as showing that, qualitatively, trapped surface formation for self-similar solutions in vacuum and all spacetime dimensions works in the same fashion.
\begin{theorem} Let $\epsilon > 0$ be a parameter and suppose we have a self-similar solution produced by Theorem~\ref{classifyself} with 
\[\slashed{g}|_{\mathcal{S}_{-1,0}} = \slashed{g}_0,\qquad {\rm tf}\left(\mathcal{L}^{n/2}_v\slashed{g}\right)|_{\mathcal{S}_{-1,0}} = \epsilon^{-1/2}h.\]

Then, for $0 < \epsilon \ll 1$, the self-similar solution will exist in the region $\frac{v}{|u|} \leq B\epsilon^{\frac{1}{n}}$ for a constant $B$ independent of $\epsilon$. If we furthermore assume that $\inf_{\mathcal{S}}\left|h\right| > 0$, then a trapped surface will form when $\frac{v}{|u|} \geq b\epsilon^{\frac{2}{n}}$ for a constant $b$ independent of $\epsilon$. An analogous statement holds for proto-ambient metrics.
\end{theorem}
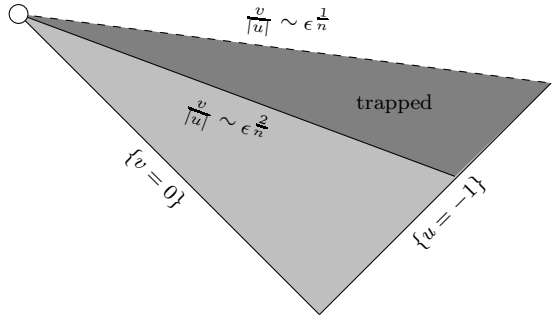
\begin{center}
\begin{tikzpicture}[scale = 2]
\fill[lightgray] (0,0)--(2,-2)--(3.54,-.46) -- (0,0);
\fill[gray] (0,0) -- (2.9,-1.08) -- (3.54,-.46) -- (0,0);

\draw (0,0) -- (2,-2) node[sloped,below,midway]{\footnotesize $\{v=0\}$};
\draw (2,-2) -- (3.54,-.46) node[sloped,below,midway]{\footnotesize $\{u=-1\}$};
\draw (0,0) -- (2.9,-1.08) node[sloped,below,midway]{\footnotesize $\frac{v}{|u|}\sim \epsilon^{\frac{2}{n}}$};
\draw (2.5,-.6) node{\footnotesize trapped};
\draw [dashed] (0,0) -- (3.54,-.46) node[sloped,above,midway]{\footnotesize $\frac{v}{|u|}\sim \epsilon^{\frac{1}{n}}$};
\path [draw=black,fill=white] (0,0) circle (1/16); 

\end{tikzpicture}
\end{center}

The reader should compare this with the work of An--Luk~\cite{anluk} which, while not directly concerning self-similar solutions, does establish a trapped surface formation result under a ``scale-invariant criterion'' in $3+1$ dimensions.

\subsection{Outline of Paper}
We close this section with an outline for the rest of the paper.

In Section~\ref{secsetupproof} we provide an outline of the proofs of Theorems~\ref{localexistenceproto},~\ref{theoextractselsimilar}, and~\ref{classifyself}. In particular, we explain the analysis of the constraint equations, how we carry out energy estimates for curvature, and how we estimate the Ricci coefficients. 

In Section~\ref{secdoublenull} we extend the well-known treatments of the double null gauge from $3+1$ dimensions to arbitrary spacetime dimensions. Among other things, we explain how the Bianchi equations can be used to carry out energy estimates, and systematically derive the correct scaling properties for the double null unknowns corresponding to a self-similar solution.

In Section~\ref{secinitialdata} we study the constraint equations for the characteristic initial data. Though it is not necessary for the proofs of our results it will be clear from the analysis how one can derive infinite order expansions for the metric. 

In Section~\ref{normssection} we define the various norms that we will use for our a priori estimates, define various renormalizations, and present useful schematic forms for commuted forms of the double null equations. The scale-invariance of the norms we define will be manifest.

In Section~\ref{aprioriestsection} we use a bootstrap argument to establish a priori estimates for the proto-ambient metrics in the norms defined in Section~\ref{normssection}. Combined with the local theory of Appendix~\ref{actuallocalexistencesec} this establishes Theorem~\ref{localexistenceproto}.

In Section~\ref{ambientregsec} we show that once the a priori estimates from Section~\ref{aprioriestsection} have been closed, and with suitable assumptions on the initial data, a preservation of regularity argument can be carried out to show that the proto-ambient metrics are quantitatively regular.

In Section~\ref{extractsec} we show how one can establish supercritical estimates for differences of rescaled proto-ambient metrics. We then use these improved estimates to establish Theorems~\ref{theoextractselsimilar} and~\ref{classifyself}.

Lastly, we have included an appendix which includes a local existence result for the characteristic initial value problem, some coordinate calculations, and a proof that if the tensor $h$ from Theorem~\ref{classifyself} satisfies an appropriate divergence condition, then the corresponding self-similar solution will have a trivial lapse and shift (see the discussion of ``straight'' solutions in Section~\ref{fgdiscuss}). 
\subsection{Acknowledgements}
The authors would like to thank Mihalis Dafermos for valuable discussions. IR acknowledges support through NSF grants DMS-1001500 and DMS-1065710. YS acknowledges support from the NSF Postdoctoral Research Fellowship under award no.~1502569.
\section{Outline of the Proofs of the Main Results}\label{secsetupproof}

In this section we turn to a more detailed outline of the argument. We will freely refer to the form of the Einstein equations in a double null foliation. See Section~\ref{secdoublenull} below for the relevant definitions and a detailed treatment of these equations. 

We quickly take the opportunity to recall that the main objects of interest in the double null formalism are the Ricci coefficients $\psi$ and curvature components $\Psi$. These satisfy a coupled system of equations, the most important of which are schematically of the form
\begin{equation}\label{transportheRicci}
\nabla_4\psi = \psi\cdot\psi  + \nabla\psi+\Psi,\qquad \nabla_3\psi = \psi\cdot\psi  + \nabla\psi+\Psi,
\end{equation}
\begin{equation}\label{dotheenergybian}
\nabla_4\Psi_1 = \mathcal{D}\Psi_2 + \psi\cdot\Psi,\qquad \nabla_3\Psi_2 = -\mathcal{D}^*\Psi_1 + \psi\cdot\Psi.
\end{equation}

Here the $\psi$'s and $\Psi$'s that show up on the right hand side of these equations denote a possibly arbitrary Ricci coefficient or curvature component respectively. The $\nabla_4$ and $\nabla_3$ are ``$\mathcal{S}_{u,v}$-projected covariant derivatives'' in the $e_4$ and $e_3$ null directions respectively, the $\nabla$ denotes the $\mathcal{S}$-gradient, $\mathcal{D}$ is a differential operator on $\mathcal{S}$, and $\mathcal{D}^*$ denotes the adjoint of $\mathcal{D}$. The equations~\eqref{transportheRicci} are usually treated as transport equations to estimate $\psi$ and the equations~\eqref{dotheenergybian} are used to carry out energy estimates for curvature.

\subsection{Self-Similar Bounds}
In Section~\ref{scalingbehav} we will provide exact formulas for how the various double null unkowns for a self-similar solution behave. Most importantly, we find that in a Lie-propagated coordinate frame $\{e_A\}$, the induced metric $\slashed{g}$ on the manifolds $\mathcal{S}$ must satisfy
\[\slashed{g}_{AB}\left(u,v,\theta\right)  = u^2\mathring{\slashed{g}}_{AB}\left(\frac{v}{u},\theta\right),\]
for some $1$-parameter family of metrics $\mathring{\slashed{g}}_{AB}\left(s,\theta\right)$ on $\mathcal{S}$. 

It also follows from the formulas of Section~\ref{scalingbehav} that for any Ricci coefficient $\psi$ and curvature component $\Psi$ we have
\[\left|\psi\right| \sim |u|^{-1},\qquad \left|\Psi\right| \sim |u|^{-2},\]
where these norms are taken with respect to the induced metrics $\slashed{g}$ on the manifolds $\mathcal{S}$. (A schematic way to remember these bounds is to recall that Ricci coefficients are proportional to one derivative of $g$, curvature components $\Psi$ are proportional to two derivatives of $g$, and that differentiation always adds a $|u|^{-1}$ to the bound.) 

Our arguments will also require us to apply angular and $\nabla_4$ derivatives to our double null unknowns. An application of these derivatives raises the homogeneity by $1$ and thus a self-similar solution will satisfy
\[\left|\nabla^i\nabla_4^j\psi\right| \sim_{i,j} |u|^{-1-i-j},\qquad \left|\nabla^i\nabla_4^j\Psi\right| \sim_{i,j} |u|^{-2-i-j}.\]
(The notation $\sim_{i,j}$ means that the implied constants may depend on $i$ and $j$.)

For the lapse and shift we have the following
\[\left|\nabla^i\nabla_4^j\Omega\right| \sim_{i,j} |u|^{-i-j},\qquad |\nabla^i\nabla_4^jb| \sim_{i,j} |u|^{-i-j}.\]

However, for our main theorems it is not these exact asymptotics which are important; instead, it is only necessary that we establish \emph{upper bounds} which are consistent with self-similar behavior. This leads to the following definition.

\begin{definition}\label{selfsimilarbounds}We say that a set of double null unknowns obeys ``self-similar bounds'' if for every Ricci coefficient $\psi$ and curvature component $\Psi$ we have 
\[\left|\nabla^i\nabla_4^j\psi\right| \lesssim_{i,j} |u|^{-1-i-j},\qquad \left|\nabla^i\nabla_4^j\Psi\right| \lesssim_{i,j} |u|^{-2-i-j},\]
for the lapse $\Omega$ and shift $b$ we have
\[\left|\nabla^i\nabla_4^j\Omega\right| \lesssim_{i,j} |u|^{-i-j},\qquad |\nabla^i\nabla_4^jb| \lesssim_{i,j} |u|^{-i-j}.\]

\end{definition} 
\subsection{Constraint Equation Analysis Along $\{v = 0\}$}\label{consexpl}

Our analysis begins with a study of the null constraint equations along $\{v = 0\}$. Indeed, as is well known, along any fixed null cone, the Einstein equations essentially reduce to a system of ordinary differential equations, and thus, given suitable ``seed data'' along the cone (on an incoming cone $\underline{\chi}$, $\Omega$, and $b$ will suffice) the values of the other double null unknowns and conjugate derivatives thereof are completely determined by their values on the initial manifold $\mathcal{S}_{-1,0}$.

The first order of business is thus to find seed data and initial values along $\mathcal{S}_{-1,0}$ so that after solving the null constraint equations along $\{v = 0\}$ the resulting set of double null unknowns obey self-similar bounds in the sense of Definition~\ref{selfsimilarbounds}.

We start with the case $n = 2$ which is simpler due to the lack of singularities as $v\to 0$. We will inductively determine the value of each $\nabla_4^i\psi$ and $\nabla_4^i\Psi$. More specifically we induct upwards on the number of $\nabla_4$ derivatives and, for each $i$, (roughly) induct backwards on the signature  of $\nabla_4^i\psi$ and $\nabla_4^i\Psi$ (see Section~\ref{signaturesection} for the definition of signature). The ``base case'' of our induction corresponds to the initial assumptions that
\begin{equation}\label{basecaseformaljet}
\Omega^2|_{v=0} = 1,\qquad b|_{v=0} = 0,\qquad \slashed{g} = u^2\slashed{g}_0,
\end{equation}
from which it immediately follows that along $\{v = 0\}$
\begin{equation}\label{frombasecasejet}
\hat{\underline{\chi}} = 0,\qquad {\rm tr}\underline{\chi} = \frac{2}{u},\qquad \underline{\omega} = 0,\qquad \underline{\alpha} = 0.
\end{equation}
The specification of these values should be thought of as a choice of seed data from which the rest of the double null unknowns (and $\nabla_4$ derivatives thereof) will be determined by integrating the null constraint equations and using their initial values on the manifold $\mathcal{S}_{-1,0}$. (One can already anticipate that these are sufficient seed data by signature considerations; $\underline{\chi}$, $\underline{\omega}$, and $\underline{\alpha}$ comprise all of the Ricci and null curvature components of highest signature.)

More specifically, consider any other Ricci coefficient $\psi_s$ of signature $s$ not determined already by the above. By inspection we see that these all satisfy $\nabla_3$ equations and thus we will have an equation schematically of the form
\begin{equation}\label{nabla3eqnforpsis}
\nabla_3\psi_s + cu^{-1}\psi_s = \sum_{\psi \neq {\rm tr}\underline{\chi},\ s_1+s_2 = s+1}\psi_{s_1}\psi_{s_2} + \nabla\psi_{s+1/2} + \Psi_{s+1},
\end{equation}
where the value of $c$ is determined by the presence of a term ${\rm tr}\underline{\chi}\psi_2$ in the corresponding null structure equation. Note that since the non-vanishing Ricci coefficients of the highest signature $1$ are only ${\rm tr}\underline\chi$, we actually do not see any Ricci coefficients of signature $1$ on the right hand side. In particular, the quadratic terms on the right hand side must only contain Ricci coefficients of a strictly higher signature.  Similarly, any curvature component $\Psi_s$ of signature $s$ will satisfy an equation schematically of the form
\begin{equation}\label{nabla3eqnforPsis}
\nabla_3\Psi_s + cu^{-1}\Psi_s = \nabla\Psi_{s+1/2} + \sum_{s_1+s_2 = s+1}\psi_{s_1}\Psi_{s_2}.
\end{equation}
(Though we will suppress the discussion of them in this introductory section, it is also important to use that certain curvature components satisfy constraint equations along the manifolds $\mathcal{S}$ from Proposition~\ref{constrainteqns}.)

Now consider the equation~\eqref{nabla3eqnforpsis} in the situation where along $\{v = 0\}$ the right hand side happens to already be a known function $F$ which behaves in a self-similar fashion. (In particular, in an orthonormal frame $|F| \sim u^{-2}$.)
\begin{equation}\label{ifweknownabla3psis}
\nabla_3\psi_s + cu^{-1}\psi_s = F.
\end{equation}
One can easily integrate this o.d.e.~along $\{v = 0\}$ and find that the general solution in an orthonormal frame $\{e_A\}$ is of the form
\[(\psi_s)_{AB} = \left((\psi_s)_{AB}|_{u=-1}\right)|u|^{-c} + |u|^{-c}\int_{-1}^{u}|\tilde u|^cF\, d\tilde u.\]

Depending on the value of $c$, the analysis proceeds in various ways.
\begin{enumerate}
		\item If $c > 1$ then there is in fact a unique solution which satisfies self-similar bounds.
		\item\label{whataplhais} If $c < 1$ then any choice of $\psi_s|_{u=-1}$ leads to a solution which satisfies self-similar bounds.  
	\item If $c = 1$ and $F = 0$, then any choice of $\psi_s|_{u=-1}$ leads to a solution which satisfies self-similar bounds. In this case there is a ``self-similar freedom''  in the specification of $\psi_s$.
	\item If $c = 1$ and $F \neq 0$, then~\eqref{ifweknownabla3psis} implies that $|\psi_s| \sim |\log(u)||u|^{-1}$ and thus there are no solutions consistent with self-similar bounds.
	\end{enumerate}
(It is straightforward to extend this analysis to the case when $F$ is only assumed to satisfy self-similar bounds.)

A similar set of scenarios occurs for the analysis of the equation~\eqref{nabla3eqnforPsis} in the case the right hand side is known. (In this case the critical value of $c$ would be $2$.) By treating the Ricci coefficients and null curvature components in the right order, the above analysis allows one to systematically determine what values for the double null unknowns at $\mathcal{S}_{-1,0}$ lead to a solution satisfying self-similar bounds along $\{v = 0\}$. (Occasionally, one must also appeal to the constraint equations of Proposition~\ref{constrainteqns}.) In particular, when $n = 2$ the problem of the logarithmic divergence turns out to never occur, one finds that $\hat{\chi}$ is the only double null unknown with a ``self-similar freedom'', and one finds that $\alpha$ is free in the sense of~\ref{whataplhais} above.

Next, one can determine $\nabla_4\psi|_{v=0}$ and $\nabla_4\Psi|_{v=0}$ easily for any Ricci coefficient and null curvature component which have a $\nabla_4$ null structure or Bianchi equation. It follows immediately from the form of the equations that these satisfy self-similar bounds. The only remaining unknown $\nabla_4$ derivatives are $\nabla_4\alpha$, $\nabla_4\underline{\eta}$, and $\nabla_4\omega$. However, for these unknowns, one can commute the corresponding $\nabla_3$ equation with $\nabla_4$ and then argue as above, \emph{mutatis mutandis}, to determine the freedom in $\nabla_4\Psi$ or $\nabla_4\psi$ by integrating the o.d.e. in the $u$-direction. (Note that Definition~\ref{selfsimilarbounds} implies that in an orthonormal frame we are interested in $\left|\nabla_4\Psi\right| \lesssim |u|^{-3}$ and $\left|\nabla_4\psi\right| \lesssim u^{-2}$ and thus the critical value of $c$ for $\nabla_4\Psi$ is $3$ and for $\nabla_4\psi$ is $2$.) Continuing in this fashion allows one to determine the allowed freedom for all the double null unknowns and their $\nabla_4$ derivatives. 

Since commutation with $\nabla_4$'s raise the homogeneity and preserves the general structure of the $\nabla_3$ equations (see Lemma~\ref{34commute} below), it is clear that after a sufficient number of $\nabla_4$ commutations the integration of every $\nabla_3$ equation will result in the analogue of case~\ref{whataplhais}, i.e., all solutions to the constraint equation will satisfy self-similar bounds. In fact,  when $n=2$, this already will occur after one $\nabla_4$ commutation.

In higher dimensions it is natural to allow for solutions which are singular as $v\to 0$. To see why, and also how the analysis above can potentially incorporate such singularities, we consider a Ricci coefficient equation of the form
\[\nabla_3\psi + cu^{-1}\psi = 0.\]
If $\psi$ is assumed to continuously extend to $\{v = 0\}$, then all solutions to this equation satisfy $\left|\psi\right| \sim |u|^{-c}$; hence, if $c > 1$ we would conclude that $\psi$ must vanish, if $c< 1$ then all solutions $\psi$ will have bounds which are better than self-similar, and if $c = 1$, $\psi$ has a ``self-similar freedom''. However, if $c < 1 $, it is also possible to formally consider the singular self-similar solution $u^{-c}v^{-1+c}$. The freedom  to consider such a solution will also be considered a ``self-similar freedom''.

Let us turn now specifically to the case of $n = 3$. One proceeds as when $n = 2$ until we come to $\alpha$, which is in principle singular as $v\to 0$. The relevant equation turns out to become
\[\nabla_3\alpha + \frac{3}{2}u^{-1}\alpha = F,\]
for some $F$ whose behavior along $\{v = 0\}$ has already been determined. To any given solution $\alpha_0$ one can add an undetermined part proportional $u^{-3/2}v^{-1/2}$.  This is the source of the formally undetermined singular self-similar term when $n = 3$. Once a choice has been made for this term, the analysis can proceed analogously to the case when $n = 2$. There turn out to be no more self-similar freedoms after $\alpha$.

Next we consider the case of $n = 4$. In this case the equation for $\alpha$ becomes
\[\nabla_3\alpha + 2u^{-1}\alpha = F,\]
for some regular self-similar function $F\sim |u|^{-3}$ which, for generic choices of $\slashed{g}_0$, does \emph{not} vanish. This time it seems one can always add a formally undetermined self-similar regular term proportional to $u^{-2}$. However, we have the problem that if $F$ does not vanish then there is a logarithmic divergence and it seems no regular self-similar solutions exist:
\[\alpha(u) \sim u^{-2}\int_{-1}^uu^2F + u^{-2}{\rm data} \sim \log(u)u^{-2}\text{ as }u\to0.\] 
The resolution of this turns out to be to allow $\alpha \sim \log\left(\frac{u}{v}\right)u^{-2}$, the point being that
\[\left(\nabla_3+2u^{-2}\right)\left(\log\left(\frac{u}{v}\right)u^{-2}\right) = u^{-3}.\]
Once we allow such logarithmic singularities and the formally undetermined regular term, the analysis can proceed analogously to the case when $n = 2$. As with $n = 3$ there turn out to be no more self-similar freedoms after $\alpha$. 

The case of higher dimensional odd and even $n$ works in an analogous fashion to the cases of $n = 3$ and $n = 4$, the key difference is that the corresponding undetermined term occurs for $\nabla_4^i\alpha$ for some $i > 0$.

We emphasize that, as is manifest from the argument above, this procedure is very wasteful with regards to angular derivatives. In particular, it is clear that these formal jet computations on their own \underline{cannot} be used to establish Theorem~\ref{classifyself}. 

At this point it is illuminating to draw a connection with the asymptotic expansions of Fefferman--Graham~\cite{FG1,FG2}. Recall that in the works~\cite{FG1,FG2} Fefferman and Graham showed that once one fixed a choice of 
\[\slashed{g}|_{v=0},\qquad {\rm tf}\left(\partial_v^{\frac{n}{2}}\slashed{g}\right)|_{v=0},\]
then there was a \emph{unique} power series expansion formally corresponding to a self-similar solution.\footnote{Though Fefferman and Graham did not work in the context of a double null foliation, it is easy to convert their coordinates to a double null foliation (see Appendix~\ref{somecoordinates}).} We will \underline{not} require these infinite order expansions for the proof of our main theorems. However, the constraint equation procedure outlined above is easily seen to generate such expansions. The freedom of $ {\rm tf}\left(\partial_v^{\frac{n}{2}}\slashed{g}\right)|_{v=0}$ exactly corresponds to the aforementioned self-similar freedoms.

\subsection{Energy Estimates for Curvature}\label{curvexplained}In order to prove something which goes beyond formal power series expansions, we will need to carry out a priori estimates. In particular, we will need to carry out energy estimates for the curvature components. (See Proposition~\ref{thebianchipairs} below.) In this section we will provide a heuristic discussion of our energy estimates. 

The analysis in this section will be done under the assumption that the Ricci coefficients will ultimately be shown to approximately behave as they do along $\{v = 0\}$. In the actual argument these assumptions will have to be part of the bootstrap. In particular, in this section we assume that 
\begin{equation}\label{ricciassumptions}
\sup_{(u,v)}\frac{u^2}{v}\left[\left|{\rm tr}\underline{\chi} - \frac{n}{u}\right| + \left|\Omega^2-1\right| + \left|\hat{\underline{\chi}}\right| + \left|\eta\right| + \left|\underline{\eta}\right| +|\omega| + |\underline{\omega}|\right] \lesssim 1,\qquad \left|\hat{\chi}\right| + \left|{\rm tr}\chi\right| \lesssim 1,
\end{equation}
and that $\frac{v}{|u|} \ll 1$.

We also introduce the following convention about integration:
\begin{convention}\label{volumeformconvention}
We introduce the convention that when we do not write a volume form it is implied that it is with respect to some combination of $du$, $dv$, or $d\slashed{Vol}_0$. Here $d\slashed{Vol}_0$ denotes the volume form associated to the metric $\slashed{g}_0$. 
\end{convention}

We first discuss the case when $n = 2$. Keeping the assumptions~\eqref{ricciassumptions}, Definition~\ref{defbiancpair}, and Proposition~\ref{thebianchipairs}  in mind  we write the first Bianchi pair as
\begin{equation}\label{n2alphaeqn}
\nabla_3\alpha + \frac{1}{u}\alpha = \nabla\hat{\otimes}\beta +\cdots,
\end{equation}
\begin{equation}\label{n2betaeqn}
\nabla_4\beta + 2{\rm tr}\chi\beta = {\rm div}\alpha +\cdots.
\end{equation}
(Later in the section we will discuss how to handle the terms hiding in the ``$\cdots$''.) The usual method for carrying out an energy estimate is to contract the first equation by $\alpha$, the second by $\beta$, integrate over a characteristic rectangle $\{(u,v,\theta) \in (-1,u_0) \times (0,v_0) \times \mathcal{S}\}$, integrate by parts, and add the resulting estimates together so as to cancel the spacetime terms containing angular derivatives. However this naive scheme will fail because, among other things, there will be a spacetime term of the wrong sign (remember that $u$ is negative in the region under consideration) proportional to
\[\int_{-1}^{u_0}\int_0^{v_0}\int_{\mathcal{S}}u^{-1}\left|\alpha\right|^2,\]
(keep Convention~\ref{volumeformconvention} in mind) and if one tries to use Gr\"{o}nwall to control this term one gets a logarithmic divergence in $u$. 

We can cure this particular logarithmic divergence by first conjugating~\eqref{n2alphaeqn} and~\eqref{n2betaeqn} by $|u|^k = (-u)^k$ for $k > 1$:
\begin{equation}\label{n2alphaeqn2}
\nabla_3\left(|u|^k\alpha\right) + \frac{1-k}{u}\left(|u|^k\alpha\right) = \nabla\hat{\otimes}\left(|u|^k\beta\right) +\cdots,
\end{equation}
\begin{equation}\label{n2betaeqn2}
\nabla_4\left(|u|^k\beta\right) + 2{\rm tr}\chi\left(|u|^k\beta\right) = {\rm div}\left(|u|^k\alpha\right) +\cdots.
\end{equation}
Now the spacetime term generated by $\alpha$ will have a good sign. Of course, there is still the spacetime term coming from $\beta$'s equation, and, even more worrisome, the value of ${\rm tr}\chi$ can be a more or less arbitrary function consistent with self-similarity.\footnote{In fact, we will see later in Section~\ref{secinitialdata} that ${\rm tr}\chi|_{v=0} = -\frac{\slashed{R}_0}{u}$ where $\slashed{R}_0$ is the Gaussian curvature of $\slashed{g}_0$.} The key point, however, is that we will control an $L^2$-flux of $\beta$ in the $u$-direction and we do \emph{not} see a logarithmic divergence when we apply Gr\"{o}nwall in the $v$-direction. More concretely, under the assumption that $\frac{v}{|u|} \ll 1$, we have
\[\sup_{0 \leq v \leq v_0}\int_{-1}^{u_0}f^2(u,v)\, du \leq A_1 + A_2\int_0^{v_0}\int_{-1}^{u_0}|u|^{-1}f^2(u,v)\, dudv \Rightarrow \]
\[\sup_{0 \leq v \leq v_0}\int_{-1}^{u_0}f^2(u,v)\, du \leq A_1\exp\left(A_2\int_0^{v_0}|u|^{-1}\, dv\right) \lesssim A_1.\]
Since all null curvature components other than $\alpha$ satisfy a $\nabla_4$ equation, we see that it is only $\alpha$'s equation which directly constrains the choice of $k$ (at least as far this particular potential logarithmic divergence goes).

We still need to determine the exact value of $k$. Note that the estimate we will eventually obtain for $\alpha$ and $\beta$ is
\begin{equation}\label{fluxtake1}
\sup_{\frac{v_0}{|u_0|} \leq \epsilon}\left[\sup_{-1 \leq u \leq u_0}\int_0^{v_0}\int_{\mathcal{S}}\left|\alpha\right|^2u^{2k}\, dv + \sup_{0 \leq v \leq v_0}\int_{-1}^{u_0}\int_{\mathcal{S}}\left|\beta\right|^2u^{2k}\, du\right].
\end{equation}
If we expect to the show that the solutions become ``scale-invariant'' as $u \to 0$, it is natural to ask that these norms are invariant under the rescaling diffeormorphism $\hat{\Phi}_{\lambda}$ (see~\eqref{rescaleddiff}). A straightforward computation shows that this leads to the choice of $k = 3/2$ (which is fortunately greater than $1$). Unfortunately, in Section~\ref{secinitialdata} we will see that $\left|\beta|_{v=0}\right| \sim u^{-2}$ with a non-zero implied constant, and thus, 
the choice of $k = 3/2$ leads to the following issue:
\[\int_{-1}^{u_0}\int_{\mathcal{S}}\left|\beta\right|^2u^3\, du \sim \left|\log\left(u_0\right)\right| \to \infty\text{ as }u_0 \to 0.\]

There are various ways to deal with this problem. If we were only interested in the $n=2$ case the simplest fix would be to introduce a small constant $\delta > 0$ and replace~\eqref{fluxtake1} with
\begin{equation}\label{fluxtake2}
\sup_{\frac{v_0}{|u_0|} \leq \epsilon}\left[\sup_{-1 \leq u \leq u_0}\int_0^{v_0}\int_{\mathcal{S}}\left|\alpha\right|^2u^{2k}\, dv + \sup_{0 \leq v \leq v_0}|u_0|^{2\delta}\int_{-1}^{u_0}\int_{\mathcal{S}}\left|\beta\right|^2u^{2k-2\delta}\, du\right].
\end{equation}

However, we will take an alternative approach which for various technical reasons turns out to be more convenient. We introduce the following notation:
\begin{definition}\label{defwidetilde}For any $\mathcal{S}_{u,v}$ tensor $\Theta$, we define
\[\reallywidetilde{\Theta}|_{\mathcal{S}_{u,v}} \doteq \Theta|_{\mathcal{S}_{u,v}} - \Theta|_{\mathcal{S}_{u,0}},\]
where $\mathcal{S}_{u,v}$ and $\mathcal{S}_{u,0}$ are identified via their canonical coordinate systems. (More specifically,  starting with any coordinate system (or frame) on $\mathcal{S}_{-1,0}$, we extend the coordinates by $\partial_u$ Lie-propagation to each $\mathcal{S}_{u,0}$. Then the coordinates are extended to each $\mathcal{S}_{u,v}$ by Lie-propagation with $\partial_v$. )
\end{definition}

Then we rewrite~\eqref{n2alphaeqn} and~\eqref{n2betaeqn} in terms of $\tilde\alpha$ and $\tilde\beta$:
\begin{equation}\label{n2alphaeqn3}
\nabla_3\tilde\alpha + \frac{1}{u}\tilde\alpha = \nabla\hat{\otimes}\tilde\beta +O\left(|u|^{-3}\right)+\cdots,
\end{equation}
\begin{equation}\label{n2betaeqn3}
\nabla_4\tilde\beta + 2{\rm tr}\chi\tilde\beta = {\rm div}\tilde\alpha +O\left(|u|^{-3}\right) +\cdots.
\end{equation}
The price we pay for working with the tilded quantities is the presence of inhomogeneous terms on the right hand side coming from the values of $\alpha$ and $\beta$ along $\{v =0 \}$ (one could try to exploit cancellations between the terms from $\alpha$ and $\beta$ but this does not turn out to be necessary). However, the advantage is that in addition to using positive $u$-weights, we can now also use \emph{negative} $v$-weights. We now let $\delta > 0 $ be a small constant and conjugate by $|u|^{2-\delta}v^{-1/2+\delta}$ (the ``total weight'' still needs to be $3/2$ to maintain scale-invariance). We obtain
\begin{equation}\label{n2alphaeqn4}
\nabla_3\left(|u|^{2-\delta}v^{-1/2+\delta}\tilde\alpha\right) + \frac{\delta-1}{u}\left(|u|^{2-\delta}v^{-1/2+\delta}\tilde\alpha\right) = \nabla\hat{\otimes}\left(|u|^{2-\delta}v^{-1/2+\delta}\tilde\beta\right) +O\left(|u|^{-1-\delta}v^{-1/2+\delta}\right)+\cdots,
\end{equation}
\begin{align}\label{n2betaeqn4}
\nabla_4\left((|u|^{2-\delta}v^{-1/2+\delta}\tilde\beta\right) &+\frac{1}{2}v^{-1}\left((|u|^{2-\delta}v^{-1/2+\delta}\tilde\beta\right)+ 2{\rm tr}\chi\left((|u|^{2-\delta}v^{-1/2+\delta}\tilde\beta\right) = 
\\ \nonumber &{\rm div}\left((|u|^{2-\delta}v^{-1/2+\delta}\tilde\alpha\right) +O\left(|u|^{-1-\delta}v^{-1/2+\delta}\right) +\cdots.
\end{align}

Now, it is easy to see that the intial data fluxes are finite and the energy estimate even generates a good spacetime term proportional $v^{-1}\left||u|^{2-\delta}v^{-1/2+\delta}\tilde\beta\right|^2$ in the equation for $\tilde \beta$. The $v$-weight is sufficiently strong to absorb the errors associated to the ${\rm tr}\chi\beta$ term. Finally, one checks that the inhomogeneous terms turn out not to be a problem (the $\delta > 0$ is necessary to avoid a logarithmic divergence). 

The next Bianchi pair is $\left(\beta,\left(\rho,\sigma\right)\right)$. We can carry out the analogous estimate; except this time we are not scared of a bad spacetime term for $\beta$ because we can control with the good spacetime term we obtained when we estimated the Bianchi pair $\left(\alpha,\beta\right)$. Repeating these estimates allows one to work down the whole Bianchi hierarchy. Note how important the coefficient of $u^{-1}$ in~\eqref{n2alphaeqn} is, and how the analogous coefficients in all of the other Bianchi equations are not as relevant.

Let's now consider the case when $n=3$. The first Bianchi pair is then
\begin{equation}\label{n3alphaeqn}
\nabla_3\alpha_{AB} + \frac{3}{2}u^{-1}\alpha_{AB} = -\nabla^C\nu_{C(AB)} + \nabla_{(A}\beta_{B)} + \cdots,
\end{equation}
\begin{equation}\label{n3betaeqn}
\nabla_4\beta_A = \nabla^B\alpha_{BA} + \cdots,
\end{equation}
\begin{equation}\label{n3nueqn}
\nabla_4\nu_{ABC} = -2\nabla_{[A}\alpha_{B]C} + \cdots.
\end{equation}
(Given our experience with the case of $n=2$ we have put the terms proportional to ${\rm tr}\chi\beta$ and ${\rm tr}\chi\nu$ into the ``$\cdots$''.) There are two fundamental differences with the case of $n=2$. First of all, the coefficient of $u^{-1}$ is now $\frac{3}{2}$ instead of $1$. Secondly, the best we can say about $\alpha$ along $\{u = -1\}$ is that $\left|\alpha\right||_{u=-1} \lesssim v^{-1/2}$. In particular, $\alpha$ is \emph{not} square-integrable initially. 

The term $\frac{3}{2}u^{-1}$ suggests that we need to use a $u$-weight with a power greater than or equal $\frac{3}{2}$. It is easiest to use a weight greater than $\frac{3}{2}$ and thus, to maintain scale-invariance, we need to use some negative $v$-weights; this naturally leads us to the use of $\tilde\alpha$, $\tilde\beta$, and $\tilde\nu$. However, since $\alpha$ is singular as $v\to 0$, it is not a priori clear what exactly $\tilde\alpha$ means. This leads to the following renormalization scheme:
\begin{definition}Let $n = 3$. Along $\{u = -1\}$ there exists $h_{AB}$ such that (in the coordinate frame)
\[\alpha_{AB}|_{u=-1} = v^{-1/2}h_{AB} + O\left(1\right).\]

Then we extend $h_{AB}$ to the whole spacetime by Lie-propagation (see the discussion in Definition~\ref{defwidetilde}) and set
\[\alpha'_{AB} \doteq \alpha_{AB} - v^{-1/2}|u|^{1/2}h_{AB}.\]
\end{definition}

The point of this renormalization scheme is that
\[\nabla_u\left(v^{-1/2}|u|^{1/2}h_{AB}\right) + \frac{3}{2}u^{-1}\left(v^{-1/2}|u|^{1/2}h_{AB}\right) = O\left(v^{1/2}|u|^{-7/2}\right).\]
(This is, of course, directly related to why in the formal analysis we discussed in Section~\ref{consexpl} the $v^{-1/2}$ term in the Taylor expansion of $\alpha$ is formally undetermined.)

For $\alpha'$ it makes sense to discuss $\reallywidetilde{\alpha'}$ and we eventually obtain 
\begin{equation}\label{n3alphaeqn2}
\nabla_3\reallywidetilde{\alpha'}_{AB} + \frac{3}{2}u^{-1}\reallywidetilde{\alpha'}_{AB} = -\nabla^C\reallywidetilde{\nu}_{C(AB)} + \nabla_{(A}\reallywidetilde{\beta}_{B)} + O\left(|u|^{-3}\right) + \cdots,
\end{equation}
\begin{equation}\label{n3betaeqn2}
\nabla_4\reallywidetilde{\beta}_A = \nabla^B\reallywidetilde{\alpha'}_{BA} + O\left(v^{-1/2}|u|^{-3/2}\right)+\cdots,
\end{equation}
\begin{equation}\label{n3nueqn2}
\nabla_4\reallywidetilde{\nu}_{ABC} = -2\nabla_{[A}\reallywidetilde{\alpha'}_{B]C} +O\left(v^{-1/2}|u|^{-3/2}\right)+ \cdots.
\end{equation}
Now we can carry out an estimate analogously to the $n=2$ case. The singular inhomogeneous terms on the right hand side of $\tilde\beta$'s and $\tilde\nu$'s equation turns out not to be a problem because the negative $v$-weight will generate a spacetime term with a $v^{-1}$ weight for $\tilde\beta$ and $\tilde\nu$. The rest of the Bianchi pairs may be treated similarly. (Note, however, that it is strictly easier to treat all other Bianchi pairs since there are no more analogues of the difficulties associated to $\alpha$.) As with $n=2$ we note the privileged role of the coefficient of $u^{-1}$ in~\eqref{n3alphaeqn}.

Next, let's discuss $n = 4$. This time we have
\begin{equation}\label{n4alphaeqn}
\nabla_3\alpha_{AB} + 2u^{-1}\alpha_{AB} = -\nabla^C\nu_{C(AB)} + \nabla_{(A}\beta_{B)} + \cdots,
\end{equation}
\begin{equation}\label{n4betaeqn}
\nabla_4\beta_A = \nabla^B\alpha_{BA} + \cdots,
\end{equation}
\begin{equation}\label{n4nueqn}
\nabla_4\nu_{ABC} = -2\nabla_{[A}\alpha_{B]C} + \cdots.
\end{equation}
Now the coefficient in front of $u^{-1}\alpha$ is $2$. Thus the smallest $u$-weight that we can hope to use is $2$ and we need a negative $v$-weight of at least $-\frac{1}{2}$ to maintain scale-invariance. This requires special care because we will not get a good spacetime term for $\alpha$ and if we naively use the scheme from $n=2$ or $n=3$ the use of a $v^{-1/2}$ weight will lead to a logarithmic divergence (recall that previously we used $v^{-1/2+\delta}$). We also need to keep in mind that along $\{u = -1\}$ we can have that $\alpha$ blows-up logarithmically as $v\to 0$. 

We start with a renormalization of $\alpha$ analogous to the case of $n=3$. The reader may find it useful to recall, for the specific case of $n \geq 4$ and even, both Definition~\ref{admissibleconjugatedata} and the discussion of the constraints from Section~\ref{consexpl}.
\begin{definition}Let $n = 4$. Along $\{u = -1\}$ there exists $\mathcal{O}_{AB}$ such that (in the coordinate frame)
\[\alpha_{AB}|_{u=-1} = \log(v)\mathcal{O}_{AB} + O\left(1\right).\]

Then we extend $\mathcal{O}_{AB}$ to the whole spacetime by Lie-propagation and set
\[\alpha'_{AB} \doteq \alpha_{AB} - \log\left(\frac{v}{u}\right)\mathcal{O}_{AB}.\]
\end{definition}

The key point for avoiding the feared logarithmic divergence is that if one recalls from Section~\ref{consexpl} how the logarithmic term in $\alpha$ is produced, then when we write~\eqref{n4alphaeqn},~\eqref{n4betaeqn}, and ~\eqref{n4nueqn} in terms of $\reallywidetilde{\alpha'}$, $\reallywidetilde{\beta}$, and $\reallywidetilde{\nu}$, then the inhomogeneous term produced in $\reallywidetilde{\alpha'}$ will in fact decay as $v\to 0$:
\begin{equation}\label{n4alphaeqn2}
\nabla_3\reallywidetilde{\alpha'}_{AB} + 2u^{-1}\reallywidetilde{\alpha'}_{AB} = -\nabla^C\reallywidetilde{\nu'}_{C(AB)} + \nabla_{(A}\reallywidetilde{\beta'}_{B)} + O\left(\log\left(\frac{v}{u}\right)\frac{v}{u^4}\right) + \cdots,
\end{equation}
\begin{equation}\label{n4betaeqn2}
\nabla_4\reallywidetilde{\beta'}_A = \nabla^B\reallywidetilde{\alpha'}_{BA} + O\left(\log\left(\frac{v}{u}\right)|u|^{-3}\right)\cdots,
\end{equation}
\begin{equation}\label{n4nueqn2}
\nabla_4\reallywidetilde{\nu'}_{ABC} = -2\nabla_{[A}\reallywidetilde{\alpha'}_{B]C} + O\left(\log\left(\frac{v}{u}\right)|u|^{-3}\right)+\cdots.
\end{equation}
Now we can carry out the energy estimates in (essentially) the same fashion as when $n=2$ and $n=3$. (In the actual estimates, analogously to ~\eqref{fluxtake2}, we actually put a weight $v^{-1/2+\delta}$ inside the integral and a $v_0^{-\delta}$ outside the integral.) The rest of the Bianchi pairs may be estimated similarly. Yet again, just as with $n=2$ and $n=3$, we note the privileged role of the coefficient of $u^{-1}$ in~\eqref{n4alphaeqn}.

Finally, we need to discuss how the case of $n > 4$ is handled. In general, the equation for $\alpha$ looks like
\[\nabla_3\alpha_{AB} + \frac{n}{2}u^{-1}\alpha_{AB} = -\nabla^C\nu_{C(AB)} + \nabla_{(A}\beta_{B)} + \cdots.\]
This suggests that we need a $u$-weight of at least $\frac{n}{2}$. Scale-invariance would then require that we use a negative $v$-weight of $3-\frac{n}{2}$. However, it is clear that for sufficiently large $n$, one cannot hope for the initial fluxes of the tilded quantities to be finite for such a large negative $v$-weight. One approach would be to subtract further terms in the Taylor expansion, but this quickly becomes very awkward. First suppose that $n$ is odd. Then we can commute the equation for $\alpha$ by $\nabla_4^{\frac{n-3}{2}}$ to obtain an equation of the form
\[\nabla_3\left(\nabla_4^{\frac{n-3}{2}}\alpha\right)_{AB} + \frac{n}{2}u^{-1}\nabla_4^{\frac{n-3}{2}}\alpha_{AB} = -\nabla^C\nabla_4^{\frac{n-3}{2}}\nu_{C(AB)} + \nabla_{(A}\nabla_4^{\frac{n-3}{2}}\beta_{B)} + \cdots.\]
(Of course, many additional terms are produced when we commute that have to be tracked carefully in the actual proof.) The application of the $\nabla_4$ derivatives raises the homogeneity of $\alpha$ and the estimate produced after conjugation with $|u|^{\frac{n}{2}+\frac{1}{2}-\delta}v^{-\frac{1}{2}+\delta}$ is now scale-invariant. Thus, after this commutation we can proceed essentially as we did when $n = 3$. A similar scheme works for large even $n$. Note that integrating from $\{v = 0\}$ allows us to recover $\nabla_4^i\alpha$ from $\nabla_4^{i+1}\alpha$. (For this it is important that the most singular term in $v$ we ever see is $v^{-1/2}$ which is integrable.)

Of course, in order to control the nonlinear terms hiding in the ``$\cdots$'' we will need to apply Sobolev inequalities and these in turn require commutations with angular derivatives $\nabla$. The basic principle is that for any curvature component $\Psi$, $\left|u\nabla\Psi\right|$ should satisfy the same estimates as $\left|\Psi\right|$ (cf.~the discussion in the introduction of~\cite{anluk}). Briefly, the reason we expect this to work is that when we commute a $\nabla_3$ Bianchi equations with $\nabla$ the formula from Lemma~\ref{3commute} shows that the coefficient of $u^{-1}$ will increase by $1$. Once a sufficient number of these commutations have been carried out, almost all of the nonlinear terms may be controlled in a bootstrap setting in a standard way using Sobolev inequalities on $\mathcal{S}$.

For the nonlinear terms on the right hand side of $\alpha$'s equation we have to be more careful. For example, consider a term on the right hand proportional to $\psi\alpha$. When $n = 3$ it seems that such a term could, in the worse case scenario, be proportional to $u^{-5/2}v^{-1/2}$ and thus will eventually produce a logarithmic divergence. When $n = 4$ then we must be worried that such a term will produce a spacetime term containing $\alpha$ of the wrong sign, which cannot be absorbed into anything since we do not produce a good spacetime term for $\alpha$. Similar worries occur for general $n$. However, as it turns out, by signature considerations the $\psi$ in such a nonlinear term must be one of the Ricci coefficients from~\eqref{ricciassumptions} which actually \emph{vanish} as $v\to 0$. Thus, these worst case scenarios do not actually happen. However, we do conclude that it is of fundamental importance that when we estimate the Ricci coefficients we do indeed recover the good $v$-weights in~\eqref{ricciassumptions}.

Before we close the section, we draw attention to the close connection between  the importance of $\alpha$ in the above energy estimate scheme and how in the constraint equation analysis of Section~\ref{consexpl}, $\alpha$ is the source of all of the ``self-similar freedom'' in the initial characteristic data. In particular, a heuristic reason for understanding the need to commute with $\nabla_4$ in higher dimensions is that we desire to carry out the top-order energy estimates with quantities whose initial characteristic data satisfies self-similar bounds without any fine-tuning. 

\subsection{Estimates for the Ricci Coefficients}
In this section we will briefly discuss the strategy for estimating the Ricci coefficients. Most importantly, see Section~\ref{curvexplained} above, we will need to verify that~\eqref{ricciassumptions} holds. In particular, if a Ricci coeffiicent $\psi$ vanishes along $\{v = 0\}$, then this needs to be remembered by the corresponding estimte, i.e., we expect to show estimates consistent with 
\begin{equation}\label{betterricci}
\left|\psi\right| \lesssim \frac{v}{u^2}.
\end{equation}

The most straightforward situation occurs when the Ricci coefficient $\psi$ we desire to estimate satisfies a $\nabla_4$ equation of the form
\[\nabla_4\psi = \psi\cdot\psi + \Psi + \nabla\psi.\]
One can then simply  integrates in the $v$-direction to show that the bounds from curvature are inherited by $\psi$. If $\psi$ vanishes when $\{v = 0\}$, then we will obtain an estimate consistent with~\eqref{betterricci}. (Of course the $\nabla\psi$ on the right hand side can produce a loss of an angular derivative in the estimate. In reality, at the top order, we need to couple the transport estimates with elliptic estimates.)

When $\psi$ satisfies a $\nabla_3$ equation instead, we have to be a bit more careful. These will be schematically of the form
\begin{equation}\label{aformofnabla32}
\nabla_3\psi + \frac{c}{u}\psi = \psi\cdot\psi + \Psi,
\end{equation}
In contrast to the case of the $\nabla_4$ equations, the second term will produce a logarithmic divergence if we treat it as an error and put it on the right hand side.

Inspired by the energy estimates from Section~\ref{curvexplained}, we conjugate~\eqref{aformofnabla32} by $u^c$ and obtain an equation schematically like 
\begin{equation}\label{aformofnabla322}
\nabla_3\left(u^c\psi\right) = u^c\psi\cdot\psi + u^c\Psi + \cdots.
\end{equation}

Integrating this equation, we can expect an estimate schematically of the form
\begin{align}\label{theeasyestimate2}
\int_{\mathcal{S}_{\hat u,\hat v}}&|\hat u|^{2c}\left|\psi\right|^2 
\\ \nonumber &\lesssim \left(\int_{-1}^{\hat u}\left(\int_{\mathcal{S}_{u,\hat{v}}}|u|^{2c}\left|\psi\right|^2\left|\psi\right|^2\right)^{1/2}\, du\right)^2 + \left(\int_{-1}^{\hat u}\left(\int_{\mathcal{S}_{\hat u,\hat v}}|u|^{2c}\left|\Psi\right|^2\right)^{1/2}\, du\right)^2  + \int_{\mathcal{S}_{-1,\hat v}}\left|\psi\right|^2.
\end{align}

Now there are at least three potential problems which can occur. 
\begin{enumerate}
	\item If, for example, $c = 2$ then in order to get a scale invariant estimate, we will need to divide everything by $|\hat u|^2$. However, if the Ricci coefficients $\psi$  which shows up in the quadratic term $\psi\cdot\psi$ only satisfies an estimate like $\left|\psi\right| \lesssim |u|^{-1}$, then, after the division by $|\hat{u}|^2$, on the right hand side we can at best hope to see
	\[|\hat{u}|^{-2}\left(\int_{-1}^{\hat{u}}\, dh\right)^2 \sim \left|\hat{u}\right|^{-2} \to \infty\text{ as }\hat{u}\to 0.\]
	Hence, depending on the value of $c$, it may be necessary that (one or both of) the Ricci coefficients which show up in the quadratic term are controlled by $\frac{v}{u^2}$ instead of $|u|^{-1}$. If we expect to in fact show that the left hand side vanishes as $v\to 0$, then we may need to have even stronger estimates for the quadratic term. 
	\item For the same reason as above, the initial data term $\int_{\mathcal{S}_{-1,\hat v}}\left|\psi\right|^2$ might need to decay as $v\to 0$. 
	\item Finally, the same problems can occur for the curvature component $\Psi$. There is even an additional difficulty however; if we expect that the left hand side to vanish as $v\to 0$, then we will need for $\Psi$ to vanish also. Unfortunately, the best that our energy estimate scheme is consistent with (see Section~\ref{curvexplained}) is that if $\Psi$ vanishes when $\{v = 0\}$, then $\left|\Psi\right| \lesssim \frac{v^{1/2}}{|u|^{5/2}}$. This is not strong enough to establish an estimate consistent with $\left|\psi\right| \lesssim \frac{v}{u^2}$! 
\end{enumerate}

It turns out that the first difficulty is resolved by the fact that whenever it is needed, the quadratic terms which show up do in fact have the desired extra decay as $v\to 0$. (That this happens can be anticipated by signature considerations.) Similarly, the initial data term turns out to always decay when it is needed to. However, the potential problem with $\Psi$ turns out to require a little more work. If $n > 4$, and $\Psi$ vanishes on $\{v = 0\}$ then our energy estimates will in fact control $\nabla_4\Psi$. Then we could hope to obtain the desired decay via the fundamental theorem of calculus. However, this cannot possibly work for $n = 2,3,4$ where we do not commute the energy estimates with $\nabla_4$.  Instead, if $\Psi \neq \alpha$ (we will not in fact need to estimate $\alpha$ in this way), we could use that $\Psi$ satisfies a $\nabla_4$ Bianchi equation and use the fundamental theorem of calculus in the $v$-direction. This gives the desired extra decay for $\Psi$ as the expense of losing an angular derivative; fortunately, the derivative may be recovered after a bit of work with elliptic estimates.

Finally, when $n > 4$ we also need estimates for an appropriate number of $\nabla_4$ derivatives applied to the Ricci coefficients. Fortunately, these follow by arguing as above with the commuted versions of the various equations. It is in fact strictly easier since we do not need to establish any improved vanishing behavior as $v\to 0$.

\subsection{Supercritical Estimates and Self-Similar Extraction}

Now we turn to a sketch of the proof of Theorem~\ref{theoextractselsimilar}. We start with the case $n = 2$.

Let's recall the scaling properties of our norms. Let $0 < \lambda \ll 1$ and, for concreteness, let's consider the component $\alpha$ and it's associated rescaling $\alpha_{\lambda}$ associated to $g_{\lambda}$. Let $\slashed{g}_{\lambda}$ denote the analogue of $\slashed{g}$ for $g_{\lambda}$. Using the explicit formulas from Section~\ref{scalingbehav}, we see that in the coordinate frame
\[\left(\alpha_{\lambda}\right)_{AB}\left(u,v,\theta\right) = \alpha_{AB}\left(\lambda u,\lambda v,\theta\right).\]

The $v$-energy flux we control for $\alpha$ is
\begin{equation}\label{scaledestalpha}
\sup_{u \in [-1,0)}\int_0^{-\epsilon u}\int_{\mathcal{S}}\slashed{g}^{AC}\slashed{g}^{BD}\reallywidetilde\alpha_{AB}\reallywidetilde\alpha_{CD}|u|^{4-2\delta}v^{-1+2\delta}\, dv.
\end{equation}
The scale-invariance of this estimate is reflected in the fact that a change of variables and the formulas from Section~\ref{scalingbehav} show that this is equal to
\[ \sup_{u \in [-\lambda^{-1},0)}\int_0^{-\epsilon u}\int_{\mathcal{S}}\left(\slashed{g}_{\lambda}\right)^{AC}\left(\slashed{g}_{\lambda}\right)^{BD}\left(\reallywidetilde\alpha_{\lambda}\right)_{AB}\left(\reallywidetilde\alpha_{\lambda}\right)_{CD}|u|^{4-2\delta}v^{-1+2\delta}\, dv.\] 
Restricting to $u \in [-1,0)$ we thus see that the scale-invariant bound for $\alpha$ we have proven for $g$ immediately implies that the same scale-invariant bounds hold for $g_{\lambda}$. More generally, we automatically get scale-invariant bounds for all the double null unknowns of $g_{\lambda}$. At this stage, it would be natural to try to use these uniform bounds along with a compactness argument to show that there exists a sequence $\lambda_i\to 0$ such that $g_{\lambda_i} \to g_{\rm self}$. However, such a compactness argument cannot yield that the limit is unique or even exists along a different choice of $\{\lambda_i\}$. 

Instead, we make the following observation. If in the estimate~\eqref{scaledestalpha} we replaced $|u|^{4-2\delta}$ with $|u|^{4-2\delta-2\kappa}$ for some $\kappa > 0$ and still managed to achieve a uniform bound, then we would obtain the following improved estimate for the rescaled quantity $\alpha_{\lambda}$:
\begin{align}\label{thisissupersuperalpha}
 \sup_{u \in [-\lambda^{-1},0)}&\int_0^{-\epsilon u}\int_{\mathcal{S}}\left(\slashed{g}_{\lambda}\right)^{AC}\left(\slashed{g}_{\lambda}\right)^{BD}\left(\reallywidetilde\alpha_{\lambda}\right)_{AB}\left(\reallywidetilde\alpha_{\lambda}\right)_{CD}|u|^{4-2\delta-2\kappa}v^{-1+2\delta}\, dv 
 \\ \nonumber &= \lambda^{2\kappa}\sup_{u \in [-1,0)}\int_0^{-\epsilon u}\int_{\mathcal{S}}\slashed{g}^{AC}\slashed{g}^{BD}\reallywidetilde\alpha_{AB}\reallywidetilde\alpha_{CD}|u|^{4-2\delta-2\kappa}v^{-1+2\delta}\, dv \to 0\text{ as }\lambda \to 0.
 \end{align}
Of course, the conclusion one draws from this is that we should not expect such a ``supercritical'' estimate for $\alpha$ to be obtainable.\footnote{We call this supercritical because when we scale towards the origin, where we expect the solution to be most singular, the estimate becomes better.} It is instructive to observe that the reason the proof of this improved estimate would break down is that the inhomogeneous terms (which are produced  by the initial data along $\{v=0\}$) on the right side of the Bianchi equations would produce errors during the energy estimates which would eventually not be integrable.

The key realization is the following. Suppose we instead try to derive the supercritical estimate for the differences $\alpha-\alpha_{\lambda}$, $\beta-\beta_{\lambda}$, etc.~First of all, after we derive equations for the differences of the Ricci coefficients and curvature components, the structure in the nonlinear terms which allowed us to prove our original estimates turns out to be preserved (cf.~the analysis of differences of the double null equations in~\cite{impulse1,impulse2}). Second of all, because the initial data along $\{v = 0\}$ is (mostly) scale invariant, the corresponding initial data (mostly) \emph{vanishes} for the differences of double null unknowns. In particular, the problematic inhomogeneous terms do not appear! This eventually allows us to establish a uniform bound on the analogue of~\eqref{thisissupersuperalpha} with $\alpha$ replaced by $\alpha_{\lambda} - \alpha$. By the same rescaling argument we thus conclude that $\{g_{\lambda}\}_{\lambda > 0}$ is Cauchy as $\lambda \to 0$.

It remains to argue that the limit $g_{\rm sim}$ is self-similar. Let $\left\vert\left\vert\cdot\right\vert\right\vert$ schematically denote the supercritical norm in which we control the differences $g-g_{\lambda}$. Let $s > 0$ and consider the rescaled metric $\left(g_{\rm sim}\right)_s$. For any $\lambda > 0$, rescaling yields the following:
\begin{align*}
\left\vert\left\vert g_{\rm sim}-\left(g_{\rm sim}\right)_s\right\vert\right\vert &\leq \left\vert\left\vert g_{\rm sim} - g_{\lambda s}\right\vert\right\vert + \left\vert\left\vert g_{\lambda s} - \left(g_{\rm sim}\right)_s\right\vert\right\vert
\\ \nonumber &\lesssim \left(\lambda s\right)^{\kappa} + s^{\kappa}\left\vert\left\vert g_{\lambda} - g_{\rm sim}\right\vert\right\vert
\\ \nonumber &\lesssim \left(\lambda s\right)^{\kappa}.
\end{align*}
Since $\lambda > 0$ is arbitrary we conclude that $g_{\rm sim} = \left(g_{\rm sim}\right)_s$. This finishes the sketch of Theorem~\ref{theoextractselsimilar}. Furthermore, the class of initial data we can allow turns out to also allow us to establish the existence part of Theorem~\ref{classifyself}.

To prove the uniqueness part of Theorem~\ref{classifyself}, a further analysis indicates that if two self-similar have the same values for $\slashed{g}|_{\mathcal{S}_{-1,0}}$ and ${\rm tf}\left(\mathcal{L}_v^{\frac{n}{2}}\slashed{g}\right)|_{\mathcal{S}_{-1,0}}$, then the difference will satisfy a supercritical estimate. A rescaling argument implies that the two solutions are in fact equal. 

Finally, we note that when $n\geq 3$ and odd, essentially the same proof works. When $n \geq 4$ and even there is a twist; if we lower the $u$-weight then there will be a problem closing the energy estimates for $\alpha$ (see Section~\ref{curvexplained}). Instead it turns to be possible to lower the $v$-weight slightly and prove supercritical estimates. The need to lower the $v$-weight explains the presence of the $\iota$'s in Definition~\ref{admissibleconjugatedata}.


\section{The equations of the double null gauge}\label{secdoublenull}

In this section we will present the equations of  the double null gauge in an arbitrary dimension. This extends the well-known treatments of the double null gauges in $3+1$ dimensions, i.e., when $n=2$, carried out in the works~\cite{KN,christ}. We emphasize that the calculations in this section do not rely on any topological assumptions for $\mathcal{S}$.
\subsection{The basic coordinate system}
We start with a metric $g$ 
in the double null gauge:
\[g = -2\Omega^2 \left(du\otimes dv + dv \otimes du\right)+ \slashed{g}_{AB}\left(d\theta^A - b^Adu\right)\otimes\left(d\theta^B - b^Bdu\right).\]

\underline{We do not assume at this point that $g$ satisfies the Einstein equations.}

The $\{\theta^A\}$ are local coordinates on a closed $n$-manifold $\mathcal{S}$ (not necessarily a sphere!) for $n \geq 2$.  We will refer to the specific copy of $\mathcal{S}$ sitting at the coordinates $(u,v) = (u_0,v_0)$ by $\mathcal{S}_{u_0,v_0}$. Unless said otherwise, \underline{in this section (and this section only)} the reader should assume that the metric $g$ is smooth. (Of course, later we will want to consider non-smooth metrics $g$. We discuss the necessary adjustments to the equations in Section~\ref{weak}.) Finally, we define the null vector fields
\[e_4 \doteq \Omega^{-1}\partial_v,\qquad e_3 \doteq \Omega^{-1}\left(\partial_u + b^A\partial_A\right).\]
These satisfy
\[g\left(e_3,e_4\right)  = -2.\]

\subsection{The Ricci coefficients and the null curvature components}\label{riccicoeff}Let $D$ denote the Levi-Civita connection associated to $g$. The Ricci coefficients are the following quantities:
\[\chi_{AB} \doteq g\left(D_Ae_4,e_B\right),\qquad \underline{\chi}_{AB} = g\left(D_Ae_3,e_B\right),\]
\[\eta_A \doteq -\frac{1}{2}g\left(D_3e_A,e_4\right),\qquad \underline{\eta}_A \doteq -\frac{1}{2}g\left(D_4e_A,e_3\right),\]
\[\omega \doteq -\frac{1}{4}g\left(D_4e_3,e_4\right),\qquad \underline{\omega} \doteq -\frac{1}{4}g\left(D_3e_4,e_3\right),\]
\[\zeta_A \doteq \frac{1}{2}g\left(D_Ae_4,e_3\right).\]

The $1$-form $\zeta_A$ is often referred to as ``torsion''. Also, many times we will split $\chi$ and $\underline{\chi}$ into their trace and trace-free parts:
\[\chi_{AB} \doteq \hat{\chi}_{AB} + \frac{1}{n}{\rm tr}\chi\slashed{g}_{AB},\qquad \underline{\chi}_{AB} \doteq \hat{\underline{\chi}}_{AB} + \frac{1}{n}{\rm tr}\underline{\chi}\slashed{g}_{AB}.\]
We will often use $\psi$ to stand for a generic Ricci coefficient.

All of these quantities are $\mathcal{S}_{u,v}$ tensors (see~\cite{KN,christ} for the precise definitions when $n=2$; the generalization to higher dimensions is immediate). We will denote the induced connection on $\mathcal{S}$ by $\nabla_A$ and the projection of $D_3$ and $D_4$ to $\mathcal{S}$ by $\nabla_3$ and $\nabla_4$. Observe that all of these definitions are exactly the same as the case of $n=2$.

We use the curvature convention
\[\left[D_i,D_j\right]\phi_k = R_{ijk}^{\ \ \ l}\phi_l.\]

The null curvature components are defined as follows:
\[\alpha_{AB} \doteq R\left(e_A,e_4,e_B,e_4\right),\qquad \underline{\alpha}_{AB} \doteq R\left(e_A,e_3,e_B,e_3\right),\]
\[\beta_A \doteq \frac{1}{2}R\left(e_A,e_4,e_3,e_4\right),\qquad \underline{\beta}_A \doteq \frac{1}{2}R\left(e_A,e_3,e_3,e_4\right),\]
\[\rho \doteq \frac{1}{4}R\left(e_4,e_3,e_4,e_3\right),\qquad \sigma_{AB} \doteq \frac{1}{2}\left(R\left(e_3,e_A,e_4,e_B\right) - R\left(e_3,e_B,e_4,e_A\right)\right),\]
\[\tau_{AB} \doteq \frac{1}{2}\left(R\left(e_3,e_A,e_4,e_B\right) + R\left(e_3,e_B,e_4,e_A\right)\right),\]
\[\nu_{ABC} = R\left(e_A,e_B,e_C,e_4\right),\qquad \underline{\nu}_{ABC} = R\left(e_A,e_B,e_C,e_3\right).\]
We will often use $\Psi$ to stand for a generic curvature component (a null curvature component or $R_{ABCD}$).

We also have the induced curvature tensor on $\mathcal{S}$
\[\slashed{Riem}_{ABCD}.\]

It will be convenient sometimes to use $\hat{\tau}_{AB}$ to denote the trace-free part of $\tau_{AB}$:
\[\hat{\tau}_{AB} \doteq \tau_{AB} - \frac{1}{n}\slashed{g}_{AB}{\rm tr}\tau.\]
Similarly, we will use $\hat{\alpha}$ and $\underline{\hat{\alpha}}$ to denote the trace-free parts of $\alpha$ and $\underline{\alpha}$ respectively.

In contrast to the $n=2$ case we have the additional null curvature components $\tau_{AB}$, $\nu_{ABC}$, and $\underline{\nu}_{ABC}$. Furthermore, the curvature component $\sigma$ from the $n=2$ case must now be considered a $2$-form $\sigma_{AB}$. After we have listed the null structure, constraint, and Bianchi equations it will become clear why these additional components become necessary in higher dimensions.

The Einstein equations imply certain algebraic relations between the curvature components. We record the most important of these below. 

\begin{lemma}\label{curvids}We have
\[{\rm tr}\alpha = {\rm Ric}_{44},\qquad {\rm tr}\underline{\alpha} = {\rm Ric}_{33},\qquad {\rm tr}\tau = {\rm Ric}_{34}-2\rho,\qquad \tau_{AB} = \slashed{g}^{CD}R_{CADB}-{\rm Ric}_{AB},\]
\[\nu_{AB}^{\ \ \ B} = \beta_A-{\rm Ric}_{A4},\qquad \underline{\nu}_{AB}^{\ \ \ B} = -\underline{\beta}_A -{\rm Ric}_{A3},\]
\begin{equation}\label{nuiden}
\nu_{(ABC)} = 0,\qquad \nu_{A[BC]} = \frac{1}{2}\nu_{CBA}, \qquad \nu_{ABC} = \frac{4}{3}\nu_{A(BC)} + \frac{2}{3}\nu_{C(BA)}.
\end{equation}
\begin{equation*}
\underline\nu_{(ABC)} = 0,\qquad \underline\nu_{A[BC]} = \frac{1}{2}\underline\nu_{CBA}, \qquad \underline\nu_{ABC} = \frac{4}{3}\underline\nu_{A(BC)} + \frac{2}{3}\underline\nu_{C(BA)}.
\end{equation*}
\end{lemma}
\begin{proof}The first five identities are immediate consequences of the definition of Ricci curvature.

The first equality in~\eqref{nuiden} follows from the first Bianchi identity. The second equality in~\eqref{nuiden} follows from the first as follows:
\[\nu_{ABC} = -\nu_{BCA} - \nu_{CAB} = \nu_{CBA} + \nu_{ACB}.\]

The final equality in~\eqref{nuiden} can then easily be derived:
\begin{align*}
\nu_{ABC} &= \nu_{A(BC)} + \nu_{A[BC]}
\\ \nonumber &= \nu_{A(BC)} + \frac{1}{2}\nu_{CBA}
\\ \nonumber &= \nu_{A(BC)} + \frac{1}{2}\nu_{C(BA)} + \frac{1}{2}\nu_{C[BA]}
\\ \nonumber &= \nu_{A(BC)} + \frac{1}{2}\nu_{C(BA)} + \frac{1}{4}\nu_{ABC} \Rightarrow
\\ \nonumber  \nu_{ABC} &= \frac{4}{3}\nu_{A(BC)} + \frac{2}{3}\nu_{C(BA)}.
\end{align*}

Clearly, the same arguments work for $\underline\nu$.
\end{proof}

Next, we record the identities which link $D$ to the projected $S_{u,v}$ derivatives $\nabla_3$ and $\nabla_4$ and to the induced covariant derivative $\nabla$ on $\mathcal{S}$. (Again we refer the reader to~\cite{KN,christ} for precise definitions.)
\begin{lemma}\label{D}
\[D_4e_4 = - 2\omega e_4,\qquad D_4e_3 = 2\omega e_3 + 2\underline{\eta}^Ae_A,\qquad D_4e_A = \underline{\eta}_Ae_4 + \nabla_4e_A, \]
\[D_3e_4 = 2\underline{\omega}e_4 + 2\eta^Ae_A,\qquad D_3e_3 = -2\underline{\omega}e_3,\qquad D_3e_A = \eta_Ae_3 + \nabla_3e_A,\]
\[D_Ae_4 = -\zeta_A e_4 + \chi_A^{\ B}e_B,\qquad D_Ae_3 = \zeta_A e_3 + \underline{\chi}_A^{\ B}e_B,\qquad D_Ae_B = \frac{1}{2}\underline{\chi}_{AB}e_4 + \frac{1}{2}\chi_{AB}e_3 + \nabla_Ae_B.\]
\end{lemma}
\begin{proof}This is a straightforward calculation.
\end{proof}

\subsection{Signature}\label{signaturesection}
An important role in our analysis within the double null gauge will be played by ``signature'' considerations. For any Ricci coefficient and null curvature component $\phi$ we define the signature of $\phi$ by
\[s\left(\phi\right) \doteq 1\cdot N_3\left(\phi\right) + \frac{1}{2}\cdot N_A\left(\phi\right) + 0\cdot N_4\left(\phi\right) - 1.\]

Here $N_3$ denotes the number of $3$'s used in the definition of $\phi$, $N_A$ denotes the number of angular indices used in the definition, and $N_4$ denotes the number of $e_4$'s used. For concreteness we list the signature of the various Ricci coefficients:
\[s\left(\chi_{AB}\right) = 0,\qquad s\left(\omega\right) = 0,\]
\[s\left(\eta_A\right) = \frac{1}{2},\qquad s\left(\underline{\eta}_A\right) = \frac{1}{2},\qquad s\left(\zeta_A\right) = \frac{1}{2},\]
\[s\left(\underline{\chi}_{AB}\right) = 1,\qquad s\left(\underline{\omega}\right) = 1,\]
and then the signatures of the curvature components:
\[s\left(\alpha_{AB}\right) = 0,\]
\[s\left(\beta_A\right) = \frac{1}{2},\qquad s\left(\nu_{ABC}\right) = \frac{1}{2},\]
\[s\left(\rho\right) =1,\qquad s\left(\sigma_{AB}\right) = 1,\qquad s\left(\tau_{AB}\right) = 1,\qquad s\left(R_{ABCD}\right) = 1,\]
\[s\left(\underline{\beta}_A\right) = \frac{3}{2},\qquad s\left(\underline{\nu}_{ABC}\right) = \frac{3}{2},\]
\[s\left(\underline{\alpha}\right) = 2.\]

We also have the rules that
\[s\left(e_4\right) = -1,\qquad s\left(e_A\right) = -\frac{1}{2},\qquad s\left(e_3\right) = 0,\]
\[s\left(\nabla_3\phi\right) = s\left(D_3\phi\right) =  1 + s\left(\phi\right),\qquad s\left(D_A\phi\right) = s\left(\nabla_A\phi\right) = 1/2 + s\left(\phi\right) ,\]
\[s\left(\nabla_4\phi\right) = s\left(D_4\phi\right) \doteq s\left(\phi\right),\qquad s\left(\phi_1\phi_2\right) = s\left(\phi_1\right) + s\left(\phi_2\right).\]

This particular notion of signature was originally introduced in~\cite{KR} where it was used in the study the problem of trapped surface formation.

By direct inspection of  Lemma~\ref{D}, we immediately obtain the following simple but fundamental lemma:
\begin{lemma}\label{preservesig}Covariant differentiation preserves signature.
\end{lemma}

\subsection{Metric Equations}
In this section we will present the equations for the metric quantities. These equations relate derivatives of the metric components to Ricci coefficients.
\begin{proposition}\label{metriceqn}We have
\[\mathcal{L}_4\slashed{g}_{AB} = 2\chi_{AB},\qquad \mathcal{L}_3\slashed{g}_{AB} = 2\underline{\chi}_{AB},\]
\[\omega = -\frac{1}{2}\nabla_4\left(\log\Omega\right),\qquad \underline{\omega} = -\frac{1}{2}\nabla_3\left(\log\Omega\right),\]
\[\zeta_A = \frac{1}{4}g\left(\left[e_3,e_4\right],e_A\right) = -\frac{1}{4}\Omega^{-1}e_4\left(b^B\right)\slashed{g}_{AB},\]
\[\eta_A = \zeta_A + \nabla_A\left(\log\Omega\right),\qquad \underline{\eta}_A = -\zeta_A + \nabla_A\left(\log\Omega\right).\]
\end{proposition}
\begin{proof}These follow exactly as in the case of $n=2$.
\end{proof}

\subsection{Null structure equations}
In this section we will present the null-structure equations. These equations relate the $\nabla_3$ and $\nabla_4$ derivatives of certain Ricci coefficients to a null curvature component plus a sum of angular derivatives of Ricci coefficients and quadratic combinations of Ricci coefficients, and possibly a component of Ricci curvature.

Before presenting the equations we recall the trace-free symmetrized product $\hat\otimes$ of two $\mathcal{S}_{u,v}$ $1$-forms $\psi_A$ and $\phi_B$:
\[\left(\psi{\hat\otimes}\phi\right)_{AB} \doteq \psi_A\phi_B + \psi_B\phi_A - \frac{2}{n}\slashed{g}_{AB}\psi^C\phi_C.\]
\begin{proposition}\label{nullstruct}We have
\begin{align*}
\nabla_4{\rm tr}\chi + \frac{1}{n}\left({\rm tr}\chi\right)^2 &=-{\rm Ric}_{44} -\left|\hat{\chi}\right|^2 - 2\omega{\rm tr}\chi,
\\ \nonumber \nabla_4\hat{\chi}_{AB}+\frac{2}{n}{\rm tr}\chi \hat{\chi}_{AB} &= -\hat{\alpha}_{AB} -2\omega\hat{\chi}_{AB} +\left(\frac{\slashed{g}_{AB}}{n}\left|\hat{\chi}\right|^2 - \hat{\chi}_{(A}^{\ \ \ C}\hat{\chi}_{B)C}\right)
\\ \nonumber \nabla_3{\rm tr}\underline{\chi} + \frac{1}{n}\left({\rm tr}\underline{\chi}\right)^2 &=-{\rm Ric}_{33} -\left|\hat{\underline{\chi}}\right|^2 - 2\underline{\omega}{\rm tr}\underline{\chi},
\\ \nonumber \nabla_3\underline{\hat{\chi}}_{AB}+\frac{2}{n}{\rm tr}\underline{\chi} \underline{\hat{\chi}}_{AB} &= -\underline{\hat{\alpha}}_{AB} -2\underline{\omega}\underline{\hat{\chi}}_{AB} +\left(\frac{\slashed{g}_{AB}}{n}\left|\underline{\hat{\chi}}\right|^2 - \underline{\hat{\chi}}_{(A}^{\ \ \ C}\underline{\hat{\chi}}_{B)C}\right),
\\ \nonumber \nabla_3\hat{\chi}_{AB} +\frac{1}{n}{\rm tr}\underline{\chi}\hat{\chi}_{AB}&= -\hat{\tau}_{AB} + 2\underline{\omega}\hat{\chi}_{AB} + \left(\nabla\hat\otimes \eta\right)_{AB} + \left(\eta\hat\otimes \eta\right)_{AB} - \frac{1}{n}{\rm tr}\chi \hat{\underline{\chi}}_{AB}  +
\\ \nonumber &\qquad -\left(\hat{\underline{\chi}}_{(A}^{\ \ \ C}\hat{\chi}_{B)C} - \frac{1}{n}\hat{\underline{\chi}}\cdot\hat{\chi}\slashed{g}_{AB}\right),
\\ \nonumber \nabla_3{\rm tr}\chi + \frac{1}{n}{\rm tr}\chi{\rm tr}\underline{\chi} &= 2\rho -{\rm Ric}_{34}+ 2\underline{\omega}{\rm tr}\chi + 2{\rm div}\eta + 2\left|\eta\right|^2 - \hat{\chi}\cdot\hat{\underline{\chi}},
\\ \nonumber \nabla_4\hat{\underline{\chi}}_{AB} + \frac{1}{n}{\rm tr}\chi \hat{\underline{\chi}}_{AB}&= -\hat{\tau}_{AB} + 2\omega\hat{\underline{\chi}}_{AB} + \left(\nabla\hat\otimes \underline{\eta}\right)_{AB} + \left(\underline{\eta}\hat\otimes \underline{\eta}\right)_{AB}  - \frac{1}{n}{\rm tr}\underline{\chi}\hat{\chi}_{AB} +
\\ \nonumber &\qquad -\left(\hat{\underline{\chi}}_{(A}^{\ \ \ C}\hat{\chi}_{B)C} - \frac{1}{n}\hat{\underline{\chi}}\cdot\hat{\chi}\slashed{g}_{AB}\right)
\\ \nonumber \nabla_4{\rm tr}\underline{\chi} + \frac{1}{n}{\rm tr}\chi{\rm tr}\underline{\chi} &= 2\rho -{\rm Ric}_{34}+ 2\omega{\rm tr}\underline{\chi} + 2{\rm div}\underline{\eta} + 2\left|\underline\eta\right|^2 - \hat{\chi}\cdot\hat{\underline{\chi}},
\\ \nonumber \nabla_4\eta &= -\chi\cdot\left(\eta-\underline{\eta}\right) - \beta,
\\ \nonumber \nabla_3\underline{\eta} &= -\underline{\chi}\cdot\left(\underline\eta-\eta\right) + \underline{\beta},
\\ \nonumber \nabla_4\underline{\omega} &= \frac{1}{2}\rho + \frac{1}{4}\left|\underline\eta\right|^2 - \frac{1}{4}\left|\eta\right|^2 + 2\underline\omega \omega + 3\left|\zeta\right|^2 - \left|\nabla \log\Omega\right|^2,
\\ \nonumber \nabla_3\omega &= \frac{1}{2}\rho + \frac{1}{4}\left|\eta\right|^2 - \frac{1}{4}\left|\underline\eta\right|^2 + 2\underline\omega \omega + 3\left|\zeta\right|^2 - \left|\nabla \log\Omega\right|^2,
\end{align*}
\end{proposition}
\begin{proof} The derivation of these formulas in all dimensions is completely analogous to the $n=2$ case discussed in~\cite{KN}. The $1/n$'s that appear in the equation arise when a tensor is decomposed into its trace and trace-free part. 
\end{proof}

\begin{remark}Note that these equations are almost identical to the $n=2$ case. One key difference is that the terms in the parenthesis on the right hand sides of the $\nabla_4\hat{\chi}_{AB}$, $\nabla_3\hat{\underline{\chi}}_{AB}$,  $\nabla_3\hat{\chi}_{AB}$, and $\nabla_3\hat{\chi}_{AB}$ equations vanish identically when $n=2$. Furthermore, if we additionally assume that ${\rm Ric}(g) = 0$, then, as we will see below, $\hat{\tau}$ vanishes when $n=2$. Thus, when $n=2$ and ${\rm Ric}(g) = 0$, the equations for $\nabla_3\hat{\chi}_{AB}$ and $\nabla_4\hat{\underline{\chi}}_{AB}$ do not have any curvature components on the right hand side.
\end{remark}
\subsection{Constraint Equations}
The double null gauge induces various equations intrinsic to the manifolds $\mathcal{S}_{u,v}$. We will derive these equations in this section. 

First we have a definition:
\begin{definition}We will use a slash to denote curvature quantities associated with the manifolds $\mathcal{S}_{u,v}$, that is, we denote the full curvature tensor, the Ricci tensor, and scalar curvature of the $\mathcal{S}_{u,v}$'s by
\[\slashed{Riem}_{ABCD},\qquad \slashed{Ric}_{AB},\qquad \slashed{R},\]
respectively.
\end{definition}

The following collection of constraint equations are natural genearalizations from the $n=2$ case.
\begin{proposition}\label{constrainteqns}We have
\begin{align}
\label{itsgauss}\slashed{Riem}_{ABCD} &= R_{ABCD} + \frac{1}{2}\left(\underline{\chi}_{BC}\chi_{AD} + \chi_{BC}\underline{\chi}_{AD} - \underline{\chi}_{AC}\chi_{BD} - \chi_{AC}\underline{\chi}_{BD}\right).
\\ \label{slashric} \slashed{Ric}_{AB} &= \tau_{AB} + {\rm Ric}_{AB} - \frac{1}{2}{\rm tr}\chi\underline{\chi}_{AB} - \frac{1}{2}{\rm tr}\underline{\chi}\chi_{AB} + \chi^C_{\ (A}\underline{\chi}_{B)C},
\\ \label{slashr} \slashed{R} &= -2\rho + R +2 {\rm Ric}_{34}+\frac{1-n}{n} {\rm tr}\chi{\rm tr}\underline{\chi} + \hat{\chi}\cdot\hat{\underline{\chi}},
\\ \label{cod1} \nabla_A\chi_{BC} - \nabla_B\chi_{AC} &= \nu_{ABC} + \chi_{AC}\zeta_B - \chi_{BC}\zeta_A,
\\ \label{cod2} \nabla_A\underline{\chi}_{BC} - \nabla_B\underline{\chi}_{AC} &= \underline{\nu}_{ABC} - \underline{\chi}_{AC}\zeta_B + \underline{\chi}_{BC}\zeta_A,
\\ \label{tcod1} \nabla^A\chi_{AB}-\nabla_B{\rm tr}\chi &=  -\beta_B+ {\rm Ric}_{4B}+ {\rm tr}\chi \zeta_B - \zeta^A\chi_{AB},
\\ \label{tcod2} \nabla^A\underline\chi_{AB}-\nabla_B{\rm tr}\underline\chi  &= \underline\beta_B + {\rm Ric}_{3B}- {\rm tr}\underline\chi \zeta_B + \zeta^A\underline\chi_{AB},
\\ \label{antisig} \nabla_A\eta_B - \nabla_B\eta_A &= \sigma_{AB} + \frac{1}{2}\left(\underline{\hat\chi}_A^{\ C}\hat\chi_{CB} - \underline{\hat\chi}_B^{\ C}\hat\chi_{CA}\right),
\\ \label{antisig2} \nabla_A\underline\eta_B - \nabla_B\underline\eta_A &=-\sigma_{AB}- \frac{1}{2}\left(\underline{\hat\chi}_A^{\ C}\hat\chi_{CB} - \underline{\hat\chi}_B^{\ C}\hat\chi_{CA}\right).
\end{align}
\end{proposition}
\begin{proof}
The Gauss equation reads
\begin{equation*}
\slashed{Riem}_{ABCD} = R_{ABCD} + \frac{1}{2}\left(\underline{\chi}_{BC}\chi_{AD} + \chi_{BC}\underline{\chi}_{AD} - \underline{\chi}_{AC}\chi_{BD} - \chi_{AC}\underline{\chi}_{BD}\right).
\end{equation*}

Tracing this once yields~\eqref{slashric} and tracing this twice yields~\eqref{slashr}.

Equations~\eqref{cod1} and~\eqref{cod2} are simply the Codazzi equations associated to $\chi$ and $\underline{\chi}$. Tracing then yields~\eqref{tcod1} and~\eqref{tcod2}.

A straightforward computation yields (just as in the $n=2$ case) that 
\begin{align}\label{3chi}
\nabla_3\chi_{AB} &= -\tau_{AB} -\sigma_{AB} + 2\underline{\omega}\chi_{AB} + 2\nabla_A\eta_B + 2\eta_A\eta_B - \underline{\chi}_A^{\ C}\chi_{CB}.
\end{align}

Taking the anti-symmetric part of~\eqref{3chi} yields~\eqref{antisig}. Similarly, we obtain~\eqref{antisig2}.
\end{proof}

\begin{remark}\label{allofEinstein}It is now manifest that if we set all of the Ricci curvature terms to vanish in the null structure and constraint equations, then we have recovered the entire content of the Einstein equations.
\end{remark}

As is well known, despite Remark~\ref{allofEinstein}, for carrying out a priori estimates of solutions to the Einstein equations, it is very useful to work with various additional equations, such as the Bianchi system. With this in mind, from this point on, we will consider metrics $g$ which satisfy ${\rm Ric}(g) = 0$.

First, we take a moment to observe that when $n=2$ there are various simplifications.
\begin{lemma}\label{simpn2}Suppose that ${\rm Ric}(g) = 0$ and $n=2$. Then we have
\[\hat{\tau}_{AB} = 0,\]
and if we let $\slashed{\epsilon}_{AB}$ denote a (locally defined) volume form on $\mathcal{S}$, we have
\[\nu_{ABC} = \slashed{\epsilon}_{AB}\tilde\nu_C,\qquad \underline\nu_{ABC} = \slashed{\epsilon}_{AB}\underline{\tilde\nu}_C,\]
for $1$-forms $\tilde\nu_C$ and $\underline{\tilde\nu}_C$ which must satisfy
\[\beta_A = *\tilde\nu_A,\qquad \underline{\beta}_A = -*\underline{\tilde \nu}_A,\]
where $*$ denotes the Hodge star operator. 

Furthermore, if $\sigma$ is defined in the usual way by $(1/4)\left(*R\right)_{3434}$, then we have
\[\sigma_{AB} = \sigma \epsilon_{AB}.\]
\end{lemma}
\begin{proof}This is straightforward.
\end{proof}

When $n > 2$ there are various ``extra'' constraint equations that may be derived. The following two are the most important:
\begin{proposition}\label{extracons}Suppose that $g$ satisfies ${\rm Ric}(g) = 0$. Then we have 
\begin{align*}
&\nabla^B\tau_{AB} + \nabla_A\rho - \frac{1}{2}\nu_{ABC}\underline{\chi}^{BC}  - \frac{1}{2}\underline{\nu}_{ABC}\chi^{BC}
\\ \nonumber & - \frac{1}{2}\underline{\chi}_A^{\ C}\beta_C + \frac{1}{2}{\rm tr}\chi\zeta_C\underline{\chi}_A^{\ C} +\frac{1}{2}\chi_A^{\ C}\underline{\beta}_C - \frac{1}{2}\zeta_C\chi_A^{\ C}{\rm tr}\underline{\chi} \\ \nonumber & + \frac{1}{2}{\rm tr}\underline{\chi}\beta_A + \frac{1}{2}{\rm tr}\underline{\chi}\chi_{AB}\zeta^B
-\frac{1}{2}{\rm tr}\chi\underline{\beta}_A  - \frac{1}{2}{\rm tr}\chi\zeta^B\underline{\chi}_{BA} = 0.
\end{align*}

\end{proposition}
\begin{proof}This is a straightforward if tedious consequence of the identity
\[\slashed{\nabla}^B\left(\slashed{Ric}_{AB} - \frac{1}{2}\slashed{g}_{AB}\slashed{R}\right) = 0,\]
and the constraint equations~\eqref{slashric}-\eqref{tcod2}.

\end{proof}

\begin{proposition}\label{extraconsthebest}Suppose that $g$ satisfies ${\rm Ric}(g) = 0$. Then we have 
\[\nabla^AR_{ABCD} = 2\nabla_{[C}\tau_{D]B} + \chi\cdot\underline{\nu} + \underline{\chi}\cdot\nu + \zeta\cdot\chi\cdot\underline{\chi},\]
where the final three terms are schematics for all possible contractions of the written terms. 
\end{proposition}
\begin{proof}We have 
\[R_{ABCD} = \slashed{Riem}_{ABCD} - \frac{1}{2}\left(\underline{\chi}_{BC}\chi_{AD} + \chi_{BC}\underline{\chi}_{AD} - \underline{\chi}_{AC}\chi_{BD} - \chi_{AC}\underline{\chi}_{BD}\right),\]
\[\slashed{Ric}_{AB} = \tau_{AB} - \frac{1}{2}{\rm tr}\chi\underline{\chi}_{AB} - \frac{1}{2}{\rm tr}\underline{\chi}\chi_{AB} + \chi^C_{\ (A}\underline{\chi}_{B)C}.\]

Thus,
\begin{align*}
\nabla^AR_{ABCD} &= \nabla^A\slashed{Riem}_{ABCD} - \frac{1}{2}\left(\chi_{AD}\nabla^A\underline{\chi}_{BC} + \underline{\chi}_{AD}\nabla^A\chi_{BC} - \underline{\chi}_{AC}\nabla^A\chi_{BD} - \chi_{AC}\nabla^A\underline{\chi}_{BD}\right)
\\ \nonumber &\qquad  - \frac{1}{2}\left(\underline{\chi}_{BC}\nabla^A\chi_{AD} + \chi_{BC}\nabla^A\underline{\chi}_{AD} - \nabla^A\underline{\chi}_{AC}\chi_{BD} - \nabla^A\chi_{AC}\underline{\chi}_{BD}\right).
\end{align*}

Recall that
\[\nabla_A\slashed{Riem}_{BCDE} + \nabla_B\slashed{Riem}_{CADE} + \nabla_C\slashed{Riem}_{ABDE} = 0.\]
Tracing on $B$ and $D$ yields 
\[\nabla_A\slashed{Ric}_{CE} + \nabla^B\slashed{Riem}_{CABE} - \nabla_C\slashed{Ric}_{AE} = 0 \Rightarrow\]
\begin{align*}
\nabla^A\slashed{Riem}_{ABCD} &= \nabla_C\slashed{Ric}_{DB} - \nabla_D\slashed{Ric}_{CB}
\\ \nonumber &= \nabla_C\tau_{DB} -\frac{1}{2}\left(\nabla_C{\rm tr}\chi\right)\underline{\chi}_{DB} - \frac{1}{2}\left(\nabla_C{\rm tr}\underline{\chi}\right)\chi_{DB}  - \frac{1}{2}{\rm tr}\chi \nabla_C\underline{\chi}_{DB} - \frac{1}{2}{\rm tr}\underline{\chi}\nabla_C\chi_{DB} + \nabla_C\left(\chi^A_{\ (D}\underline{\chi}_{B)A} \right)
\\ \nonumber &\qquad -\nabla_D\tau_{CB} +\frac{1}{2}\left(\nabla_D{\rm tr}\chi\right)\underline{\chi}_{CB} + \frac{1}{2}\left(\nabla_D{\rm tr}\underline{\chi}\right)\chi_{CB} + \frac{1}{2}{\rm tr}\chi \nabla_D\underline{\chi}_{CB} + \frac{1}{2}{\rm tr}\underline{\chi}\nabla_D\chi_{CB} - \nabla_D\left(\chi^A_{\ (C}\underline{\chi}_{B)A} \right).
\end{align*}
Combining the two identities leads to
\begin{align*}
\nabla^AR_{ABCD} &= \nabla_C\tau_{DB} - \nabla_D\tau_{CB}
\\ \nonumber &\qquad - \frac{1}{2}\left(\underline{\chi}_{BC}\left(\nabla^A\chi_{AD} -\nabla_D{\rm tr}\chi\right) + \chi_{BC}\left(\nabla^A\underline{\chi}_{AD} -\nabla_D{\rm tr}\underline{\chi}\right)\right)
\\ \nonumber &\qquad + \frac{1}{2}\left(\chi_{BD}\left(\nabla^A\underline{\chi}_{AC} -\nabla_C{\rm tr}\underline{\chi}\right) + \underline{\chi}_{BD}\left(\nabla^A\chi_{AC} - \nabla_C{\rm tr}\chi\right)\right)
\\ \nonumber &\qquad - \frac{1}{2}\left({\rm tr}\chi\left(\nabla_D\underline{\chi}_{CB} - \nabla_C\underline{\chi}_{DB}\right) + {\rm tr}\underline{\chi}\left(\nabla_D\chi_{CD} - \nabla_C\chi_{DB}\right)\right)
\\ \nonumber &\qquad -\frac{1}{2}\left(\chi_{AD}\left(\nabla^A\underline{\chi}_{BC} - \nabla^C\underline{\chi}_{AB}\right) + \underline{\chi}_{AD}\left(\nabla^A\chi_{BC} - \nabla^C\chi_{AB}\right)\right)
\\ \nonumber &\qquad + \frac{1}{2}\left(\chi_{AC}\left(\nabla^A\underline{\chi}_{BD} - \nabla^D\underline{\chi}_{AB}\right) + \underline{\chi}_{AC}\left(\nabla^A\chi_{BD} - \nabla^D\chi_{AB}\right)\right)
\\ \nonumber &\qquad - \frac{1}{2}\left(\underline{\chi}_{AB}\left(\nabla_D\chi^A_{\ \ C} - \nabla_C\chi^A_{\ \ D}\right) + \chi^A_{\ \ B}\left(\nabla_D\underline{\chi}_{AC} - \nabla_C\underline{\chi}_{DA}\right)\right).
\end{align*}
The proof is then finished by appealing to the Codazzi equations.

\end{proof}

\begin{remark}
One can also derive equations which link either $\nabla^A\nu_{A[BC]}$ or $\nabla^C\nu_{ABC}$ to $\nabla_{[A}\beta_{B]}$ or $\nabla^A\underline{\nu}_{A[BC]}$ and $\nabla^C\underline{\nu}_{ABC}$ to $\nabla_{[A}\underline{\beta}_{B]}$ along with lower order terms, but, we will not directly need these equations for this paper. 
\end{remark}

\subsection{Commutation}
The following commutator estimates will play a fundamental role in our analysis.
\begin{lemma}\label{4commute}Let $\phi_{A_1\cdots A_r}$ be a $(0,r)$ $\mathcal{S}_{u,v}$ tensor. Then, for every $m \geq 0$, we have that

\[\left|\nabla_4\nabla^m\phi - \nabla^m\nabla_4\phi\right|  \lesssim \sum_{i+j+k = m-1}\left|\nabla^i\psi^{j+1}\right|\left|\nabla^k\nabla_4\phi\right| + \sum_{i+j+k = m}\left|\nabla^i\psi^{j+1}\right|\left|\nabla^k\phi\right|,\]
where $\nabla^i\psi_j$ and $\nabla^i\psi^{j+1}$ are schematic notations for all possible ways of distributing $i$ angular derivatives over a product of $j$ or $j+1$ Ricci coeffiients.

\end{lemma}
\begin{proof}

The formula is well known in the case $n=2$, e.g., see~\cite{lukchar,taylor}. The proof in the general case $n \geq 2$ is completely analogous.
\end{proof}

Next, we have
\begin{lemma}\label{4Acommute}Let $\phi_{A_1\cdots A_r}$ be a $(0,r)$ $\mathcal{S}_{u,v}$ tensor. Then, for every $m \geq 0$, we have that
\begin{align*}
\left|\nabla_4^m\nabla_A\phi - \nabla_A\nabla_4^m\phi\right| &\lesssim \sum_{i+j +k= m}\left|\nabla_4^i\psi^{j+1}\right|\left|\nabla_4^k\phi\right| + \sum_{i+j+k = m-1}\left|\nabla\nabla_4^i\psi^{j+1}\right|\left|\nabla_4^k\phi\right| 
\\ \nonumber &\qquad + \sum_{i+j+k=m-1}\left|\nabla_4^i\psi^{j+1}\right|\left|\nabla\nabla_4^k\phi\right|.
\end{align*}

\end{lemma}
\begin{proof}This is proven in an analogous fashion to Lemma~\ref{4commute}.
\end{proof}

For $\nabla_3$ equations it will be useful to track ${\rm tr}\underline{\chi}$ a little more explicitly. 
\begin{lemma}\label{3commute}Let $\phi_{A_1\cdots A_r}$ be a $(0,r)$ $\mathcal{S}_{u,v}$ tensor. Then, for every $m \geq 0$, we have that
\[\nabla_3\phi + \frac{c}{n}{\rm tr}\underline{\chi} \phi= F \Rightarrow \]
\begin{align*}
&\left|\nabla_3\nabla^m\phi + \frac{c+m}{n}{\rm tr}\underline{\chi}\nabla^m\phi - \nabla^mF\right| \lesssim 
\\ \nonumber &\ \ \sum_{i+j+k = m,\ k \neq m}\left|\nabla^i\left(\eta+\underline{\eta}\right)^j\right|\left|\nabla^kF\right| + \sum_{i+j+(k_1,k_2)+l=m,\ (k_2,l) \neq (0,m)}\left|\nabla^i\psi^j\right|\left[\left|\nabla^{k_1}\hat{\underline{\chi}}\right|+ \left|\nabla^{k_2}{\rm tr}\underline{\chi}\right|\right]\left|\nabla^l\phi\right|.
\end{align*}

We are using the same schematic notation as in Lemma~\ref{4commute}.
\end{lemma}
\begin{proof}As with Lemma~\ref{4commute}, it is well-known that such an estimate holds in the $n=2$ case and the generalization to $n \geq 2$ is straightforward.
\end{proof}

Finally, we will also need to commute with $\nabla_4$.
\begin{lemma}\label{34commute}
Let $\phi_{A_1\cdots A_r}$ be a $(0,r)$ $\mathcal{S}_{u,v}$ tensor. Then, for every $m \geq 0$, we have that
\[\nabla_3\phi= F \Rightarrow \]
\begin{align*}
&\left|\nabla_3\nabla_4^m\phi - \nabla_4^mF\right| \lesssim 
\\ \nonumber &\sum_{i+j+k = m,\ k \neq m}\left|\nabla_4^i\psi^j\right|\left|\nabla_4^kF\right| \\ \nonumber &+\sum_{i+j_1+j_2+k = m}\left[\left|\nabla_4^i\mathring{\psi}\right| + \left|\nabla\nabla_4^{i-1}\mathring{\psi}\right|\right]\left[\left|\nabla\nabla_4^{j_1-1}\psi^{j_2}\right| + \left|\nabla_4^{j_1}\psi^{j_2}\right|\right]\left[\left|\nabla\nabla_4^{k-1}\phi\right| + \left|\nabla_4^k\phi\right|\right],
\end{align*}
where we are employing the same schematic notation as before and $\mathring{\psi}$ denotes a Ricci coefficient which is one of $\underline{\omega}$, $\eta$, $\underline{\eta}$, or $\hat{\underline{\chi}}$.
\end{lemma}
\begin{proof}A standard calculation yields
\[\nabla_3\phi= F \Rightarrow \]
\begin{align*}
&\left|\nabla_3\nabla_4^m\phi - \nabla_4^mF\right| \lesssim 
\\ \nonumber &\ \ \sum_{i+j+k+l=m-1}\left|\nabla_4^i\omega^j\right|\left[\left|\nabla_4^k\eta\right| + \left|\nabla_4^k\underline{\eta}\right|\right]\left[\left|\nabla_4^l\nabla\phi\right| + \left|\left[\nabla_4^l,\nabla\right]\phi\right|\right]
\\ \nonumber &\ \ +\sum_{i+j+k = m-1}\left[\left|\nabla^i_4\omega^{j+1}\right|\left|\nabla_4^kF\right| + \left|\nabla_4^i\underline{\omega}^{j+1}\right|\left|\nabla_4^{k+1}\phi\right|\right]
\\ \nonumber &\ \ +\sum_{i+j+k+l = m-1}\left|\nabla_4^i\omega^j\right|\left[\left|\nabla_4^j\left(\eta+\underline{\eta}\right)^2\right| + \left|\nabla_4^j\sigma\right|\right]\left|\nabla_4^l\phi\right|.
\end{align*}
\end{proof}

\subsection{Schematic Notation for Error Terms}
The analysis of various nonlinear error terms will be achieved almost entirely based on signature considerations.

The following notation will be used to refer to nonlinear terms on the right hand side of the Bianchi equations which will always be handled perturbatively. 
\begin{definition}Let $s \in \left\{0,\frac{1}{2},1,\frac{3}{2},2,\frac{5}{2}\right\}$. Then we introduce the schematic notation
\begin{align*}
 \mathscr{E}^{(3)}_s  &\doteq \sum_{s_1+s_2 = s,\ s_1 \neq 1}\psi_{s_1}\Psi_{s_2}+ \zeta\sum_{s_1+s_2 = s-1/2}\psi_{s_1}\psi_{s_2},
 \\ \nonumber \mathscr{E}^{(4)}_s&\doteq \sum_{s_1+s_2 = s}\psi_{s_1}\Psi_{s_2}+ \zeta\sum_{s_1+s_2 = s-1/2}\psi_{s_1}\psi_{s_2}.
 \end{align*}
 Here $\psi_s$ refers to a Ricci coefficient of signature $s$ and $\Psi_s$ refers to a curvature component of signature $s$.

\end{definition}



\subsection{Bianchi Equations}
In this section we turn to the Bianchi equations which we will eventually use to carry out energy estimates. It turns out that our energy estimate scheme will only require us to calculate the equations up to error terms which can be controlled by $\mathscr{E}_s$. This will simplify the relevant calculations.

\begin{proposition}\label{Bianchit}Suppose that $g$ satisfies ${\rm Ric}(g) = 0$. Then we have
\begin{align*}
 \nabla_3\alpha_{AB} &= - \nabla^C\nu_{C(AB)}  + \nabla_{(A}\beta_{B)} - \frac{1}{2}{\rm tr}\underline{\chi}\alpha_{AB} + 4\underline{\omega}\alpha_{AB}  + \mathscr{E}^{(3)}_1,
 \\ \nonumber \nabla_4\beta_A &= \nabla^B\alpha_{BA} + \mathscr{E}^{(4)}_{1/2},
 \\ \nonumber \nabla_4\nu_{ABC} &= -2\nabla_{[A}\alpha_{B]C} +\mathscr{E}^{(4)}_{1/2},
 \\ \nonumber 
   \\ \nonumber 
   \\ \nonumber \nabla_3\nu_{ABC} &= \nabla_B\tau_{AC} -\nabla_A\tau_{BC} -\nabla_A\sigma_{BC}+\nabla_B\sigma_{AC} - \frac{2}{n}{\rm tr}\underline{\chi}\nu_{ABC} 
   \\ \nonumber &\qquad +\hat{\underline{\chi}}_A^{\ D}\nu_{DBC} - \hat{\underline{\chi}}_B^{\ D}\nu_{DAC} + 2\underline{\omega}\nu_{ABC} + \mathscr{E}^{(3)}_{3/2},
\\ \nonumber   \nabla_4R_{ABCD} &= \nabla_B\nu_{CDA} - \nabla_A\nu_{CDB} + \mathscr{E}^{(4)}_1,
 \\ \nonumber \nabla_4\sigma_{AB} &= \nabla^C\nu_{ABC} + \mathscr{E}^{(4)}_1,
\\ \nonumber 
\\ \nonumber \nabla_3R_{ABCD} &= -\nabla_A\underline{\nu}_{CDB} + \nabla_B\underline{\nu}_{CDA} -\frac{1}{2}\underline{\chi}_{AC}\left(\tau_{BD} + \sigma_{BD}\right) + \frac{1}{2}\underline{\chi}_{AD}\left(\tau_{BC} + \sigma_{BC}\right)
\\ \nonumber &\qquad +\frac{1}{2}\underline{\chi}_{BC}\left(\tau_{AD} + \sigma_{AD}\right) - \frac{1}{2}\underline{\chi}_{BD}\left(\tau_{AC} + \sigma_{AC}\right) 
\\ \nonumber &\qquad -\frac{2}{n}{\rm tr}\underline{\chi}R_{ABCD}+ \underline{\hat{\chi}}_A^{\ E}R_{BECD} + \underline{\hat{\chi}}_B^{\ E}R_{EACD} + \mathscr{E}^{(3)}_{2},
\\ \nonumber \nabla_3\tau_{AB} &= \nabla_{(A}\underline{\beta}_{B)} - \nabla^C\underline{\nu}_{C(AB)} +\underline{\chi}_{AB}\rho -\left(\frac{1}{n}+\frac{1}{2}\right){\rm tr}\underline{\chi}\tau_{AB}  + \underline{\hat{\chi}}^{CD}R_{C(AB)D} + \mathscr{E}^{(3)}_{2},
\\ \nonumber \nabla_3\rho &= -\nabla^A\underline{\beta}_A - \left(1+\frac{1}{n}\right){\rm tr}\underline{\chi}\rho - \frac{1}{2}\hat{\underline{\chi}}^{AB}\tau_{AB} + \mathscr{E}^{(3)}_2,
\\ \nonumber \nabla_3\sigma_{AB} &= -\nabla^C\underline{\nu}_{ABC} - \left(1+\frac{1}{n}\right){\rm tr}\underline{\chi}\sigma_{AB} - \frac{1}{2}\hat{\underline{\chi}}^C_{\ [A}R_{B]4C3} + \mathscr{E}^{(3)}_2, 
\\ \nonumber \nabla_4\underline{\nu}_{ABC} &= -2\nabla_{[A}\tau_{B]C} + 2\nabla_{[A}\sigma_{B]C}  + \mathscr{E}^{(4)}_{3/2},
\\ \nonumber
\\ \nonumber \nabla_3\underline{\nu}_{ABC} &= - 2\underline{\omega}\nu_{ABC} - 2\nabla_{[A}\underline{\alpha}_{B]C} - 2\underline{\chi}_{C[A}\underline{\beta}_{B]}-\frac{3}{n}{\rm tr}\underline{\chi}\underline{\nu}_{ABC} 
\\ \nonumber &\qquad + 2\underline{\hat{\chi}}^D_{\ [A}\underline{\nu}_{B]DC} + \underline{\hat{\chi}}_A^{\ D}\underline{\nu}_{CDB} -\underline{\hat{\chi}}_B^{\ D}\underline{\nu}_{CDA} + \mathscr{E}^{(3)}_{5/2},
\\ \nonumber \nabla_3\underline{\beta}_A &= -\nabla^B\underline{\alpha}_{BA} -2\underline{\omega}\underline{\beta}_A + \left(1+\frac{2}{n}\right){\rm tr}\underline{\chi}\underline{\beta}_A  + \hat{\underline{\chi}}^{BC}\underline{\nu}_{ABC} - \hat{\underline{\chi}}_A^{\ B}\underline{\beta}_B + \mathscr{E}^{(3)}_{5/2},
\\ \nonumber\nabla_4\underline{\alpha}_{AB} &= -\nabla^C\underline{\nu}_{C(AB)}-\nabla_{(A}\underline\beta_{B)} + \mathscr{E}^{(4)}_2.
\end{align*}
\end{proposition}
\begin{remark}We have mostly grouped the equations in the way that corresponds to the ``Bianchi-pair'' notion introduced in Definition~\ref{defbiancpair} (see also Proposition~\ref{thebianchipairs}). The exception is that there are extra $\nabla_3$ equations for $\tau_{AB}$ and $\rho$ which are useful to have written explicitly. (They are just suitable traces of the $R_{ABCD}$ equation.)
\end{remark}
\begin{proof}Let's have a few general remarks before we dive into the calculations. There are two ways we generate equations for the $\nabla_3$ or $\nabla_4$ derivatives of curvature components. The first is the second Bianchi identity:
\begin{equation}\label{2ndbian}
D_iR_{jklm} + D_jR_{klmi} + D_kR_{lmil} = 0.
\end{equation}
This is effective if the curvature component we are interested in has $\mathcal{S}$ indices in the first two slots. Otherwise we will use the well-known fact that Ricci flatness and the second Bianchi identity imply that the curvature tensor is divergence free:
\begin{equation}\label{divriem}
D^iR_{ijkl} = 0.
\end{equation}

In either case, having written down the equation in terms of the connection $D$ of $g$, we use the formulas for the connection coefficients from Lemma~\ref{D} to re-write everything in terms of $\nabla_3$, $\nabla_4$, and $\nabla_A$. This turns out to not be as complicated as one might fear because, first of all, this procedure only requires us to consider the component of $D$ which is normal to $\mathcal{S}$, and second of all, due to Lemma~\ref{preservesig} and the fact that we allow terms proportional to the appropriate error term $\mathscr{E}$, for lower order terms we only need to track the ${\rm tr}\underline{\chi}$. Thus we can effectively work with the following table
\[D_4e_4 = \mathscr{E},\qquad D_4e_3 = \mathscr{E},\qquad \left(D_4e_A\right)^{\perp} =\mathscr{E}, \]
\[D_3e_4 = 2\underline{\omega}e_4 + \mathscr{E} ,\qquad D_3e_3 = -2\underline{\omega}e_3,\qquad \left(D_3e_A\right)^{\perp} = \mathscr{E},\]
\[D_Ae_4 = \mathscr{E},\qquad D_Ae_3 = \underline{\chi}_A^{\ B}e_B + \mathscr{E},\qquad \left(D_Ae_B\right)^{\perp} = \frac{1}{2}\underline{\chi}_{AB}e_4 +\mathscr{E}.\]
Here $\mathscr{E}$ denotes a product of a Ricci coefficient not of signature $1$ and an element of the null frame $\left(e_3,e_4,\{e_A\}\right)$ which is consistent with signature.

We start with $\tau_{AB}$, using Lemma~\ref{curvids} we write
\begin{equation}\label{tauexp}
\tau_{AB} = \slashed{g}^{CD}R_{CADB}.
\end{equation}
Note that the signature of $R_{CADB}$ is $1$. Applying $\nabla_4$ yields
\[\nabla_4\tau_{AB} = \slashed{g}^{CD}\nabla_4R_{CADB}.\]
Then we can express $\nabla_4R_{CADB}$ in terms of $D_4R_{CADB}$ and use the second Bianchi identity~\eqref{2ndbian} to obtain
\[D_4R_{CABD} + D_CR_{A4DB} + D_AR_{4CDB} = 0.\]
Converting the covariant derivatives $D$ into $\mathcal{S}_{u,v}$ derivatives leads to 
\begin{align*}
&\nabla_4R_{CADB} + \nabla_CR_{A4DB} + \nabla_AR_{4CDB} - \frac{1}{2}\underline{\chi}_{CD}R_{A44B} 
\\ \nonumber &\qquad - \frac{1}{2}\underline{\chi}_{CB}R_{A4D4} - \frac{1}{2}\underline{\chi}_{AD}R_{4C4B} - \frac{1}{2}\underline{\chi}_{AB}R_{4CD4} + \mathscr{E}^{(3)}_1 = 0.
\end{align*}
Tracing over $C$ and $D$ and symmetrizing in $A$ and $B$ then leads to
\[\nabla_4\tau_{AB} + \nabla^C\nu_{C(AB)} + \nabla_{(A}\beta_{B)} + \frac{1}{2}{\rm tr}\underline{\chi}\alpha_{AB} - \underline{\chi}_{(A}^{\ \ C}\alpha_{B)C} + \mathscr{E}^{(3)}_1 = 0.\]
This establishes the desired $\nabla_4$ equation for $\tau$. We have kept the more precise error term $\mathscr{E}^{(3)}_1$ since it will be important for establishing the equation for $\nabla_3\alpha_{AB}$.

Next we compute the equation for $\nabla_3\alpha_{AB}$. For this we use~\eqref{divriem}: 
\[-\frac{1}{2}D_3R_{4A4B} - \frac{1}{2}D_4R_{3A4B} + \slashed{g}^{CD}D_CR_{DA4B} = 0.\]
After symmetrizing in $A$ and $B$ and using that we already have computed the $\nabla_4$ equation for $\tau$, we eventually obtain
\[\nabla_3\alpha_{AB}  + \nabla^C\nu_{C(AB)} - \nabla_{(A}\beta_{B)} + \frac{1}{2}{\rm tr}\underline{\chi}\alpha_{AB} -4\underline{\omega}\alpha_{AB} + \mathscr{E}^{(3)}_1 = 0.\]

Next we compute the equation for $\nabla_3\nu_{ABC}$.  Working in the usual fashion with~\eqref{2ndbian} yields
\begin{equation}\label{nu1}
\nabla_3\nu_{ABC} = -\nabla_A\tau_{BC} + \nabla_B\tau_{AC} - \nabla_A\sigma_{BC} + \nabla_B\sigma_{AC} - \frac{2}{n}{\rm tr}\underline{\chi}\nu_{ABC}+\hat{\underline{\chi}}_A^{\ D}\nu_{DBC} - \hat{\underline{\chi}}_B^{\ D}\nu_{DAC} + 2\underline{\omega}\nu_{ABC} + \mathscr{E}^{(3)}_{3/2}.
\end{equation}


Next, we compute the equation for $\nabla_4R_{ABCD}$. We start with
\[D_4R_{ABCD} + D_AR_{B4CD} + D_BR_{4ACD} = 0,\]
and then obtain
\[\nabla_4R_{ABCD} = \nabla_B\nu_{CDA} - \nabla_A\nu_{CDB} + \mathscr{E}^{(4)}_1.\]

For $\nabla_4\sigma_{AB}$ we recall that the first Bianchi identity implies
\[\sigma_{AB} = \frac{1}{2}R_{34AB}.\]
Using also that
\[-\frac{1}{2}D_4R_{34AB} + \slashed{g}^{CD}D_CR_{D4AB} = 0,\]
we obtain
\[\nabla_4\sigma_{AB} = \nabla^C\nu_{ABC} + \mathscr{E}^{(4)}_1.\]

Next, we compute the equation for $\nabla_3R_{ABCD}$. We start with
\[D_3R_{ABCD} + D_AR_{B3CD} + D_BR_{3ACD} = 0.\]
Eventually we obtain
\begin{align*}
&\nabla_3R_{ABCD} + \nabla_A\underline{\nu}_{CDB} - \nabla_B\underline{\nu}_{CDA} +\frac{1}{2}\underline{\chi}_{AC}\left(\tau_{BD} + \sigma_{BD}\right) - \frac{1}{2}\underline{\chi}_{AD}\left(\tau_{BC} + \sigma_{BC}\right)
\\ \nonumber &\qquad -\frac{1}{2}\underline{\chi}_{BC}\left(\tau_{AD} + \sigma_{AD}\right) + \frac{1}{2}\underline{\chi}_{BD}\left(\tau_{AC} + \sigma_{AC}\right) +\frac{2}{n}{\rm tr}\underline{\chi}R_{ABCD}+ \underline{\hat{\chi}}_A^{\ E}R_{BECD} - \underline{\hat{\chi}}_B^{\ E}R_{EACD} + \mathscr{E}^{(3)}_{2} = 0.
\end{align*}

Tracing once and symmetrizing leads to
\begin{align*}
&\nabla_3\tau_{AB} -\nabla_{(A}\underline{\beta}_{B)} + \nabla^C\underline{\nu}_{C(AB)} -\underline{\chi}_{AB}\rho +\left(\frac{1}{n}+\frac{1}{2}\right){\rm tr}\underline{\chi}\tau_{AB}  - \underline{\hat{\chi}}^{CD}R_{C(AB)D} + \mathscr{E}^{(3)}_{2} = 0.
\end{align*}
Tracing one final time leads to 
\[\nabla_3\rho + \nabla^A\underline{\beta}_A + \left(1+\frac{1}{n}\right){\rm tr}\underline{\chi}\rho + \frac{1}{2}\hat{\underline{\chi}}^{AB}\tau_{AB} + \mathscr{E}^{(3)}_2 = 0.\]

Now we come to $\nabla_3\sigma_{AB}$. As when we computed $\nabla_4\sigma_{AB}$, we start with
\[\sigma_{AB} = -\frac{1}{2}R_{43AB}.\]
Then we appeal to
\[-\frac{1}{2}D_3R_{43AB} + \slashed{g}^{CD}D_CR_{D3AB} = 0.\]
We eventually obtain
\[\nabla_3\sigma_{AB} + \nabla^C\underline{\nu}_{ABC} + \left(1+\frac{1}{n}\right){\rm tr}\underline{\chi}\sigma_{AB} + \frac{1}{2}\hat{\underline{\chi}}^C_{\ [A}R_{B]4C3} + \mathscr{E}^{(3)}_2.\]

Next we compute the equation for $\nabla_4\underline{\nu}_{ABC}$. We proceed analogously to how we computed the equation for $\nabla_4\nu_{ABC}$. We eventually obtain
\[\nabla_4\underline{\nu}_{ABC} = -2\nabla_{[A}\tau_{B]C} + 2\nabla_{[A}\sigma_{B]C} + \mathscr{E}^{(4)}_{3/2}.\]

Similarly, we may compute $\nabla_3\underline{\nu}_{ABC}$. We eventually obtain
\[\nabla_3\underline{\nu}_{ABC} + 2\underline{\omega}\nu_{ABC} + 2\nabla_{[A}\underline{\alpha}_{B]C} + 2\underline{\chi}_{C[A}\underline{\beta}_{B]}+\frac{3}{n}{\rm tr}\underline{\chi}\underline{\nu}_{ABC} - 2\underline{\hat{\chi}}^D_{\ [A}\underline{\nu}_{B]DC} - \underline{\hat{\chi}}_A^{\ D}\underline{\nu}_{CDB} + \underline{\hat{\chi}}_B^{\ D}\underline{\nu}_{CDA} + \mathscr{E}^{(3)}_{5/2} = 0.\]


Tracing this yields
\[\nabla_3\underline{\beta}_A + 2\underline{\omega}\beta_A +\nabla^B\underline{\alpha}_{AB} + \left(1+\frac{2}{n}\right){\rm tr}\underline{\chi}\underline{\beta}_A - \hat{\underline{\chi}}^{BC}\underline{\nu}_{ABC} + \hat{\underline{\chi}}_A^{\ B}\underline{\beta}_B + \mathscr{E}^{(3)}_{5/2} = 0. \]

Finally we compute the equation for $\nabla_4\underline{\alpha}_{AB}$. We use~\eqref{divriem} with $jkl = A3B$ and use that we already have an equation for $\nabla_3\tau_{AB}$. We obtain
\begin{align*}
\nabla_4\underline{\alpha}_{AB} = -\nabla^C\underline{\nu}_{C(AB)}+\nabla_{(A}\beta_{B)}+ \mathscr{E}^{(4)}_2.
\end{align*}
\end{proof}

In order to use these equations to carry out energy estimates it will be useful to group various subsets into so-called \emph{Bianchi pairs}.

\begin{definition}\label{defbiancpair} We say that the two tuples $\left(\left(\Psi_{i_1},\cdots,\Psi_{i_l}\right),\left(\Psi_{j_1},\cdots,\Psi_{j_m}\right)\right)$ of curvature components form a ``Bianchi pair'' if, first of all, they satisfy a system of equations of the form
\[\nabla_3\Psi_{i_1} = \mathcal{D}^{(i_1)}\left(\Psi_{j_1},\cdots,\Psi_{j_m}\right) + {\rm l.o.t.},\]
\[\vdots\]
\[\nabla_3\Psi_{i_l} = \mathcal{D}^{(i_l)}\left(\Psi_{j_1},\cdots,\Psi_{j_m}\right)+ {\rm l.o.t.},\]
\[\nabla_4\Psi_{j_1} = \mathcal{D}^{(j_1)}\left(\Psi_{i_1},\cdots,\Psi_{i_l}\right) + {\rm l.o.t.},\]
\[\vdots\]
\[\nabla_4\Psi_{j_m} = \mathcal{D}^{(j_m)}\left(\Psi_{i_1},\cdots,\Psi_{i_l}\right) + {\rm l.o.t.},\]
where the l.o.t. denotes lower order terms and the $\mathcal{D}$'s are linear first order differential operators on $\mathcal{S}$, and, secondly, there exists positive constants 
\[c_{i_1},\cdots,c_{i_l},c_{j_1},\cdots,c_{j_m},\]
such that
\[\sum_{k=1}^lc_{i_k}\int_{\mathcal{S}}\mathcal{D}^{(i_k)}\left(\Psi_{j_1},\cdots,\Psi_{j_m}\right)\Psi_{i_k} + \sum_{p=1}^mc_{j_p}\int_{\mathcal{S}}\mathcal{D}^{(j_p)}\left(\Psi_{i_1},\cdots,\Psi_{i_l}\right)\Psi_{j_p} = {\rm l.o.t.}\]
\end{definition}

\begin{remark}If $\left(\left(\Psi_{i_1},\cdots,\Psi_{i_l}\right),\left(\Psi_{j_1},\cdots,\Psi_{j_m}\right)\right)$ form a Bianchi pair, then we will have 
\[\sum_{k=1}^lc_{i_k}\int_{\mathcal{S}}\nabla_3\left|\Psi_{i_1}\right|^2+ \sum_{p=1}^mc_{j_p}\int_{\mathcal{S}}\nabla_4\left|\Psi_{j_p}\right|^2 = {\rm l.o.t}.\]
Integrating in the $3$ and $4$ directions then yields an energy estimate.
\end{remark}

We have
\begin{proposition}\label{thebianchipairs}The following tuples of curvature components form Bianchi pairs:
\[\left(\alpha_{AB},\left(\beta_A,\nu_{ABC}\right)\right),\qquad \left(\nu_{ABC},\left(R_{ABCD},\sigma_{AB}\right)\right), \qquad \left(\left(R_{ABCD},\sigma_{AB}\right),\underline{\nu}_{ABC}\right),\qquad \left(\left(\underline{\beta}_A,\underline{\nu}_{ABC}\right),\underline{\alpha}_{AB}\right).\]
The grouping here corresponds exactly to how we grouped the equations in Proposition~\ref{Bianchit} into the corresponding groups. 

The relevant constants $c_i$ (see Definition~\ref{defbiancpair}) may be taken to be
\[\left(2,\left(2,1\right)\right),\qquad \left(2,\left(1,1\right)\right),\qquad \left(\left(1,1\right),2\right),\qquad \left(\left(2,1\right),2\right).\]
\end{proposition}
\begin{proof}This follows easily from Propositions~\ref{Bianchit} and~\ref{extraconsthebest} and a straightforward series of integration by parts on $\mathcal{S}$.
\end{proof}



\subsection{Regular Solutions}\label{weak}
We will want to consider solutions to the Einstein equations which are not necessarily smooth.  In this section we will discuss in what sense these are solutions.

We assume that our background differentiable manifold is given by 
\begin{equation}\label{whereu0isref}
\mathcal{M} = (\mathring{u},0) \times [0,\mathring{v}) \times \mathcal{S}
\end{equation}
 where $\mathcal{S}$ is a closed $n$-dimensional manifold and we use the $u$-coordinate to parametrize $(\mathring{u},0)$ and the $v$-coordinate to parametrize $[0,\mathring{v})$.

\begin{definition}We call any collection of $\mathcal{S}_{u,v}$ tensors on $\mathcal{M}$:
\[\slashed{g}_{AB},\qquad b^A,\qquad \Omega,\qquad \chi_{AB},\qquad \underline{\chi}_{AB},\qquad \omega,\qquad \underline{\omega},\qquad \zeta_A,\]
\[\eta_A,\qquad \underline{\eta}_A,\qquad \alpha_{AB},\qquad \underline{\alpha}_{AB},\qquad \beta_A,\qquad \underline{\beta}_A,\]
\[\nu_{ABC},\qquad \underline{\nu}_{ABC},\qquad \rho,\qquad \tau_{AB},\qquad \sigma_{AB},\qquad R_{ABCD}\]
a set of ``double null unknowns.''
\end{definition}

The next proposition is the statement that the equations of the double null gauge imply Ricci flatness of the metric, at least when the metric is $C^2$.
\begin{proposition}\label{theyimply}Suppose we have a set of ``double null unknowns'' such that $\slashed{g}_{AB}$, $b^A$, and $\log\Omega$ are twice continuously differentiable on $\mathcal{M}$, and we have classical solutions of all of the equations listed in Propositions~\ref{metriceqn},~\ref{nullstruct}, and~\ref{constrainteqns} where all of the Ricci curvature terms are set to $0$. 

Then the $C^2$-metric $g_{\alpha\beta}$ defined by 
\[g = -2\Omega^2 \left(du\otimes dv + dv \otimes du\right)+ \slashed{g}_{AB}\left(d\theta^A - b^Adu\right)\otimes\left(d\theta^B - b^Bdu\right),\]
is a classical solution to the Einstein vacuum equations
\[Ric\left(g\right) = 0.\]
\end{proposition}
\begin{proof}This the content of Remark~\ref{allofEinstein}.
\end{proof}

We first consider the case when our solution $g$ is an regular solution.

\begin{definition}\label{whatisweak}We say that a metric $g$ is a \textbf{regular solution} to the Einstein vacuum equations if $g$ is regular (see Definition~\ref{ambientregular}) and if 
\begin{enumerate}
	\item When $n = 2$ or $n > 4$ the metric defined by 
	\[g = -2\Omega^2 \left(du\otimes dv + dv \otimes du\right)+ \slashed{g}_{AB}\left(d\theta^A - b^Adu\right)\otimes\left(d\theta^B - b^Bdu\right),\]
is a classical solution to the Einstein vacuum equations.
\item When $n=3$ or $n=4$, $g$ is a classical solution to the Einstein vacuum equations when $v > 0$, the correspondingly defined double null unknowns are classical solutions everywhere to the constraint equations of Proposition~\ref{constrainteqns}, and the double null unknows are weak solutions to the equations of Propositions~\ref{metriceqn},~\ref{nullstruct}, and~\ref{Bianchit}, where, letting $\mathcal{D}$ denote an arbitrary $1$-order skew-adjoint differential operator on $\mathcal{S}$, we say the equation 
\[\nabla_3\psi_1 + \mathcal{D}\psi_2 = F,\]
is satisfied weakly in $\mathcal{M}$ if for almost every $v_1 \in [0,\mathring{v})$, smooth $\varphi\left(u,\theta^A\right)$, and $\mathring{u} < u_1 < u_2 < 0$ we have
\[\int_{u_1}^{u_2}\int_{\mathcal{S}}\left[-\psi_1\cdot\nabla_3^*\left(\varphi\sqrt{\slashed{g}}\right) -\psi_2\mathcal{D}^*\left(\varphi\sqrt{\slashed g}\right) + F\cdot\varphi \sqrt{\slashed{g}}\right]|_{v= v_1}\, du\, d\theta^A \]
\[+ \int_{\mathcal{S}_{u_2,v_1}}\psi_1\cdot\varphi\sqrt{\slashed{g}}\, d\theta^A - \int_{\mathcal{S}_{u_1,v_1}}\psi_1\cdot\varphi\sqrt{\slashed{g}}\, d\theta^A = 0.\]
We have an analogous definition for $\nabla_4$ equations.
\end{enumerate}
\end{definition}
\begin{remark}Note that the only double null unknown which does not necessarily extend continuously to $\{v =0\}$ is $\alpha$ when $n = 3$ or $n=4$.\footnote{To see this, first consider the case when $n$ is odd. Then we recall that the first possibly non-smooth term in the Taylor expansion of $g$ near $\{v = 0\}$ is a term proportional to $v^{\frac{n}{2}}$. If $n \geq 5$, then it is clear that all second derivatives of $g$ (and hence all curvature components) will extend continuously to $\{v = 0\}$. When $n=3$, the only problematic component is one that involves two $v$-derivatives of $g$, and this is exactly the $\alpha$ component of curvature. When $n$ is even, the first possibly non-smooth component is the term proportional to $v^{\frac{n}{2}}\log(v)$. Now one can argue as in the odd case.}Hence, once one observes that Proposition~\ref{Bianchit} does not have a $\nabla_4$ equation for $\alpha$, it is straightforward to see that  Definition~\ref{whatisweak} makes sense.
\end{remark}
\begin{remark}By revisiting the proof of Proposition~\ref{theyimply}, it is straightforward to check that this definition implies that in the $(u,v,\theta^A)$ coordinates, the Einstein equations $R_{\mu\nu} = 0$ are satisfied weakly in $L^2$.
\end{remark}

Finally, the metrics produced by Theorem~\ref{localexistenceproto} are potentially more singular than the above. However, they will be constructed as uniform limits in the $\mathfrak{E}$ norm (see Definition~\ref{totalnorm}) of regular solutions. This allows in  a straightforward manner for the interpretation of them as solutions to the Einstein equations.

\subsection{Scaling numerology}\label{scalingbehav}
In this section we record the behavior of the metric components, Ricci coefficients, and curvature components under the action of 
\begin{equation}\label{rescaleddiff}
\hat{\Phi}_{\lambda} \doteq \lambda^{-2}\Phi_{\lambda},
\end{equation}
(see~\eqref{rescalingmap}). In the following formulas it is understood that tensors are always evaluated in the coordinate frame. We also suppress the dependence on the $\theta^A$ coordinates.

\[\hat{\Phi}_{\lambda}\Omega\left(u,v\right) = \Omega\left(\lambda u,\lambda v\right), \qquad \left(\hat{\Phi}_{\lambda}\slashed{g}_{AB}\right)\left(u,v\right) = \lambda^{-2}\slashed{g}_{AB}\left(\lambda u,\lambda v\right),\qquad \hat{\Phi}_{\lambda}b^A\left(u,v\right) = \lambda b^A\left(\lambda u,\lambda v\right),\]

\[\hat{\Phi}_{\lambda}\chi_{AB}\left(u,v\right) = \lambda^{-1}\chi_{AB}\left(\lambda u,\lambda v\right),\qquad \hat{\Phi}_{\lambda}{\rm tr}\chi\left(u,v\right) = \lambda {\rm tr}\chi\left(\lambda u,\lambda v\right),\]

\[\hat{\Phi}_{\lambda}\hat{\chi}_{AB}\left(u,v\right) = \lambda^{-1}\hat{\chi}_{AB}\left(\lambda u,\lambda v\right),\qquad \hat{\Phi}_{\lambda}w\left(u,v\right) = \lambda \omega\left(\lambda u,\lambda v\right),\]

\[\hat{\Phi}_{\lambda}\alpha_{AB}\left(u,v\right) = \alpha_{AB}\left(\lambda u,\lambda v\right),\qquad \hat{\Phi}_{\lambda}\zeta_A\left(u,v\right) = \zeta_A\left(\lambda u,\lambda v\right),\]

\[\hat{\Phi}_{\lambda}\eta_A\left(u,v\right) = \eta_A\left(\lambda u,\lambda v\right),\qquad \hat{\Phi}_{\lambda}\underline{\eta}_A\left(u,v\right) = \underline{\eta}_A\left(\lambda u,\lambda v\right),\]

\[\hat{\Phi}_{\lambda}\beta_A\left(u,v\right) = \lambda \beta_A\left(\lambda u,\lambda v\right),\qquad \hat{\Phi}_{\lambda}\underline{\chi}_{AB}\left(u,v\right) = \lambda^{-1}\underline{\chi}_{AB}\left(\lambda u,\lambda v\right),\]

\[\hat{\Phi}_{\lambda}{\rm tr}\underline{\chi}\left(u,v\right) = \lambda {\rm tr}\underline{\chi}\left(\lambda u,\lambda v\right),\qquad \hat{\Phi}_{\lambda}\hat{\underline{\chi}}_{AB}\left(u,v\right) = \lambda^{-1}\hat{\underline{\chi}}_{AB}\left(\lambda u,\lambda v\right),\]

\[\hat{\Phi}_{\lambda}\underline{\alpha}_{AB}\left(u,v\right) = \underline{\alpha}_{AB}\left(\lambda u,\lambda v\right),\qquad \hat{\Phi}_{\lambda}\underline{\beta}_A\left(u,v\right) = \lambda \underline{\beta}_A\left(\lambda u,\lambda v\right),\qquad \hat{\Phi}_{\lambda}R_{ABCD}\left(u,v\right) = \lambda^{-2}R_{ABCD}\left(\lambda u,\lambda v\right),\]

\[\hat{\Phi}_{\lambda}\underline{\omega}\left(u,v\right) = \lambda \underline{\omega}\left(\lambda u,\lambda v\right),\qquad \hat{\Phi}_{\lambda}\sigma_{AB}\left(u,v\right) =  \sigma\left(\lambda u,\lambda v\right),\qquad \hat{\Phi}_{\lambda}\rho\left(u,v\right) = \lambda^2\rho\left(\lambda u,\lambda v\right).\]

\begin{remark}\label{remarkaboutscalingbehav}In particular, this calculation shows that for a self-similar solution we expect Ricci coefficients $\psi$ to satisfy
\[\left|\psi\right|_{\slashed{g}} = f\left(\frac{v}{u},\theta^A\right)|u|^{-1},\]
for some non-negative function $f$.

For a curvature component $\Psi$, we expect
\[\left|\Psi\right|_{\slashed{g}} = F\left(\frac{v}{u},\theta^A\right)u^{-2},\]
for some non-negative function $F$.

To make this more concrete, let's do the calculation explicitly for $\hat{\chi}$. Using the relations above, we see that in the coordinate frame
\[\hat{\chi}_{AB}\left(u,v\right) = -u\hat{\chi}_{AB}\left(-1,\frac{v}{-u}\right),\qquad \slashed{g}^{AB}\left(v,u\right) = u^{-2}\slashed{g}^{AB}\left(-1,\frac{v}{-u}\right).\]
Thus
\begin{align*}
\left|\hat{\chi}\right||_{(u,v)} &= \sqrt{\slashed{g}^{AC}\slashed{g}^{BD}\hat{\chi}_{AB}\hat{\chi}_{CD}}|_{(u,v)}
\\ \nonumber &= |u|^{-1} \sqrt{\slashed{g}^{AC}\slashed{g}^{BD}\hat{\chi}_{AB}\hat{\chi}_{CD}}|_{(-1,\frac{v}{-u})}.
\end{align*}

\end{remark}


\section{Analysis of the Initial Characteristic Data}\label{secinitialdata}
We will be interested in solving a characterstic initial value problem, and thus we must prescribe data along the two null hypersurfaces $\{v = 0\}$ and $\{u = -1\}$. As we have already recalled in Section~\ref{consexpl}, the Einstein equations are over-determined and the solution cannot be prescribed arbitrarily along these null hypersurfaces; instead the various null constraint equations must be satisfied. In order to prove Theorem~\ref{localexistenceproto} we face the additional necessity that the characterstic data satisfy ``self-similar bounds'' (see the dicussion in Section~\ref{consexpl}.)

By a standard density argument (see Appendix~\ref{actuallocalexistencesec}), in order to prove Theorem~\ref{localexistenceproto} we can work with data which is qualitatively more regular than Definition~\ref{admissibleconjugatedata}, as long as the quantitative estimates respect the regularity of Definition~\ref{admissibleconjugatedata}. In particular, it will suffice to work with initial data of the following type.

\begin{definition}\label{thisdataisreallyreallyregular}Let $\left(\mathcal{S},\slashed{g}_0\right)$ be a closed orientable $n$-dimensional Riemannian manifold. Then we say a $1$-parameter family $\hat{\slashed{g}}(v)$, $v \in [0,\epsilon)$, of Riemannian metrics on $\mathcal{S}$ is ``regular conjugate data'' if 
\begin{enumerate}
	\item
	\begin{enumerate}
		\item When $n=2$, $\hat{\slashed{g}}(v)$ is smooth in $v$.
		\item When $n \geq 3$ and odd, $\hat{\slashed{g}}(v) = \hat{\slashed{g}}^{(1)}(v) + v^{\frac{n}{2}}\hat{\slashed{g}}^{(2)}(v)$ for smooth $\hat{\slashed{g}}^{(1)}$ and $\hat{\slashed{g}}^{(2)}$.
		\item When $n \geq 4$ and even, $\hat{\slashed{g}}(v) = \hat{\slashed{g}}^{(1)}(v) + \log(v)v^{\frac{n}{2}}\hat{\slashed{g}}^{(2)}(v)$ for smooth $\hat{\slashed{g}}^{(1)}$ and $\hat{\slashed{g}}^{(2)}$.
	\end{enumerate}

	\item $\hat{\slashed{g}}\left(0\right)  = \slashed{g}_0$.
		\end{enumerate}
\end{definition}

As is well-known, at least when $n = 2$, the Einstein equations are locally well-posed when data are posed on two \emph{regular} tranversally intersecting null hypersurfaces (see~\cite{lukchar,ren}). For every $u_0 < 0$, we can appeal to these results to get a local solution associated to the characterstic data on $\{v=0\}\cap \{|u| \geq u_0\}$ and $\{u=-1\}$:
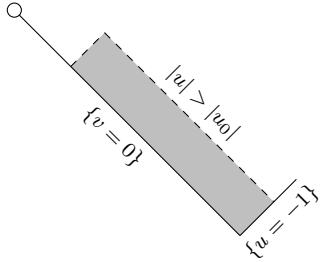
\begin{center}
\begin{tikzpicture}[scale =1.5]
\fill[lightgray] (2,-2)--(2.3,-1.7) --  (.8,-.2)--(.5,-.5); 

\draw (0,0) -- (2,-2) node[sloped,below,midway]{\footnotesize $\{v=0\}$};
\draw (2,-2) -- (2.5,-1.5) node[sloped,below,midway]{\footnotesize $\{u=-1\}$};
\draw [dashed] (.8,-.2)--(2.3,-1.7) node[sloped,above,midway]{\footnotesize $|u| > |u_0|$};
\draw [dashed] (.5,-.5) -- (.8,-.2);
\path [draw=black,fill=white] (0,0) circle (1/16); 

\end{tikzpicture}
\end{center}
Here the width of the rectangle will depend on a lower bound for $|u_0|$; in particular it does not follow from this result that the width is uniformly bounded from below as $u_0\to 0$.

 As is dicussed in Appendix~\ref{actuallocalexistencesec}, we can then take the union over $u_0 < 0$ to establish the following local existence theorem for regular solutions to the Einstein equations (which in fact holds in all dimensions).
\begin{theorem}\label{localprotoambientyay}Let $\left(\mathcal{S},\slashed{g}_0\right)$ be a closed orientable $n$-dimensional Riemannian manifold and $\hat{\slashed{g}}(v)$ be regular conjugate data. Then, after possibly taking $\epsilon$ smaller, there exists a function $0 < f(u) \ll \epsilon$ with $\lim_{u\to 0}f(u) = 0$, an open set $\mathcal{M}_0 \doteq \{(u,v,\theta^A) \in \mathcal{M}: v \leq f(u)\}$, and a unique regular metric $g$ on $\mathcal{M}_0$ solving the Einstein equations such that in the corresponding double null gauge we have 
\begin{enumerate}
\item 
\[\slashed{g}|_{v = 0} = u^2\slashed{g}_0.\]
\item 
\[\zeta|_{(u,v) = (-1,0)} = 0.\]
\item 
\[\Omega^2|_{\{v = 0\} \cup \{u=-1\}} = 1.\]
\item 
\[b|_{\{v = 0\}} = 0.\]
\item 
\begin{equation}\label{trchieqn}
{\rm tr}\chi|_{(u,v) = (-1,0)} = \frac{\slashed{R}_0}{n-1}.
\end{equation}
\item There exists a function $\Phi\left(v,\theta^A\right)$ with 
\[\slashed{g}|_{u=-1} = \Phi^2\hat{\slashed{g}}.\]
\end{enumerate}
\end{theorem}

\begin{remark}\label{somestuff}One can easily establish that
\[\hat{\chi}|_{\{u=-1\}} = \frac{1}{2}\Phi^2\mathcal{L}_v\hat{\slashed{g}},\qquad {\rm tr}\chi = n\partial_v\log\left(\Phi\right).\]
\end{remark}
\begin{remark}\label{specifyalpha}Using Remark~\ref{somestuff}  and the formula (see the first two equations of Proposition~\ref{nullstruct})
\[\mathcal{L}_v\chi_{AB} = \slashed{g}^{CD}\chi_{AC}\chi_{BD} - \alpha_{AB},\]
which holds in a Lie-propagated frame along $\{u=-1\}$, one can easily verify that the specification of $\{\mathcal{L}^k_v\hat{\slashed{g}}|_{v=0}\}_{k=1}^M$ is equivalent to the specification of $\{\nabla_4^k\alpha|_{v=0}\}_{k=1}^M$, and that in this fashion one can arrange for $\{\nabla_4^k\alpha|_{v=0}\}_{k=1}^M$ to take any set of values. Similarly, a logarithmic divergence of $\alpha$ as $v\to 0$ can be specified by an appropriate logarithmic divergence of $\hat{\slashed{g}}$ as $v\to 0$.
\end{remark}
Below we depict the region of existence given by this theorem:
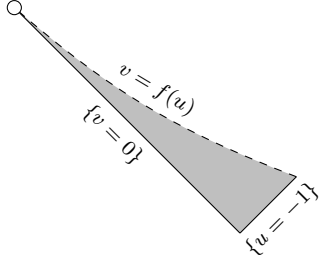
\begin{center}
\begin{tikzpicture}[scale = 1.5]
\fill[lightgray] (0,0)--(2,-2)--(2.5,-1.5)  to [bend left = 10] (0,0);

\draw (0,0) -- (2,-2) node[sloped,below,midway]{\footnotesize $\{v=0\}$};
\draw (2,-2) -- (2.5,-1.5) node[sloped,below,midway]{\footnotesize $\{u=-1\}$};
\draw [dashed] (2.5,-1.5)to [bend left= 10] node[sloped,above,midway]{\footnotesize $v=f(u)$} (0,0) ;
\path [draw=black,fill=white] (0,0) circle (1/16); 

\end{tikzpicture}
\end{center}
We emphasize once again that the theorem does not provide an estimate for the size of $\mathcal{M}_0$; in particular, there are no estimates for $f(u)$ from below and hence the region covered is certainly not sufficient to prove Theorem~\ref{localexistenceproto}.

Next, we give a definition of ``compatible regular conjugate data''. These will be choices of regular conjugate data whose corresponding metrics $g$ induced by Theorem~\ref{localprotoambientyay} have an asymptotic behavior along $\{v = 0\}$ which is consistent with being asymptotically self-similar.

\begin{definition}\label{defcompatible}We say that regular conjugate data $\hat{\slashed{g}}(v)$  is ``compatible'' if the metric $g$ induced by Theorem~\ref{localprotoambientyay} satisfies the following.
\begin{enumerate}
	\item If $n = 2$, every Ricci coefficient $\psi$ and curvature component $\Psi$ obey the estimate
	\begin{equation}\label{arequirement}
	\left|\nabla^i\psi\right||_{v=0} \lesssim_i |u|^{-1-i},\qquad \left|\nabla^i\Psi\right||_{v=0} \lesssim_i |u|^{-2-i},\qquad \forall i \geq 0.
	\end{equation}
	\item When $n \geq 3$ and odd:
		\begin{enumerate}
			\item Every Ricci coefficient $\psi$ and curvature component $\Psi$ not equal to $\alpha$ 					satisfy~\eqref{arequirement}.
			\item For all $0 \leq j < \frac{n-3}{2}$ we have 
				\[\left|\nabla^i\nabla_4^j\alpha\right||_{v=0} \lesssim_i |u|^{-2-j-i},\qquad \forall i \geq 0.\]
			\item There exists a tracefree and symmetric $\mathcal{S}_{u,v}$ two-tensor $h_{AB}$ 				which is independent of $u$ and $v$ and  such that
					\[\lim_{v\to 0}\left|\nabla^i\left(\nabla_4^{\frac{n-3}{2}}\alpha - h v^{-1/2}u^{1/2}\right)\right| 								\lesssim_i |u|^{-2-i-\frac{n-3}{2}},\qquad \forall i \geq 0.\]
		 \end{enumerate}
	\item When $n \geq 4$ and even: 
	\begin{enumerate}
			\item Every Ricci coefficient $\psi$ and curvature component $\Psi$ not equal to $\alpha$ 					satisfy~\eqref{arequirement}.
			\item For all $0 \leq j < \frac{n-4}{2}$ we have 
				\[\left|\nabla^i\nabla_4^j\alpha\right||_{v=0} \lesssim_i |u|^{-2-i-j},\qquad \forall i\geq 0.\]
			\item There exists a tracefree and symmetric $\mathcal{S}_{u,v}$ two-tensor $\mathcal{O}_{AB}$  which is independent of $u$ and $v$ and  such that
					\[\lim_{v\to 0}\left|\nabla^i\left(\nabla_4^{\frac{n-4}{2}}\alpha - \mathcal{O}_{AB}\log\left(\frac{v}{|u|}\right)\right)\right| 								\lesssim_i |u|^{-2-i-\frac{n-4}{2}}, \qquad \forall i \geq 0.\]
		 \end{enumerate}
	
\end{enumerate}
\end{definition}

The goal of the remainder of this section will be to classify compatible regular conjugate data and provide some estimates along $\{v = 0\}$ of the corresponding metrics.

We start with the simplest case, $n = 2$.
\begin{proposition}\label{incomingdatan2}When $n = 2$ any choice of conjugate data is compatible. 

Furthermore, if we define a symmetric traceless $\mathcal{S}_{u,v}$ tensor $h_{AB}$ to be independent of $u$ and $v$ and to satisfy 
\[h_{AB}|_{\{(u,v) = (-1,0)\}} = \hat{\chi}_{AB}|_{\{(u,v) = (-1,0)\}},\]
then we will have the following behavior along $\{v = 0\}$ for any metric $g$ induced by Theorem~\ref{localprotoambientyay}:

\[\omega = 0,\qquad \underline{\omega} = 0,\qquad {\rm tr}\chi = -\frac{2k_0}{u},\qquad \hat{\chi}_{AB} = -uh_{AB},\]
\[ {\rm tr}\underline{\chi} = \frac{2}{u},\qquad \eta = 0,\qquad \underline{\eta} = 0,\qquad \hat{\underline{\chi}} = 0,\qquad \zeta = 0,\]
\[\alpha_{AB} = -\left(\nabla \hat{\otimes} \left[{\rm div}_0\left(h\right) - \nabla k\right]\right)_{AB} + uf_{AB},\qquad \beta_A = u^{-1}\left[{\rm div}_0\left(h\right)_A - \nabla_Ak\right],\]
\[\rho = 0,\qquad \sigma = 0,\qquad \underline{\beta} = 0,\qquad \underline{\alpha} = 0.\]

Here $f_{AB}$ is a smooth tensor independent on $u$, $k_0$ denotes the Gaussian curvature of $\slashed{g}_0$, the ${\rm div}_0$ refers to the divergence operator of $\slashed{g}_0$. Finally, we note that since $n =2$, $\nu$ and $\underline{\nu}$ may be easily recovered from $\beta$ and $\underline{\beta}$ and that $\hat{\tau}$ vanishes identically. 

\end{proposition}
\begin{proof}Throughout this proof, unless noted otherwise, all quantities are considered to be evaluated along the cone $\{v = 0\}$.

First of all, since $\Omega$ is identically $1$ along $\{v = 0\}$, we in particular obtain
\[\underline{\omega} = 0.\]

Next, since we have by assumption
\[\slashed{g}\left(u,\theta\right) = u^2\slashed{g}_0\left(\theta\right).\]

A straightforward calculation then yields
\[\underline{\chi} = u\slashed{g}_0\left(\theta\right) \Rightarrow \]
\[{\rm tr}\underline{\chi} = \frac{2}{u},\qquad \hat{\underline{\chi}} = 0.\]

Then the $\nabla_3$ equation for $\hat{\underline{\chi}}$ immediately implies that
\[\underline{\alpha} = 0.\]

Next, we turn to $\zeta$. Using that along $\{v = 0\}$ we have $\eta = \zeta$ and $\underline{\eta} = -\zeta$, the $\nabla_3$ equation for $\underline{\eta}$ becomes
\begin{equation}\label{forundbeta}
\nabla_3\zeta_A + \frac{2}{u}\zeta_A = -\underline{\beta}_A.
\end{equation}

Next, we can use the constraint equation~\eqref{tcod2} to obtain

\[\underline{\beta}_A = \frac{1}{u}\zeta_A.\]

Combining the two equations yields
\[\nabla_3\zeta_A + \frac{3}{u}\zeta_A = 0.\]
Since $\zeta$ vanishes when $u = -1$ we conclude that
\[\zeta = \eta = \underline{\eta} = 0.\]

In turn,~\eqref{forundbeta} yields
\[\underline{\beta} = 0.\]

The constraint equation~\eqref{antisig} immediately implies
\[\sigma = 0.\]

Next we turn to ${\rm tr}\chi$. We have
\[\nabla_3{\rm tr}\chi + \frac{1}{u}{\rm tr}\chi = 2\rho.\]

From~\eqref{slashr} we have
\[\rho = -\frac{k_0}{u^2} - \frac{1}{2u}{\rm tr}\chi,\]
where $k_0$ is the Gaussian curvature of $\slashed{g}_0$.

Combining the two equations yields 
\[\nabla_3{\rm tr}\chi + \frac{2}{u}{\rm tr}\chi = -\frac{2k_0}{u^2}.\]
Recall that we already have specified that 
\[{\rm tr}\chi|_{(u,v) = (-1,0)} = 2k_0.\]
Solving the o.d.e. along $\{v = 0\}$ then yields
\[{\rm tr}\chi = -\frac{2k_0}{u}.\]

In turn, we then obtain 
\[\rho = 0.\]

Next, the $\nabla_3$ equation for $\omega$ now implies that $\omega$ is constant. Since $\omega$ vanishes when $u = -1$, we obtain
\[\omega = 0.\]

Keeping in mind that Lemma~\ref{simpn2} implies that
\[\hat{\tau} = 0,\]

we obtain for $\hat{\chi}$ that
\[\nabla_3\hat{\chi} + \frac{1}{u}\hat{\chi} = 0.\]

We immediately obtain
\[\hat{\chi} = -uh.\]

Next, the traced Codazzi equation~\eqref{tcod1} implies
\begin{equation}\label{betais}
\beta_A = u^{-1}\left[{\rm div}_0\left(h\right)_A - \nabla_Ak\right],
\end{equation}
where the divergence is with respect to $\slashed{g}_0$.

This formula for $\beta$ also determines $\nu$ via Lemma~\ref{simpn2}:
\[\nu_{ABC} = -\slashed{\epsilon}_{AB}\left(*\beta\right)_C.\]

Lastly, we come to $\alpha$. This satisfies the following equation:
\[\nabla_3\alpha + \frac{1}{u}\alpha = \nabla \hat{\otimes}\beta + \mathscr{E}_1.\]

Direct consideration of the possible terms in $\mathscr{E}_1$ easily implies
\[\mathscr{E}_1 = 0,\]
hence,
\[\nabla_3\alpha + \frac{1}{u}\alpha = \nabla \hat{\otimes}\beta.\]

Writing this out in a Lie-propagated frame yields the equation
\[\partial_u\left(\alpha_{AB}\right) - \frac{1}{u}\alpha_{AB} =u^{-1}\left(\nabla \hat{\otimes} \left[{\rm div}_0\left(h\right)- \nabla k\right]\right)_{AB} \Leftrightarrow \]
\[\partial_u\left(u^{-1}\alpha_{AB}\right) = \Omega_0u^{-2}\left(\nabla \hat{\otimes} \left[{\rm div}_0\left(h\right)- \nabla k\right]\right)_{AB}.\]

The unique solution to this o.d.e. is
\[\alpha = -\left(\nabla \hat{\otimes} \left[{\rm div}_0\left(h\right)- \nabla k\right]\right) + uf_{AB},\]
for some symmetric traceless $2$-tensor $f_{AB}$ depending on $h$, $\nabla k$, and $\alpha|_{(u,v) = (-1,0)}$.

Note that the estimates needed to conclude that $\hat{\slashed{g}}(v)$ is compatible, i.e.~\eqref{arequirement},  follow immediately form these specific formulas.
\end{proof}

We now turn to the analogous proposition for $n \geq 3$ and odd. Here we will not provide exact formulas (since the full set of prescribed values gets more and more complicated as the dimension increases) but, in addition to determining when $\hat{\slashed{g}}(v)$ is compatible, we are also interested in finding which double null unknowns will always vanish on $\{v = 0\}$.

\begin{proposition}\label{incomingdatanodd}Let $n \geq 3$ and odd. Then there exists symmetric $2$-tensors on $\mathcal{S}$, $\left\{\hat{\slashed{g}}^{(i)}\right\}_{i=1}^{\frac{n-1}{2}}$, such that $\hat{\slashed{g}}(v)$ is compatible if and only if
\[\mathcal{L}_v^i\hat{\slashed{g}}|_{v=0} = \hat{\slashed{g}}^{(i)},\qquad \forall i = 1,\cdots,\frac{n-1}{2}.\]

We furthermore have that the following Ricci coefficients and null curvature components vanish along $\{v = 0\}$: 
\[\omega,\qquad \underline{\omega},\qquad \eta,\qquad \underline{\eta},\qquad \hat{\underline{\chi}},\qquad \zeta,\qquad  \rho, \qquad  \sigma,\qquad \tau,\qquad \underline{\beta},\qquad \underline{\nu},\qquad \underline{\alpha}.  \]

\end{proposition}
\begin{proof}
We proceed in an analogous fashion to the case when $n = 2$. First of all, we must have
\[\underline{\omega} = 0,\qquad \slashed{g}\left(u,\theta\right) = u^2\slashed{g}_0\left(u,\theta\right),\]
which in turn implies
\[\hat{\underline{\chi}} = 0,\qquad {\rm tr}\underline\chi = \Omega^{-1}\frac{n}{u}.\]

Next, arguing just as in the $n=2$ case we easily obtain that 
\[\underline{\alpha} = \underline{\beta} = \eta = \underline{\eta} = \zeta = 0.\]

The constraint equation~\eqref{cod2} then immediately implies that
\[\underline{\nu} = 0.\]

Arguing just as in $n=2$ we also find that
\[\sigma = 0,\qquad \rho = 0,\qquad {\rm tr}\chi = -\frac{\slashed{R}_0}{n-1}u^{-1}. \]

Now we turn to $\tau$. We have
\[\nabla_3\hat{\tau} + \left(\frac{n}{2}+1\right)u^{-1}\hat{\tau} = \mathscr{E}^{(3)}_2.\]
Signature considerations allow one to deduce that $\mathscr{E}^{(3)}_2 = 0$. Thus

\[\nabla_3\hat{\tau} + \left(\frac{n}{2}+1\right)u^{-1}\hat{\tau} = 0.\]

Since $n > 2$ the only way we hope to have that $\left|\tau\right| \lesssim u^{-2}$ is for $\tau$ to vanish identically. Of course, this happens if and only if $\hat{\tau}$ vanishes when $u = -1$. The constraint equation~\eqref{slashric} yields
\[\hat{\tau} = \hat{\slashed{Ric}_0} + \left(\frac{n}{2}-1\right)\hat{\chi},\]
where $\hat{\slashed{Ric}}$ denotes the trace-free part of the Ricci tensor of $\slashed{g}_0$.

Thus, we see that in order for $\hat{\slashed{g}}(v)$ to be compatible, it must be the case that
\begin{equation}\label{thismusthold}
 \hat{\slashed{Ric}} + \left(\frac{n}{2}-1\right)\hat{\chi} = 0.
 \end{equation}
Remark~\ref{somestuff} implies that is equivalent to $\mathcal{L}_v\hat{\slashed{g}}|_{v=0}$ being prescribed. Finally, the Gauss equation determines $R_{ABCD}$ and we see that it satisfies the desired self-similar bound.

Now we proceed under the assumption that $\mathcal{L}_v\hat{\slashed{g}}$ has the correct value so that~\eqref{thismusthold} is true.

Just as for $n= 2$, we easily obtain that
\[\omega = 0.\]

Next,~\eqref{cod1} and~\eqref{tcod1} determine $\beta$ and $\nu$. (One can show that $\beta$ in fact vanishes, but this won't matter for us.)

Finally, we come to $\alpha$. First we specialize to $n = 3$. In this case, our definition of regular allows for $\alpha$ in principle to not extend continuously to $\{v = 0\}$ and instead have a term which blows up like $v^{-1/2}$ in its Taylor expansion. Signature considerations and the previous estimates yield that 
\[\mathscr{E}^{(3)}_1 = O\left(v^{1/2}\right)\qquad \text{as }v\to 0,\]
on any compact set of $u \in [-1,0)$. In particular, for each $v > 0$ we can write $\alpha$'s Bianchi equation as
\begin{equation}\label{awaytowritealphaeqn}
\nabla_3\alpha_{AB} + \frac{3}{2}u^{-1}\alpha_{AB} = f_{AB}u^{-1} + O\left(v^{1/2}\right),
\end{equation}
where the $O$ holds for $u$ in any compact set $[-1,0)$.

Next, along $\{u = -1\}$, the regularity assumption implies that we can Taylor expand $\alpha$ as follows:
\[\alpha_{AB}|_{u=-1} = h_{AB}v^{-1/2} + d_{AB} + O\left(v^{1/2}\right).\]
After extending $h$ to be independent of $u$, we may integrate~\eqref{awaytowritealphaeqn} from $-1$ to $u$ and obtain
\[\lim_{v\to 0}\left(\alpha - hv^{-1/2}u^{1/2}\right) = \tilde f^{(1)} + \tilde f^{(2)} u^{1/2} ,\]
for tensors $\tilde f^{(1)}$ and $\tilde f^{(2)}$ independent of $u$. This concludes the proof for $n = 3$.

Now, if $n$ is odd and $n > 3$, then, since our metric is regular $\alpha$ will continuously extend to $\{v = 0\}$. Along $\{v = 0\}$ we will have 
\[\nabla_3\alpha_{AB} + \frac{n}{2}u^{-1}\alpha_{AB} = l_{AB}u^{-1},\]
for some tensor $l_{AB}$ independent of $u$ and $v$ and which is given explicitly in terms of $\slashed{g}_0$. Integrating this o.d.e. yields
\[\alpha_{AB}(u) = \tilde l^{(1)}_{AB} + \tilde l^{(2)}_{AB}u^{\frac{4-n}{2}},\]
for tensors $\tilde l^{(1)}$ and $\tilde l^{(2)}$ both independent of $u$ and $v$. Furthermore, there is a unique choice of $\alpha_{AB}|_{u=-1}$ which will make $\tilde l^{(2)}$ vanish. Of course, we must make this choice if the conjugate data is to be compatible. Using Remark~\ref{specifyalpha} this forces the $\mathcal{L}^2_v\hat{\slashed{g}}|_{v=0}$ to take a specific value.

The next step in the analysis requires us to examine $\nabla_4\alpha$. (Of course, we expect to have to argue in a special way for $n = 5$.) We may derive an equation for $\nabla_3\left(\nabla_4\alpha\right)$ by commuting the Bianchi equation for $\alpha$ with $\nabla_4$ and using the commutation formula from Lemma~\ref{34commute}. However, in order to use this equation, we need to already have expressions for all of the terms on the right hand side. (We also need to produce estimates for all of the other $\nabla_4\Psi$'s and $\nabla_4\psi$'s.) Fortunately, the Bianchi equations immediately give expressions for $\nabla_4\Psi$ for any curvature component $\Psi$ which is not equal to $\alpha$. It immediately follows from the form of these equations that $\nabla_4\Psi$ satisfies self-similar bounds. Similarly, the null structure equations determine $\nabla_4\psi$ for all Ricci coefficients except for $\underline{\eta}$ and $\omega$ and we see that all such $\nabla_4\psi$ satisfy self-similar bounds. Next, for $\underline{\eta}$ and $\omega$ we first note that it suffices to estimate $\zeta$ and $\omega$ and then observe that $\nabla_4\zeta$ can be computed by differentiation of the constraint equation for $\underline{\beta}$. For $\omega$, one commutes its $\nabla_3$ equation with $\nabla_4$ to derive and an equation of the form $\nabla_3\left(\nabla_4\omega\right) = F$ where we will have an explicit expression for $F$. In particular, $F$ will not contain $\nabla_4\alpha$ or $\nabla_4\omega\cdot\psi$ and all terms in $F$ will satisfy self-similar bounds. Then we can integrate from $u = -1$ where we know that $\nabla_4\omega$ vanishes to determine $\nabla_4\omega$ everywhere and see that $\nabla_4\omega$ satisfies the self-similar bounds. Note that one consequence of this analysis is that for every curvature component $\Psi$ not equal to $\alpha$ and Ricci coefficient $\psi$ we have
\[\left|\nabla^j\nabla_4\Psi\right||_{v=0} \lesssim |u|^{-2-j-i},\qquad \left|\nabla^j\nabla_4\psi\right||_{v=0} \lesssim |u|^{-1-j-i},\]
where $j$ is arbitrary.

Finally, we will have an equation for $\nabla_4\alpha$ of the form
\[\nabla_3\left(\nabla_4\alpha\right) + \frac{n}{2}u^{-1}\nabla_4\alpha =\  {\rm known},\]
where we furthermore have that $\left|{\rm known}\right| \lesssim |u|^{-4}$.

At this point, if $n = 5$ we can argue exactly as we did before for $n = 3$, or if $n > 5$ we can argue as we did before for $n > 3$. Continuing to commute with $\nabla_4$ and repeating the analysis above \emph{mutatis mutandis} eventually concludes the proof for all $n$.

\end{proof}

Finally, we consider the case when $n \geq 4$ and even.
\begin{proposition}\label{incomingdataneven}Let $n \geq 4$ and even. Then there exists symmetric $2$-tensors on $\mathcal{S}$, $\left\{\hat{\slashed{g}}^{(i)}\right\}_{i=1}^{\frac{n}{2}-1}$ and $\mathcal{O}$, such that $\hat{\slashed{g}}(v)$ is compatible if and only if in a Lie-propagated from
\[\hat{\slashed{g}}(v)  = \hat{\slashed{g}}^{(0)} + v\hat{\slashed{g}}^{(1)} + \cdots + \frac{v^{\frac{n}{2}-1}}{\left(\frac{n}{2}-1\right)!}\hat{\slashed{g}}^{\left(\frac{n}{2}-1\right)} + v^{\frac{n}{2}}\log\left(v\right)\mathcal{O} + O\left(v^{\frac{n}{2}}\right).\]

We furthermore have that the following Ricci coefficients and null curvature components vanish along $\{v = 0\}$: 
\[\omega,\qquad \underline{\omega},\qquad \eta,\qquad \underline{\eta},\qquad \hat{\underline{\chi}},\qquad \zeta,\qquad  \rho, \qquad  \sigma,\qquad \tau,\qquad \underline{\beta},\qquad \underline{\nu},\qquad \underline{\alpha}.  \]


\end{proposition}

\begin{proof}We start just as in the proof of Proposition~\ref{incomingdatanodd}. It is easy to see that everything goes through until the analysis of $\alpha$. First we specialize to the case when $n =4$. In this case $\alpha$ may blow-up logarithmically as $v \to 0$. The equation for $\alpha$ may be written as
\begin{equation}\label{thisiswhatweintegrate}
\nabla_3\alpha_{AB} + \frac{2}{u}\alpha_{AB} = f_{AB}u^{-1} + O\left(v\log(v)\right),
\end{equation}
where the big $O$ error estimate holds over any compact region of $u$ and $f_{AB}$ is as in Proposition~\ref{incomingdatanodd}. Let $\mathring{\alpha}_{AB}$ be the tensor defined by extending $\alpha|_{u=-1}$ all $u$ and $v$ by making it independent of $u$. For $v > 0$, integrating this equation~\eqref{thisiswhatweintegrate} from $u = -1$ yields
\[\alpha_{AB} = \tilde f_{AB} \log(u) + \mathring{\alpha}_{AB} + O\left(v\log(v)\right).\]
It is now clear that our conjugate data will be compatible if and only if $\mathring{\alpha}_{AB}$ has a specific logarithmic singularity as $v\to 0$ determined by the tensor $\tilde f_{AB}$. Using Remark~\ref{specifyalpha} this concludes the proof when $n = 4$.

For $n > 4$ one argues as in Proposition~\ref{incomingdatanodd} with further commutations of $\nabla_4$. We omit the details.

\end{proof}


\section{Norms, Renormalizations, and the Commuted Equations}\label{normssection}

In this section we will introduce the various scale-invariant norms we will use for our a priori estimates. Before we jump into the definitions, we introduce two important small parameters $\epsilon,\delta > 0$. Our a priori estimates will all cover the region where $\frac{v}{|u|} \leq \epsilon$. The parameter $\delta$ will show up in various norms below; its presence is used to avoid various logarithmic divergences. 
\subsection{Renormalizations}\label{renormalizationswoo}

We start with some notation which we will use for various renormalizations. 

It is often useful to subtract off the values of various tensors along $\{v = 0\}$:
\begin{definition}\label{thisistilde}For any $\mathcal{S}_{u,v}$ tensor $\Theta$, we define
\[\reallywidetilde{\Theta}|_{\mathcal{S}_{u,v}} \doteq \Theta|_{\mathcal{S}_{u,v}} - \Theta|_{\mathcal{S}_{u,0}},\]
where $\mathcal{S}_{u,v}$ and $\mathcal{S}_{u,0}$ are identified via their canonical coordinate systems.
\end{definition}

When $n \geq 3$ and odd, then the most singular term we will be confronted with is $\alpha$. It turns out to be useful to work with a quantity $\alpha'$ where the most singular part of $\alpha$ has been subtracted off.
\begin{definition}Let $n \geq 3$ and odd. Given a solution arising from compatible regular conjugate data, we define
\[\alpha'_{AB} \doteq \alpha_{AB} - \frac{1}{\left(\frac{n-3}{2}\right)!}v^{\frac{n-4}{2}}u^{\frac{4-n}{2}}h_{AB}.\]
See Definition~\ref{defcompatible}.
\end{definition}

Next we have the analogous normalizations for $n \geq 4$ and even. 

\begin{definition}Let $n \geq 4$ and even. Given a solution arising from compatible regular conjugate data, we define
\[\alpha'_{AB} \doteq \alpha_{AB} - \frac{1}{\left(\frac{n-4}{2}\right)!}v^{\frac{n-4}{2}}u^{\frac{4-n}{2}}\log\left(\frac{v}{-u}\right)\mathcal{O}_{AB}.\]
See Definition~\ref{defcompatible}.
\end{definition}

\begin{remark}When $n=2$ we will not see any singular behavior as $v\to 0$ and hence we do not need to renormalize $\alpha$.
\end{remark}

Finally, it will be useful to explicitly introduce some notation for a curvature component which, if equal to $\alpha$, has been renormalized.

\begin{definition}We will use the notation $\Psi'$ to refer to a generic curvature component which, if equal to $\alpha$, has been renormalized.
\end{definition}

\subsection{Top Level Energy Norms for Curvature}
We start with the energy norms for curvature. Choosing this norm correctly is the most subtle part of the entire set-up and is different for the case of $n$ odd, $n \geq 4$ and even, and $n = 2$. See the discussion in the introduction.

Before we give the full definition of the energy norms, we introduce the number $N = N(n)$ which can be taken to only depend on the dimension $n$ of $\mathcal{S}$ and will denote the total number of angular derivatives that we will apply to our curvature components in the top energy norm. The main requirement for $N$ is that it is large enough so that the corresponding Sobolev space $H^N\left(\mathcal{S}\right)$ forms an algebra, but we will make no effort to optimize the choice of $N$. From this point we take $N$ to be fixed and sufficiently large.

Next, we introduce the notation $\mathcal{R}_{\tilde u,\tilde v}$ to refer to the characteristic rectangle $[-1,\tilde u] \times [0,\tilde v]$. One peculiar feature of our norms is that they will be stated with respect to an arbitrary characterstic rectangle which fits in the region $\{v \leq \epsilon |u|\}$: 

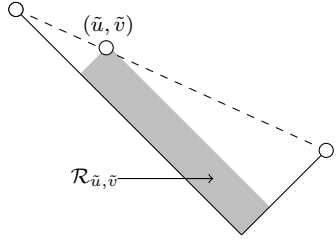
\begin{center}
\begin{tikzpicture}[scale = 1.5]
\fill[lightgray] (2,-2)--(2.25,-1.75) --  (.82,-.32)--(.57,-.57); 

\draw (0,0) -- (2,-2);
\draw (2,-2) -- (2.75,-1.25);
\path [draw=black,fill=white] (0,0) circle (1/16); 

\draw [dashed] (0,0)--(2.75,-1.25);
\path [draw=black,fill=white] (0,0) circle (1/16); 
\path [draw=black,fill=white] (2.75,-1.25) circle (1/16); 
\path [draw=black,fill=white] (.8,-.34) circle (1/16); 
\draw (.82,-.32) node[above]{\footnotesize $(\tilde u,\tilde v)$};
\draw [->]  (.9,-1.5) -- (1.75,-1.5);
\draw (1,-1.5) node[left]{\footnotesize $\mathcal{R}_{\tilde u,\tilde v}$};



\end{tikzpicture}
\end{center}

Finally, we emphasize the Convention~\ref{volumeformconvention} that all integration over any $\mathcal{S}_{u,v}$ is with respect to the volume form of $\slashed{g}_0$.

\begin{definition}(Odd Dimensions) Suppose that $n$ is odd. Let $\Psi$ be a curvature component not equal to $\alpha$ or $\underline{\alpha}$, and $(\tilde u,\tilde v)$ satisfy $\frac{\tilde v}{|\tilde u|} \leq \epsilon$. We define the top order energy norm by
\begin{align*}
\left\vert\left\vert\Psi\right\vert\right\vert_{\mathfrak{T}_{\tilde u,\tilde v}}^2 \doteq \sup_{0\leq j \leq N}\sup_{(u_0,v_0)\in \mathcal{R}_{\tilde u,\tilde v}}\Bigg[&\int_{-1}^{u_0}\int_{\mathcal{S}}\left|\reallywidetilde{\nabla^j\nabla_4^{\frac{n-3}{2}}\Psi}\right|^2u^{n+1-2\delta+2j}v^{-1+2\delta}
\\ \nonumber &+\int_0^{v_0}\int_{\mathcal{S}}\left|\reallywidetilde{\nabla^j\nabla_4^{\frac{n-3}{2}}\Psi}\right|^2u^{n+1-2\delta+2j}v^{-1+2\delta}
\\ \nonumber &+\int_{-1}^{u_0}\int_0^{v_0}\int_{\mathcal{S}}\left|\reallywidetilde{\nabla^j\nabla_4^{\frac{n-3}{2}}\Psi}\right|^2u^{n+1-2\delta+2j}v^{-2+2\delta}\Bigg].
\end{align*}

For $\alpha'$, we drop the $u$-flux:
\begin{align*}
\left\vert\left\vert\alpha'\right\vert\right\vert_{\mathfrak{T}_{\tilde u,\tilde v}}^2 \doteq \sup_{0\leq j \leq N}\sup_{(u_0,v_0)\in\mathcal{R}_{\tilde u, \tilde v}}\Bigg[&\int_0^{v_0}\int_{\mathcal{S}}\left|\reallywidetilde{\nabla^j\nabla_4^{\frac{n-3}{2}}\alpha'}\right|^2u^{n+1-2\delta+2j}v^{-1+2\delta}
\\ \nonumber &+\int_{-1}^{u_0}\int_0^{v_0}\int_{\mathcal{S}}\left|\reallywidetilde{\nabla^j\nabla_4^{\frac{n-3}{2}}\alpha'}\right|^2u^{n-2\delta+2j}v^{-1+2\delta}\Bigg].
\end{align*}

For $\underline{\alpha}$, we drop the $v$-flux:
\begin{align*}
\left\vert\left\vert\underline{\alpha}\right\vert\right\vert_{\mathfrak{T}_{\tilde u,\tilde v}}^2 \doteq \sup_{0\leq j \leq N} \sup_{(u_0,v_0)\in \mathcal{R}_{\tilde u, \tilde v}}\Bigg[&\int_{-1}^{u_0}\int_{\mathcal{S}}\left|\reallywidetilde{\nabla^j\nabla_4^{\frac{n-3}{2}}\underline{\alpha}}\right|^2u^{n+1-2\delta+2j}v^{-1+2\delta}
\\ \nonumber &+\int_{-1}^{u_0}\int_0^{v_0}\int_{\mathcal{S}}\left|\reallywidetilde{\nabla^j\nabla_4^{\frac{n-3}{2}}\underline{\alpha}}\right|^2u^{n+1-2\delta+2j}v^{-2+2\delta}\Bigg].
\end{align*}
\end{definition}

For large even dimensions we have similar definition.

\begin{definition}(Large Even Dimensions) Suppose that $n$ is even and $n \geq 4$. Let $\Psi$ be a curvature component not equal to $\alpha'$ or $\underline{\alpha}$, and $(\tilde u, \tilde v)$ satisfy $\frac{\tilde v}{|\tilde u|} \leq \epsilon$. We define the top order energy norm by
\begin{align*}
\left\vert\left\vert\Psi\right\vert\right\vert_{\mathfrak{T}_{\tilde u,\tilde v}}^2 \doteq \sup_{0\leq j \leq N} \sup_{(u_0,v_0) \in \mathcal{R}_{\tilde u, \tilde v}}\Bigg[&\int_{-1}^{u_0}\int_{\mathcal{S}}\left|\reallywidetilde{\nabla^j\nabla_4^{\frac{n-4}{2}}\Psi}\right|^2u^{n+2j}v^{-1}
\\ \nonumber &+v_0^{-2\delta}\int_0^{v_0}\int_{\mathcal{S}}\left|\reallywidetilde{\nabla^j\nabla_4^{\frac{n-4}{2}}\Psi}\right|^2u^{n+2j}v^{-1+2\delta}
\\ \nonumber &+v_0^{-2\delta}\int_{-1}^{u_0}\int_0^{v_0}\int_{\mathcal{S}}\left|\reallywidetilde{\nabla^j\nabla_4^{\frac{n-4}{2}}\Psi}\right|^2u^{n+2j}v^{-2+2\delta}\Bigg].
\end{align*}

For $\alpha'$, we drop the $u$-flux and the spacetime term:
\begin{align*}
\left\vert\left\vert\alpha'\right\vert\right\vert_{\mathfrak{T}_{\tilde u,\tilde v}}^2 \doteq \sup_{0\leq j \leq N}\sup_{(u_0,v_0)\in\mathcal{R}_{\tilde u, \tilde v}}v_0^{-2\delta}\int_0^{v_0}\int_{\mathcal{S}}\left|\reallywidetilde{\nabla^j\nabla_4^{\frac{n-4}{2}}\alpha'}\right|^2u^{n+2j}v^{-1+2\delta}.
\end{align*}

For $\underline{\alpha}$, we drop the $v$-flux:
\begin{align*}
\left\vert\left\vert\underline{\alpha}\right\vert\right\vert_{\mathfrak{T}_{\tilde u,\tilde v}}^2 \doteq \sup_{0\leq j \leq N}\sup_{(u_0,v_0)\in \mathcal{R}_{\tilde u, \tilde v}}\Bigg[&\int_{-1}^{u_0}\int_{\mathcal{S}}\left|\reallywidetilde{\nabla^j\nabla_4^{\frac{n-4}{2}}\underline{\alpha}}\right|^2u^{n+2j}v^{-1}
\\ \nonumber &+v_0^{-2\delta}\int_{-1}^{u_0}\int_0^{v_0}\int_{\mathcal{S}}\left|\reallywidetilde{\nabla^j\nabla_4^{\frac{n-4}{2}}\underline{\alpha}}\right|^2u^{n+2j}v^{-2+2\delta}\Bigg].
\end{align*}
\end{definition}

Finally we have the case of $n = 2$.

\begin{definition}($n=2$) Suppose that $n = 2$. Let $\Psi$ be a curvature component not equal to $\alpha$ or $\underline{\alpha}$, and $(\tilde u,\tilde v)$ satisfy $\frac{\tilde v}{|\tilde u|} \leq \epsilon$. We define the top order energy norm by
\begin{align*}
\left\vert\left\vert\Psi\right\vert\right\vert_{\mathfrak{T}_{\tilde u,\tilde v}}^2 \doteq \sup_{0\leq j \leq N}\sup_{(u_0,v_0)\in \mathcal{R}_{\tilde u,\tilde v}}\Bigg[&\int_{-1}^{u_0}\int_{\mathcal{S}}\left|\reallywidetilde{\nabla^j\Psi}\right|^2u^{4-2\delta+2j}v^{-1+2\delta}
\\ \nonumber &+\int_0^{v_0}\int_{\mathcal{S}}\left|\reallywidetilde{\nabla^j\Psi}\right|^2u^{4-2\delta+2j}v^{-1+2\delta}
\\ \nonumber &+\int_{-1}^{u_0}\int_0^{v_0}\int_{\mathcal{S}}\left|\reallywidetilde{\nabla^j\Psi}\right|^2u^{4-2\delta+2j}v^{-2+2\delta}\Bigg].
\end{align*}

For $\alpha$, we drop the $u$-flux:
\begin{align*}
\left\vert\left\vert\alpha'\right\vert\right\vert_{\mathfrak{T}_{\tilde u,\tilde v}}^2 \doteq \sup_{0\leq j \leq N}\sup_{(u_0,v_0)\in\mathcal{R}_{\tilde u, \tilde v}}\Bigg[&\int_0^{v_0}\int_{\mathcal{S}}\left|\reallywidetilde{\nabla^j\nabla^j\alpha'}\right|^2u^{4-2\delta+2j}v^{-1+2\delta}
\\ \nonumber &+\int_{-1}^{u_0}\int_0^{v_0}\int_{\mathcal{S}}\left|\reallywidetilde{\nabla^j\nabla_4^{\frac{n-3}{2}}\alpha'}\right|^2u^{3-2\delta+2j}v^{-1+2\delta}\Bigg].
\end{align*}

For $\underline{\alpha}$, we drop the $v$-flux:
\begin{align*}
\left\vert\left\vert\underline{\alpha}\right\vert\right\vert_{\mathfrak{T}_{\tilde u,\tilde v}}^2 \doteq \sup_{0\leq j \leq N} \sup_{(u_0,v_0)\in \mathcal{R}_{\tilde u, \tilde v}}\Bigg[&\int_{-1}^{u_0}\int_{\mathcal{S}}\left|\reallywidetilde{\nabla^j\underline{\alpha}}\right|^2u^{4-2\delta+2j}v^{-1+2\delta}
\\ \nonumber &+\int_{-1}^{u_0}\int_0^{v_0}\int_{\mathcal{S}}\left|\reallywidetilde{\nabla^j\underline{\alpha}}\right|^2u^{4-2\delta+2j}v^{-2+2\delta}\Bigg].
\end{align*}
\end{definition}

\subsection{Lower Energy Norms for Curvature}
This next set of norms are $L^2$ norms for curvature components which involve less that the highest possible number of $\nabla_4$ derivatives. The basic rule is that if one takes $i$ less $\nabla_4$ derivatives then one is allowed to take $i$ more angular derivatives.

\begin{definition}(Odd Dimensions) Suppose that $n\geq 5$ is odd. Let $\Psi$ be a curvature component not equal to $\alpha$ or $\underline{\alpha}$, and $(\tilde u,\tilde v)$ satisfy $\frac{\tilde v}{|\tilde u|} \leq \epsilon$. We define the lower order energy norm by
\begin{align*}
\left\vert\left\vert\Psi\right\vert\right\vert_{\mathfrak{L}_{\tilde u,\tilde v}}^2 \doteq \sup_{1\leq i \leq \frac{n-3}{2}} \sup_{0\leq k \leq N+i}\sup_{(u_0,v_0) \in \mathcal{R}_{\tilde u,\tilde v}}\Bigg[&\int_{-1}^{u_0}\int_{\mathcal{S}}\left|\reallywidetilde{\nabla^k\nabla_4^{\frac{n-3}{2}-i}\Psi}\right|^2u^{n+1-2\delta+2(k-i)}v^{-1+2\delta}
\\ \nonumber &+\int_0^{v_0}\int_{\mathcal{S}}\left|\reallywidetilde{\nabla^k\nabla_4^{\frac{n-3}{2}-i}\Psi}\right|^2u^{n+1-2\delta+2(k-i)}v^{-1+2\delta}
\\ \nonumber &+\int_{-1}^{u_0}\int_0^{v_0}\int_{\mathcal{S}}\left|\reallywidetilde{\nabla^k\nabla_4^{\frac{n-3}{2}-i}\Psi}\right|^2u^{n+1-2\delta+2(k-i)}v^{-2+2\delta}\Bigg].
\end{align*}

For $\left\vert\left\vert \alpha'\right\vert\right\vert_{\mathfrak{L}_{\tilde u,\tilde v}}^2$, the spacetime weight improves:
\begin{align*}
\left\vert\left\vert\alpha'\right\vert\right\vert_{\mathfrak{T}_{\tilde u,\tilde v}}^2 \doteq \sup_{1\leq i \leq \frac{n-3}{2}} \sup_{0\leq k\leq N+i}\sup_{(u_0,v_0)\in\mathcal{R}_{\tilde u, \tilde v}}\Bigg[&\int_0^{v_0}\int_{\mathcal{S}}\left|\reallywidetilde{\nabla^k\nabla_4^{\frac{n-3}{2}-i}\alpha'}\right|^2u^{n+1-2\delta+2(k-i)}v^{-1+2\delta}
\\ \nonumber &+\int_{-1}^{u_0}\int_0^{v_0}\int_{\mathcal{S}}\left|\reallywidetilde{\nabla^k\nabla_4^{\frac{n-3}{2}-i}\alpha'}\right|^2u^{n+2(k-i)}v^{-1}\Bigg].
\end{align*}

 $\left\vert\left\vert \alpha'\right\vert\right\vert_{\mathfrak{L}_{\tilde u,\tilde v}}^2$ and $\left\vert\left\vert \underline\alpha\right\vert\right\vert_{\mathfrak{L}_{\tilde u,\tilde v}}^2$ are defined analogously.

\end{definition}

For large even dimensions the definition is completely analogous except that $\alpha'$ gains the spacetime term 
\[v_0^{-2\delta}\int_{-1}^{u_0}\int_0^{v_0}\int_{\mathcal{S}}\left|\reallywidetilde{\nabla^k\nabla_4^{\frac{n-3}{2}-i}\alpha'}\right|^2u^{n+2(k-i)-2\delta}v^{-2+4\delta}\]

For $n = 2,3,4$ one does not take any $\nabla_4$ derivatives of curvature in the top order energy norm and thus there is no lower order energy norm.

\begin{definition}We will use $\tilde N$ to denote the maximum number of angular derivatives that ever are applied in the above norms. 
\begin{enumerate}
	\item For $n=2$ we have $\tilde N = N$.

	\item For $n \geq 3$ and odd, we have $\tilde N = N + \frac{n-3}{2}$. 

	\item For $n \geq 4$ and even we have $\tilde N = N + \frac{n-4}{2}$. 

\end{enumerate}
\end{definition}

\subsection{$L^{\infty}_{u,v}$-Vanishing Norms for Curvature}

The following norms will quantify the fact that certain curvature components and their angular derivatives vanish as $\frac{v}{|u|} \to 0$. It will be necessary to allow for the loss of $1$ angular derivative. 

\begin{definition}Let $n \geq 2$. Let $\mathring{\Psi}$ denote one of
\[\rho,\qquad \sigma,\qquad \tau,\qquad \underline{\beta},\qquad \underline{\nu},\qquad \underline{\alpha}.\]
Let $(\tilde u,\tilde v)$ satisfy $\frac{\tilde v}{|\tilde u|} \leq \epsilon$.
 Then we define
\[\left\vert\left\vert \mathring\Psi\right\vert\right\vert_{\mathfrak{U}_{\tilde u,\tilde v}}^2 \doteq \sup_{(u,v) \in \mathcal{R}_{\tilde u,\tilde v}}\sup_{0\leq j \leq \tilde N-1}\int_{\mathcal{S}_{u,v}}\left|\nabla^j\Psi\right|^2u^{2j+6-4\delta}v^{-2+4\delta}.\]
\end{definition}

\subsection{$L^{\infty}_{u,v}$-Vanishing Norms for Ricci Coefficients} 
Lastly, it will also be important to quantify the vanishing of the following Ricci coefficients as $\frac{v}{|u|} \to 0$:  $\underline{\omega}$, $\hat{\underline{\chi}}$, ${\rm tr}\underline{\chi}'$, $\eta$, and $\underline{\eta}$.

\begin{definition}Let $n \geq 2$ and  $(\tilde u,\tilde v)$ satisfy $\frac{\tilde v}{|\tilde u|} \leq \epsilon$. Then, for $\mathring{\psi}$ denoting any of $\hat{\underline{\chi}}$, $\reallywidetilde{{\rm tr}\underline{\chi}}$, $\eta$, $\underline{\eta}$, or $\omega$ we set 

\[\left\vert\left\vert \underline{\omega}\right\vert\right\vert_{\mathfrak{V}_{\tilde u,\tilde v}}^2 \doteq \sup_{\{(u,v) \in \mathcal{R}_{\tilde u,\tilde v}\}}\sup_{0\leq j \leq \tilde N}\int_{\mathcal{S}_{u,v}}\left|\nabla^j\underline{\omega}\right|^2u^{2j+6-8\delta}v^{-4+8\delta},\]
\[\left\vert\left\vert \nabla_4\underline{\omega}\right\vert\right\vert_{\mathfrak{V}_{\tilde u,\tilde v}}^2 \doteq \sup_{\{(u,v) \in \mathcal{R}_{\tilde u,\tilde v}\}}\sup_{0\leq j \leq \tilde N-1}\int_{\mathcal{S}_{u,v}}\left|\nabla^j\nabla_4\underline{\omega}\right|^2u^{2j+4-4\delta}v^{-2+4\delta},\]

\[\left\vert\left\vert \mathring{\psi}\right\vert\right\vert_{\mathfrak{V}_{\tilde u,\tilde v}}^2 \doteq\sup_{\{(u,v) \in \mathcal{R}_{\tilde u,\tilde v}\}}\sup_{0\leq j \leq \tilde N}\int_{\mathcal{S}_{u,v}}\left|\nabla^j\mathring{\psi}\right|^2u^{2j+4-4\delta}v^{-2+4\delta}.\]

\end{definition}

\subsection{$L^{\infty}_{u,v}$-Norms for Ricci Coefficients}
Next we have $L^{\infty}_{u,v}$ norms for the Ricci coefficients and for $\slashed{Riem}_{ABCD}$, the Riemann tensor of $\slashed{g}$.

\begin{definition}\label{vanishnormricci}Let $n \geq 2$ and let $\psi$ denote any Ricci coefficient. Let $(\tilde u,\tilde v)$ satisfy $\frac{\tilde v}{|\tilde u|} \leq \epsilon$. Then we set

\[\left\vert\left\vert \psi\right\vert\right\vert_{\mathfrak{S}_{\tilde u,\tilde v}}^2 \doteq \sup_{(u,v) \in \mathcal{R}_{\tilde u,\tilde v}}\sup_{0\leq i \leq {\rm max}\left(\lfloor\frac{n-3}{2}\rfloor,0\right)}\sup_{0\leq j +i\leq \tilde N}\int_{\mathcal{S}_{u,v}}\left|\reallywidetilde{\nabla^j\nabla_4^i\psi}\right|^2u^{2(i+j)+3-4\delta}v^{-1+4\delta},\]
\[\left\vert\left\vert \slashed{Riem}\right\vert\right\vert_{\mathfrak{S}_{\tilde u,\tilde v}}^2 \doteq \sup_{(u,v) \in \mathcal{R}_{\tilde u,\tilde v}}\sup_{0\leq i \leq {\rm max}\left(\lfloor\frac{n-3}{2}\rfloor,0\right)}\sup_{0\leq j +i\leq \tilde N-1}\int_{\mathcal{S}_{u,v}}\left|\reallywidetilde{\nabla^j\nabla_4^i\slashed{Riem}}\right|^2u^{2(i+j)+3-4\delta}v^{-1+4\delta},\]
\end{definition}

\subsection{Total Norm}
It is convenient to bundle all of the norms into the following definition.
\begin{definition}\label{totalnorm}Let $\left(\mathcal{M},g\right)$ be produced by Theorem~\ref{localprotoambientyay}. Let $(\tilde u,\tilde v)$ be given  and suppose that the solution exists in the characteristic rectangle defined by $(u,v) \in [-1,\tilde u] \times [0,\tilde v]$, then we define the norm $\left\vert\left\vert\cdot\right\vert\right\vert_{\mathfrak{E}_{\tilde u,\tilde v}}$ by
\[\left\vert\left\vert g\right\vert\right\vert_{\mathfrak{E}_{\tilde u,\tilde v}} \doteq \left\vert\left\vert g\right\vert\right\vert_{\mathfrak{T}_{\tilde u,\tilde v}} +\left\vert\left\vert g\right\vert\right\vert_{\mathfrak{L}_{\tilde u,\tilde v}} + \left\vert\left\vert g\right\vert\right\vert_{\mathfrak{U}_{\tilde u,\tilde v}} + \left\vert\left\vert g\right\vert\right\vert_{\mathfrak{S}_{\tilde u,\tilde v}} + \left\vert\left\vert g\right\vert\right\vert_{\mathfrak{V}_{\tilde u,\tilde v}},\]
where the norms on the right hand side are defined in terms of the corresponding double null knowns and the definitions of the previous sections. 
\end{definition}

\subsection{Renormalized and Commuted Equations for $n \geq 3$ and odd}\label{renormn3}
In this section we will present schematic forms for the commuted and renormalized Bianchi equations odd $n$.

First we introduce some useful notation.
\begin{definition}For $i,j\geq 0$ we let $\mathscr{R}_{ij}$ denote the schematic expression
\[\mathscr{R}_{ij} \doteq \nabla^i\nabla_4^j\left[\left(\mathring{\psi}+\mathring{\phi}\right)|u|^{-\frac{n-4}{2}}v^{\frac{n-4}{2}}h\right],\]
where $\mathring{\psi}$ is as in Definition~\ref{vanishnormricci} and $\mathring{\phi}$ denotes either $u^{-1}b$ or $u^{-1}\left(\Omega^{-1}-1\right)$.
\end{definition}

\begin{definition}For $m,l \geq 0$ we let $\mathscr{F}_{ml}$ denote the schematic expression
\begin{align*}
\mathscr{F}_{ml} &\doteq \sum_{i+j +k = m+l,i\leq l}\nabla^k\nabla_4^i\left(\psi^{j+1}\Psi'\right) + \sum_{i+j=m+l,i \leq l}\nabla^j\nabla_4^i\left(\zeta\psi\psi\right)
\\ \nonumber &\qquad + \sum_{i+j= m-1}\nabla^i\left(\slashed{Riem}^{j+1}\nabla_4^l\Psi'\right). 
\end{align*}

We define $\mathscr{F}'_{ml}$ by the definition except that when $\Psi' = \alpha'$ then one of the Ricci coefficients multiplying $\Psi'$ is a $\mathring{\psi}$ as in Definition~\ref{vanishnormricci}.
\end{definition}

The following proposition is the key result of the section.

\begin{proposition}\label{writesomeeqns}The Bianchi equation for $\alpha$ can be re-written as
\[\nabla_3\nabla^i\nabla_4^j\alpha'_{AB} + \frac{n+i}{2}u^{-1}\nabla^i\nabla_4^j\alpha'_{AB} = -\nabla^C\nabla^i\nabla_4^j\nu_{C(AB)} + \nabla_{(A}\nabla^i\nabla_4^j\beta_{B)}  + \mathscr{R}_{ij} + \mathscr{F}'_{ij},\]

In every other Bianchi equation except for the $\nabla_4$ equations for $\beta$ and $\nu$, we can replace each curvature component $\Psi$ with $\nabla^i\nabla_4^j\Psi'$ at the expense of replacing the error terms $\mathscr{E}$ with $\mathscr{R}_{ij} +\mathscr{F}_{ij}$. In the $\nabla_4$ equation for $\beta$ and $\nu$ we must add an additional error term proportional to 
\[|u|^{-\frac{n+2+2i}{2}}v^{\frac{n-4-2j}{2}}.\]
\end{proposition}
\begin{proof}We start with $\alpha$. The equation for $\alpha$ may be written as
\[\nabla_3\alpha_{AB} + \frac{n}{2}u^{-1}\alpha_{AB} = -\nabla^C\nu_{C(AB)} + \nabla_{(A}\beta_{B)} + \mathscr{E}_1.\]

\emph{Computing in a Lie-propagated frame}, we find
\begin{align*}
\nabla_3\left(u^{\frac{4-n}{2}}v^{\frac{n-4}{2}}h\right)_{AB} &= -\frac{n}{2}u^{\frac{2-n}{2}}v^{\frac{n-4}{2}}h_{AB} + \left(\Omega^{-1}-1\right)\left(\frac{4-n}{2}\right)u^{\frac{2-n}{2}}v^{\frac{n-4}{2}}h_{AB} \\ \nonumber &\qquad + \Omega^{-1}u^{\frac{4-n}{2}}v^{\frac{n-4}{2}}b^C\partial_C\left(h_{AB}\right) 
 - 2u^{\frac{4-n}{2}}v^{\frac{n-4}{2}}\reallywidetilde{{\rm tr}\underline{\chi}}h_{AB} - 2u^{\frac{4-n}{2}}v^{\frac{n-4}{2}}\hat{\underline{\chi}}_{(A}^Ch_{B)C}
 \\ \nonumber &\doteq -\frac{n}{2}u^{\frac{2-n}{2}}v^{\frac{n-4}{2}}h_{AB} + \mathscr{R}_{00}.
\end{align*}

In particular, we obtain
\begin{equation}\label{renormalizedalpha}
\nabla_3\alpha'_{AB} + \frac{n}{2}u^{-1}\alpha'_{AB} = -\nabla^C\nu_{C(AB)} + \nabla_{(A}\beta_{B)} + \mathscr{R}_{00}+\mathscr{E}_1.
\end{equation}

Next, it follows easily from signature considerations that $\mathscr{E}_1 = \mathscr{R}_{00} + \mathscr{F}_{00}$. Thus we obtain:
\begin{equation}\label{renormalizedalpha2}
\nabla_3\alpha'_{AB} + \frac{n}{2}u^{-1}\alpha'_{AB} = -\nabla^C\nu_{C(AB)} + \nabla_{(A}\beta_{B)} + \mathscr{R}_{00}+\mathscr{F}_{00}.
\end{equation}

The desired commuted equations for $\alpha$ follow from first commuting with $\nabla_4^j$ and then with $\nabla^i$ and using the formulas of Lemma~\ref{4commute} and~\ref{34commute}.

The rest of the Bianchi equations are handled in an analogous fashion. (The extra term in the $\nabla_4$ equation of $\beta$ and $\nu$ come from the present of $\nabla\alpha$ on the right hand side of those equations.)

\end{proof}
\subsection{Renormalized and Commuted Equations for $n \geq 4$ and even}\label{renormn4}
In this section we will present schematic forms for the commuted and renormalized Bianchi equations even $n\geq 4$. Our notation will mirror that of the case of odd $n$.

\begin{definition}For $i,j\geq 0$ we let $\mathscr{R}_{ij}$ denote the schematic expression
\[\mathscr{R}_{ij} \doteq \nabla^i\nabla_4^j\left[\left(\mathring{\psi}+\mathring{\phi}\right)|u|^{-\frac{n-4}{2}}v^{\frac{n-4}{2}}\log\left(\frac{v}{|u|}\right)\mathcal{O}_{AB}\right],\]
where $\mathring{\psi}$ is as in Definition~\ref{vanishnormricci} and $\mathring{\phi}$ denotes either $u^{-1}b$ or $u^{-1}\left(\Omega^{-1}-1\right)$.
\end{definition}

\begin{definition}For $m,l \geq 0$ we let $\mathscr{F}_{ml}$ denote the schematic expression
\begin{align*}
\mathscr{F}_{ml} &\doteq \sum_{i+j +k = m+l,i\leq l}\nabla^k\nabla_4^i\left(\psi^{j+1}\Psi'\right) + \sum_{i+j=m+l,i \leq l}\nabla^j\nabla_4^i\left(\zeta\psi\psi\right)
\\ \nonumber &\qquad + \sum_{i+j= m-1}\nabla^i\left(\slashed{Riem}^{j+1}\nabla_4^l\Psi'\right). 
\end{align*}

We define $\mathscr{F}'_{ml}$ by the same definition except that when $\Psi' = \alpha'$ then one of the Ricci coefficients multiplying $\Psi'$ is a $\mathring{\psi}$ as in Definition~\ref{vanishnormricci}.

\end{definition}

The following proposition is the key result of the section.

\begin{proposition}\label{renormeven}In every Bianchi equation except for the $\nabla_4$ equations for $\beta$ and $\nu$, we can replace each curvature component $\Psi$ with $\nabla^i\nabla_4^j\Psi'$ at the expense of replacing the error terms $\mathscr{E}$ with $\mathscr{R}_{ij} +\mathscr{F}_{ij}$. 

In the $\nabla_4$ equation for $\beta$ and $\nu$ we must add an additional error term proportional to 
\[|u|^{-\frac{n+2+2i}{2}}v^{\frac{n-4-2j}{2}}\log\left(\frac{v}{|u|}\right).\]

The Bianchi equation for $\alpha$ can be re-written as
\[\nabla_3\nabla^i\nabla_4^j\alpha'_{AB} + \frac{n+i}{2}u^{-1}\nabla^i\nabla_4^j\alpha'_{AB} = -\nabla^C\nabla^i\nabla_4^j\nu_{C(AB)} + \nabla_{(A}\nabla^i\nabla_4^j\beta_{B)}  + \mathscr{R}_{ij} + \mathscr{F}'_{ij}.\]

Furthermore, when we consider $\nabla^i\nabla_4^{\frac{n-4}{2}}\alpha'$ we can assume that the right hand side vanishes and thus we can write
\[\nabla_3\nabla^i\nabla_4^{\frac{n-4}{2}}\alpha' + \frac{n+i}{2}u^{-1}\nabla^i\nabla_4^{\frac{n-4}{2}}\alpha' = -\nabla^C\reallywidetilde{\nabla^i\nabla_4^{\frac{n-4}{2}}\nu_{C(AB)} }+ \nabla_{(A}\reallywidetilde{\nabla^i\nabla_4^{\frac{n-4}{2}}\beta_{B)}}  + \reallywidetilde{\mathscr{R}}_{i\frac{n-4}{2}} + \reallywidetilde{\mathscr{F}'}_{i\frac{n-4}{2}},\]
where we recall the tilde notation from Definition~\ref{thisistilde}.
\end{proposition}
\begin{proof}Everything is done in the same fashion as $n$ odd except for the claim that the right hand side of the equation for  $\nabla^i\nabla_4^{\frac{n-4}{2}}\alpha'$ can be assumed to vanish. However, when $i=0$, the desired assertion follows immediately from the part of the proof of Proposition~\ref{incomingdataneven} that shows how the logarithmic component of $\alpha$ is determined. For  $i \geq 1$ one commutes and uses Lemma~\ref{3commute}.

\end{proof}

\section{A Priori Estimates for Proto-Ambient Metrics: The Main Bootstrap}\label{aprioriestsection}

The key result of the section is the following.
\begin{theorem}\label{thefundamentalestimate}Assume that we have a proto-ambient metric $(\mathcal{M},g)$ which arises from compatible regular conjugate data and exists in a characteristic rectangle $\mathcal{R}_{\tilde u,\tilde v}$ with $\frac{\tilde v}{|\tilde u|} \leq \epsilon$ for $\epsilon > 0$ sufficiently small.

Suppose we have the bootstrap assumption
\begin{equation}\label{bootstrap}
\mathfrak{T}_{\tilde u,\tilde v} + \mathfrak{L}_{\tilde u,\tilde v} +\mathfrak{U}_{\tilde u,\tilde v} + \mathfrak{S}_{\tilde u,\tilde v} + \mathfrak{V}_{\tilde u,\tilde v} \leq A.
\end{equation}

Then, for $\epsilon > 0$ sufficient small there exists a constant $C \geq 1$ depending only on the size of the initial data such that
\begin{equation}\label{quiteaniceconclusionifidontsaysomyself}
\mathfrak{T}_{\tilde u,\tilde v} + \mathfrak{L}_{\tilde u,\tilde v} +\mathfrak{U}_{\tilde u,\tilde v} + \mathfrak{S}_{\tilde u,\tilde v} + \mathfrak{V}_{\tilde u,\tilde v} \leq C.
\end{equation}

The constant $C$ can be taken to depend on 
\begin{equation}\label{whatCdependson}
\mathfrak{T}_{-1,\tilde v} + \mathfrak{L}_{-1,\tilde v} +\mathfrak{U}_{-1,\tilde v} + \mathfrak{S}_{-1,\tilde v} + \mathfrak{V}_{-1,\tilde v},\qquad \sup_{j \leq \tilde N}\sup_{i \leq F(n)}\left|\partial_{\theta}^j\mathcal{L}_v^i\slashed{g}|_{\mathcal{S}_{-1,0}}\right|,\qquad h_{AB},
\end{equation}
where $F(2) = 2$, when $n \geq 3$ and odd, $F(n) = \frac{n+1}{2}$, and when $n \geq 4$ and even, $F(n) = \frac{n}{2}$.
\end{theorem}

In the sections that follow we will allow the constant $C$ to grow from line to line as needed. Also, when there is no risk of confusion we will drop the $(\tilde u,\tilde v)$ from the subscripts of $\mathfrak{T},\cdots,\mathfrak{V}$ or $\mathcal{R}$.

As we explain in Appendix~\ref{actuallocalexistencesec}, Theorem~\ref{localexistenceproto} follows from Theorem~\ref{thefundamentalestimate}.
\subsection{Estimates for Initial Data}
We start with a basic estimate at the level of the initial data along $\{v = 0\}$ which quantifies the estimates carried out in Propositions~\ref{incomingdatan2},~\ref{incomingdatanodd}, and~\ref{incomingdataneven}.

\begin{lemma}\label{stuffonv0}For every curvature component $\Psi$ not equal to $\alpha$ and Ricci coefficient $\psi$ we have
\[\left|\nabla^j\nabla_4^i\Psi\right||_{v=0} \leq C|u|^{-2-j-i},\qquad \left|\nabla^j\nabla_4^i\psi\right||_{v=0} \leq C|u|^{-1-j-i},\]
where $j$ is arbitrary, $i \leq {\rm max}\left(\lfloor \frac{n-3}{2}\rfloor,0\right)$, and the constant $C$ only depends on~\eqref{whatCdependson}.
\end{lemma}
\begin{proof}Without the quantification of the constant $C$, this estimate is already contained in Propositions~\ref{incomingdatan2},~\ref{incomingdatanodd}, and~\ref{incomingdataneven} (recall Definition~\ref{defcompatible}). However, re-running through the proofs of those propositions easily shows that $C$ can be taken to depend on~\eqref{whatCdependson}.
\end{proof}

\subsection{Estimates for the metric and the curvature of $\mathcal{S}$}
In this section we will establish estimates for the metric coefficients and use this to establish Sobolev inequalities and elliptic estimates.

The following Sobolev inequality will play a fundamental role in the analysis.
\begin{proposition}\label{letsdoasob}Let $a > 0$ be sufficiently small, $M(n)$ be sufficiently large depending on $n$, and assume that in each coordinate patch we have
\begin{equation}\label{metricscloseenough}
\sup_{1 \leq j \leq M(n)}\sum_{AB}u^j\left|\mathcal{L}^j_{\theta}\left(u^{-2}\slashed{g}_{AB}|_{\mathcal{S}_{u,v}}-\left(\slashed{g}_0\right)_{AB}|_{\mathcal{S}_{u,0}}\right)\right| \leq a.
\end{equation}

Then, for any $(0,s)$ tensor-field $\phi$ and $(u,v) \in \mathcal{R}_{\tilde u, \tilde v}$, we have
\begin{equation}\label{sobolev}
\sup_{\mathcal{S}_{u,v}}\left|\phi\right| \leq C_s \sum_{i=0}^{\lceil \frac{n}{2}\rceil+1}\left(\int_{\mathcal{S}_{u,v}}|u|^{2i}\left|\nabla^i\phi\right|^2\right)^{1/2},
\end{equation}
for a constant $C_s$ which only depends on $s$. 
\end{proposition}
\begin{proof}This is standard; one simply shows that the control of the $\slashed{g}$ covariant derivatives implies control of the corresponding $\slashed{g}_0$ covariant derivatives and then applies a Sobolev inequality relative to $\slashed{g}_0$.
\end{proof}
\begin{remark}Remember that our convention is that, unless said otherwise, all integrals over $\mathcal{S}_{u,v}$ are with respect to the volume form $\slashed{dVol}_0$ associated to $\slashed{g}_0$. The reason we must assume~\eqref{metricscloseenough} is because for tensorial quantities, the norm $\left|\cdot\right|$ and covariant derivative $\nabla$ depends on $\slashed{g}$. Finally, we assume that $N$ is chosen so that $M(n) \ll N$. 
\end{remark}

We have
\begin{proposition}\label{metricest}We have
\begin{equation}\label{metrictoprove}
\sup_{(u,v) \in \mathcal{R}}\left[\left|\Omega^{-1}-1\right|\frac{|u|^{2-3\delta}}{v^{2-3\delta}} + \sum_{AB}\left|\slashed{g}_{AB} - u^2\left(\slashed{g}_0\right)_{AB}|_{\mathcal{S}_{u,0}}\right|\frac{1}{|u|v} + |b|\frac{|u|^{2-3\delta}}{v^{2-3\delta}}\right]\leq C.
\end{equation}
Here, $\slashed{g}$ is expressed in the canonical coordinate frame (and we sum over a family of coordinate patches that covers all of $\mathcal{S}$).

Also, we have
\begin{equation}\label{yaywecandotheobolev}
\sup_{1 \leq j \leq M(n)}\sum_{AB}u^j\left|\mathcal{L}^j_{\theta}\left(u^{-2}\slashed{g}_{AB}|_{\mathcal{S}_{u,v}}-\left(\slashed{g}_0\right)_{AB}|_{\mathcal{S}_{u,0}}\right)\right| \leq \frac{a}{2},
\end{equation}
where $a$ and $M(n)$ are from Proposition~\ref{letsdoasob}, and 
\begin{align}\label{metricbetter}
\sup_{1 \leq j \leq \tilde N}\sup_{(u,v) \in \mathcal{R}}\int_{\mathcal{S}_{u,v}}&\Bigg[\left|\nabla^j\Omega^{-1}\right|^2v^{4-6\delta}u^{-4+6\delta}+\left|\nabla^jb\right|v^{4-6\delta}u^{-4+6\delta}\Bigg]u^{2j} \leq C.
\end{align}
\end{proposition}
\begin{proof}We start with the proof of~\eqref{yaywecandotheobolev}. Using a standard bootstrap argument, we can assume that~\eqref{yaywecandotheobolev} holds with $\frac{3a}{4}$ replacing $\frac{a}{2}$ on the right hand side. In particular, we can arrange for~\eqref{metricscloseenough} to hold and thus we can freely appeal to the Sobolev inequality~\eqref{sobolev}.


We have the following equation for $\slashed{g}$:
\begin{equation}\label{tocontrolslashg}
\mathcal{L}_v\slashed{g}_{AB} = 2\Omega \chi_{AB}.
\end{equation}
Integrating in $v$ along coordinate patches yields the estimate
\[\sum_{AB}\sup_{\mathcal{S}_{u,v}}\left|\left(\slashed{g}_{AB} - u^2\left(\slashed{g}_0\right)_{AB}|_{\mathcal{S}_{u,0}}\right)\right| \leq 2\int_0^v\sum_{AB}\sup_{\mathcal{S}_{u,v}}\left|\Omega\right|\left|\chi_{AB}\right|. \]

The bootstrap assumption that~\eqref{yaywecandotheobolev} holds with $\frac{3a}{4}$ is easily seen to imply that 
\[\sum_{AB}\sup_{\mathcal{S}_{u,v}}\left|\Omega\right|\left|\chi_{AB}\right| \lesssim \sup_{\mathcal{S}_{u,v}}\left[|u|^2\left|\chi\right|\left|\Omega\right|\right].\]

In particular, using the bootstrap assumption~\eqref{bootstrap} and a Sobolev inequality yields:
\[\sum_{AB}\left|\left(\slashed{g}_{AB} - u^2\left(\slashed{g}_0\right)_{AB}|_{\mathcal{S}_{u,0}}\right)\right| \lesssim A^2v\left|u\right|.\]

Similarly, keeping in mind that $M(n) \ll N$, we can commute~\eqref{tocontrolslashg} with Lie derivatives along $\mathcal{S}$ and argue analogously to find that 
\[\sup_{1 \leq j \leq M(n)}\sum_{AB}u^j\left|\mathcal{L}^j_{\theta}\left(u^{-2}\slashed{g}_{AB}|_{\mathcal{S}_{u,v}}-\left(\slashed{g}_0\right)_{AB}|_{\mathcal{S}_{u,0}}\right)\right| \lesssim A^2\frac{v}{|u|}.\]

Now we can just take $\epsilon$ sufficiently small to obtain~\eqref{yaywecandotheobolev}.

Next we turn to the lapse $\Omega$. From Proposition~\ref{metriceqn} we easily deduce
\[\partial_v\left(\Omega^{-1}\right) = 2\omega.\]
Integrating in $v$ and using that $\Omega^{-1}$ is identically $1$ on $\{v = 0\}$ yields
\begin{align*}
\sup_{\mathcal{S}_{u,v}}\left|\left(\Omega^{-1}-1\right)\right| &\leq Cv^{2-2\delta}u^{-2+2\delta}\mathfrak{S}
\\ \nonumber &\leq Cv^{2-3\delta}u^{-2+3\delta}.
\end{align*}
In the last line we used the bootstrap assumption~\eqref{bootstrap} and took $\epsilon$ sufficiently small. This yields the estimate for $\Omega^{-1}$ in~\eqref{metrictoprove}.

Similarly, we can inductively commute with $\nabla_A^j$, integrating in $v$, and use again that $\Omega^{-1}$ is identically $1$ on $\{v = 0\}$ to obtain for each $1 \leq j \leq N$ that 
\begin{align*}
\int_{\mathcal{S}_{u,v}}\left|\nabla^j\left(\Omega^{-1}-1\right)\right|^2 &\leq Cv^{4-4\delta}u^{-4+4\delta-2j}\mathfrak{S}
\\ \nonumber &\leq Cv^{4-6\delta}u^{-4+6\delta-2j}.
\end{align*}
In the last line we used the bootstrap assumption~\eqref{bootstrap} and took $\epsilon$ sufficiently small. This establishes the estimate for $\Omega^{-1}$ in~\eqref{metricbetter}.

The estimates for $b$ work in an analogous fashion using the equations:
\[\mathcal{L}_vb^A = -4\Omega^2\zeta^A.\]

\end{proof}

The next lemma is useful when we integrate $\nabla_3$ equations.
\begin{lemma}For every point $p = (u_0,v_0,\theta_0) \in \mathcal{R} \times \mathcal{S}$ there is an integral curve of $e_3$ which connects $p$ to $\{u = -1\} \cap \{v \leq \epsilon\}$.
\end{lemma}
\begin{proof}This follows easily from the bounds on the shift $b$ given by Proposition~\ref{metricest}.
\end{proof}

Finally, our estimates of the metric allow us to carry out elliptic estimates along $\mathcal{S}$:
\begin{lemma}\label{dotheelliptic}

Suppose $\phi_{A_1\cdots A_k}$ is an $\mathcal{S}_{u,v}$ tensor such that
\[\nabla^B\phi_{BA_2\cdots A_k} = F_{A_2\cdots A_k},\qquad \nabla_{[A_0}\phi_{A_1]A_2\cdots A_k} = G_{A_0\cdots A_k}.\]

Then, for every $1 \leq j \leq \tilde N$ we have
\[\int_{\mathcal{S}_{u,v}}\left|\nabla^j\phi\right|^2 \leq C\left(\int_{\mathcal{S}_{u,v}}\left|\phi\right|^2|u|^{-2j} + \int_{\mathcal{S}_{u,v}}\left[\left|\nabla^{j-1}F\right|^2 + \left|\nabla^{j-1}G\right|^2\right]\right).\]
\end{lemma}
\begin{proof}We first discuss the case when $j=1$. One integrates the identity
\[\nabla^B\phi_{BA_2\cdots A_k}\nabla_C\phi^{CA_2\cdots A_k} = \left|F\right|^2\]
over $\mathcal{S}_{u,v}$ and integrates by parts. Eventually we obtain
\[\int_{\mathcal{S}_{u,v}}\left|\nabla\phi\right|^2\slashed{dVol} \leq C\int_{\mathcal{S}_{u,v}}\left(\left|\slashed{Riem}\right|\left|\phi\right|^2 + \left|G\right|\left|\nabla\phi\right| + \left|F\right|^2\right)\slashed{dVol},\]
from which the desired estimate follows easily after controlling the curvature term in $L^{\infty}$ with Lemma~\ref{stuffonv0}, the bootstrap assumption~\eqref{bootstrap}, and a Sobolev inequality.

For the higher order estimate, we induct in $j$ and commute the equation with $\nabla^{j-1}$ to obtain
\[\left|\nabla^B\nabla^{j-1}\phi_{BA_2\cdots A_k}\right| \lesssim \left|\nabla^{j-1}F_{A_2\cdots A_k}\right| + \sum_{i+k+l=j-2}\left|\nabla^i\slashed{Riem}^{k+1}\right|\left|\nabla^l\phi\right|,\]
\[\left|\nabla_{[A_0}\nabla^{j-1}\phi_{A_1]A_2\cdots A_k}\right| \lesssim \left|\nabla^{j-1}G_{A_0\cdots A_k}\right| + \sum_{i+k+l=j-2}\left|\nabla^i\slashed{Riem}^{k+1}\right|\left|\nabla^l\phi\right|,\]
where the $\slashed{Riem}^{k+1}$ refers to the usual schematic notation for any possible product of $k+1$ components of $\slashed{Riem}$.

Since the product $\left|\nabla^i\slashed{Riem}^{k+1}\right|\left|\nabla^l\phi\right|$ can easily be controlled in $L^2$ by taking the term with the most derivatives in $L^2$ and applying a Sobolev inequality to handle the remaining terms, we can simply apply the same integration by parts argument to obtain the desired estimate.
\end{proof}

We will also need elliptic estimates for the Laplacian $\slashed{\Delta}$ on $\mathcal{S}$.
\begin{lemma}\label{dotheomegaelliptic}

Let $f$ be a scalar function on $\mathcal{S}_{u,v}$ . Then, for every $2 \leq j \leq N$ we have
\[\int_{\mathcal{S}_{u,v}}\left|\nabla^jf\right|^2 \leq C\int_{\mathcal{S}_{u,v}} \left[\left|\nabla^{j-2}\left(\slashed{\Delta} f\right)\right|^2 + |u|^{-2j}f^2\right].\]
\end{lemma}
\begin{proof}

When $j = 0$ this is a standard elliptic estimate. The error terms from the curvature of $\mathcal{S}$ and commutating are handled as in the proof of Lemma~\ref{dotheelliptic}.
\end{proof}
\subsection{Estimates for $\mathfrak{T}$}
In this section we will carry out the top order energy estimates. The argument is a bit different for $n \geq 3$ and odd, $n \geq 4$ and even, and $n = 2$. 
\subsubsection{$n \geq 3$ and Odd}
We start with the case of odd $n$.

\begin{proposition}\label{controlTnodd}Let $n \geq 3$ and odd. Then we have
\[\mathfrak{T} \leq C.\]
\end{proposition}
\begin{proof}We will carry estimates for each Bianchi pair (see Definition~\ref{defbiancpair} and Proposition~\ref{thebianchipairs}) . We start with $\left(\alpha,\left(\beta,\nu\right)\right)$. For every $0 \leq i \leq N$ we have
\begin{equation}\label{renormalthealpha}
\nabla_3\nabla^i\nabla_4^{\frac{n-3}{2}}\alpha'_{AB} + \frac{n+i}{2}u^{-1}\nabla^i\nabla_4^{\frac{n-3}{2}}\alpha'_{AB} = -\nabla^C\nabla^i\nabla_4^{\frac{n-3}{2}}\nu_{C(AB)} + \nabla_{(A}\nabla^i\nabla_4^{\frac{n-3}{2}}\beta_{B)}  + \mathscr{R}_{i\frac{n-3}{2}} + \mathscr{F}'_{i\frac{n-3}{2}},
\end{equation}
\[ \nabla_4\nabla^i\nabla_4^{\frac{n-3}{2}}\beta_A = \nabla^B\nabla^i\nabla_4^{\frac{n-3}{2}}\alpha'_{BA} + \mathscr{R}_{i\frac{n-3}{2}} + \mathscr{F}_{i\frac{n-3}{2}} + O\left(|u|^{-\frac{n+2i}{2}}v^{-1/2}\right),\]
\[\nabla_4\nabla^i\nabla_4^{\frac{n-3}{2}}\nu_{ABC} = -2\nabla_{[A}\nabla^i\nabla_4^{\frac{n-3}{2}}\alpha'_{B]C} +\mathscr{R}_{i\frac{n-3}{2}} + \mathscr{F}_{i\frac{n-3}{2}}+O\left(|u|^{-\frac{n+2i}{2}}v^{-1/2}\right).\]

In order to make the notation more compact let's introduce the notation 
\[\mathfrak{D} \doteq \nabla^i\nabla_4^{\frac{n-3}{2}}.\]

Writing the equations in terms of $\reallywidetilde{\mathfrak{D}\alpha'}$, $\reallywidetilde{\mathfrak{D}\beta}$, and $\reallywidetilde{\mathfrak{D}\nu}$, using Lemma~\ref{stuffonv0} yields, and conjugating the equations by the weight 
\[w \doteq |u|^{\frac{n+1+2i}{2}-\delta}v^{-1/2 + \delta},\]
eventually yields

\begin{align}\label{blahblahblahblah}
\Omega^{-1}\nabla_{u+b^A\partial_A}\left(w\reallywidetilde{\mathfrak{D}\alpha'}\right)_{AB} &+ \frac{1-2\delta}{2}|u|^{-1}w\reallywidetilde{\mathfrak{D}\alpha'}_{AB} = 
\\ \nonumber &-w\nabla^C\reallywidetilde{\mathfrak{D}\nu}_{C(AB)} + w\nabla_{(A}\reallywidetilde{\mathfrak{D}\beta}_{B)}  + w\mathscr{R}_{i\frac{n-3}{2}} + w\mathscr{F}'_{i\frac{n-3}{2}} + wO\left(u^{-\frac{n+3+2i}{2}}\right),
\end{align}
\begin{align*}
\Omega^{-1}\nabla_v\left(w\reallywidetilde{\mathfrak{D}\beta}\right)_A &+ \frac{1-2\delta}{2}v^{-1}w\reallywidetilde{\mathfrak{D}\beta}_A= 
\\ \nonumber &w\nabla^B\reallywidetilde{\mathfrak{D}\alpha'}_{BA} + w\mathscr{R}_{i\frac{n-3}{2}} + \mathscr{F}_{i\frac{n-3}{2}} + wO\left(|u|^{-\frac{n+2+2i}{2}}v^{-1/2}\right),
\end{align*}
\begin{align*}
\Omega^{-1}\nabla_v\left(w\reallywidetilde{\mathfrak{D}\nu}\right)_{ABC}&+ \frac{1-2\delta}{2}v^{-1}w\reallywidetilde{\mathfrak{D}\nu}=
\\ \nonumber &-2w\nabla_{[A}\reallywidetilde{\mathfrak{D}\alpha'}_{B]C} +w\mathscr{R}_{i\frac{n-3}{2}} + w\mathscr{F}_{i\frac{n-3}{2}}+wO\left(|u|^{-\frac{n+2+2i}{2}}v^{-1/2}\right).
\end{align*}

Now we multiply the first equation by 2$w\reallywidetilde{\mathfrak{D}\alpha'}^{AB}$, the second by 2$w\reallywidetilde{\mathfrak{D}\beta}^A$, and the third by $w\reallywidetilde{\mathfrak{D}\nu}^{ABC}$. Adding the three identities together and integrating by parts over $\mathcal{R}_{\tilde u,\tilde v}$ (see Proposition~\ref{thebianchipairs}), using Proposition~\ref{metricest}, and applying Cauchy Schwarz yields the following basic energy estimate:
\begin{align}\label{thebasicenergyestimate}
&\sup_{-1 \leq u \leq \tilde u}\int_0^{\tilde v}\int_{\mathcal{S}}w^2\left|\reallywidetilde{\mathfrak{D}\alpha'}\right|^2 + \int_{-1}^{\tilde u}\int_0^{\tilde v}\int_{\mathcal{S}}|u|^{-1}w^2\left|\reallywidetilde{\mathfrak{D}\alpha'}\right|^2 + 
\\ \nonumber &\sup_{0 \leq v \leq \tilde v}\int_{-1}^{\tilde u}\int_{\mathcal{S}}w^2\left[\left|\reallywidetilde{\mathfrak{D}\beta}\right|^2 + \left|\reallywidetilde{\mathfrak{D}\nu}\right|^2\right]+ \int_{-1}^{\tilde u}\int_0^{\tilde v}\int_{\mathcal{S}}v^{-1}w^2\left[\left|\reallywidetilde{\mathfrak{D}\beta}\right|^2 + \left|\reallywidetilde{\mathfrak{D}\nu}\right|^2\right] 
\\ \nonumber &\qquad \lesssim \int_{-1}^{\tilde u}\int_0^{\tilde v}\int_{\mathcal{S}}\left[|u|w^2\left(\left|\mathscr{R}_{i\frac{n-3}{2}}\right|^2 + \left|\mathscr{F}'_{i\frac{n-3}{2}}\right|^2 + |u|^{-n-3-2i}\right)\right] +
\\ \nonumber &\qquad \qquad \int_{-1}^{\tilde u}\int_0^{\tilde v}\int_{\mathcal{S}}\left[vw^2\left(\left|\mathscr{R}_{i\frac{n-3}{2}}\right|^2 + \left|\mathscr{F}_{i\frac{n-3}{2}}\right|^2 + |u|^{-n-2-2i}v^{-1}\right)\right]+\int_0^{\tilde v}\int_{\mathcal{S}}w^2\left|\reallywidetilde{\mathfrak{D}\alpha'}\right|^2|_{u=-1}.
\end{align}

Observe that the left hand side of~\eqref{thebasicenergyestimate} is exactly the top-order energy norms of $\alpha'$, $\beta$, and $\nu$ that we wish to control (except for a $v$-flux of $\beta$ and $\nu$ that will come from the next Bianchi pair). We now turn to an analysis of the error terms on the right hand side of~\eqref{thebasicenergyestimate}. First of all
\[\int_{-1}^{\tilde u}\int_0^{\tilde v}\int_{\mathcal{S}}w^2|u|^{-n-2-2i} \lesssim \int_{-1}^{\tilde u}\int_0^{\tilde v}\int_{\mathcal{S}}|u|^{-1-2\delta}v^{-1+2\delta} \lesssim \left(\frac{v}{|u|}\right)^{2\delta} \lesssim \epsilon^{2\delta}.\]

Next, let's consider the term with $\mathscr{R}_{i\frac{n-3}{2}}$. Recall that this stands for the schematic error term
\begin{equation}\label{schematicforR}
\nabla^i\nabla_4^{\frac{n-3}{2}}\left[\left(\mathring{\psi}+\mathring{\phi}\right)|u|^{-\frac{n-4}{2}}v^{\frac{n-4}{2}}h\right].
\end{equation}

The most singular situation that can arise is when all of the $\nabla_4$ derivatives fall on the $v^{\frac{n-4}{2}}$, so let's consider that situation first. In this case the error term will look like
\[\mathscr{R}_{i\frac{n-3}{2}} \sim \sum_{j=0}^i|u|^{-j-\frac{n}{2}}\left[\left|\nabla^{i-j}\mathring{\psi}\right| + \left|\nabla^{i-j}\mathring{\phi}\right|\right]v^{-1/2}.\]
Thus, in this case,
\begin{align}
\int_{-1}^{\tilde u}\int_0^{\tilde v}\int_{\mathcal{S}}|u|w^2\left|\mathscr{R}_{i\frac{n-3}{2}}\right|^2 &\lesssim \sum_{j=0}^i\int_{-1}^{\tilde u}\int_0^{\tilde v}\int_{\mathcal{S}}|u|^{2(1+i-j-\delta)}\left[\left|\nabla^{i-j}\mathring{\psi}\right|^2 + \left|\nabla^{i-j}\mathring{\phi}\right|^2\right]v^{-2+2\delta}
\\ \nonumber &\lesssim \left[\left\vert\left\vert \mathring{\psi}\right\vert\right\vert_{\mathfrak{V}}^2 + \sup_{k\leq  N}\sup_{u,v}\int_{\mathcal{S}}\left|\nabla^k\mathring{\phi}\right|^2|u|^{2k-4-4\delta}v^{-2+4\delta}\right]\int_{-1}^{\tilde u}\int_0^{\tilde v}|u|^{2\delta - 2}v^{-2\delta}
\\ \nonumber &\lesssim A\epsilon^{1-2\delta}
\\ \nonumber &\lesssim 1,
\end{align}
where we used the bootstrap assumption and Proposition~\ref{metricest}.

Next, we consider error terms in~\eqref{schematicforR} where we don't have that all of the $\nabla_4$ derivatives fall on the $v^{\frac{n-4}{2}}$, i.e.,
\begin{equation}\label{schematicforR2}
\nabla^i\nabla_4^{\frac{n-3}{2}-1}\left[\left(\nabla_4\mathring{\psi}+\nabla_4\mathring{\phi}\right)|u|^{-\frac{n-4}{2}}v^{\frac{n-4}{2}}h\right].
\end{equation}

In this case we have 
\begin{align}
\int_{-1}^{\tilde u}\int_0^{\tilde v}\int_{\mathcal{S}}|u|w^2\left|\mathscr{R}_{i\frac{n-3}{2}}\right|^2 &\lesssim \sum_{k=0}^{\frac{n-3}{2}-1}\sum_{j=0}^i\int_{-1}^{\tilde u}\int_0^{\tilde v}\int_{\mathcal{S}}|u|^{2(1+i-j-\delta)}\left[\left|\nabla^{i-j}\nabla_4^{k+1}\mathring{\psi}\right|^2 + \left|\nabla^{i-j}\nabla_4^{k+1}\mathring{\phi}\right|^2\right]v^{2\delta+2k}
\\ \nonumber &\lesssim \left[1+\left\vert\left\vert \psi\right\vert\right\vert_{\mathfrak{S}}^2\right]\int_{-1}^{\tilde u}\int_0^{\tilde v}|u|^{-2-2\delta}v^{2\delta}
\\ \nonumber &\lesssim \left(1+A\right)\epsilon^{1-2\delta}
\\ \nonumber &\lesssim 1,
\end{align}
where we used the bootstrap assumption, Proposition~\ref{metricest}, Lemma~\ref{stuffonv0}, and the fact that $\nabla_4\mathring{\phi}$ can be expressed in terms of combinations of metric coefficients and Ricci coefficients.

Next we turn to the error term proportional to $\mathscr{F}'_{i\frac{n-3}{2}}$. Recall that
\begin{align}\label{formofFprime}
\mathscr{F}'_{ml} &\doteq \sum_{i+j +k = m+l,i\leq l}\nabla^k\nabla_4^i\left(\psi^{j+1}\Psi'\right) + \sum_{i+j=m+l,i \leq l}\nabla^j\nabla_4^i\left(\zeta\psi\psi\right)
\\ \nonumber &\qquad + \sum_{i+j= m-1}\nabla^i\left(\slashed{Riem}^{j+1}\nabla_4^l\Psi'\right). 
\end{align}

with the rule that if $\Psi'$ is equal to $\alpha'$ then one of the Ricci coefficients is given by a $\mathring{\psi}$.

We start with the first term on the right hand side of~\eqref{formofFprime} in the case when $\Psi' = \alpha'$ and where the maximum number of $\nabla_4$ derivatives fall on $\alpha'$. In this case we will have
\begin{align}
\int_{-1}^{\tilde u}\int_0^{\tilde v}\int_{\mathcal{S}}|u|w^2\left|\mathscr{F}'_{i\frac{n-3}{2}}\right|^2 &\lesssim \sum_{j+k=i}\int_{-1}^{\tilde u}\int_0^{\tilde v}\int_{\mathcal{S}}|u|^{2j}\left|\nabla^j\mathring{\psi}\right|^2|u|^{2k}\left|\nabla^k\nabla_4^{\frac{n-3}{2}}\alpha'\right|^2|u|^{n+2-2\delta}v^{-1+2\delta}
\\ \nonumber &\lesssim \sup_{j,k\leq N}\int_{-1}^{\tilde u}\int_0^{\tilde v}\left(\int_{\mathcal{S}}|u|^{2j}\left|\nabla^j\mathring{\psi}\right|^2\right)\left(\int_{\mathcal{S}}|u|^{2k}\left|\nabla^k\nabla_4^{\frac{n-3}{2}}\alpha'\right|^2\right)|u|^{n+2-2\delta}v^{-1+2\delta}
\\ \nonumber &\lesssim \sup_{k \leq N}\left\vert\left\vert \mathring{\psi}\right\vert\right\vert_{\mathfrak{V}}^2\int_{-1}^{\tilde u}\int_0^{\tilde v}\int_{\mathcal{S}}|u|^{2k}\left|\nabla^k\nabla_4^{\frac{n-3}{2}}\alpha'\right|^2|u|^{n-2\delta}v^{-1+2\delta}|u|^{-2+4\delta}v^{2-4\delta}
\\ \nonumber &\lesssim A^2\epsilon^{2-4\delta}
\\ \nonumber &\lesssim 1.
\end{align}
Here we used a Sobolev inequality on $\mathcal{S}$, the bootstrap assumption, and Lemma~\ref{stuffonv0}.

Next we consider the first term on the right hand side of~\eqref{formofFprime} in the case when $\Psi' \neq \alpha'$ and where the maximum number of $\nabla_4$ derivatives fall on $\Psi'$. In this case we will have
\begin{align}
\int_{-1}^{\tilde u}\int_0^{\tilde v}\int_{\mathcal{S}}|u|w^2\left|\mathscr{F}'_{i\frac{n-3}{2}}\right|^2 &\lesssim \sum_{j+k=i}\int_{-1}^{\tilde u}\int_0^{\tilde v}\int_{\mathcal{S}}|u|^{2j}\left|\nabla^j\psi\right|^2|u|^{2k}\left|\nabla^k\nabla_4^{\frac{n-3}{2}}\Psi'\right|^2|u|^{n+2-2\delta}v^{-1+2\delta}
\\ \nonumber &\lesssim \sup_{j,k\leq N}\int_{-1}^{\tilde u}\int_0^{\tilde v}\left(\int_{\mathcal{S}}|u|^{2j}\left|\nabla^j\psi\right|^2\right)\left(\int_{\mathcal{S}}|u|^{2k}\left|\nabla^k\nabla_4^{\frac{n-3}{2}}\Psi'\right|^2\right)|u|^{n+2-2\delta}v^{-1+2\delta}
\\ \nonumber &\lesssim \sup_{k \leq N}\left(1+\epsilon^{1-4\delta}\left\vert\left\vert\psi\right\vert\right\vert_{\mathfrak{S}}^2\right)\int_{-1}^{\tilde u}\int_0^{\tilde v}\int_{\mathcal{S}}|u|^{2k}\left|\nabla^k\nabla_4^{\frac{n-3}{2}}\Psi'\right|^2|u|^{n-2\delta}v^{-1+2\delta}
\\ \nonumber &\lesssim \left(1+\epsilon^{1-4\delta}A\right)\epsilon^{2\delta}+
\\ \nonumber &\qquad \left(1+\epsilon^{1-4\delta}A\right) \sup_{k\leq N}\int_{-1}^{\tilde u}\int_0^{\tilde v}\int_{\mathcal{S}}|u|^{2k}\left|\reallywidetilde{\nabla^k\nabla_4^{\frac{n-3}{2}}\Psi'}\right|^2|u|^{n+1-2\delta}v^{-2+2\delta}\frac{v}{|u|}
\\ \nonumber &\lesssim 1 + \left(1+\epsilon^{1-4\delta}A\right)\left\vert\left\vert \Psi'\right\vert\right\vert_{\mathfrak{T}}^2\epsilon
\\ \nonumber &\lesssim 1 + \left(1+\epsilon^{1-4\delta}A\right)A\epsilon
\\ \nonumber &\lesssim 1.
\end{align}
Here we used a Sobolev inequality, the bootstrap assumption, and Lemma~\ref{stuffonv0}. Next, it is easy to see that essentially the same argument also covers the case when not all of the $\nabla_4$ derivatives fall on the curvature term $\Psi'$; the point being that once we are not in the case when all of the $\nabla_4$ derivatives fall potentially on $\alpha'$, then the $\mathfrak{L}$ norm for every curvature component except $\alpha'$ controls a spacetime term with a weight $v^{-2+4\delta}$ and for $\alpha'$ we control $v^{-1}$ (as opposed to just $v^{-1+2\delta}$). This allows for the above argument to go through. It is easy to check that similar arguments suffice for the other terms in $\mathscr{F}'_{i\frac{n-3}{2}}$.

Next, we have the error terms proportional to $vw^2\left(\left|\mathscr{R}_{i\frac{n-3}{2}}\right|^2 + \left|\mathscr{F}_{i\frac{n-3}{2}}\right|^2\right)$. The $\mathscr{R}$ term is, of course, strictly easier to control than the previous $\mathscr{R}$ terms we discussed. The only difference between $\mathscr{F}$ and $\mathscr{F}'$ is for the curvature terms $\alpha'$; however, the fact that $\mathscr{F}$ is multiplied by $vw^2$ instead of $|u|w^2$ allows for the same argument that worked for the $\mathscr{F}'$ term to work for $\mathscr{F}$. Finally, the initial data term satisfies
\[\int_0^{\tilde v}\int_{\mathcal{S}}w^2\left|\reallywidetilde{\mathfrak{D}\alpha'}\right|^2|_{u=-1} \lesssim 1.\]

The up-shot is that we have 
\begin{align}\label{thebasicenergyestimate2}
&\sup_{-1 \leq u \leq \tilde u}\int_0^{\tilde v}\int_{\mathcal{S}}w^2\left|\reallywidetilde{\mathfrak{D}\alpha'}\right|^2 + \int_{-1}^{\tilde u}\int_0^{\tilde v}\int_{\mathcal{S}}|u|^{-1}w^2\left|\reallywidetilde{\mathfrak{D}\alpha'}\right|^2 + 
\\ \nonumber &\sup_{0 \leq v \leq \tilde v}\int_{-1}^{\tilde u}\int_{\mathcal{S}}w^2\left[\left|\reallywidetilde{\mathfrak{D}\beta}\right|^2 + \left|\reallywidetilde{\mathfrak{D}\nu}\right|^2\right]+ \int_{-1}^{\tilde u}\int_0^{\tilde v}\int_{\mathcal{S}}v^{-1}w^2\left[\left|\reallywidetilde{\mathfrak{D}\beta}\right|^2 + \left|\reallywidetilde{\mathfrak{D}\nu}\right|^2\right]  \lesssim 1.
\end{align}

The next Bianchi pair is $\left(\nu_{ABC},\left(R_{ABCD},\sigma_{AB}\right)\right)$. Similarly to our treatment of $\left(\alpha,\left(\beta,\nu\right)\right)$, we obtain:
\begin{align*}
\nabla_3\left(w\reallywidetilde{\mathfrak{D}\nu}\right)_{ABC} &+O\left(|u|^{-1}\right)w\reallywidetilde{\mathfrak{D}\nu}_{ABC} = 
\\ \nonumber &w2\nabla_{[A}\reallywidetilde{\mathfrak{D}\tau}_{A]C} -2 w\nabla_{[A}\reallywidetilde{\mathfrak{D}\sigma}_{B]C}  + w\mathscr{R}_{i\frac{n-3}{2}} + w\mathscr{F}_{i\frac{n-3}{2}} + wO\left(u^{-\frac{n+3+2i}{2}}\right),
\end{align*}
\begin{align*}
\nabla_4\left(w\reallywidetilde{\mathfrak{D}R}\right)_{ABCD} &+ \frac{1-2\delta}{2}v^{-1}w\reallywidetilde{\mathfrak{D}R}_{ABCD}= 
\\ \nonumber &w\nabla_B\reallywidetilde{\mathfrak{D}\nu}_{CDA} -w\nabla_A\reallywidetilde{\mathfrak{D}\nu}_{CDB}+ w\mathscr{R}_{i\frac{n-3}{2}} + \mathscr{F}_{i\frac{n-3}{2}} + wO\left(u^{-\frac{n+3+2i}{2}}\right),
\end{align*}
\begin{align*}
\nabla_4\left(w\reallywidetilde{\mathfrak{D}\sigma}\right)_{AB}&+ \frac{1-2\delta}{2}v^{-1}w\reallywidetilde{\mathfrak{D}\sigma}_{AB}=
\\ \nonumber &2w\nabla^C\reallywidetilde{\mathfrak{D}\nu}_{ABC} +w\mathscr{R}_{i\frac{n-3}{2}} + w\mathscr{F}_{i\frac{n-3}{2}}+wO\left(u^{-\frac{n+3+2i}{2}}\right).
\end{align*}

Now we will be able to mostly proceed analogously to the estimates we carried out for $\left(\alpha,\left(\beta,\nu\right)\right)$, but there are two important differences. Most noticeably, the second term on the right hand side of $\reallywidetilde{\mathfrak{D}\nu}$'s equation does not necessarily have a good sign and hence we will generate a spacetime term with an unfavorable sign proportional to
\[\int_{-1}^{\tilde u}\int_0^{\tilde v}\int_{\mathcal{S}}|u|^{-1}w^2\left|\reallywidetilde{\mathfrak{D}\nu}\right|^2.\]
However, the key point is that 
\[\int_{-1}^{\tilde u}\int_0^{\tilde v}\int_{\mathcal{S}}|u|^{-1}w^2\left|\reallywidetilde{\mathfrak{D}\nu}\right|^2 \lesssim \epsilon \int_{-1}^{\tilde u}\int_0^{\tilde v}\int_{\mathcal{S}}v^{-1}w^2\left|\reallywidetilde{\mathfrak{D}\nu}\right|^2,\]
and thus we can control this potentially dangerous spacetime term with the already established estimate~\eqref{thebasicenergyestimate2}. The second difference is that when we apply Cauchy-Schwarz after integrating by parts we should use the spacetime term from~\eqref{thebasicenergyestimate2} to absorb the $\reallywidetilde{\mathfrak{D}\nu}$ term. Other than these two caveats, everything proceeds as before.

It is now clear that one can systematically work down the entire Bianchi hierarchy, and eventually we will establish the desired estimates for each curvature component.

\end{proof}
\subsubsection{$n \geq 4$ and Even}
We now turn to the case of $n \geq 4$ and even.

\begin{proposition}\label{controlTneven}Let $n \geq 4$ and even. Then we have
\[\mathfrak{T} \leq C.\]
\end{proposition}
\begin{proof}Just as when $n$ was odd, we conjugate the equations for the Bianchi pari $\left(\alpha,\left(\beta,\nu\right)\right)$ by the weight $w$. However, in this case we set 
\[\mathfrak{D} \doteq \nabla^i\nabla_4^{\frac{n-4}{2}},\qquad w \doteq |u|^{\frac{n+2i}{2}}v^{-1/2 + \delta}.\]
We obtain (keeping Proposition~\ref{renormeven} in mind)
\begin{align*}
\nabla_3\left(w\reallywidetilde{\mathfrak{D}\alpha'}\right)_{AB} & = 
\\ \nonumber &-w\nabla^C\reallywidetilde{\mathfrak{D}\nu}_{C(AB)} + w\nabla_{(A}\reallywidetilde{\mathfrak{D}\beta}_{B)}  + w\reallywidetilde{\mathscr{R}}_{i\frac{n-4}{2}} + w\reallywidetilde{\mathscr{F}'}_{i\frac{n-4}{2}},
\end{align*}
\begin{align*}
\nabla_4\left(w\reallywidetilde{\mathfrak{D}\beta}\right)_A &+ \frac{1-2\delta}{2}v^{-1}w\reallywidetilde{\mathfrak{D}\beta}_A= 
\\ \nonumber &w\nabla^B\reallywidetilde{\mathfrak{D}\alpha'}_{BA} + w\mathscr{R}_{i\frac{n-4}{2}} + \mathscr{F}_{i\frac{n-4}{2}} + wO\left(|u|^{-\frac{n+2+2i}{2}}\log\left(\frac{v}{|u|}\right)\right),
\end{align*}
\begin{align*}
\nabla_4\left(w\reallywidetilde{\mathfrak{D}\nu}\right)_{ABC}&+ \frac{1-2\delta}{2}v^{-1}w\reallywidetilde{\mathfrak{D}\nu}=
\\ \nonumber &-2w\nabla_{[A}\reallywidetilde{\mathfrak{D}\alpha'}_{B]C} +w\mathscr{R}_{i\frac{n-4}{2}} + w\mathscr{F}_{i\frac{n-4}{2}}+wO\left(|u|^{-\frac{n+2+2i}{2}}\log\left(\frac{v}{|u|}\right)\right).
\end{align*}

There are two key differences with the case of $n$ odd. First of all, in $\alpha'$'s equation, the conjugation by the weight $w$ completely cancels the lower order term instead of leaving a positive multiple as in the case of $n$ odd. This will prevent us from getting a good spacetime estimate for $\alpha'$. Second of all, instead of seeing singularities like $v^{-1/2}$, the worst singular terms when $n$ is even blow-up like $\log\left(v\right)$ as $v \to 0$.

Carrying out the energy estimate as we did when $n$  was odd yields 
\begin{align}\label{thebasicenergyestimateeven}
&\sup_{\mathcal{R}}\left(\hat{v}^{-2\delta}\int_0^{\hat v}\int_{\mathcal{S}}w^2\left|\reallywidetilde{\mathfrak{D}\alpha'}\right|^2 + \hat v^{-2\delta}\int_{-1}^{\hat u}\int_{\mathcal{S}}w^2\left[\left|\reallywidetilde{\mathfrak{D}\beta}\right|^2 + \left|\reallywidetilde{\mathfrak{D}\nu}\right|^2\right]+ \hat v^{-2\delta}\int_{-1}^{\hat u}\int_0^{\hat v}\int_{\mathcal{S}}v^{-1}w^2\left[\left|\reallywidetilde{\mathfrak{D}\beta}\right|^2 + \left|\reallywidetilde{\mathfrak{D}\nu}\right|^2\right] \right)
\\ \nonumber &\qquad \lesssim \sup_{\mathcal{R}}\Bigg(\hat v^{-2\delta}\int_{-1}^{\hat u}\int_0^{\hat v}\int_{\mathcal{S}}\left[w^2\left(\left|\tilde{\mathscr{R}}_{i\frac{n-4}{2}}\right| + \left|\reallywidetilde{\mathscr{F}'}_{i\frac{n-4}{2}}\right|\right)\left|\reallywidetilde{\mathfrak{D}\alpha'}\right|\right] +
\\ \nonumber &\qquad \qquad \hat v^{-2\delta}\int_{-1}^{\hat u}\int_0^{\hat v}\int_{\mathcal{S}}\left[vw^2\left(\left|\mathscr{R}_{i\frac{n-4}{2}}\right|^2 + \left|\mathscr{F}_{i\frac{n-4}{2}}\right|^2 + |u|^{-n-2-2i}\log^2\left(\frac{v}{|u|}\right)\right)\right]+\hat v^{-2\delta}\int_0^{\hat v}\int_{\mathcal{S}}w^2\left|\reallywidetilde{\mathfrak{D}\alpha'}\right|^2|_{u=-1}\Bigg).
\end{align}

Note the fundamental difference with the case of $n$ odd in that we put a weight $\hat v^{-2\delta}$ outside in order to get the a scale-invariant norm; this creates some extra annoyance in estimating the nonlinear error terms.

We now need to discuss how we will control the various terms on the right hand side. For any $a > 0$, we have 
\[\hat v^{-2\delta}\int_{-1}^{\hat u}\int_0^{\hat v}\int_{\mathcal{S}}w^2\left|\reallywidetilde{\mathscr{R}}_{i\frac{n-4}{2}}\right|\left|\reallywidetilde{\mathfrak{D}\alpha'}\right| \lesssim a{\hat v}^{-2\delta}\sup_{\hat u}\int_0^{\hat v}\int_{\mathcal{S}}w^2\left|\reallywidetilde{\mathfrak{D}\alpha'}\right|^2 + a^{-1}\hat v^{-2\delta}\left(\int_{-1}^{\hat u}\left(\int_0^{\hat v}\int_{\mathcal{S}}w^2\left|\reallywidetilde{\mathscr{R}}_{i\frac{n-4}{2}}\right|^2\right)^{1/2}\right)^2.\]

The first term on the right hand side here can be absorbed into the left hand side of~\eqref{thebasicenergyestimateeven}. As for the second term, we first recall the schematic form:
\[\mathscr{R}_{ij} \doteq \nabla^i\nabla_4^j\left[\left(\mathring{\psi}+\mathring{\phi}\right)|u|^{-\frac{n-4}{2}}v^{\frac{n-4}{2}}\log\left(\frac{v}{|u|}\right)\mathcal{O}_{AB}\right].\]

As we did with $n$ odd, we first consider the maximally singular situation when all of the $\nabla_4$ derivatives fall on the $v^{\frac{n-4}{2}}$. In this case
\begin{align*}
&\hat v^{-2\delta}\left(\int_{-1}^{\hat u}\left(\int_0^{\hat v}\int_{\mathcal{S}}w^2\left|\reallywidetilde{\mathscr{R}}_{i\frac{n-4}{2}}\right|^2\right)^{1/2}\right)^2 
\\ \nonumber &\qquad \lesssim \hat v^{-2\delta}\sup_{j\leq N} \left(\int_{-1}^{\hat u}\left(\int_0^{\hat v}\int_{\mathcal{S}}|u|^{2j}\left[\left|\nabla^j\mathring{\psi}\right|^2 + \left|\nabla^j\mathring{\phi}\right|^2\right]v^{-1+2\delta}\log^2\left(\frac{v}{|u|}\right)\right)^{1/2}\right)^2
\\ \nonumber &\qquad \lesssim \hat v^{-2\delta}\left[\left\vert\left\vert \mathring{\psi}\right\vert\right\vert_{\mathfrak{V}}^2 + \sup_{j \leq N}\sup_{u,v}\int_{\mathcal{S}}\left|\nabla^j\mathring{\phi}\right|^2|u|^{2j+4-4\delta}v^{-2+4\delta}\right]\left(\int_{-1}^{\hat u}\left(\int_0^{\hat v}|u|^{-4+4\delta}v^{1-2\delta}\log^2\left(\frac{v}{|u|}\right)\right)^{1/2}\right)^2
\\ \nonumber &\qquad \lesssim A\hat v^{-2\delta}\left(\int_{-1}^{\hat u}\left(\int_0^{\hat v}|u|^{-4+6\delta}v^{1-4\delta}\right)^{1/2}\right)^2
\\ \nonumber &\qquad \lesssim A\frac{\hat v^{2-6\delta}}{|\hat u|^{2-6\delta}}
\\ \nonumber &\qquad \lesssim 1.
\end{align*}

Next, like in the case of $n$ odd, we consider the situation where not all of the derivatives fall on the $v$. In this case
\[\mathscr{R}_{ij} \doteq \nabla^i\nabla_4^{j-1}\left[\left(\nabla_4\mathring{\psi}+\nabla_4\mathring{\phi}\right)|u|^{-\frac{n-4}{2}}v^{\frac{n-4}{2}}\log\left(\frac{v}{|u|}\right)\mathcal{O}_{AB}\right].\]
Thus we will have
\begin{align*}
&\hat v^{-2\delta}\left(\int_{-1}^{\hat u}\left(\int_0^{\hat v}\int_{\mathcal{S}}w^2\left|\reallywidetilde{\mathscr{R}}_{i\frac{n-4}{2}}\right|^2\right)^{1/2}\right)^2 
\\ \nonumber &\qquad \lesssim \sum_{k=0}^{\frac{n-4}{2}-1}\sum_{j=0}^i\hat v^{-2\delta}\sup_{j\leq N} \left(\int_{-1}^{\hat u}\left(\int_0^{\hat v}\int_{\mathcal{S}}|u|^{2(i-j)}\left[\left|\nabla^{i-j}\nabla_4^{k+1}\mathring{\psi}\right|^2 + \left|\nabla^{i-j}\nabla_4^{k+1}\mathring{\phi}\right|^2\right]v^{2k+1+2\delta}\log^2\left(\frac{v}{|u|}\right)\right)^{1/2}\right)^2
\\ \nonumber &\qquad \lesssim \left[1 + \left\vert\left\vert \psi\right\vert\right\vert_{\mathfrak{S}}^2\right]\hat v^{-2\delta}\left(\int_{-1}^{\hat u}\left(\int_0^{\hat v}|u|^{-4}v^{1+2\delta}\log^2\left(\frac{v}{|u|}\right)\right)^{1/2}\right)^2
\\ \nonumber &\qquad \lesssim \left(1+A\right)\hat v^{-2\delta}\left(\int_{-1}^{\hat u}\left(\int_0^{\hat v}|u|^{-4+\delta}v^{1+\delta}\right)^{1/2}\right)^2
\\ \nonumber &\qquad \lesssim \left(1+A\right)\frac{\hat v^{2-\delta}}{|\hat u|^{2-\delta}}
\\ \nonumber &\qquad \lesssim 1.
\end{align*}

Next we turn to the error terms from the $\reallywidetilde{\mathscr{F}'}$.  Recall that
\begin{align}\label{formofFprime2}
\mathscr{F}'_{ml} &\doteq \sum_{i+j +k = m+l,i\leq l}\nabla^k\nabla_4^i\left(\psi^{j+1}\Psi'\right) + \sum_{i+j=m+l,i \leq l}\nabla^j\nabla_4^i\left(\zeta\psi\psi\right)
\\ \nonumber &\qquad + \sum_{i+j= m-1}\nabla^i\left(\slashed{Riem}^{j+1}\nabla_4^l\Psi'\right),
\end{align}
with the rule that if $\Psi'$ is equal to $\alpha'$ then one of the Ricci coefficients is given by a $\mathring{\psi}$.

We start with the first term on the right hand side of~\eqref{formofFprime2} in the case when $\Psi' = \alpha'$ and where the maximum number of $\nabla_4$ derivatives fall on $\alpha'$. In this case we will have 
\begin{align*}
&\hat{v}^{-2\delta}\left(\int_{-1}^{\hat u}\left(\int_0^{\hat v}\int_{\mathcal{S}}w^2\left|\reallywidetilde{\mathscr{F}'}_{i\frac{n-4}{2}}\right|^2\right)^{1/2}\right)^2
\\ \nonumber &\qquad \lesssim \hat{v}^{-2\delta}\sum_{j+k=i}\left(\int_{-1}^{\hat{u}}\left(\int_0^{\hat{v}}\int_{\mathcal{S}}|u|^{2j}\left|\nabla^j\mathring{\psi}\right|^2|u|^{2k}\left|\nabla^k\nabla_4^{\frac{n-4}{2}}\alpha'\right|^2|u|^nv^{-1+2\delta}\right)^{1/2}\right)^2
\\ \nonumber &\qquad \lesssim  \hat{v}^{-2\delta}\sup_{j,k \leq N}\left(\int_{-1}^{\hat{u}}\left(\int_0^{\hat{v}}\left(\int_{\mathcal{S}}|u|^{2j}\left|\nabla^j\mathring{\psi}\right|^2\right)\left(\int_{\mathcal{S}}|u|^{2k}\left|\nabla^k\nabla_4^{\frac{n-4}{2}}\alpha'\right|^2\right)|u|^nv^{-1+2\delta}\right)^{1/2}\right)^2
\\ \nonumber &\qquad \lesssim \left\vert\left\vert \mathring{\psi}\right\vert\right\vert_{\mathfrak{V}}^2\hat{v}^{-2\delta}\sup_{k\leq N}\left(\int_{-1}^{\hat{u}}\left(\int_0^{\hat{v}}\int_{\mathcal{S}}|u|^{2k}\left|\nabla^k\nabla_4^{\frac{n-4}{2}}\alpha'\right|^2|u|^{n-4+4\delta}v^{1-2\delta}\right)^{1/2}\right)^2
\\ \nonumber &\qquad \lesssim \left\vert\left\vert \mathring{\psi}\right\vert\right\vert_{\mathfrak{V}}^2\epsilon^{2-2\delta} + \left\vert\left\vert \mathring{\psi}\right\vert\right\vert_{\mathfrak{V}}^2\sup_{k\leq N}\left(\int_{-1}^{\hat{u}}\left(\int_0^{\hat{v}}\int_{\mathcal{S}}|u|^{2k}\left|\reallywidetilde{\nabla^k\nabla_4^{\frac{n-4}{2}}\alpha'}\right|^2|u|^{n-4+4\delta}v^{1-2\delta}\right)^{1/2}\right)^2
\\ \nonumber &\qquad \lesssim \epsilon^{2-2\delta}\left(A + A^2\right)
\\ \nonumber &\qquad \lesssim 1,
\end{align*}
where we used a Sobolev inequality, Lemma~\ref{stuffonv0}, and the bootstrap assumption.

Next we consider the first term on the right hand side of~\eqref{formofFprime} in the case when $\Psi' \neq \alpha'$ and where the maximum number of $\nabla_4$ derivatives fall on $\Psi'$. In this case we will have
\begin{align*}
&\hat{v}^{-2\delta}\left(\int_{-1}^{\hat u}\left(\int_0^{\hat v}\int_{\mathcal{S}}w^2\left|\reallywidetilde{\mathscr{F}'}_{i\frac{n-4}{2}}\right|^2\right)^{1/2}\right)^2
\\ \nonumber &\qquad \lesssim \hat{v}^{-2\delta}\sum_{j+k=i}\left(\int_{-1}^{\hat{u}}\left(\int_0^{\hat{v}}\int_{\mathcal{S}}|u|^{2j}\left|\reallywidetilde{\nabla^j\psi}\right|^2|u|^{2k}\left|\nabla^k\nabla_4^{\frac{n-4}{2}}\Psi'\right|^2|u|^nv^{-1+2\delta}\right)^{1/2}\right)^2
\\ \nonumber &\qquad \qquad + \hat{v}^{-2\delta}\sum_{j+k=i}\left(\int_{-1}^{\hat{u}}\left(\int_0^{\hat{v}}\int_{\mathcal{S}}|u|^{2j}\left|\nabla^j\psi\right|^2|u|^{2k}\left|\reallywidetilde{\nabla^k\nabla_4^{\frac{n-4}{2}}\Psi'}\right|^2|u|^nv^{-1+2\delta}\right)^{1/2}\right)^2
\\ \nonumber &\qquad \doteq I + II.
\end{align*}

We have
\begin{align*}
\left|I\right| &\lesssim \left\vert\left\vert \psi\right\vert\right\vert_{\mathfrak{S}}\hat{v}^{-2\delta}\sup_{k \leq N}\left(\int_{-1}^{\hat{u}}\left(\int_0^{\hat{v}}\int_{\mathcal{S}}|u|^{2k}\left|\nabla^k\nabla_4^{\frac{n-4}{2}}\Psi'\right|^2|u|^{n-3+4\delta}v^{-2\delta}\right)^{1/2}\right)^2
\\ \nonumber &\lesssim \left\vert\left\vert \psi\right\vert\right\vert_{\mathfrak{S}}\left(\frac{\hat{v}^{1-4\delta}}{|\hat u|^{1-4\delta}} + \hat v^{-2\delta}\sup_{k \leq N}\left(\int_{-1}^{\hat{u}}\left(\int_0^{\hat{v}}\int_{\mathcal{S}}|u|^{2k}\left|\reallywidetilde{\nabla^k\nabla_4^{\frac{n-4}{2}}\Psi'}\right|^2|u|^{n-3+4\delta}v^{-2\delta}\right)^{1/2}\right)^2
\right)
\\ \nonumber &\lesssim A\left(\epsilon^{1-4\delta} + \hat{v}^{-2\delta}|\hat{u}|^{-2+4\delta}\int_{-1}^{\hat{u}}\int_0^{\hat{v}}\int_{\mathcal{S}}|u|^{2k}\left|\reallywidetilde{\nabla^k\nabla_4^{\frac{n-4}{2}}\Psi'}\right|^2|u|^nv^{-2+2\delta}v^{2-4\delta}\right)
\\ \nonumber &\lesssim A\left(\epsilon^{1-4\delta} + A\epsilon^{2-4\delta}\right)
\\ \nonumber &\lesssim 1.
\end{align*}

Similarly,
\begin{align*}
\left|II\right| &\lesssim \hat{v}^{-2\delta}\sum_{j+k=i}\left(\int_{-1}^{\hat{u}}\left(\int_0^{\hat{v}}\int_{\mathcal{S}}|u|^{2j}\left|\nabla^j\psi\right|^2|u|^{2k}\left|\reallywidetilde{\nabla^k\nabla_4^{\frac{n-4}{2}}\Psi'}\right|^2|u|^nv^{-1+2\delta}\right)^{1/2}\right)^2
\\ \nonumber &\lesssim \left(1 + \epsilon^{1-4\delta}\left\vert\left\vert \psi\right\vert\right\vert_{\mathfrak{S}}\right)\hat{v}^{-2\delta}\sup_{k \leq N}\left(\int_{-1}^{\hat{u}}\left(\int_0^{\hat{v}}\int_{\mathcal{S}}|u|^{2k}\left|\reallywidetilde{\nabla^k\nabla_4^{\frac{n-4}{2}}\Psi'}\right|^2|u|^{n-2}v^{-1+2\delta}\right)^{1/2}\right)^2
\\ \nonumber &\lesssim  \left(1 + \epsilon^{1-4\delta}\left\vert\left\vert \psi\right\vert\right\vert_{\mathfrak{S}}\right)\hat{v}^{-2\delta}|\hat u|^{-1}\sup_{k \leq N}\int_{-1}^{\hat{u}}\int_0^{\hat v}\int_{\mathcal{S}}|u|^{2k}\left|\reallywidetilde{\nabla^k\nabla_4^{\frac{n-4}{2}}\Psi'}\right|^2|u|^nv^{-2+2\delta}v
\\ \nonumber &\lesssim  \left(1 + \epsilon^{1-4\delta}A\right)A\epsilon
\\ \nonumber &\lesssim 1.
\end{align*}
As usual, we have used a Sobolev inequality, Lemma~\ref{stuffonv0}, and the bootstrap assumption.

Next, just as in the case of $n$ odd, it is easy to see that essentially the same argument also covers the case when not all of the $\nabla_4$ derivatives fall on the curvature term $\Psi'$; the point being that once we are not in the case when all of the $\nabla_4$ derivatives fall potentially on $\alpha'$, then the $\mathfrak{L}$ norm for every curvature components controls a spacetime term with a weight $v^{-2+4\delta}$. This allows for the above argument to go through. It is easy to check that similar arguments suffice for the other terms in $\reallywidetilde{\mathscr{F}'}_{i\frac{n-4}{2}}$.

This concludes the treatment of the nonlinear error terms produced by the right hand side of $\alpha'$'s Bianchi equation. Next, we have the terms proportional to $vw^2\left(\left|\mathscr{R}_{i\frac{n-4}{2}}\right|^2 + \left|\mathscr{F}_{i\frac{n-4}{2}}\right|^2 + |u|^{-n-2-2i}\log^2\left(\frac{v}{|u|}\right)\right)$. However, this may be deal with by adapting the estimates above in a straightforward fashion; the point is simply that the presence of the $v$-weight compensates for the replacement of $\reallywidetilde{\mathscr{F}}$ by $\mathscr{F}$.

Finally, it is immediate that
\[\sup_{\hat v} \hat v^{-2\delta}\int_0^{\hat v}\int_{\mathcal{S}}w^2\left|\reallywidetilde{\mathfrak{D}\alpha'}\right|^2|_{u=-1} \lesssim \sup_{\hat v} \hat v^{-2\delta}\int_0^{\hat v}\int_{\mathcal{S}}v^{-1+2\delta}\left|\reallywidetilde{\mathfrak{D}\alpha'}\right|^2|_{u=-1} \lesssim 1.\]

\end{proof}
\subsubsection{$n = 2$}
Finally, we also have $n=2$.

\begin{proposition}Let $n = 2$.
\[\mathfrak{T} \leq C.\]
\end{proposition}
\begin{proof}
We omit the proof of this as it is analogous to and strictly easier than the case of $n \geq 3$ odd; the point is that following the same procedure as with $n$ odd yields a spacetime for $\alpha$, but one does not even have to deal with any singular terms as $v\to 0$.
\end{proof}
\subsection{Estimates for $\mathfrak{L}$}
In this section we will establish the desired estimates for $\mathfrak{L}$. Note that this estimate is only necessary if $n \geq 5$.
\subsubsection{The case of odd $n$}
We first consider the case of odd $n$.
\begin{proposition}Let $n \geq 5$ and odd. Then we have
\[\mathfrak{L} \leq C.\]
\end{proposition}
\begin{proof}For $i =1,\cdots,\frac{n-3}{2}$ and $0 \leq j \leq N+i$. We need to prove estimates for $\nabla^j\nabla_4^{\frac{n-3}{2}-i}\Psi'$. The proof will be by induction on $i$. 

Let's consider the base case $i = 1$ and $0 \leq j \leq N$. In this case we can just integrate in $v$ and use that we already control $\mathfrak{T}$. Let $\mathfrak{D} = \nabla^j\nabla_4^{\frac{n-3}{2}-1}$. We first consider the case of $\alpha'$.

Using Lemma~\ref{4commute} we find that
\[\left|\nabla_4\reallywidetilde{\mathfrak{D}\alpha'} \right| \lesssim \left|\nabla^j\nabla_4^{\frac{n-3}{2}}\alpha'\right| +\sum_{i+m+k = j-1}\left|\nabla^i\psi^{m+1}\right|\left|\nabla^k\nabla_4^{\frac{n-3}{2}}\alpha'\right| + \sum_{i+m+k = j}\left|\nabla^i\psi^{m+1}\right|\left|\nabla^k\nabla_4^{\frac{n-3}{2}-1}\alpha'\right|. \]

Integrating in the $v$-direction and applying Lemma~\ref{stuffonv0} we obtain
\begin{align*}
\left|\reallywidetilde{\mathfrak{D}\alpha'}\right| &\lesssim vu^{-\frac{2j+n+1}{2}} + 
\\ \nonumber &\qquad \int_0^v\left[\left|\reallywidetilde{\nabla^j\nabla_4^{\frac{n-3}{2}}\alpha'}\right| +\sum_{i+m+k = j-1}\reallywidetilde{\left|\nabla^i\psi^{m+1}\right|\left|\nabla^k\nabla_4^{\frac{n-3}{2}}\alpha'\right|} + \sum_{i+m+k = j}\reallywidetilde{\left|\nabla^i\psi^{m+1}\right|\left|\nabla^k\nabla_4^{\frac{n-3}{2}-1}\alpha'\right|}\right].
\end{align*}

From this we easily obtain, using the bootstrap assumption,
\begin{align}\label{someestimateforalpha}
\int_{\mathcal{S}}\left|\reallywidetilde{\mathfrak{D}\alpha'}\right|^2u^{n-2+2j}v^{-1} \lesssim v|u|^{-3} + \left(1+A\right)v^{1-2\delta}u^{-3+2\delta} \lesssim |u|^{-2}.
\end{align}

After integrating in $v$ or in $u$ and $v$ this is easily seen to establish the desired estimate. A completely analogous argument works for the other curvature components $\Psi$.

The more difficult case is when $j = N+1$ and thus  $\mathfrak{D} = \nabla^{N+1}\nabla_4^{\frac{n-3}{2}-1}$. Now, integrating in the $v$-direction would lose too many angular derivatives. Instead we proceed as follows:

We start by establishing a spacetime estimate for $\alpha'$.  By a slight variation of the steps in the proof of Proposition~\ref{writesomeeqns}, we have the following
\[\reallywidetilde{\mathfrak{D}\nabla_4\beta_A} = \nabla^B\left(\reallywidetilde{\mathfrak{D}\alpha'}\right)_{BA} + \reallywidetilde{\mathscr{R}}_{j\frac{n-3}{2}-1}+\reallywidetilde{\mathscr{F}}_{j\frac{n-3}{2}-1} + O\left(|u|^{-\frac{n+2+2j}{2}}v^{1/2}\right),\]
\[\reallywidetilde{\mathfrak{D}\nabla_4\nu_{ABC}} = -2\nabla_{[A}\left(\reallywidetilde{\mathfrak{D}\alpha'}\right)_{B]C} + \reallywidetilde{\mathscr{R}}_{j\frac{n-3}{2}-1}+\reallywidetilde{\mathscr{F}}_{j\frac{n-3}{2}-1}+ O\left(|u|^{-\frac{n+2+2j}{2}}v^{1/2}\right).\]

Now, the elliptic estimate of Lemma~\ref{dotheelliptic} and the estimate~\eqref{someestimateforalpha} yields
\begin{align}\label{whatwegetfromelliptic}
\int_{\mathcal{S}}&\left|\nabla\reallywidetilde{D\alpha'}\right|^2|u|^{n+2j}v^{-1} \lesssim 
\\ \nonumber &\int_{\mathcal{S}}\left[\left|\reallywidetilde{\mathfrak{D}\nabla_4\beta}\right|^2 + \left|\reallywidetilde{\mathfrak{D}\nabla_4\nu}\right|^2+\left|\reallywidetilde{\mathscr{R}}_{j\frac{n-3}{2}-1}\right|^2+\left|\reallywidetilde{\mathscr{F}}_{j\frac{n-3}{2}-1}\right|^2\right]|u|^{n+2j}v^{-1} + |u|^{-2}.
\end{align}

By Proposition~\ref{controlTnodd} we also have
\[\int_{-1}^{\hat{u}}\int_0^{\hat v}\int_{\mathcal{S}}\left[\left|\reallywidetilde{\mathfrak{D}\nabla_4\beta}\right|^2 + \left|\reallywidetilde{\mathfrak{D}\nabla_4\nu}\right|^2\right]|u|^{n+1-2\delta+2j}v^{-2+2\delta} \lesssim \mathfrak{T} \lesssim 1.\]

Combining these estimates with~\eqref{whatwegetfromelliptic} yields
\[\int_{-1}^{\hat u}\int_0^{\hat v}\int_{\mathcal{S}}\left|\nabla\reallywidetilde{\mathfrak{D}\alpha'}\right|^2|u|^{n+2j}v^{-1} \lesssim \int_{-1}^{\hat u}\int_0^{\hat v}\int_{\mathcal{S}}\left[\left|\reallywidetilde{\mathscr{R}}_{j\frac{n-3}{2}-1}\right|^2+\left|\reallywidetilde{\mathscr{F}}_{j\frac{n-3}{2}-1}\right|^2\right]|u|^{n+2j}v^{-1}  + 1.\]

Finally, the nonlinear terms in $\reallywidetilde{\mathscr{R}}_{j\frac{n-3}{2}-1}$ and $\reallywidetilde{\mathscr{F}}_{j\frac{n-3}{2}-1}$ can be estimated,  \emph{mutatis mutandis}, as we did in the equations for $\alpha'$ in Propositions~\ref{controlTnodd} and~\ref{controlTneven}. We finally obtain 

\begin{equation}\label{woohoo32}
\int_{-1}^{\hat u}\int_0^{\hat v}\int_{\mathcal{S}}\left|\nabla\reallywidetilde{\mathfrak{D}\alpha'}\right|^2|u|^{n+2j}v^{-1} \lesssim 1.
\end{equation}

Now we observe the following: The only difference in Proposition~\ref{controlTnodd} that occurs if we commute the Bianchi equations with $\nabla^{N+1}\nabla_4^{\frac{n-3}{2}-1}$ instead of $\nabla^{N+1}\nabla_4^{\frac{n-3}{2}}$ and conjugate by $|u|^{\frac{n-1+2(N+1)}{2}-\delta}v^{-1/2+\delta}$ instead of $|u|^{\frac{n+1+2(N+1)}{2}-\delta}v^{-1/2+\delta}$ is that the second term on the right hand side of the equation~\eqref{blahblahblahblah} for $\alpha'$ will no longer necessarily have a good sign. However, since we have already established the estimate~\eqref{woohoo32} this is clearly not a problem. Hence, we may re-run the proof of Proposition~\ref{controlTnodd} to establish the desired estimates for all of the other curvature components.

We thus have thus finished the base case $i = 1$. However, it is immediately clear that the induction procedure can be successfully carried by arguing in the same fashion.
\end{proof}
\subsubsection{The case of even $n$}
The case of even $n$ is essentially the same as odd $n$.
\begin{proposition}Let $n \geq 6$ and even. Then we have
\[\mathfrak{L} \leq C.\]
\end{proposition}
\begin{proof}One can use the same proof as in the case of $n$ odd.
\end{proof}
\subsection{Estimates for $\mathfrak{U}$}
In this section we will show that $\mathfrak{U}$ is bounded.
\begin{proposition}Let $n \geq 2$. Then we have
\[\mathfrak{U} \leq C.\]
\end{proposition}
\begin{proof}Let $\mathring{\Psi}$ be any of $\rho$, $\sigma$, $\tau$, $\underline{\beta}$, $\underline{\nu}$, or $\underline{\alpha}$. First of all, we know from Propositions~\ref{incomingdatan2},~\ref{incomingdatanodd}, and ~\ref{incomingdataneven} that $\mathring{\Psi}|_{v = 0} = 0$. Next, for every $0 \leq j \leq \tilde N-1$, the $\nabla_4$ equation for each $\mathring{\Psi}$ implies the following estimate:
\[\left|\nabla_4\nabla^j\mathring{\Psi}\right| \lesssim |u|^{-3-j} + \left|\reallywidetilde{\nabla^{j+1}\Psi'}\right| + \mathscr{R}_{j0} + \mathscr{F}_{j0},\]
where $\Psi' \neq \underline{\alpha}$. Furthermore, it follows easily from signature considerations and Proposition~\ref{Bianchit} that $\underline{\alpha}$ does not appear in $\mathscr{R}_{j0}$ or $\mathscr{F}_{j0}$.

Now, the desired estimate follows in a straightforward fashion after integrating over $\mathcal{S}$ and applying the fundamental theorem of calculus in the $v$-direction. It is straightforward to control the nonlinear error terms $\mathscr{R}_{j0}$ or $\mathscr{F}_{j0}$ using the previously established estimates for curvature and the bootstrap assumption for the Ricci coefficients.
\end{proof}
\subsection{Estimates for $\mathfrak{V}$}\label{estV}
The goal of this section will be to establish the desired bounds to $\mathfrak{V}$.

We proceed in two steps. First we follow naive strategy of estimating the relevant Ricci coefficients by a straightforward integrating of the null structure equations. This will yield the correct estimates, but only at the cost of losing an angular derivative (except for $\eta$ and ${\rm tr}\underline{\chi}'$).
\begin{proposition}\label{easiervanish}
We have 
\begin{equation}\label{estforeta}
\sup_{0\leq j \leq \tilde N}\int_{\mathcal{S}_{u,v}}\left|\nabla^j\eta\right|^2u^{4-4\delta+2j}v^{-2+4\delta}\lesssim 1,
\end{equation}
\begin{equation}\label{estforundetaoff}
\sup_{0\leq j \leq \tilde N-1}\int_{\mathcal{S}_{u,v}}\left|\nabla^j\underline{\eta}\right|^2u^{4-4\delta+2j}v^{-2+4\delta}\lesssim 1,
\end{equation}
\begin{equation}\label{undomegaestoff}
\sup_{0\leq j \leq \tilde N-1}\int_{\mathcal{S}_{u,v}}\left|\nabla^j\underline{\omega}\right|^2u^{6-8\delta+2j}v^{-4+8\delta}\lesssim 1,
\end{equation}
\begin{equation}\label{undomegaestoff2}
\sup_{0\leq j \leq \tilde N-1}\int_{\mathcal{S}_{u,v}}\left|\nabla^j\nabla_4\underline{\omega}\right|^2u^{4-4\delta+2j}v^{-2+4\delta}\lesssim 1,
\end{equation}
\begin{equation}\label{estfortrchi}
\sup_{0\leq j \leq \tilde N}\int_{\mathcal{S}_{u,v}}\left|\reallywidetilde{\nabla^j{\rm tr}\underline{\chi}}\right|^2u^{4-4\delta+2j}v^{-2+4\delta}\lesssim 1,
\end{equation}
\begin{equation}\label{estfortrchioff}
\sup_{0\leq j \leq \tilde N-1}\int_{\mathcal{S}_{u,v}}\left|\nabla^j\underline{\hat{\chi}}\right|^2u^{4-4\delta+2j}v^{-2+4\delta}\lesssim 1.
\end{equation}
\end{proposition}

\begin{proof}We first give the proof for $n$ odd.

Let's start with $\eta$. Let $m \leq N$. Commuting $\eta$'s $\nabla_4$ equation with $\nabla^m$ and using Lemma~\ref{4commute} yields
\[\left|\nabla_4\nabla^m\eta\right| \lesssim \sum_{i+j+k = m-1}\left|\nabla^i\psi^{j+1}\right|\left|\nabla^k\left(\chi\left(\eta-\underline{\eta}\right) - \beta\right)\right| + \sum_{i+j+k = m-1}\left|\nabla^i\psi^{j+1}\right|\left|\nabla^k\eta\right| .\]
Integrating the $v$-direction, using that $\eta$ vanishes when $v = 0$, and using the bootstrap assumptions immediately yields
\[\int_{\mathcal{S}}\left|\nabla^m\eta\right|^2 \lesssim v^{2-4\delta}u^{-4+4\delta-2m},\]
which establishes~\eqref{estforeta}.

For $\underline{\eta}$ we will need to use the $\nabla_3$ equation. Before entering into the details we give some general remarks about estimates for the Ricci coefficients via $\nabla_3$ equations. Consider an equation for a Ricci coefficient $\mathring{\psi}$ of the form
\begin{equation}\label{aformofnabla3}
\nabla_3\mathring{\psi} + \frac{c}{u}\mathring{\psi} = \mathring{\psi_1}\cdot\psi_2 + \mathring{\Psi},
\end{equation}
where $c \in \{0,1,2\}$, $\mathring{\psi}_1\cdot\psi_2$ denotes a general quadratic term involving Ricci coefficients, at least one of which is a $\mathring{\psi}$, and $\mathring{\Psi}$ denotes a curvature component which vanishes at $\{v = 0\}$. The first observation is that, as opposed to the case when we integrate $\nabla_4$ equations, we cannot put the second term on the left hand side onto the right hand side and treat it like an error term; if we did so it is easy to see that we would obtain a logarithmic divergence.

In order to avoid the difficulty of the logarithmic divergence, we conjugate~\eqref{aformofnabla3} by $u^c$ and obtain
\begin{equation}\label{aformofnabla32}
\nabla_3\left(u^c\mathring{\psi}\right) = \frac{\left(\Omega^{-1}-1\right)}{u}u^c\mathring{\psi} + u^c\mathring{\psi_1}\cdot\left(\psi_2|_{v=0} + \reallywidetilde{\psi}_2\right) + u^c\mathring{\Psi}.
\end{equation}

Now we multiply by $2u^c\mathring{\psi}$ and obtain
\begin{equation}\label{aformofnabla33}
\nabla_3\left|u^c\mathring{\psi}\right|^2 = \frac{\left(\Omega^{-1}-1\right)}{u}\left|u^c\mathring{\psi}\right|^2 + u^{2c}\mathring{\psi_1}\cdot\left(\psi_2|_{v=0} + \reallywidetilde{\psi}_2\right)\cdot\mathring{\psi} + u^{2c}\mathring{\Psi}\cdot\mathring{\psi}.
\end{equation}

Let $\kappa > 0$. Using our previous estimates for $\Omega^{-1}-1$ and $b$ from Proposition~\ref{metricest}, we may then integrate in the $u$ -direction, apply a Gronwall and Sobolev inequality, apply Cauchy-Schwarz, and obtain
\begin{align}\label{theeasyestimate}
\int_{\mathcal{S}_{\hat u,\hat v}}&|\hat u|^{2c}\left|\mathring\psi\right|^2 
\\ \nonumber &\lesssim \left(\int_{-1}^{\hat u}\left(\int_{\mathcal{S}_{u,\hat{v}}}|u|^{2c}\left|\mathring\psi_1\right|^2\left|\psi_2\right|^2\right)^{1/2}\, du\right)^2 + \left(\int_{-1}^{\hat u}\left(\int_{\mathcal{S}_{\hat u,\hat v}}|u|^{2c}\left|\mathring\Psi\right|^2\right)^{1/2}\, du\right)^2  + \int_{\mathcal{S}_{-1,\hat v}}\left|\mathring\psi\right|^2.
\\ \nonumber &\lesssim_{\kappa} |\hat u|^{-2\kappa}\int_{-1}^{\hat u}\int_{\mathcal{S}_{u,\hat v}}\left[\left|\mathring\psi_1\right|^2\left[\left|\psi_2|_{v=0}\right|^2 + \left|\reallywidetilde\psi_2\right|^2\right]+\left|\mathring\Psi\right|^2\right]|u|^{2c+1+2\kappa}  +   \int_{\mathcal{S}_{-1,\hat v}}\left|\mathring\psi\right|^2.
\end{align}

Let's now turn specifically to $\underline{\eta}$ with no angular derivatives. We write the equation for $\underline{\eta}$ as
\[\nabla_3\underline{\eta}+ \frac{1}{u}\underline{\eta} = -\frac{1}{n}\reallywidetilde{{\rm tr}\underline{\chi}}\underline{\eta} + \underline{\chi}\eta + \underline{\beta}.\]
This corresponds to~\eqref{aformofnabla3} with $c = 1$.

In this context~\eqref{aformofnabla32} becomes
\[\nabla_3\left(u\underline{\eta}\right)+ \frac{1-\Omega^{-1}}{u}\left(u\underline{\eta}\right) = -\frac{u}{n}\reallywidetilde{{\rm tr}\underline{\chi}}\underline{\eta} + u\hat{\underline\chi}\eta + \eta + \frac{u}{n}\reallywidetilde{{\rm tr}\underline{\chi}}\eta + u\underline{\beta}.\]

Now we estimate as in~\eqref{theeasyestimate}. Let's the examine the various terms on the right hand side starting with the curvature term:
\begin{align*}
|\hat{u}|^{-2\kappa}\int_{-1}^{\hat u}\int_{\mathcal{S}_{u,\hat v}}\left|\underline\beta\right|^2|u|^{3+2\kappa} &= |\hat u|^{-2\kappa}\int_{-1}^{\hat u}\int_{\mathcal{S}_{u,\hat v}}\left|\underline\beta\right|^2|u|^{6-4\delta}\hat v^{-2+4\delta}u^{-3+4\delta-2\kappa}\hat v^{2-4\delta}
\\ \nonumber &\lesssim \mathfrak{U}\hat v^{2-4\delta}|\hat u|^{-2\kappa}\int_{-1}^{\hat u}u^{-3+4\delta + 2\kappa}\, du
\\ \nonumber &\lesssim \frac{v^{2-4\delta}}{|\hat u|^{2-4\delta}}.
\end{align*}

Next, using the previously established estimate for $\eta$, we consider the $\eta$ term:
\begin{align*}
|\hat{u}|^{-2\kappa}\int_{-1}^{\hat{u}}\int_{\mathcal{S}_{u,\hat v}}\left|\eta\right|^2|u|^{1+2\kappa} &= |\hat u|^{-2\kappa}\int_{-1}^{\hat u}\int_{\mathcal{S}_{u,\hat v}}|\eta|^2|u|^{4-4\delta}\hat v^{-2+4\delta}|u|^{-3+2\kappa+4\delta}\hat{v}^{2-4\delta}
\\ \nonumber &\lesssim |\hat u|^{-2\kappa}\hat v^{2-4\delta}\int_{-1}^{\hat u}|u|^{-3+2\kappa+4\delta}\, du
\\ \nonumber &\lesssim \frac{v^{2-4\delta}}{|\hat u|^{2-4\delta}}.
\end{align*}

Lastly, using a Sobolev inequality and the bootstrap assumption we treat the quadratic terms $\mathring{\psi}\cdot\mathring{\psi}$:
\begin{align*}
&|\hat u|^{-2\kappa}\int_{-1}^{\hat u}\int_{\mathcal{S}_{u,\hat v}}\left|\mathring{\psi}_1\right|^2\left|\mathring{\psi}_2\right|^2|u|^{3+2\kappa}\, du 
\\ \nonumber &\lesssim  \sup_{j,k \leq \tilde N}|\hat u|^{-2\kappa}\int_{-1}^{\hat u}\left(\int_{\mathcal{S}_{u,\hat v}}|u|^{2j}\left|\nabla^j\mathring{\psi}_1\right|^2\right)\left(\int_{\mathcal{S}_{u,\hat v}}|u|^{2k}\left|\nabla^k\mathring{\psi}_2\right|^2\right)|u|^{3+2\kappa}\, du
\\ \nonumber &\lesssim A\hat v^{4-8\delta}|\hat u|^{-2\kappa}\int_{-1}^{\hat u}|u|^{-5+4\delta+2\kappa}\, du
\\ \nonumber &\lesssim A\frac{\hat v^{4-8\delta}}{|\hat u|^{4-8\delta}}
\\ \nonumber &\lesssim \frac{v^{2-4\delta}}{|\hat u|^{2-4\delta}}.
\end{align*}

The above estimate along with Proposition~\ref{incomingdatanodd} imply
\begin{align}\label{someprogressiguess}
\int_{\mathcal{S}}\left|\underline{\eta}\right|^2&\lesssim \frac{v^{2-4\delta}}{|u|^{5-4\delta}} .
\end{align}

Note that the use of  $\mathfrak{U}$ involves the ``loss'' of an angular derivative; this is what ultimately restricts us to establishing the estimate for up $\tilde N - 1$ derivatives of $\eta$. 

Arguing in a similar fashion after commuting with $\nabla^m$ and using Lemma~\ref{3commute} allows one to conclude that 
\[\sup_{0\leq j \leq \tilde N-1}\int_{\mathcal{S}_{u,v}}\left|\nabla^j\underline{\eta}\right|^2u^{4-4\delta+2j}v^{-2+4\delta}\lesssim 1.\]

For $\underline{\omega}$ we have
\[\nabla_4\underline{\omega} = \frac{1}{2}\rho + \frac{1}{4}\left|\underline\eta\right|^2 - \frac{1}{4}\left|\eta\right|^2 + 2\underline\omega \omega + 3\left|\zeta\right|^2 - \left|\nabla \log\Omega\right|^2,\]
which can be written schematically as
\[\nabla_4\underline{\omega} = \mathring{\Psi} + \mathring{\psi}\cdot\mathring{\psi}.\]

In particular, arguing as above easily leads to the estimate 
\begin{equation}\label{undomegaest}
\sup_{0\leq j \leq \tilde N-1}\int_{\mathcal{S}_{u,v}}\left|\nabla^j\underline{\omega}\right|^2u^{6-8\delta+2j}v^{-4+8\delta}\lesssim 1.
\end{equation}

Similarly, we obtain~\eqref{undomegaestoff2}.

Next we turn to $\reallywidetilde{{\rm tr}\underline{\chi}}$. From the $\nabla_3$ equation for ${\rm tr}\underline{\chi}$, we may derive
\[\nabla_3{\rm tr}\reallywidetilde{\underline{\chi}} + \frac{2}{u}\reallywidetilde{{\rm tr}\underline{\chi}} = \frac{(\Omega^{-1}-1)n}{u^2} - \frac{1}{n}\left(\reallywidetilde{{\rm tr}\underline{\chi}}\right)^2 -\left|\hat{\underline{\chi}}\right|^2 - 2\underline{\omega}\left(\frac{n}{u} +\reallywidetilde{{\rm tr}\underline{\chi}}\right).\]
Multiplying through by $u^2$ then yields
\[\nabla_3\left(u^2\reallywidetilde{{\rm tr}\underline{\chi}}\right) =  -\frac{2\left(1-\Omega^{-1}\right)}{u}\left(u^2\reallywidetilde{{\rm tr}\underline{\chi}}\right) + (\Omega^{-1}-1)n - u^2\left(\reallywidetilde{{\rm tr}\underline{\chi}}\right)^2 -u^2\left|\hat{\underline{\chi}}\right|^2 - 2\underline{\omega}\left(un +u^2{\rm tr}\reallywidetilde{\underline{\chi}}\right).\]
Now we proceed exactly like we did for $\eta$. The only difference is that when bound the term coming from $\underline{\omega}$ we need to use the improved decay in $v$. We eventually obtain
\[\int_{\mathcal{S}_{u,v}}\left|\reallywidetilde{{\rm tr}\underline{\chi}}\right|^2u^{4-4\delta+2j}v^{-2+4\delta}\lesssim 1.\]

Since there is no curvature term on the right hand side, we do not have the derivative loss problem we experience in the other $\nabla_3$ equations. After commuting with $\nabla^j$ and repeating the estimate above, we easily obtain
\[\sup_{0\leq j \leq N}\int_{\mathcal{S}_{u,v}}\left|\nabla^j\reallywidetilde{{\rm tr}\underline{\chi}}\right|^2u^{4-4\delta+2j}v^{-2+4\delta}\lesssim 1.\]

Finally, the desired estimates for $\hat{\underline\chi}$ follows from its $\nabla_4$ equation in a similar to the estimates above. The presence of the $\nabla\underline\eta$ of the right hand side of the equation is the reason we only get estimates up to $\nabla^{N-1}\hat{\underline{\chi}}$.

It is clear that the proof for $n$-even works analogously.
\end{proof}

In the next lemma, we will use elliptic estimates to recover the estimates for the top-order angular derivatives of the Ricci coefficients. 
\begin{proposition}We have 
\begin{equation}\label{estforundeta}
\int_{\mathcal{S}_{u,v}}\left|\nabla^{\tilde N}\underline{\eta}\right|^2u^{4-4\delta+2N}v^{-2+4\delta}\leq C,\end{equation}
\begin{equation}\label{undomegaest}
\int_{\mathcal{S}_{u,v}}\left|\nabla^{\tilde N}\underline{\omega}\right|^2u^{6-8\delta+2N}v^{-4+8\delta}\leq C,
\end{equation}
\begin{equation}\label{estfortrchi}
\int_{\mathcal{S}_{u,v}}\left|\nabla^{\tilde N}\underline{\hat{\chi}}\right|^2u^{4-4\delta+2N}v^{-2+4\delta}\leq C.\end{equation}
\end{proposition}
\begin{proof}We will give the proof for odd $n$. The case of even $n$ is completely analogous.

We start with $\underline{\eta}$. The $\nabla_3$ equation for $\underline{\eta}$ can be written schematically as
\[\nabla_3\underline{\eta} + \frac{1}{u}\underline{\eta} = \mathring{\psi}\cdot\underline{\eta} + \psi\cdot\mathring{\psi} +\underline{\beta}.\]
Let $m = \tilde N - 1$, we commute with $\nabla^m$ and use Lemma~\ref{3commute} to obtain
\[\left|\nabla_3\nabla^m\underline{\eta} + \frac{1+m}{u}\nabla^m\underline{\eta} - \nabla^m\underline{\beta}\right| \lesssim \sum_{i+j+k=m,}\nabla^i\mathring{\psi}^{j+1}\nabla^k\psi.\]

Finally, we commute through divergence relative to the last index of $\nabla^m\underline{\eta}$:
\begin{equation}\label{commuteyayyay}
\left|\nabla_3\nabla^A\nabla^m\underline{\eta}_A + \frac{2+m}{u}\nabla^A\nabla^m\underline{\eta}_A - \nabla^A\nabla^m\underline{\beta}_A\right| \lesssim \sum_{i+j+k=m+1,}\nabla^i\mathring{\psi}^{j+1}\nabla^k\psi
\end{equation}

Next, we observe that by signature considerations, in the $\nabla_3$ equation for $\rho$, the only curvature terms that appear on the right hand side are $\mathring{\Psi}$. Thus, a slight variation of Proposition~\ref{writesomeeqns} yields the following equation for $\rho$:
\begin{align*}
&\left|\nabla_3\nabla^m\rho + \nabla^A\nabla^m\underline{\beta}_A\right| \lesssim
\\ \nonumber &\qquad \sum_{i+j+k=m+1,}\nabla^i\mathring{\psi}^{j+1}\nabla^k\psi + \sum_{i+j = m}\nabla^i\psi\nabla^j\mathring{\Psi} + \sum_{i+j = m}\nabla^i\left(\slashed{Riem}^{j+1}\mathring{\Psi}\right).
\end{align*}

Combining this with~\eqref{commuteyayyay} yields
\begin{align}\label{iguessthisiswillhavetodo}
&\left|\nabla_3\left(\nabla^A\nabla^m\underline{\eta}_A+\nabla^m\rho\right) + \frac{2+m}{u}\left(\nabla^A\nabla^m\underline{\eta}_A+\nabla^m\rho\right)\right| \lesssim 
\\ \nonumber &\qquad \sum_{i+j+k=m+1,}\nabla^i\mathring{\psi}^{j+1}\nabla^k\psi + \sum_{i+j = m}\left[\nabla^i\psi + |u|^{-i-i}\right]\nabla^j\mathring{\Psi} + \sum_{i+j = m}\nabla^i\left(\slashed{Riem}^{j+1}\mathring{\Psi}\right).
\end{align}

Integrating this and estimating as in Propsition~\ref{easiervanish}, using our previous estimates, using Proposition~\ref{incomingdatanodd}, and appealing to the bootsrap assumption yields
\[\int_{\mathcal{S}_{u,v}}\left|\nabla^A\nabla^m\underline{\eta}_A + \nabla^m\rho\right|^2u^{4+2m} \lesssim \frac{v^{2-4\delta}}{| u|^{2-4\delta}} \Rightarrow \]
\begin{equation}\label{estimatefordereta}
\int_{\mathcal{S}_{u,v}}\left|\nabla^A\nabla^m\underline{\eta}_A\right|^2\lesssim \frac{v^{2-4\delta}}{|u|^{6+2m-4\delta}}
\end{equation}

Next, arguing in a similar fashion using the constraint equation~\eqref{antisig2}, the bootstrap assumption, and our previous estimates, one easily establishes 
\[\int_{\mathcal{S}_{u,v}}\left|\nabla_{[A}\nabla^m\underline{\eta}_{B]}\right|^2 \lesssim C\frac{v^{2-4\delta}}{|u|^{6+2m-4\delta}} .\]

Combining these two estimates with Lemma~\ref{dotheelliptic} finally yields
\[\int_{\mathcal{S}_{u,v}}\left|\nabla^{m+1}\underline{\eta}\right|^2 \leq C\frac{v^{2-4\delta}}{|u|^{6+2m-4\delta}} .\]

This establishes~\eqref{estforundeta}.

Now we turn to $\underline{\omega}$. This time we set $m = \tilde N - 2$ and we start by deriving, using signature considerations as we did above,
\begin{align}\label{woohoo}
-\nabla_4\nabla^m\underline{\beta}_A &= \nabla_A\nabla^m\rho + \nabla^B\nabla^m\sigma_{BA}  
\\ \nonumber &\qquad + \sum_{i+j+k=m+1,}\nabla^i\mathring{\psi}^{j+1}\nabla^k\psi + \sum_{i+j = m}\left[\nabla^i\psi\nabla^j\mathring{\Psi} + \nabla^i\mathring{\psi}\nabla^i\Psi\right] + \sum_{i+j = m}\nabla^i\left(\slashed{Riem}^{j+1}\mathring{\Psi}\right).
\end{align}

Next, using the anti-symmetry of $\sigma$, we note that
\begin{align*}
\nabla^A\nabla^B\nabla_{C_1C_2\cdots C_m}\sigma_{AB} &= \left[\nabla^A,\nabla^B\right]\nabla_{C_1C_2\cdots C_m}\sigma_{AB}
\\ \nonumber &= \sum_{i=1}^m\slashed{R}^{BA\ D}_{\ \ \ C_i}\nabla_{C_1\cdots C_{i-1}DC_{i+1}\cdots C_m}\sigma_{AB} + \slashed{R}^{BA\ D}_{\ \ \ A}\nabla^m\sigma_{DB} + \slashed{R}^{BA\ D}_{\ \ \ B}\nabla^m\sigma_{AD}
\\ \nonumber &= \slashed{Riem}\cdot\nabla^m\sigma.
\end{align*}

In particular, we easily obtain
\begin{align}\label{woohoo2}
-\nabla_4\nabla^A\nabla^m\underline{\beta}_A &= \slashed{\Delta}\nabla^m\rho + \slashed{Riem}\cdot\nabla^m\sigma
\\ \nonumber &\qquad + \sum_{i+j+k=m+2}\nabla^i\mathring{\psi}^{j+1}\nabla^k\psi + \sum_{i+j = m+1}\left[\nabla^i\psi\nabla^j\mathring{\Psi} + \nabla^i\mathring{\psi}\nabla^i\Psi\right]+ \sum_{i+j = m+1}\nabla^i\left(\slashed{Riem}^{j+1}\mathring{\Psi}\right).
\end{align}

We also may easily derive
\[\nabla_4\slashed{\Delta}\nabla^m\underline{\omega} = \frac{1}{2}\slashed{\Delta}\nabla^m\rho + \sum_{i+j+k=m+2}\nabla^i\mathring{\psi}^{j+1}\nabla^k\mathring{\psi}.\]
In turn, we obtain
\begin{align}\label{gettingthereundomega}
&\left|\nabla_4\left(\slashed{\Delta}\nabla^m\underline{\omega} + \nabla^A\nabla^m\underline{\beta}_A\right)\right| \\ \nonumber  &\lesssim \slashed{Riem}\cdot\nabla^m\sigma + \sum_{i+j+k=m+2}\nabla^i\mathring{\psi}^{j+1}\nabla^k\psi + \sum_{i+j = m+1}\left[\nabla^i\psi\nabla^j\mathring{\Psi} + \nabla^i\mathring{\psi}\nabla^i\Psi\right] + \sum_{i+j = m+1}\nabla^i\left(\slashed{Riem}^{j+1}\mathring{\Psi}\right).
\end{align}

Arguing as above, we easily obtain
\[\int_{\mathcal{S}_{u,v}}\left|\Delta \nabla^m\underline{\omega} + \nabla^A\nabla^m\underline{\beta}_A\right|^2 \lesssim \frac{v^{4-8\delta}}{u^{2m+10-8\delta}}.\]
This then implies 
\[\int_{\mathcal{S}_{u,v}}\left|\slashed{\Delta} \nabla^m\underline{\omega} \right|^2 \lesssim \frac{v^{4-8\delta}}{u^{2m+10-8\delta}}.\]
Now the estimate for $\underline{\omega}$ is finished by an application of Lemma~\ref{dotheomegaelliptic}.

Lastly, there is $\hat{\underline{\chi}}$. The desired estimates follow in a straightfoward fashion using the elliptic estimate of Lemma~\ref{dotheelliptic}, the constraint equations~\eqref{cod2} and ~\eqref{tcod2}, and the fact that we already have the desired estimates for $\nabla{\rm tr}\underline{\chi}$ from Proposition~\ref{easiervanish}.

\end{proof}

Finally, combining the above two propositions immediately yields

\begin{proposition}We have
\[\mathfrak{V} \leq C.\]
\end{proposition}

\subsection{Estimates for $\mathfrak{S}$}\label{estS}
In this section we will provide the estimates for $\mathfrak{S}$.
\begin{proposition}We have
\[\mathfrak{S} \leq C.\]
\end{proposition}
\begin{proof} 

If $\psi$ is one of ${\rm tr}\chi$, $\hat{\chi}$, $\eta$, or $\underline{\omega}$, then $\psi$ will satisfy an equation schematically of the form
\[\nabla_4\psi = \psi\cdot\psi + \Psi.\]
Commuting with $\nabla^i$ and $\nabla_4^j$ yields the schematic equation
\[\nabla_4\nabla^i\nabla_4^j\psi = \sum_{k+l=i}\nabla^k\nabla_4^j\left(\psi^{l+1}\psi\right) + \nabla^k\nabla_4^j\Psi \Rightarrow \]
\[\nabla_4\reallywidetilde{\nabla^i\nabla_4^j\psi} = O\left(u^{-2}\right) + \sum_{k+l=i}\nabla^k\reallywidetilde{\nabla_4^j\left(\psi^{l+1}\psi\right)} + \reallywidetilde{\nabla^k\nabla_4^j\Psi}.\]

The desired estimate for $\reallywidetilde{\nabla^i\nabla_4^j\psi}$ then follows in a straightforward fashion by inducting on $i$ and $j$, integrating the $v$-direction, using the previously established estimates, and the bootstrap assumption.

Now we turn the case when $\psi$ is one of the Ricci coefficients $\reallywidetilde{{\rm tr}\underline{\chi}}$, $\hat{\underline{\chi}}$, $\underline{\eta}$ or $\omega$. Let's denote these by $\psi_0$. For these we will instead use the corresponding $\nabla_3$ equation. After replacing ${\rm tr}\underline{\chi}$ with $\frac{n}{u} + \reallywidetilde{{\rm tr}\underline{\chi}}$, the general form of these equations is as follows:
\[\nabla_3\psi_0 + \frac{c}{u}\psi_0 = \psi_1\cdot\psi_2 + \mathring{\psi}_3u^{-1} + \Psi.\]
For $\reallywidetilde{{\rm tr}\underline{\chi}}$, $\hat{\underline{\chi}}$, $\underline{\eta}$ or $\omega$ we will have $c = 2,2,1,0$ respectively. Furthermore, the $\mathring{\psi}_3$ term only appears in the case of $\reallywidetilde{{\rm tr}\underline{\chi}}$. We now explain the general heuristic for estimating $\psi$. (This is, of course, similar to the method in Proposition~\ref{easiervanish}.) For simplicity we will first prove the estimate without any angular differentiation. It will be clear then how to handle the general case.

Let $m$ denote the maximal number of of $\nabla_4$ derivatives (this will be $0$ if $n=2$, $\frac{n-3}{2}$ if $n \geq 3$ and odd, and $\frac{n-4}{2}$ if $n \geq 4$ and even). The first observation is that the following strong estimates for $\nabla^j\psi_0$ with $j < m$ follow immediately from the bootstrap assumption:
\begin{align}\label{wowsomuchbetter}
\int_{\mathcal{S}_{\hat u,\hat v}}\left|\reallywidetilde{\nabla_4^j\psi_0}\right|^2 &\lesssim \int_{\mathcal{S}_{\hat u,\hat v}}\left(\int_0^{\hat v}\left|\nabla^{j+1}\psi_0\right|\, dv\right)^2
\\ \nonumber &\lesssim \int_{\mathcal{S}_{\hat u,\hat v}}\left(\int_0^{\hat v}\left(u^{-j-2}+\left|\reallywidetilde{\nabla^{j+1}\psi_0}\right|\right)\, dv\right)^2
\\ \nonumber &\lesssim \frac{\hat{v}^2}{|\hat{u}|^{2j+4}} + \frac{\hat{v}^{3-4\delta}}{|\hat u|^{2j+5-4\delta}}\int_{\mathcal{S}_{\hat u,\hat v}}\left|\reallywidetilde{\nabla_4^{j+1}\psi_0}\right|^2|\hat u|^{2(j+1)+3-4\delta}\hat v^{-1+4\delta}
\\ \nonumber &\lesssim \frac{\hat{v}^2}{|\hat{u}|^{2j+4}}. 
\end{align}
Note that this is stronger than what the norm $\mathfrak{S}$ dictates.

Similarly, one may establish
\[\sup_{j \leq m-1}\sup_{i \leq N + m-1-j}\int_{\mathcal{S}_{\hat u,\hat v}}\left|\reallywidetilde{\nabla^i\nabla_4^j\psi_0}\right|^2|\hat u|^{2(i+j)+4\delta}|\hat v|^{-2} \lesssim 1.\]

For $\nabla_4^m\psi_0$ we commute with $\nabla_4^m$ and use Lemma~\ref{34commute}:
\begin{align}\label{commutethatnabla3equation}
\nabla_3\nabla_4^m\psi_0 &+ \frac{c}{u}\nabla_4^m\psi_0 = 
\\ \nonumber &\sum_{m_1+m_2=m}\nabla_4^{m_1}\psi_1\cdot\nabla_4^{m_2}\psi_2 + \nabla_4^m\Psi +\nabla_4^m\mathring{\psi}_3|u|^{-1}
\\ \nonumber &+ \sum_{i+j+k = m, k\neq m}\left|\nabla_4^i\psi^j\right|\left|\nabla_4^k\left(\psi_1\cdot\psi_2 + \Psi+\mathring{\psi}_3|u|^{-1}\right)\right|
\\ \nonumber &\sum_{i+j_1+j_2+k=m}\left[\left|\nabla_4^i\mathring{\psi}\right| + \left|\nabla\nabla_4^{i-1}\mathring{\psi}\right|\right]\left[\left|\nabla\nabla_4^{j_1-1}\psi^{j_2}\right|+\left|\nabla_4^{j_1}\psi^{j_2}\right|\right]\left[\left|\nabla\nabla_4^{k-1}\psi_0\right| + \left|\nabla_4^k\psi_0\right|\right] \doteq \mathscr{H}_0
\end{align}

Next, we write the equation in terms of $\reallywidetilde{\nabla^m\psi_0}$ and obtain:
\begin{equation*}
\nabla_3\reallywidetilde{\nabla_4^m\psi_0} + \frac{c}{u}\reallywidetilde{\nabla_4^m\psi_0} = \reallywidetilde{\mathscr{H}_0}.
\end{equation*}

Let $\kappa > 0$. Now, conjugating by $|u|^c$, multiplying by $|u|^c\reallywidetilde{\nabla_4^m\psi_0}$, integrating and estimating as in Proposition~\ref{easiervanish} yields
\begin{equation}\label{wowthisiswhatweendedupwith}
\int_{\mathcal{S}_{\hat u,\hat v}}|u|^{2c}\left|\reallywidetilde{\nabla_4^m\psi_0}\right|^2 \lesssim |\hat u|^{-2\kappa}\int_{-1}^{\hat{ u}}\int_{\mathcal{S}_{u,\hat v}}\left|\mathscr{H}_0\right|^2|u|^{2c+1+2\kappa}.
\end{equation}

Before we deal with the generic case, we will need to argue a little more carefully in the case when $m = 0$, i.e., when $n=2,3,4$. In this case Proposition~\ref{easiervanish} has already dealt with all the necessary estimates except for $\omega$. For $\omega$ we have that $c = 0$ and
\[\left|\mathscr{H}_0\right| \lesssim \left|\rho\right| + \left|\eta\right|^2 + \left|\underline{\eta}\right|^2 + \left|\underline\omega\right|\left|\omega\right|.\]
It then straightforward to see that the desired estimate follows from the previously established estimates. The key points are that $\eta$ and $\underline{\eta}$ enter quadratically and that $\underline{\omega}$ has an improved decay estimate in $v$.

Now we return to the general case of estimating the right hand side of~\eqref{wowthisiswhatweendedupwith} when  $m\geq 1$. In what follows we will freely use Sobolev inequalities and the bootstrap without explicit comment. Keeping in mind that $\left|\reallywidetilde{fg}\right| \lesssim \left|f\reallywidetilde{g}\right| + \left|\reallywidetilde{f}g\right|$, let's consider that various terms that show up when we try to estimate the right hand side of~\eqref{wowthisiswhatweendedupwith}. Note that the curvature term is never given by $\alpha$ and thus it will  be straightforward to estimate. When $n \geq 3$ and odd we have:
\begin{align*}
|\hat u|^{-2\kappa}\int_{-1}^{\hat u}\int_{\mathcal{S}_{u,\hat v}}\left|\reallywidetilde{\nabla_4^m\Psi}\right|^2|u|^{2c+1+2\kappa}  &\lesssim |\hat u|^{-2\kappa}\int_{-1}^{\hat u}\int_{\mathcal{S}_{u,\hat v}}\left|\reallywidetilde{\nabla_4^m\Psi}\right|^2|u|^{4+2m-2\delta}\hat{v}^{-1+2\delta}|u|^{2c-3-2m+2\delta+2\kappa}\hat{v}^{1-2\delta}
\\ \nonumber &\lesssim \frac{\hat v^{1-2\delta}}{|u|^{3+2m-2c-2\delta}}.
\end{align*}
The key point here is that $3+2m-2c-2\delta > 0$ since $m \geq 1$ and $c \leq 2$.

Next we consider the terms on the right hand side of~\eqref{wowthisiswhatweendedupwith} which do not contain any $\nabla_4^m$ derivative or an angular derivative. A generic such term can be represented by letting $i+j+k = m$ for $i \neq m$ and considering the following:
\begin{align*}
|\hat u|^{-2\kappa}\int_{-1}^{\hat u}\int_{\mathcal{S}_{u,\hat v}}\left|\reallywidetilde{\nabla_4^i\psi}\right|^2\left|\nabla_4^j\psi^{k+1}\right|^2|u|^{2c+1+2\kappa}  &\lesssim |\hat u|^{-2\kappa}\int_{-1}^{\hat u}\int_{\mathcal{S}_{u,\hat v}}\left|\reallywidetilde{\nabla_4^i\psi}\right|^2|u|^{2c-1+2\kappa-2j-2k}
\\ \nonumber &\lesssim  |\hat u|^{-2\kappa}\int_{-1}^{\hat u}\int_{\mathcal{S}_{u,\hat v}}\left|\reallywidetilde{\nabla_4^i\psi}\right|^2\frac{|u|^{2i+4}}{\hat v^2}|u|^{2c-5+2\kappa-2m}\hat v^2
\\ \nonumber &\lesssim \hat{v}^2|\hat u|^{-2\kappa}\int_{-1}^{\hat{u}}|u|^{2c-5+2\kappa-2m}
\\ \nonumber &\lesssim \frac{\hat v^2}{|u|^{4+2m-2c-4}}.
\end{align*}
Here we used the estimate~\eqref{wowsomuchbetter}, and the fact that $-2c+5-2\kappa+2m > 1$. Similarly, still relying on~\eqref{wowsomuchbetter}, a analogous estimate will work for terms which contain an angular derivative as long as they are proportional to $\nabla\nabla_4^j$ for $j \leq m-2$. 

Next we consider the term where a term with $\nabla\nabla_4^{m-1}$ shows up. Here we cannot use~\eqref{wowsomuchbetter}. However, the point will be that such terms in~\eqref{commutethatnabla3equation} are always multiplied by a $\mathring{\psi}$ and this yields the desired decay:
\begin{align*}
|\hat u|^{-2\kappa}\int_{-1}^{\hat u}\int_{\mathcal{S}_{u,\hat v}}\left|\mathring{\psi}\right|^2\left|\nabla\nabla_4^{m-1}\psi\right|^2|u|^{2c+1+2\kappa}  &\lesssim |\hat u|^{-2\kappa}\int_{-1}^{\hat u}\int_{\mathcal{S}_{u,\hat v}}\left|\mathring{\psi}\right|^2|u|^{2c-1-2m+2\kappa} 
\\ \nonumber &\lesssim |\hat u|^{-2\kappa}\int_{-1}^{\hat u}\int_{\mathcal{S}_{u,\hat v}}\left|\mathring{\psi}\right|^2\frac{|u|^{4-4\delta}}{\hat v^{2-4\delta}}|u|^{2c-5-2m+4\delta+2\kappa} \hat v^{2-4\delta}
\\ \nonumber &\lesssim \hat v^{-2+4\delta}|\hat u|^{-2\kappa}\int_{-1}^{\hat u}|u|^{2c-5-2m+4\delta+2\kappa}
\\ \nonumber &\lesssim \frac{\hat v^{2-4\delta}}{|u|^{4+2m-2c-4\delta}}.
\end{align*}

Finally, we come to terms with $\nabla_4^m$. These terms can arise in the second and fourth line of~\eqref{commutethatnabla3equation}. If they arrive in the fourth line, then they come with a $\mathring{\psi}$ and they can be estimated as above. When they arrive in the first line we have to be a little more careful. We start with the terms $\psi_1\cdot\reallywidetilde{\nabla_4^m\psi_2}$ or $\reallywidetilde{\nabla_4^m\psi_1}\cdot\psi_2$. (The terms $\reallywidetilde{\psi_1}\cdot\nabla_4^m\psi_2$ are straightforward to estimate as before using the bootstrap assumption and the estimate~\eqref{wowsomuchbetter}.) Here, the point is simply that direct inspection of the equations shows that either the term $\reallywidetilde{\nabla_4^m\psi}$ has already been estimated with a $\nabla_4$ equation, or it is multiplying a $\mathring{\psi}$. The case where the term is multiplied by a $\mathring{\psi}$ can be handled in an analogous fashion to the above estimates. Let's consider the case where $\reallywidetilde{\nabla_4^m\psi_1}$ has already been estimated by a $\nabla_4$ equation:

\begin{align*}
|\hat u|^{-2\kappa}\int_{-1}^{\hat u}\int_{\mathcal{S}_{u,\hat v}}\left|\reallywidetilde{\nabla_4^m\psi_1}\right|^2\left|\psi_2\right|^2|u|^{2c+1+2\kappa}  &\lesssim |\hat u|^{-2\kappa}|\hat v|^{1-4\delta}\int_{-1}^{\hat u}\int_{\mathcal{S}_{u,\hat v}}|u|^{2c-4-2m+4\delta+2\kappa} 
\\ \nonumber &\lesssim \frac{\hat v^{1-4\delta}}{|u|^{3+2m-2c-4\delta}}.
\end{align*}
Here we used that  we have that $2c-4-2m+4\delta+2\kappa < -1$.

The term $\psi_3u^{-2}$ may be treated similarly after noting that $\psi_3$ is always equal to $\underline{\omega}$ and hence is estimated by a $\nabla_4$ equation.

Finally, the desired estimates for $\slashed{Riem}$ are straightforward to establish using  the Gauss equation~\eqref{itsgauss}  which relates $\slashed{Riem}_{ABCD}$ and $R_{ABCD}$ and then exploiting the $\nabla_4$ equation for $R_{ABCD}$. We omit the details.

\end{proof}

\section{Regular Estimates}\label{ambientregsec}
In this section we will show that if we have a proto-ambient metric $\left(\mathcal{M},g\right)$ which arises from compatible regular conjugate data and exists in a suitably small characteristic rectangle $\mathcal{R}_{\tilde u,\tilde v}$, then we have quantitative estimates for the regularity of the solution (see Definition~\ref{ambientregular}). More precisely, we have the following three propositions:

When $n = 2$ we have full $C^{\infty}$ estimates.

\begin{proposition}\label{ambientreg2}Suppose that $n=2$ and that $\left(\mathcal{M},g\right)$ is a proto-ambient metric arising from compatible regular conjugate data and exists in characteristic rectangle $\mathcal{R}_{\tilde u,\tilde v}$ with $\frac{\tilde v}{|\tilde u|} \leq \epsilon$ for $\epsilon > 0$ suitable small. Then, for every curvature component $\Psi$ and Ricci coefficient $\psi$ and $(i,j,k) \in \mathbb{N}\times \mathbb{N}\times \mathbb{N}$, there exists a constant $C_{ijk}$ such that for all $(u,v,\theta^A) \in \mathcal{R}$ we have
\[\left|\nabla_4^i\nabla_3^j\nabla^k\Psi\right| \leq C_{ijk}|u|^{-2-i-j-k},\qquad \left|\nabla_4^i\nabla_3^j\nabla^k\psi\right| \leq C_{ijk}|u|^{-1-i-j-k}.\]

Similar estimates hold directly for the metric components  $\slashed{g}$, $\Omega$, and $b$.
\end{proposition}

When $n \geq 3$ and odd, we have estimates which, among other things, allow for a quantitative estimate for the little $o$ error in the expansion~\eqref{expandmetricodd}.
\begin{proposition}\label{ambientregodd}Suppose that $n\geq 3$ and odd and that $\left(\mathcal{M},g\right)$ is a proto-ambient metric arising from compatible regular conjugate data and exists in characteristic rectangle $\mathcal{R}_{\tilde u,\tilde v}$ with $\frac{\tilde v}{|\tilde u|} \leq \epsilon$ for $\epsilon > 0$ suitable small. Then, for every curvature component $\Psi$ and Ricci coefficient $\psi$, $0 \leq i \leq \frac{n-3}{2}$, and $(j,k) \in \mathbb{N}\times \mathbb{N}\times \mathbb{N}$, there exists a constant $C_{ijk}$ such that for all $(u,v,\theta^A) \in \mathcal{R}$ we have
\[\left|\nabla_4^i\nabla_3^j\nabla^k\Psi'\right| \leq C_{ijk}|u|^{-2-i-j-k},\qquad \left|\nabla_4^i\nabla_3^j\nabla^k\psi\right| \leq C_{ijk}|u|^{-1-i-j-k}.\]

For $i = \frac{n-1}{2}$, $(j,k) \in \mathbb{N}\times \mathbb{N}$, there is $C_{jk}$ such that for all $(u,v,\theta^A) \in \mathcal{R}$ we have
\[\left|\nabla_4^{\frac{n-1}{2}}\nabla_3^j\nabla^k\Psi'\right| \leq C_{ jk}|v|^{-1/2}|u|^{-3/2-\frac{n-1}{2}-j-k},\qquad \left|\nabla_4^{\frac{n-1}{2}}\nabla_3^j\nabla^k\hat{\chi}\right| \leq C_{ jk}|v|^{-1/2}|u|^{-1/2-\frac{n-1}{2}-j-k},\]
while for any $\psi \neq \hat{\chi}$, for all $(u,v,\theta^A) \in \mathcal{R}$ we have
\[\left|\nabla_4^{\frac{n-1}{2}}\nabla_3^j\nabla^k\psi\right| \leq C_{ jk}|u|^{-1-\frac{n-1}{2}-j-k}.\]

Similar estimates hold directly for the metric components $\slashed{g}$, $\Omega$, and $b$.
\end{proposition}

Finally, when $n \geq 4$ and even, we have estimates which, among other things, allow for a quantitative estimate for the little $o$ error in the expansion~\eqref{expandmetriceven}.
\begin{proposition}\label{ambientregeven}Suppose that $n\geq 4$ and odd and that $\left(\mathcal{M},g\right)$ is a proto-ambient metric arising from compatible regular conjugate data and exists in characteristic rectangle $\mathcal{R}_{\tilde u,\tilde v}$ with $\frac{\tilde v}{|\tilde u|} \leq \epsilon$ for $\epsilon > 0$ suitable small. Then, for every curvature component $\Psi$ and Ricci coefficient $\psi$, $0 \leq i \leq \frac{n-4}{2}$, and $(j,k) \in \mathbb{N}\times \mathbb{N}\times \mathbb{N}$, there exists a constant $C_{ijk}$ such that for all $(u,v,\theta^A) \in \mathcal{R}$ we have
\[\left|\nabla_4^i\nabla_3^j\nabla^k\Psi'\right| \leq C_{ijk}|u|^{-2-i-j-k},\qquad \left|\nabla_4^i\nabla_3^j\nabla^k\psi\right| \leq C_{ijk}|u|^{-1-i-j-k}.\]

For $i = \frac{n-2}{2}$, $(j,k) \in \mathbb{N}\times \mathbb{N}$, there is $C_{ jk}$ such that for all $(u,v,\theta^A) \in \mathcal{R}$ we have
\[\left|\nabla_4^{\frac{n-2}{2}}\nabla_3^j\nabla^k\Psi'\right| \leq C_{ jk}\log\left(\frac{v}{u}\right)|u|^{-2-\frac{n-2}{2}-j-k},\qquad \left|\nabla_4^{\frac{n-2}{2}}\nabla_3^j\nabla^k\hat{\chi}\right| \leq C_{jk}\log\left(\frac{v}{u}\right)|u|^{-1-\frac{n-2}{2}-j-k},\]
while for any $\psi \neq \hat{\chi}$, for all $(u,v,\theta^A) \in \mathcal{R}$ we have
\[\left|\nabla_4^{\frac{n-1}{2}}\nabla_3^j\nabla^k\psi\right| \leq C_{ jk}|u|^{-1-\frac{n-1}{2}-j-k}.\]

Similar estimates hold directly for the metric components  $\slashed{g}$, $\Omega$, and $b$.
\end{proposition}

\begin{remark}If one unwinds the renormalization in $\Psi'$ and expresses the above propositions in terms of the expansions from Definition~\ref{ambientregular} then the result is a quantitative estimate for the constant in the big $O$ in terms of the initial data.  In fact, when further regularity assumptions are made on the conjugate data, the method of proof given for these propositions can be extended to yield full asymptotic expansions consistent with the expectations from~\cite{FG2}.
\end{remark}
\subsection{Estimates for $n=2$}
We start with the proof of Proposition~\ref{ambientreg2}. 
\begin{proof}This is a relatively standard combination of a preservation of regularity argument with the already established scale invariant estimates, so we will just sketch the proof.

First of all, it is clear that Theorem~\ref{thefundamentalestimate} can be run with any number of angular derivatives. Thus, in the various norms controlled by Theorem~\ref{thefundamentalestimate} we can apply Sobolev inequalities to replace any $L^2$ norm over $\mathcal{S}$ with the corresponding $L^{\infty}$ norm over $\mathcal{S}$. This establishes $L^{\infty}\left(\mathcal{R}\right)$ estimates for the Ricci coefficients. 

To obtain similar estimates $L^{\infty}\left(\mathcal{R}\right)$ for the $\Psi$'s (and any angular derivative thereof) we need to use the Bianchi equations. This is most straightforward for any curvature component $\Psi$ not equal to $\alpha$; in this case, we can simply integrate the corresponding $\nabla_4$ Bianchi equation use Theorem~\ref{thefundamentalestimate} to control the $L^1$ norm of the right hand side. Commuting with angular derivatives then establishes the corresponding estimates for any number of angular derivatives these components.

For $\alpha$, we have to integrate the $\nabla_3$ equation. Given the previous estimates we can write the $\nabla_3$ equation for $\alpha$ as 
\[\nabla_3\alpha + u^{-1}\alpha = O\left(\frac{v^{1-\delta}}{|u|^{2-\delta}}\right)\alpha + O\left(u^{-3}\right).\]
We conjugate the equation through by $|u|$ to obtain
\[\nabla_3\left(|u|\alpha\right) = O\left(\frac{v^{1-\delta}}{|u|^{2-\delta}}\right)|u|\alpha + O\left(u^{-2}\right).\]
Then we can integrate from the initial data along $\{u=-1\}$ and apply Gronwall to establish the desired estimate.

Next, given these estimates on curvature, it is straightforward to get estimates on $\nabla_4\psi$ and $\nabla_3\psi$ via the null structure equations. (As we have already seen multiple times before, if, say, one does not have an explicit null structure equation giving an expression for $\nabla_4\psi$, then one commutes the $\nabla_3$ equation to derive an equation for $\nabla_3\nabla_4\psi$ and then integrates as we did for $\alpha$. See the estimates in Section~\ref{estV} and~\ref{estS}.)

Revisiting Bianchi allows one to obtain $L^{\infty}$ estimates for $\nabla_4$ and $\nabla_3$ derivatives of curvature. The proof for an arbitrary number of derivatives then follows by a straightforward induction argument.

\end{proof}
\subsection{Estimates for $n\geq 3$ and odd}
We next give the proof of Proposition~\ref{ambientregodd}. 
\begin{proof}The estimates when one commutes with at most $\frac{n-3}{2}$ $\nabla_4$ derivatives are straightforward and follow analogously to the case of $n=2$. We omit the details.

The estimate for the maximal number, $\frac{n-1}{2}$, of $\nabla_4$ derivatives requires a little more care due to the singular terms generated by the $v^{\frac{n-4}{2}}|u|^{\frac{n}{2}+2}h_{AB}$ which is present in the various nonlinear terms. As usual, the estimate is easiest to obtain for any curvature component $\Psi$ other than $\alpha$; these all satisfy a $\nabla_4$ equation and via Proposition~\ref{writesomeeqns} the desired estimate follows immediately. For $\alpha$ we have to commute it's $\nabla_3$ equation with $\nabla_4^{\frac{n-1}{2}}$. It follows from Proposition~\ref{writesomeeqns} and the previous estimates we have established that we have 
\begin{equation*}
\nabla_3\nabla_4^{\frac{n-1}{2}}\alpha'_{AB} + \frac{n}{2}u^{-1}\nabla_4^{\frac{n-1}{2}}\alpha'_{AB} = O\left(v^{-1/2}|u|^{-5/2-\frac{n-1}{2}}\right),
\end{equation*}
from which the desired estimate for $\nabla_4^{\frac{n-1}{2}}\alpha'$ follow in a straightforward fashion. (The reason we see singular terms like $v^{-1/2}$ is that in the nonlinear terms from Proposition~\ref{writesomeeqns} there are terms proportional to  
\[\nabla_4^{\frac{n-1}{2}}\left(\mathring{\psi}h_{AB}v^{\frac{n-4}{2}}|u|^{-\frac{n}{2}+2}\right),\]
and the best estimate we have for $\mathring{\psi}$ is $\left|\mathring{\psi}\right| \lesssim v|u|^{-2}$.) From here the desired estimate for $\nabla_4^{\frac{n-1}{2}}$ follows after conjugating the equation with $|u|^{\frac{n}{2}}$ and integrating in the $\nabla_3$ direction.

Estimates for additional derivatives follow in a straightforward fashion.
\end{proof}
\subsection{Estimates for $n \geq 4$ and odd}
The proof of Proposition~\ref{ambientregeven} is completely analogous to the proof of Proposition~\ref{ambientregodd} and hence we omit it.
\section{Self-Similar Extraction}\label{extractsec}
In this section we will discuss the extraction of the self-similar solution from successive rescalings of a proto-ambient metric.

\subsection{Initial Data Analysis}
Suppose that $\left(\mathcal{M},g\right)$ is a proto-ambient metric produced by Theorem~\ref{localexistenceproto}. Let $1 > \lambda > \mu > 0$ and consider the rescaled metrics $\left(\mathcal{M},g_{\lambda}\right)$ and $\left(\mathcal{M},g_{\mu}\right)$ defined by
\[g_{\lambda} \doteq \lambda^{-2}\Phi_{\lambda}^*g,\qquad g_{\mu} \doteq \mu^{-2}\Phi_{\mu}^*g,\]
where the rescaling diffeomorphism $\Phi$ is defined by~\eqref{rescalingmap}.

We will refer to metric coefficients, Ricci coefficients, and curvature components of $g_{\lambda}$ by $\phi_{\lambda}$, $\psi_{\lambda}$, $\Psi_{\lambda}$. We adopt a similar notation for $g_{\mu}$. We will put an overline over quantities which denote the difference of two double null quantities corresponding to $g_{\lambda}$ and $g_{\mu}$ respectively. For example, we set
\[\overline{\alpha} \doteq \alpha_{\lambda} - \alpha_{\mu}.\]

In the following sequence of lemmas, we will show that the data corresponding to $\overline{\phi}$, $\overline{\psi}$, and $\overline{\Psi}$ along $\{v = 0\}$ exhibits many cancellations.

We start with the case of $n=2$.
\begin{proposition}\label{initsupern2}When $n = 2$, we have that every $\overline{\phi}$, $\overline{\psi}$, and $\overline{\Psi}$ vanishes along $\{v = 0\}$ except for $\overline{\alpha}$, which satisfies
\[\left|\overline{\alpha}|_{\{v=0\}}\right| \lesssim \left|\lambda - \mu\right||u|^{-1}.\]

Furthermore, we have the following improved vanishing for some Ricci coefficients for any $i$
:
\begin{equation}\label{supervanishriccin2}
\left[\left|\nabla^i\overline{\hat{\underline{\chi}}}\right| + \left|\nabla^i\overline{{\rm tr}\underline{\chi}}\right| + \left|\nabla^i\overline{\underline{\eta}}\right| + \left|\nabla^i\overline{\omega}\right| \right]|_{\{u=-1\}} \lesssim v^{2-\delta}.
\end{equation}

The implied constants just depend on $i$ and~\eqref{whatCdependson} (with $N$ taken suitably large.)

\end{proposition}
\begin{proof}Other than~\eqref{supervanishriccin2}, this follows immediately from the formulas in Section~\ref{scalingbehav} and Proposition~\ref{incomingdatan2}. (Keep in mind that the formulas in Section~\ref{scalingbehav} are given with respect to a Lie-propagated frame!)

Before we show how to obtain~\eqref{supervanishriccin2}, we first note that by integrating the $\nabla_4$ Bianchi equations and using the estimates from  Theorem~\ref{localexistenceproto} and Sobolev inequalities, we obtain that for $\Psi \neq \alpha$ we have $\left|\reallywidetilde{\nabla^i\Psi}\right| \leq C \frac{v^{1-\delta}}{|u|^{i+3-\delta}}$. Using the $\nabla_4$ Bianchi equations once again yields $\left|\nabla^i\reallywidetilde{\nabla_4\Psi}\right| \leq C \frac{v^{1-\delta}}{|u|^{i+4-\delta}}$ for all $\Psi \neq \alpha,\beta$. 

Now let's turn to~\eqref{supervanishriccin2}. Let's start with the estimate for $\overline{\hat{\underline{\chi}}}$. First of all, it follows immediately from the corresponding $\nabla_4$ null structure equation that $\left|\nabla^i\reallywidetilde{\hat{\chi}}\right| \leq C\frac{v^{1-\delta}}{|u|^{2+i-\delta}}$. In turn this immediately implies that  $\left|\nabla^i\overline{\hat{\chi}}\right| \leq C\frac{v^{1-\delta}}{|u|^{2+i-\delta}}$. Then the desired estimate for $\overline{\hat{\underline{\chi}}}$ just follows by integrating it's $\nabla_4$ equation and using that everything on the right hand side vanishes in $v$. The estimates for $\overline{{\rm tr}\underline{\chi}}$ are analogous except that one also needs to use that $\left|\nabla^i\overline{\rho}\right| \leq C\frac{v^{1-\delta}}{|u|^{2+i-\delta}}$ which itself follows from the uniform bound for $\nabla_4\rho$ we derived in the previous paragraph and the fact that $\rho$ vanishes on $\{v =0 \}$.

Next, we turn to $\omega$. First of all, using the estimates we have for $\nabla_4\rho$, it follows easily by commuting with $\nabla_4$ and integrating the $\nabla_3$ equation for $\omega$ that $\left|\nabla^i\reallywidetilde{\nabla_4\omega}\right| \lesssim \frac{v^{1-\delta}}{|u|^{3+i+\delta}}$. We then obtain that  $\left|\nabla^i\overline{\nabla_4\omega}\right| \lesssim \frac{v^{1-\delta}}{|u|^{3+i+\delta}}$ and the desired estimate for $\overline{\omega}$ follows from the fundamental theorem of calculus. A similar argument works for $\underline{\eta}$.

\end{proof}

Next we have $n\geq 3$ and odd.
\begin{proposition}\label{initsupernodd}When $n\geq 3$ and odd, we have that for every $0 \leq i \leq \frac{n-3}{2}$, $\mathcal{L}^{i+1}_v\overline{\phi}$, $\mathcal{L}_v^{i+1}\overline{\psi}$, and $\mathcal{L}_v^i\overline{\Psi}$ vanish along $\{v = 0\}$, except for $\mathcal{L}^{\frac{n-3}{2}}_v\overline{\alpha}$ and $\mathcal{L}^{\frac{n-1}{2}}\hat{\chi}$. Instead we have  
\[\left|\mathcal{L}^{\frac{n-3}{2}}_v\overline{\alpha}|_{\{v=0\}}\right| \lesssim \left|\lambda - \mu\right|^{1/2}|u|^{-\frac{n}{2}}.\]

Furthermore, we have the following improved vanishing for some Ricci coefficients for any $i$:
\begin{equation}\label{supervanishriccinodd}
\left[\left|\nabla^i\overline{\hat{\underline{\chi}}}\right| + \left|\nabla^i\overline{{\rm tr}\underline{\chi}}\right| + \left|\nabla^i\overline{\underline{\eta}}\right| + \left|\nabla^i\overline{\omega}\right| \right]|_{\{u=-1\}} \lesssim v^{3/2-\delta}.
\end{equation}

The implied constants just depend on $i$ and~\eqref{whatCdependson} (with $N$ taken suitably large.)
\end{proposition}
\begin{proof}This is analogous to the case when $n=2$.

\end{proof}

For $n \geq 4$ and even we have 
\begin{proposition}\label{initsuperneven}When $n\geq 4$ and even, we have that for every $0 \leq i \leq \frac{n-4}{2}$, $\mathcal{L}^{i+1}_v\overline{\phi}$, $\mathcal{L}_v^{i+1}\overline{\psi}$, and $\mathcal{L}_v^i\overline{\Psi}$ vanish along $\{v = 0\}$.

Furthermore, we have the following improved vanishing for some Ricci coefficients for any $i$
:
\begin{equation}\label{supervanishriccineven}
\left[\left|\nabla^i\overline{\hat{\underline{\chi}}}\right| + \left|\nabla^i\overline{{\rm tr}\underline{\chi}}\right| + \left|\nabla^i\overline{\underline{\eta}}\right| + \left|\nabla^i\overline{\omega}\right| \right]|_{\{u=-1\}} \lesssim \left|\log(v)\right|v^{2-\delta}.
\end{equation}

The implied constants just depend on $i$ and~\eqref{whatCdependson} (with $N$ taken suitably large.)
\end{proposition}
\begin{proof}This is also analogous to the case of $n = 2$.
\end{proof}

\begin{remark}There exists many more improvements in the vanishing of the overlined double-null quantities as $v\to 0$. However, we have just listed here the ones that we will need later.
\end{remark}
\subsection{Supercritical Estimates}\label{superduperestiamtes}
In this section we will show that the various overlined quantities satisfy a supercritical estimates.
\subsubsection{The Norms}
Recall that in Section~\ref{renormalizationswoo}, we defined renormalized versions of $\alpha$ which eliminated the most singular behavior as $v\to 0$. However, Propositions~\ref{initsupernodd} and~\ref{initsuperneven} imply that for $\overline{\alpha}$ the most singular pieces always vanish. Hence, we do not need to define a renormalized $\overline{\alpha}$. 

For the norms we have following definition for $n \geq 3$ and odd or $n=2$.
\begin{definition}Let $n \geq 3 $ and odd or $n  =2$. Let $\kappa > 0$ be sufficiently small and fixed and choose $\hat{N}$ sufficiently large. Then we define the norms $\overline{\mathfrak{T}}$, $\overline{\mathfrak{L}}$, $\overline{\mathfrak{U}}$, $\overline{\mathfrak{V}}$, and $\overline{\mathfrak{S}}$, by replacing in each of the norms $\mathfrak{T}$, $\mathfrak{L}$, $\mathfrak{U}$, $\mathfrak{V}$, and $\mathfrak{S}$ the corresponding $\psi$ or $\Psi$ with $\overline{\psi}$ or $\overline{\Psi}$ and then lowering the the power of $|u|$ in the definition by $2\kappa$. Throughout we use $g_{\lambda}$ to define the covariant derivatives and contractions.

For example, when $n$ is odd,
\begin{align*}
\left\vert\left\vert\overline\Psi\right\vert\right\vert_{\overline{\mathfrak{T}}_{\tilde u,\tilde v}}^2 \doteq \sup_{0\leq j \leq \hat{N}}\sup_{(u_0,v_0)\in \mathcal{R}_{\tilde u,\tilde v}}\Bigg[&\int_{-1}^{u_0}\int_{\mathcal{S}}\left|\reallywidetilde{\nabla^j\nabla_4^{\frac{n-3}{2}}\overline{\Psi}}\right|^2u^{n+1-2\delta+2j-2\kappa}v^{-1+2\delta}
\\ \nonumber &+\int_0^{v_0}\int_{\mathcal{S}}\left|\reallywidetilde{\nabla^j\nabla_4^{\frac{n-3}{2}}\overline{\Psi}}\right|^2u^{n+1-2\delta+2j-2\kappa}v^{-1+2\delta}
\\ \nonumber &+\int_{-1}^{u_0}\int_0^{v_0}\int_{\mathcal{S}}\left|\reallywidetilde{\nabla^j\nabla_4^{\frac{n-3}{2}}\overline{\Psi}}\right|^2u^{n+1-2\delta+2j-2\kappa}v^{-2+2\delta}\Bigg].
\end{align*}\end{definition}

When $n\geq 4$ and even we have the following.
\begin{definition}Let $n \geq 4 $ and even. Let $\kappa > 0$ be sufficiently small and fixed and choose $\hat{N}$ sufficiently large. Then we define the norms $\overline{\mathfrak{T}}$, $\overline{\mathfrak{L}}$, $\overline{\mathfrak{U}}$, $\overline{\mathfrak{V}}$, and $\overline{\mathfrak{S}}$, by replacing in each of the norms $\mathfrak{T}$, $\mathfrak{L}$, $\mathfrak{U}$, $\mathfrak{V}$, and $\mathfrak{S}$ the corresponding $\psi$ or $\Psi$ with $\overline{\psi}$ or $\overline{\Psi}$ and then lowering the the power of $v$ in the definition by $2\kappa$. 

For example,
\begin{align*}
\left\vert\left\vert\Psi\right\vert\right\vert_{\mathfrak{T}_{\tilde u,\tilde v}}^2 \doteq \sup_{0\leq j \leq \hat{N}} \sup_{(u_0,v_0) \in \mathcal{R}_{\tilde u, \tilde v}}\Bigg[&\int_{-1}^{u_0}\int_{\mathcal{S}}\left|\reallywidetilde{\nabla^j\nabla_4^{\frac{n-4}{2}}\Psi}\right|^2u^{n+2j}v^{-1-2\kappa}
\\ \nonumber &+v_0^{-2\delta}\int_0^{v_0}\int_{\mathcal{S}}\left|\reallywidetilde{\nabla^j\nabla_4^{\frac{n-4}{2}}\Psi}\right|^2u^{n+2j}v^{-1-2\kappa+2\delta}
\\ \nonumber &+v_0^{-2\delta}\int_{-1}^{u_0}\int_0^{v_0}\int_{\mathcal{S}}\left|\reallywidetilde{\nabla^j\nabla_4^{\frac{n-4}{2}}\Psi}\right|^2u^{n+2j}v^{-2-2\kappa+2\delta}\Bigg].
\end{align*}
\end{definition}

\begin{remark}\label{moreangular}We have written $\hat{N}$ in the above instead of $N$ (which appeared in Section~\ref{normssection}) to emphasize that the norms of this section and the norms of Section~\ref{normssection} do not have to have the same number of angular derivative commutations. In fact, due to the quasilinear nature of the Einstein equations, we expect that in order to control the convergence of the rescaled solutions, we will be required to have established that each rescaled solution is bounded in a space that involves one more derivative. Specifically, in our context, one we have fixed the number $\hat{N}$ of angular commutation in the norm used for convergence, we will need to take $N \geq \hat{N}+1$ in the norm used for Theorem~\ref{localexistenceproto}. Since we have already not concerned ourselves with trying to minimize the use of angular derivatives in Theorem~\ref{localexistenceproto}, here we will also not attempt to optimize $\hat{N}$.
\end{remark}
\subsubsection{Double-null equations for the Overlined Quantities}We will eventually want to derive equations for the overlined double-null quantities. We first recall the various formulas relating the covariant derivatives to coordinate derivatives and Ricci coefficients:
\begin{equation}\label{nabla4coord}
\nabla_4\phi_{A_1\cdots A_k} = \Omega^{-1}\partial_v\left[\phi_{A_1\cdots A_k}\right] - \sum_{i=1}^k\chi_{A_i}^{\ \ B}\phi_{A_1\cdots \hat{A_i}B\cdots A_k},
\end{equation}
\begin{equation}\label{nabla3coord}
\nabla_3\phi_{A_1\cdots A_k} = \Omega^{-1}\left(\partial_u+b^B\partial_B\right)\left[\phi_{A_1\cdots A_k}\right] - \sum_{i=1}^k\left(\underline\chi_{A_i}^{\ \ B}-\Omega^{-1}\partial_{A_i}b^B\right)\phi_{A_1\cdots \hat{A_i}B\cdots A_k},
\end{equation}
\begin{equation}\label{nablaAcoord}
\nabla_B\phi_{A_1\cdots A_k} = \partial_B\left[\phi_{A_1\cdots A_k}\right] - \sum_{i=1}^k\slashed{\Gamma}_{A_iB}^C\phi_{A_1\cdots \hat{A_i}C\cdots A_k}.
\end{equation}

In particular, the following formulas may be easily derived:
\begin{align}\label{nabla4comp}
\Omega_{\lambda}\left(\nabla_4\right)_{\lambda}\left(\phi_{\lambda}\right)_{A_1\cdots A_k} - \Omega_{\mu}\left(\nabla_4\right)_{\mu}\left(\phi_{\mu}\right)_{A_1\cdots A_k} &= \Omega_{\lambda}\left(\nabla_4\right)_{\lambda}\overline{\phi}_{A_1\cdots \hat{A_i}B\cdots A_k}  - \sum_{i=1}^k\overline{\Omega\chi}_{A_i}^{\ \ B}\left(\phi_{\mu}\right)_{A_1\cdots \hat{A_i}B\cdots A_k},
\end{align}
\begin{align}\label{nabla3comp}
\Omega_{\lambda}\left(\nabla_3\right)_{\lambda}&\left(\phi_{\lambda}\right)_{A_1\cdots A_k} - \Omega_{\mu}\left(\nabla_3\right)_{\mu}\left(\phi_{\mu}\right)_{A_1\cdots A_k} = 
\\ \nonumber &\Omega_{\lambda}\left(\nabla_3\right)_{\lambda}\overline{\phi}_{A_1\cdots \hat{A_i}B\cdots A_k}  - \sum_{i=1}^k\left[\overline{\Omega\underline\chi}_{A_i}^{\ \ B}- \overline{\Omega\partial_{A_i}b}^B\right]\left(\phi_{\mu}\right)_{A_1\cdots \hat{A_i}B\cdots A_k} + \overline{b}^B\partial_B\left[\left(\phi_{\mu}\right)_{A_1\cdots A_k}\right].
\end{align}
A key point of these formulas is that to highest order in the overlined quantities, they are the same as the original $\lambda$ equations.

Analogous formulas hold for differences $\mathcal{D}_{\lambda}\phi_{\lambda} - \mathcal{D}_{\mu}\phi_{\mu}$ of angular operators. By using these types of formulas systematically we will be able to use the differences of the various double null equations to effectively estimate the overlined double null quantities. We turn now to the details.

\subsubsection{The Estimates}
In this section we will establish the desired estimates. As we have seen before, separate arguments are required in the case of $n \geq 3$ and odd, $n \geq 4$ and even, and $n =2$.
\begin{theorem}\label{thefundamentalestimatenodd}Let $n \geq 3$ and odd. Assume that we have a proto-ambient metric $(\mathcal{M},g)$ which arises from Theorem~\ref{localexistenceproto} and exists in a characteristic rectangle $\mathcal{R}_{\tilde u,\tilde v}$ with $\frac{\tilde v}{|\tilde u|} \leq \epsilon$ for $\epsilon > 0$ sufficiently small, and furthermore satisfies~\eqref{quiteaniceconclusionifidontsaysomyself}.

Then, for $\epsilon > 0$ sufficient small there exists a constant $C \geq 1$ depending only on the size of the initial data such that
\[\overline{\mathfrak{T}}_{\tilde u,\tilde v} + \overline{\mathfrak{L}}_{\tilde u,\tilde v} +\overline{\mathfrak{U}}_{\tilde u,\tilde v} + \overline{\mathfrak{S}}_{\tilde u,\tilde v} + \overline{\mathfrak{V}}_{\tilde u,\tilde v} \leq C.\]
\end{theorem}
\begin{proof}
We introduce the bootstrap assumption
\begin{equation}\label{bootstrap2}
\overline{\mathfrak{T}}_{\tilde u,\tilde v} + \overline{\mathfrak{L}}_{\tilde u,\tilde v} +\overline{\mathfrak{U}}_{\tilde u,\tilde v} + \overline{\mathfrak{S}}_{\tilde u,\tilde v} + \overline{\mathfrak{V}}_{\tilde u,\tilde v} \leq A.
\end{equation}
(Later in the proof we will justify why the initial fluxes when $u = -1$ are bounded.)

We start by explaining how we will establish the energy estimates behind the bound for $\overline{\mathfrak{T}}_{\tilde u,\tilde v}$. 

As usual, we start with $\left(\overline{\alpha},\left(\overline{\beta},\overline{\nu}\right)\right)$. We have the following equations for $\left(\alpha_{\lambda},\left(\beta_{\lambda},\nu_{\lambda}\right)\right)$:
\begin{align}\label{atfirstwehavethis}
\Omega_{\lambda}{\nabla}_3\left(\alpha_{\lambda}\right)_{AB} + \Omega_{\lambda}\frac{n/2}{u}\left(\alpha_{\lambda}\right)_{AB} &= \Omega_{\lambda}\left(u^2\slashed{g}_{\lambda}\right)^{CD}\nabla_D\left(\nu_{\lambda}\right)_{C(AB)} + \Omega_{\lambda}\nabla_{(A}\left(\beta_{\lambda}\right)_{B)} + \Omega_{\lambda}\mathscr{E}_{\lambda},
\\ \nonumber \Omega_{\lambda}\nabla_4\left(\beta_{\lambda}\right)_A &= \Omega_{\lambda}\left(u^2\slashed{g}_{\lambda}\right)^{BC}\nabla_C\left(\alpha_{\lambda}\right)_{BA} + \Omega_{\lambda}\mathscr{E}_{\lambda},
\\ \nonumber \Omega_{\lambda}\nabla_4\left(\nu_{\lambda}\right)_{ABC} &= -2\Omega_{\lambda}\nabla_{[A}\left(\alpha_{\lambda}\right)_{B]C}+\Omega_{\lambda}\mathscr{E}_{\lambda}.
\end{align}
Here the covariant derivatives are all defined with respect to $g_{\lambda}$ even though we have suppressed this in the notation.

We have an analogous set of equations for $\left(\alpha_{\mu},\left(\beta_{\mu},\nu_{\mu}\right)\right)$. Taking the difference of the equations and using the identities~\eqref{nabla4comp},~\eqref{nabla3comp},  and the counterparts for the angular derivatives then yields that~\eqref{atfirstwehavethis} holds again where
\begin{enumerate}
	\item We replace each $\alpha_{\lambda}$, $\beta_{\lambda}$, and $\nu_{\lambda}$ with the corresponding overlined quantities $\overline{\alpha}$, $\overline{\beta}$, $\overline{\nu}$.
	\item We must add the error term $-\Omega_{\mu}\mathscr{E}_{\mu}$.
	\item We must add an error term $\hat{\mathscr{E}}$ which results from the applications of the identities~\eqref{nabla4comp},~\eqref{nabla3comp}, and the angular variants. This will be of the schematic form:
	\[\hat{\mathscr{E}} = \left[\overline{\Omega\psi} + \nabla\overline{b}\right]\Psi_{\mu} + \left[\overline{\Omega\slashed{g}} + \overline{b}\right]\nabla\Psi_{\mu} + \overline{\Omega\slashed{\Gamma}}\Psi_{\mu}.\]
\end{enumerate}
A key point in the error term $\hat{\mathscr{E}}$ is that the highest order terms are all in term of $\Psi_{\mu}$ for which we already have estimates.

The next set is to conjugate the equation with the weight $|u|^{-\kappa} = \left(-u\right)^{-\kappa}$. This only changes the nature of the left hand side. We obtain
\begin{align}\label{itsallsetupfortheenergyestimates}
\nabla_3\left(|u|^{-\kappa}\overline\alpha\right)_{AB} &+ \frac{(n/2)+\kappa}{u}\left(|u|^{-\kappa}\overline\alpha\right)_{AB} = |u|^{-\kappa}\left(\cdots\right)
\\ \nonumber \nabla_4\left(|u|^{-\kappa}\overline\beta\right)_A &= |u|^{-\kappa}\left(\cdots\right)
\\ \nonumber \nabla_4\left(|u|^{-\kappa}\nu\right)_{ABC} &=|u|^{-\kappa}\left(\cdots\right),
\end{align}
where the terms in the $\left(\cdots\right)$  are just the same that were on the right hand side of non-conjugated equation.

We draw attention to the fact that conjugation by the weight $|u|^{-\kappa}$ has effectively left the form of the equations unchanged except for the that the second term on the left hand side of $\overline{\alpha}$'s equation has had it's coefficient slightly raised. (Keep in mind that $\kappa > 0$ is a small constant).

We conjugate the equation just as we did in to obtain~\eqref{blahblahblahblah} and then we proceed to carry out the energy estimates in the same fashion. Due to original conjugation by $|u|^{-\kappa}$, the left hand side of our estimates will control $\overline{\mathfrak{T}}$. Note that after we carry out the energy estimate via the standard integration by parts, there will be no terms with angular derivatives of $\overline{\alpha}$, $\overline{\nu}$, and $\overline{\beta}$!

When we estimate the various error terms on the right hand side, the differences from before are as follows:
\begin{enumerate}
\item  In the equation for $w|u|^{-\kappa}\overline{\alpha}$, the coefficients of the second term on the left hand side will be proportional to $1-\delta-\kappa$ instead of $1-\delta$. However, since $|\delta|,|\kappa| \ll 1$, this term is still positive.
\item When we estimate the resulting nonlinear error terms we need to take account of the presence of the $|u|^{-\kappa}$ which we conjugated the equation by. The key point is that after using inequalities of the form
\[\left|\psi_{\lambda}\Psi_{\lambda} - \psi_{\mu}\Psi_{\mu}\right| \leq \left|\overline{\psi}\right|\left|\Psi_{\lambda}\right| + \left|\psi_{\mu}\right|\left|\overline{\Psi}\right|,\]
every error term coming from $\mathscr{E}_{\lambda} - \mathscr{E}_{\mu}$ contains a product with at least one overlined quantity coming from the difference of double null quantities associated to $g_{\lambda}$ and $g_{\mu}$. This term will satisfy estimates associated to $\overline{\mathfrak{T}}$, $\overline{\mathfrak{L}}$, etc., and hence can absorb the $|u|^{-\kappa}$. Furthermore, the other terms do not pose a threat due to the scale invariance of the norms in $\mathfrak{T}$, $\mathfrak{L}$, etc., the scaling properties from Section~\ref{scalingbehav} (see especially Remark~\ref{remarkaboutscalingbehav}), and the already established estimates from Theorem~\ref{localprotoambientyay} . The additional error terms from $\hat{\mathscr{E}}$ may be handled in a similar fashion after noting that, in particular,  we already have appropriate scale-invariant bounds on $\nabla\left(\alpha_{\mu},\beta_{\mu},\nu_{\mu}\right)$. Note that even though we use that we control angular derivatives of $\alpha_{\mu}$, $\beta_{\mu}$, and $\nu_{\mu}$,  at this stage of the energy estimates we will not control any derivatives of $\overline{\alpha}$, $\overline{\beta}$, and $\overline{\nu}$.  This represents the usual loss of a derivative that occurs when studying differences of solutions of quasilinear wave equations. 

\item We have to argue that the inhomogeneous terms which are introduced when we replace $\overline\Psi$ with $\reallywidetilde{\overline\Psi}$ are integrable, even though there is an extra factor of $|u|^{-\kappa}$. However,  Lemma~\ref{initsupernodd} implies that we only see such terms arising from $\overline{\alpha}$'s equation, and that the corresponding inhomogeneous error term which arises on the right hand side of the energy estimate is $O\left(|u|^{-\frac{n+1+2\kappa+2i}{2}}\right)$ which is better than the  $O\left(|u|^{-\frac{n+3+2i}{2}}\right)$ we saw when we were estimating $\mathfrak{T}$, and hence this  does not pose any danger. 
\item We have to explain how we obtain higher order estimates.  For this we simply directly take the difference of the equations from Section~\ref{renormn3} and repeat the above analysis, modulo straightforward additional estimates of Christoffel symbols (Cf.~analogous arguments from Section 5 of~\cite{impulse1}.)

\item Finally, we have to argue that the initial fluxes along $\{u = -1\}$ are under control. However, for this we can just use the inequality $\left|\mathfrak{D}\overline{\alpha}_{\lambda}\right| \lesssim \left|\mathfrak{D}\reallywidetilde{\alpha}'_{\lambda}\right| + \left|\mathfrak{D}\reallywidetilde{\alpha}'_{\mu}\right|$ and exploit the the scale invariance of the $\mathfrak{T}$ norms and the fact that $|u|^{-\kappa} = 1$ along $\{u = -1\}$. For example, using the self-similar relations of Section~\ref{scalingbehav} and the change of variables $\tilde u = \lambda^{-1}u$, $\tilde v = \lambda^{-1}v$, we have the following estimates for the $u$-flux of $\alpha'_{\lambda}$:
\begin{align*}
\int_0^{v_0}\int_{\mathcal{S}}\left|\nabla^j\nabla_4^{\frac{n-3}{2}}\alpha'_{\lambda}\right|^2v^{-1+2\delta}|_{\{u=-1\}}\, dv &= \int_0^{v_0}\int_{\mathcal{S}}\left|\nabla^j\nabla_4^{\frac{n-3}{2}}\alpha'_{\lambda}\right|^2|u|^{n+1-2\delta+2j}v^{-1+2\delta}|_{\{u=-1\}}\, dv
\\ \nonumber &=\int_0^{\lambda v_0}\int_{\mathcal{S}}\left|\nabla^j\nabla_4^{\frac{n-3}{2}}\alpha'\right|^2|\tilde u|^{n+1-2\delta+2j}{\tilde v}^{-1+2\delta}|_{\{u=-1\}}\, d\tilde v
\\ \nonumber &\lesssim 1.
\end{align*}
\end{enumerate}

It is now clear that the estimates for $\overline{\mathfrak{T}}$ can be successfully established. Analogous arguments allow one to bound $\overline{\mathfrak{L}}$ and $\overline{\mathfrak{U}}$.

We now turn to a discussion of the estimates for $\overline{\mathfrak{V}}$ and $\overline{\mathfrak{S}}$ for the Ricci coefficients. We undertake a different strategy for estimates based on $\nabla_4$ equations or estimates based on $\nabla_3$ equations. The $\nabla_4$ equations are the most straightforward. In analogy with the Bianchi system, we simply derive an equation for $\nabla_4$ of $\overline{\psi}$ and carry out the same estimates as we did when previously estimates $\mathfrak{V}$ and $\mathfrak{S}$. Due to Lemma~\ref{initsupernodd} there is no contribution from initial data or an inhomogeneous terms, and it is clear, just as it worked for Bianchi, that each term on the right hand of the estimate will contain at least one overlined quantity which produces an extra $|u|^{2\kappa}$ and allows us to close the estimate. So that the mechanism is clear, let us consider a representative equation. For $\overline{\hat\chi}$ one derives an equation of the following form:
\begin{align*}
\Omega\nabla_4\overline{\hat{\chi}}_{AB} &= -2\overline{\Omega}\left(\chi_{\lambda}\right)_A^{\ \ B}\left(\hat{\chi}_{\lambda}\right)_{BC} - 2\Omega_{\mu}\overline{\chi}_A^{\ \ B}\left(\hat{\chi}_{\lambda}\right)_{BC} - 2\Omega_{\mu}\left(\chi_{\mu}\right)_A^{\ \ B}\overline{\hat{\chi}}_{BC}
\\ \nonumber &\qquad -\frac{2}{n}\overline{\Omega}{\rm tr}\chi_{\lambda}\left(\hat{\chi}_{\lambda}\right)_{AB} - \frac{2}{n}\Omega_{\mu}\overline{{\rm tr}\chi}\left(\hat{\chi}_{\lambda}\right)_{AB} - \frac{2}{n}\Omega_{\mu}{\rm tr}\chi_{\mu}\overline{\hat{\chi}}_{AB} - \overline{\Omega}\left(\alpha_{\lambda}\right)_{AB} - \Omega_{\mu}\overline{\alpha}_{AB} +\cdots \Rightarrow
\\ \nonumber \left|\Omega\nabla_4\overline{\hat{\chi}}\right| &\lesssim \left|\overline{\Omega}\right||u|^{-2} + \left|\overline{\chi}\right|u^{-1} + \left|\overline{\alpha}\right| + \cdots,
\end{align*}
from which it is clear that the desired estimates for $\overline{\hat{\chi}}$ can be obtained by integrating in the $v$-direction. Note that this same strategy allows one to prove the analogue of Proposition~\ref{metricest} in a way which gains a factor of $|u|^{-2\kappa}$.

For estimates of Ricci coefficients which use $\nabla_3$ equations we have to be a bit more careful because there is a contribution from from the initial data for the overlined Ricci coefficients along $\{u = -1\}$, and we have to explain how we can gain a $|u|^{2\kappa}$ from this. Fortunately, this can easily be accomplished from exploiting vanishing in $v$: More specifically, for the $\nabla_3$ based estimates involved in $\overline{\mathfrak{V}}$ we can exploit~\eqref{supervanishriccinodd}, which implies that the Ricci coefficients decay $v^{1/2}$ faster as $v \to 0$ than we could exploit in our previous analysis of the $\mathfrak{V}$ norms. Since in the region under consideration, we have  $v \leq \epsilon |u|$, the extra $v$-decay easily provides the desired extra $|u|^{-\kappa}$. 

\end{proof}

Next, we turn the case of $n \geq 4$ and even.
\begin{theorem}\label{thefundamentalestimateneven}Let $n \geq 4$ and even. Assume that we have a proto-ambient metric $(\mathcal{M},g)$ which arises from Theorem~\ref{localexistenceproto} and exists in a characteristic rectangle $\mathcal{R}_{\tilde u,\tilde v}$ with $\frac{\tilde v}{|\tilde u|} \leq \epsilon$ for $\epsilon > 0$ sufficiently small, and furthermore satisfies~\eqref{quiteaniceconclusionifidontsaysomyself}.

Then, for $\epsilon > 0$ sufficient small there exists a constant $C \geq 1$ depending only on the size of the initial data such that
\[\overline{\mathfrak{T}}_{\tilde u,\tilde v} + \overline{\mathfrak{L}}_{\tilde u,\tilde v} +\overline{\mathfrak{U}}_{\tilde u,\tilde v} + \overline{\mathfrak{S}}_{\tilde u,\tilde v} + \overline{\mathfrak{V}}_{\tilde u,\tilde v} \leq C.\]
\end{theorem}
\begin{proof}

The difference with the case of  $n \geq 3$ and odd is that we are making the norms supercritical by lowering the $v$-weight instead of lowering the $u$-weight. The reason we cannot lower the $u$-weight is because, unlike the case when $n \geq 3$ and odd, if we conjugated the $\nabla_3$ equation for $\alpha$ by an additional negative $u$-weight we would produce a lower order term of the wrong sign. However, in contrast to the case of $n \geq 3$ and odd, Proposition~\ref{initsuperneven} yields that all overlined curvature components vanish at $\{v = 0\}$. Thus, when carrying out the energy estimates via the Bianchi equations, no inhomogeneous terms are generated; this is what allows us to use a lower $v$-weight.  With these caveats, and after replacing $|u|^{-\kappa}$ with $v^{-\kappa}$, it is clear that the proof can proceed exactly as in the case of $n \geq 3$ and odd if(!) we can show that the initial energy fluxes are bounded, i.e.,
\[\overline{\mathfrak{T}}_{-1,\tilde v} + \overline{\mathfrak{L}}_{-1,\tilde v} +\overline{\mathfrak{U}}_{-1,\tilde v} + \overline{\mathfrak{S}}_{-1,\tilde v} + \overline{\mathfrak{V}}_{-1,\tilde v} \leq C.\]
(Of course, the fluxes along $\{v = 0\}$ all vanish.)

Note that we cannot directly appeal to a rescaling argument as we did when $n \geq 3$ and odd because the $v$-weight we need to control inside the integral is $v^{-1-2\delta-2\kappa}$, while rescaling will only ever produce an estimate for $v^{-1-2\delta}$. Let's specialize to the case of $n =4$ as the high dimensional case is analogous. First we consider how we control the initial $v$-flux for $\overline{\Psi} \neq \overline{\alpha}$. For any $\Psi \neq \alpha$, integrating the corresponding $\nabla_4$ Bianchi equation and using the previously established estimates immediately yields $\left|\nabla^i\reallywidetilde{\Psi}\right| \lesssim \frac{v^{1-\delta}}{|u|^{3+i-\delta}}$. We immediately obtain that $\left|\nabla^i\overline{\Psi}\right| \lesssim \frac{v^{1-\delta}}{|u|^{3+i-\delta}}$ which easily implies that the desired initial $v$-flux is finite. For $\overline{\alpha}$, we do not have a corresponding $\nabla_4$ Bianchi equation, and instead we will have to exploit assumption~\eqref{woohoothisisawiseassumptiontomake} on the initial data. However, with the estimate~\eqref{woohoothisisawiseassumptiontomake} in hand, it is clear what we need to do. Integrating $\alpha$'s $\nabla_3$ equation and using the estimates we have just established along with the estimates from Theorem~\ref{localprotoambientyay} easily yields the estimate $\left|\nabla^i\reallywidetilde{\alpha}\right||_{u,v} \lesssim \frac{v^{1-\delta}}{|u|^{3+i-\delta}} + \left|\nabla^i\reallywidetilde{\alpha}\right||_{-1,v}$. Given this, the desired initial flux estimate of $\overline{\alpha}$ follows by the same scaling argument we have used repeatedly.
\end{proof}

Finally, we have the analogous result for $n =2$.
\begin{theorem}\label{thefundamentalestimaten2}Let $n = 2$. Assume that we have a proto-ambient metric $(\mathcal{M},g)$ which arises from Theorem~\ref{localexistenceproto} and exists in a characteristic rectangle $\mathcal{R}_{\tilde u,\tilde v}$ with $\frac{\tilde v}{|\tilde u|} \leq \epsilon$ for $\epsilon > 0$ sufficiently small, and furthermore satisfies~\eqref{quiteaniceconclusionifidontsaysomyself}.


Then, for $\epsilon > 0$ sufficient small there exists a constant $C \geq 1$ depending only on the size of the initial data such that
\[\overline{\mathfrak{T}}_{\tilde u,\tilde v} + \overline{\mathfrak{L}}_{\tilde u,\tilde v} +\overline{\mathfrak{U}}_{\tilde u,\tilde v} + \overline{\mathfrak{S}}_{\tilde u,\tilde v} + \overline{\mathfrak{V}}_{\tilde u,\tilde v} \leq C.\]
\end{theorem}
\begin{proof}This is proven in exactly the same fashion as when $n \geq 3$ and odd. Note that the situation is strictly easier since there are no singular terms, and $\overline{\alpha}|_{v=0}$ blows-up at an even slower rate than when $n$ is odd.
\end{proof}
\subsection{Extracting the Limit}
In this section we will show that $g_{\lambda}$ has a unique limit as $\lambda \to 0$. 

\begin{theorem}\label{selfsimilarextract}Assume that we have a proto-ambient metric $(\mathcal{M},g)$ which arises Theorem~\ref{localexistenceproto} and exists in a characteristic rectangle $\mathcal{R}_{\tilde u,\tilde v}$ with $\frac{\tilde v}{|\tilde u|} \leq \epsilon$ for $\epsilon > 0$ sufficiently small, and furthermore satisfies~\eqref{quiteaniceconclusionifidontsaysomyself}.

Then there exists a unique metric $g_{\rm sim}$ such that $\left(\mathcal{M},h\right)$ is a self-similar solution to the Einstein equations and
\[g_{\lambda} \to g_{\rm sim}\text{ as }\lambda \to 0,\]
where the convergence is with respect to the supercritical norms of $\left\vert\left\vert\cdot\right\vert\right\vert_{\overline{\mathfrak{T}}}$, $\left\vert\left\vert\cdot\right\vert\right\vert_{\overline{\mathfrak{L}}}$, etc., applied to the various double-null unknowns associated to $g_{\lambda}$ and $h$.
\end{theorem}
\begin{proof}Let's introduce the notation 
\begin{equation}\label{bignorm}
\left\vert\left\vert g_{\lambda} - g_{\mu}\right\vert\right\vert_{\overline{\mathfrak{E}}} \doteq \left\vert\left\vert g_{\lambda} - g_{\mu}\right\vert\right\vert_{\overline{\mathfrak{T}}} +\left\vert\left\vert g_{\lambda} - g_{\mu}\right\vert\right\vert_{\overline{\mathfrak{L}}} + \left\vert\left\vert g_{\lambda} - g_{\mu}\right\vert\right\vert_{\overline{\mathfrak{U}}} + \left\vert\left\vert g_{\lambda} - g_{\mu}\right\vert\right\vert_{\overline{\mathfrak{S}}} + \left\vert\left\vert g_{\lambda} - g_{\mu}\right\vert\right\vert_{\overline{\mathfrak{V}}}.
\end{equation}

Theorems~\ref{thefundamentalestimatenodd},~\ref{thefundamentalestimateneven}, and~\ref{thefundamentalestimaten2} have shown that 
\[\left\vert\left\vert g_{\lambda}-g_{\mu}\right\vert\right\vert_{\overline{\mathfrak{E}}} \lesssim 1,\]
for any $0 < \lambda,\mu \leq 1$.

Let's suppose that $0 < \mu < \lambda < 1$. Since the overlined norms have broken the scale-invariance, it now immediately follows by rescaling that
\begin{align*}
\left\vert\left\vert g_{\lambda}-g_{\mu}\right\vert\right\vert_{\overline{\mathfrak{E}}} &= \left\vert\left\vert g_{\lambda\cdot 1}-g_{\lambda\cdot\frac{\mu}{\lambda}}\right\vert\right\vert_{\overline{\mathfrak{E}}}
\\ \nonumber &\leq \lambda^{\kappa}\left\vert\left\vert g_1-g_{\frac{\mu}{\lambda}}\right\vert\right\vert_{\overline{\mathfrak{E}}}
\\ \nonumber &\lesssim \lambda^{\kappa}.
\end{align*}

In particular, it immediately follows that there exists a unique metric $g_{\rm sim}$ such that
\[\left\vert\left\vert g_{\lambda}-g_{\rm sim}\right\vert\right\vert_{\overline{\mathfrak{E}}} \to 0\text{ as }\lambda \to 0.\]

It remains to argue that $g_{\rm sim}$ is self-similar. However, this can be easily seen as follows: Let $s > 0$ and consider the rescaled metric $\left(g_{\rm sim}\right)_s$. We then have that
\[\left\vert\left\vert g_{s\lambda}-(g_{\rm sim})_s\right\vert\right\vert_{\overline{\mathfrak{E}}} \leq s^{\kappa} \left\vert\left\vert g_{\lambda}-g_{\rm sim}\right\vert\right\vert_{\overline{\mathfrak{E}}} \to 0\text{ as }\lambda \to 0,\]
\[\left\vert\left\vert g_{s\lambda} - g_{\lambda}\right\vert\right\vert_{\overline{\mathfrak{E}}} \leq \lambda^{\kappa}\left\vert\left\vert g_{s} - g\right\vert\right\vert_{\overline{\mathfrak{E}}} \lesssim \lambda^{\kappa} \to 0\text{ as }\lambda \to 0.\]

Together this implies that $g_{\lambda} \to (g_{\rm sim})_s$ and then, by uniqueness of limits, we conclude that $g_{\rm sim} = (g_{\rm sim})_s$, i.e., $g_{\rm sim}$ is a self-similar solution.
\end{proof}

It turns out that our rescaling procedure also allows us to classify the limits $g_{\rm sim}$.
\begin{theorem}The limiting self-similar solution $g_{\rm sim}$ produced by Theorem~\ref{selfsimilarextract} depends only on the following two tensors associated to the original metric $g$:
\begin{enumerate}
	\item The original induced metric on $\{u = -1\}$, $\slashed{g}_{AB}|_{\{(u,v)=(-1,0)\}}$.
	\item 
	\begin{enumerate}
		\item When $n=2$, the tensor $\hat{\chi}_{AB}|_{(u,v) = (-1,0)}$.
		\item When $n \geq 3$ and odd, the tensor $h_{AB}$ which determines the singular behavior of $\alpha_{AB}$ as $v\to 0$.
		\item When $n \geq 4$ and even, the tensor which is equal to $\left(\nabla_4^{\frac{n-4}{2}}\alpha - \mathcal{O}_{AB}\log(v)\right)|_{(u,v) = (-1,0)}$.			\end{enumerate}
\end{enumerate}
\end{theorem}
\begin{proof}Suppose we have two self-similar metrics $g^{(1)}$ and $g^{(2)}$ for which the two tensors defined above agree. Then, if one re-runs the proof of Propositions~\ref{incomingdatan2},~\ref{incomingdatanodd}, or~\ref{incomingdataneven}, depending on the dimension $n$, one finds that the conclusions of Propositions~\ref{initsupern2},~\ref{initsupernodd}, or ~\ref{initsuperneven}, depending on the dimension, hold with $\overline{g}$ replaced by $g^{(1)}- g^{(2)}$. This allows us to run our a priori estimates and conclude that $\left\vert\left\vert g^{(1)}-g^{(2)}\right\vert\right\vert_{\overline{\mathfrak{E}}} \lesssim 1$. Arguing as Theorem~\ref{selfsimilarextract} then shows that $\left(g^{(1)}-g^{(2)}\right)_{\lambda}$ must converge to $0$ as $\lambda \to 0$. By self-similarity the two metrics must in fact be equal.
\end{proof}
\appendix
\section{Local Existence for regular Data}\label{actuallocalexistencesec}
In this section we will discuss the local theory behind the proof of Theorem~\ref{localprotoambientyay}. First of all, we recall the well-posedness of the standard characteristic initial value problem.
\begin{theorem}\label{localsmoothdata}Let $\left(\mathcal{S},\slashed{g}_0\right)$ be a closed orientable $n$-dimensional Riemannian manifold and $\hat{\slashed{g}}(v)$ be conjugate data which is smooth as $v\to 0$. Then, after possibly taking $\epsilon$ smaller, there exists an open set $\mathcal{M}_0 \subset \mathcal{M} \doteq \{(u,v,\theta^A) \in (0,-1] \times [0,\epsilon) \times \mathcal{S}$ around $\left(\{v = 0\} \cup \{u=-1\}\right) \cap \mathcal{M}$ and a unique smooth metric $g$ on $\mathcal{M}_0$ solving the Einstein equations such that in the corresponding double null gauge we have 
\begin{enumerate}
\item 
\[\slashed{g}|_{v = 0} = u^2\slashed{g}_0.\]
\item 
\[\zeta|_{(u,v) = (-1,0)} = 0.\]
\item 
\[\Omega^2|_{\{v = 0\} \cup \{u=-1\}} = 1.\]
\item 
\[b|_{\{v = 0\}} = 0.\]
\item 

\[{\rm tr}\chi|_{(u,v) = (-1,0)} = \frac{\slashed{R}_0}{n-1}\]
\item There exists a function $\Phi\left(v,\theta^A\right)$ with 
\[\slashed{g}|_{u=-1} = \Phi^2\hat{\slashed{g}}.\]
\end{enumerate}
\end{theorem}
\begin{proof}As discussed before the statement of Theorem~\ref{localprotoambientyay} in Section~\ref{secinitialdata}, in the case of $n = 2$, this follows from~\cite{lukchar}. An examination of the proof in~\cite{lukchar} shows that there are no essential changes for the case of general $n\geq 2$.
\end{proof}

Now we turn to Theorem~\ref{localprotoambientyay}.
\begin{proof}When $n = 2$, the conjugate data is always smooth, so the desired result follows immediately from Theorem~\ref{localsmoothdata}.

Next, let's consider the case when $n \geq 4$. In this case, even though the conjugate data is not smooth as $v\to 0$, the curvature components and angular derivatives thereof are all in $L^2$ on the initial conjugate null cone. An examination of the estimates and the convergence scheme of~\cite{lukchar,impulse1} and the equations we have derived in Propositions~\ref{nullstruct},~\ref{constrainteqns}, and~\ref{Bianchit} allow one to prove a local existence result in a straightforward manner. The desired regularity statement then follows in an analogous fashion to the proofs of Proposition~\ref{ambientregodd} and~\ref{ambientregeven}. (Note that we are in a strictly easier situation here since for the proof of this theorem, we do not need to track any singular behavior as $u\to 0$.)

When $n = 3$ then $\alpha\sim v^{-1/2}$ along $\{u=-1\}$ and hence is not in $L^2$ on the conjugate cone. However, we can write
\[\alpha|_{u=-1} = \alpha^{(1)} + v^{-1/2}\alpha^{(2)},\]
for $\alpha^{(1)}$ and $\alpha^{(2)}$ in $L^{\infty}$. Then we can mimic Definition~\ref{thisistilde} and set (in the coordinate frame)
\[\alpha'_{AB} \doteq \alpha_{AB} - v^{-1/2}u^{1/2}\alpha^{(2)}_{AB}|_{u=-1}.\]
In Section~\ref{aprioriestsection} we have seen that we can replace $\alpha$ with $\alpha'$ as an unknown in the double-null unknowns. The main difference is that on the right hand side of various $\nabla_4$ equations for curvature, one will now see an inhomogeneous term which is $O\left(v^{-1/2}\right)$. In the $\nabla_3$ equation for $\alpha'$, one now finds terms like $O\left(v^{-1/2}\left[\left|\hat{\underline{\chi}}\right| + \left|\underline{\omega}\right|\right]\right)$. One can now easily check that the schemes from~\cite{lukchar,impulse1} work for this modified system. The key point being that for the $\nabla_4$ Bianchi equations, one can just use that $v^{-1/2}$ is integrable, and for the $\nabla_3$ equation of $\alpha$, one will have that $\hat{\underline{\chi}}$ and $\underline{\omega}$ decay as $v\to 0$ and thus that $v^{-1/2}\left[\left|\hat{\underline{\chi}}\right| + \left|\underline{\omega}\right|\right]$ is square integrable.

\end{proof}

\begin{remark}Theorem~\ref{localprotoambientyay} could also easily be proven directly using Theorem~\ref{localsmoothdata} and a density argument which exploits the difference estimates that we developed in the study of self-similar extraction, but this is ``overkill'' in the sense that Theorem~\ref{localprotoambientyay} does not need to use any control of the solution near $\{u = 0\}$.
\end{remark}

Finally, we note that the following lemma can proved by combining the above proof with a standard last slice and density argument. 
\begin{lemma}\label{thisfinishesit}Theorem~\ref{thefundamentalestimate} implies Theorem~\ref{localexistenceproto}.
\end{lemma}
\begin{proof}This follows by arguments which are analogous to Section 5 of~\cite{impulse1}. Let's give a sketch of one possible implementation in the case $n=3$ since that is where the curvature is the most singular. By a straightforward density argument and the convergence estimates we have already established, it suffices to prove Theorem~\ref{localexistenceproto} for regular conjugate data (see Definition~\ref{thisdataisreallyreallyregular}). In this case, along $\{u = -1\}$ we will be able to write $\alpha_{AB}(v)|_{u=-1} = \left(\alpha_{\rm reg}\right)_{AB}(v) + v^{-1/2}\mathring{h}_{AB}(v)$ where $\left(\alpha_{\rm reg}\right)_{AB}$ is smooth and $\mathring{h}_{AB}(v)$ and angular derivatives thereof are bounded along $\{u = -1\}$. For every $q > 0$, let $\xi_q(v)$ be a smooth function which is identically $0$ for $v < q/2$, identically $1$ when $v > q$, and satisfies $\left|\xi\right| + q\left|\xi'_q\right| \lesssim 1$. One can find smooth conjugate data $\hat{\slashed{g}}_q(v)$ so that the corresponding $\alpha$ component of curvature, which we denote by $\alpha^{(q)}$ will satisfy $\left(\alpha^{(q)}\right)_{AB}|_{u=-1} = \left(\alpha^{(q)}_{\rm reg}\right)_{AB} +\xi_q(v)v^{-1/2}\mathring{h}^{(q)}_{AB}$ where $\alpha^{(q)}_{\rm reg}$ converges to $\alpha_{\rm reg}$ in the bootstrap norm,  $\mathring{h}^{(q)}$ and angular derivatives thereof are uniformly bounded, and for any $r > 0$, $\mathring{h}^{(q)}$ converges to $\mathring{h}$ uniformly on $\{u=-1\} \cap \{v \geq r\}$, and, also, all of the other double null unknowns converge to the desired values in the bootstrap norm. Since the initial data is now smooth we can apply the usual local existence and continuation criteria results. Now we observe that if we define an alternative renormalization $\alpha'_{AB} \doteq \alpha_{AB} - \xi_q(v)v^{-1/2}\mathring{h}_{AB}$ and repeat the arguments \emph{mutatis mutandis}, the proof of Theorem~\ref{thefundamentalestimate} works for solutions arising from the $\hat{\slashed{g}}_q$ data. (Note, in particular, that at no point in the proof of Theorem~\ref{thefundamentalestimate} does one take a $\nabla_4$ derivative of $\alpha'$). Thus we obtain that these approximate solutions exist in the desired region. Finally, using Theorem~\ref{selfsimilarextract} of this paper, one can estimate the differences of the approximate solutions and show that they in fact converge as $q\to 0$ to the desired solution. The case of general $n$ is similar.
\end{proof}

\section{Straightness}\label{appstraight}
In this section we will identify conditions on the initial data so that the corresponding self-similar solution has a trivial lapse $\Omega$ and vanishing shift $b$.

First we record some fundamental consequences of self-similarity for $\chi$, $\underline{\chi}$, $\omega$, and $\underline{\omega}$.
\begin{lemma}\label{someselfsimilarrelationsyay}Suppose we have a self-similar solution. Then 
\begin{equation}\label{undechitochi}
\underline{\chi}_{AB} +\frac{v}{u}\chi_{AB} = \Omega^{-1}u^{-1}\slashed{g}_{AB} + \frac{1}{2}\Omega^{-1}\left(\mathcal{L}_b\slashed{g}\right)_{AB},
\end{equation}
\begin{equation}\label{trundchitotrchi}
{\rm tr}\underline{\chi} +\frac{v}{u}{\rm tr}\chi = \Omega^{-1}u^{-1}n + \Omega^{-1}{\rm div}b,
\end{equation}
\begin{equation}\label{trfreeundchitotrchi}
\hat{\underline{\chi}}_{AB}+ \frac{v}{u}\hat{\chi}_{AB} =  \frac{1}{2}\Omega^{-1}\left(\nabla\hat{\otimes} b\right)_{AB},
\end{equation}
\begin{equation}\label{selfsimilaromegarelations}
\underline{\omega} = -\frac{v}{u}\omega + \frac{1}{2}b^A\partial_A\left(\Omega^{-1}\right).
\end{equation}
\end{lemma}
\begin{proof}
We recall the formulas
\[e_4 = \Omega^{-1}\partial_v,\qquad e_3 = \Omega^{-1}\left(\partial_u + b^A\partial_A\right).\]

Furthermore, the numerology of Section~\ref{scalingbehav} implies that in the coordinate frame we may write $\slashed{g}_{AB}$ as 
\[\slashed{g}_{AB}\left(u,v,\theta\right) = u^2\mathring{\slashed{g}}_{AB}\left(\frac{v}{u},\theta\right).\]
Denoting the first variable of $\mathring{g}_{AB}$ by $\rho$, we may compute in a Lie-propagated frame that
\[\chi_{AB} = \frac{1}{2}\Omega^{-1} u\left(\partial_{\rho}\mathring{\slashed{g}}\right)_{AB},\qquad \underline{\chi}_{AB} = \Omega^{-1}u\mathring{\slashed{g}}_{AB} - \frac{v}{2}\Omega^{-1}\left(\partial_{\rho}\mathring{\slashed{g}}\right)_{AB} + \frac{1}{2}\Omega^{-1}\left(\mathcal{L}_b\slashed{g}\right)_{AB}.\]

In particular, we obtain the following important identity
\begin{equation*}
\underline{\chi}_{AB} +\frac{v}{u}\chi_{AB} = \Omega^{-1}u^{-1}\slashed{g}_{AB} + \frac{1}{2}\Omega^{-1}\left(\mathcal{L}_b\slashed{g}\right)_{AB}.
\end{equation*}
Since all the terms are tensorial, we conclude that the identity holds in any frame.

Tracing yields
\begin{equation*}
{\rm tr}\underline{\chi} +\frac{v}{u}{\rm tr}\chi = \Omega^{-1}u^{-1}n + \Omega^{-1}{\rm div}b.
\end{equation*}

The trace-free part is
\begin{equation*}
\hat{\underline{\chi}}_{AB}+ \frac{v}{u}\hat{\chi}_{AB} =  \frac{1}{2}\Omega^{-1}\left(\nabla\hat{\otimes} b\right)_{AB}.
\end{equation*}

For $\omega$ and $\underline{\omega}$, we observe that
\[\partial_v\Omega^{-1} = 2\omega,\qquad \left(\partial_u+b^A\partial_A\right)\Omega^{-1} = 2\underline{\omega}.\]

Self-similarity implies that $\Omega^{-1}$ is a function of $\frac{v}{u}$ and $\theta$. Thus, from the above, we obtain~\eqref{selfsimilaromegarelations}.

\end{proof}

\begin{proposition}\label{anequationfordivb}We have
\begin{align*}
\partial_v\left(\Omega^{-1}{\rm div}b\right) + \Omega^{-1}{\rm div}b\left(v^{-1} - \frac{2u}{nv}{\rm tr}\underline{\chi}\right) &= -4\omega n u^{-1} + 2\omega\left(\Omega{\rm tr}\underline{\chi} - {\rm div}b\right) + 2\underline{\omega}\Omega\frac{u}{v}{\rm tr}\underline{\chi} 
\\ \nonumber &\qquad - \frac{1}{4}\Omega^{-1}\frac{u}{v}\left|\nabla\hat{\otimes}b\right|^2 + \Omega \frac{u}{v}\hat{\underline{\chi}}\cdot\nabla\hat{\otimes}b - \Omega^{-1}\frac{u}{n v}\left({\rm div}b\right)^2+ \frac{u}{v}b^A\partial_A\left({\rm tr}\underline{\chi}\right)
\end{align*}\end{proposition}
\begin{proof}
We start by deriving a $\nabla_4$ propagation equation for $\Omega^{-1}{\rm div}b$. The basic idea is to differentiate the relation~\eqref{trundchitotrchi} with $\nabla_4$. We will then use the self-similar relations to express $\nabla_4{\rm tr}\underline{\chi}$ in terms $\nabla_3{\rm tr}\underline{\chi}$ and lower order terms. Then the Raychaudhuri equations can be used to simplify. We now turn to the details

Self-similarity implies that 
\[{\rm tr}\underline{\chi} = u^{-1}F\left(\frac{v}{u},\theta\right),\]
for some function $F$. We thus easily obtain

\begin{equation}\label{dutrchitodv}
\partial_u{\rm tr}\underline{\chi} = -u^{-1}{\rm tr}\underline{\chi} - \frac{v}{u}\partial_v{\rm tr}\underline{\chi}.
\end{equation}

Next, we differentiate $\frac{u}{v}$ times the self-similar relation~\eqref{trundchitotrchi} between ${\rm tr}\chi$ and ${\rm tr}\underline{\chi}$:
\begin{align}\label{diffselfsimtrchi}
\nabla_4{\rm tr}\chi &= \Omega^{-1}\partial_v\left(-\frac{u}{v}{\rm tr}\underline{\chi}\right) + \Omega^{-1}\partial_v\left(\Omega^{-1}\right)v^{-1}n -\Omega^{-2}v^{-2}n + \nabla_4\left(\frac{u}{v}\Omega^{-1}{\rm div}b\right)
\\ \nonumber &= 2\Omega^{-1}\frac{u}{v^2}{\rm tr}\underline{\chi} + \Omega^{-1}\frac{u^2}{v^2}\partial_u{\rm tr}\underline{\chi} + 2\Omega^{-1}\omega v^{-1}n - \Omega^{-2}v^{-2}n -\frac{u}{v^2}\Omega^{-2}{\rm div}b + \Omega^{-1}\frac{u}{v}\partial_v\left(\Omega^{-1}{\rm div}b\right)
\\ \nonumber &= 2\Omega^{-1}\frac{u}{v^2}{\rm tr}\underline{\chi} + \frac{u^2}{v^2}\nabla_3{\rm tr}\underline{\chi}-\Omega^{-1}\frac{u^2}{v^2}b^A\partial_A\left({\rm tr}\underline{\chi}\right) +2\omega\Omega^{-1}v^{-1}n
\\ \nonumber &\qquad -\Omega^{-2}v^{-2}n - \frac{u}{v^2}\Omega^{-2}{\rm div}b + \Omega^{-1}\frac{u}{v}\partial_v\left(\Omega^{-1}{\rm div}b\right). 
\end{align}

We get
\begin{align*}
\partial_v\left(\Omega^{-1}{\rm div}b\right) - v^{-1}\Omega^{-1}{\rm div}b &= \Omega\frac{v}{u}\nabla_4{\rm tr}\chi -2v^{-1}{\rm tr}\underline{\chi} -\Omega \frac{u}{v}\nabla_3{\rm tr}\underline{\chi}+\frac{u}{v}b^A\partial_A\left({\rm tr}\underline{\chi}\right) - 2\omega u^{-1}n + \Omega^{-1}u^{-1}v^{-1}n
\\ \nonumber &= -\Omega \frac{v}{nu}\left({\rm tr}\chi\right)^2 - 2\omega u^{-1}\left(\Omega v {\rm tr}\chi\right) - \Omega \frac{v}{u}\left|\hat{\chi}\right|^2 
\\ \nonumber &\qquad +\Omega \frac{u}{vn}\left({\rm tr}\underline{\chi}\right)^2  + 2\underline{\omega}\Omega\frac{u}{v}{\rm tr}\underline{\chi} + \Omega \frac{u}{v}\left|\hat{\underline{\chi}}\right|^2
\\ \nonumber &\qquad -2v^{-1}{\rm tr}\underline{\chi} + \frac{u}{v}b^A\partial_A\left({\rm tr}\underline{\chi}\right) -2\omega n u^{-1} + \Omega^{-1}u^{-1}v^{-1} n
\\ \nonumber &= - 2\omega u^{-1}\left(\Omega v {\rm tr}\chi\right) - \Omega \frac{v}{u}\left|\hat{\chi}\right|^2 
\\ \nonumber &\qquad + 2\underline{\omega}\Omega\frac{u}{v}{\rm tr}\underline{\chi} + \Omega \frac{u}{v}\left|\hat{\underline{\chi}}\right|^2 - \Omega^{-1}\frac{u}{n v}\left({\rm div}b\right)^2
\\ \nonumber &\qquad + \frac{u}{v}b^A\partial_A\left({\rm tr}\underline{\chi}\right) -2\omega n u^{-1} + 2\frac{u}{n v}{\rm tr}\underline{\chi}{\rm div} b - 2v^{-1}{\rm div}b
\\ \nonumber &= - 2\omega u^{-1}\left(\Omega v {\rm tr}\chi\right) - \frac{1}{4}\Omega^{-1}\frac{u}{v}\left|\nabla\hat{\otimes}b\right|^2 + \Omega \frac{u}{v}\hat{\underline{\chi}}\cdot\nabla\hat{\otimes}b
\\ \nonumber &\qquad + 2\underline{\omega}\Omega\frac{u}{v}{\rm tr}\underline{\chi}  - \Omega^{-1}\frac{u}{n v}\left({\rm div}b\right)^2
\\ \nonumber &\qquad + \frac{u}{v}b^A\partial_A\left({\rm tr}\underline{\chi}\right) -2\omega n u^{-1} + 2\frac{u}{n v}{\rm tr}\underline{\chi}{\rm div} b - 2v^{-1}{\rm div}b
\end{align*}

The third and fourth equalities are obtained by squaring the self-similar relations for ${\rm tr}\chi$ and $\hat{\chi}$.

We end up with
\begin{align*}
\partial_v\left(\Omega^{-1}{\rm div}b\right) + \Omega^{-1}{\rm div}b\left(v^{-1} - \frac{2u}{nv}{\rm tr}\underline{\chi}\right) &= -4\omega n u^{-1} + 2\omega\left(\Omega{\rm tr}\underline{\chi} - {\rm div}b\right) + 2\underline{\omega}\Omega\frac{u}{v}{\rm tr}\underline{\chi} 
\\ \nonumber &\qquad - \frac{1}{4}\Omega^{-1}\frac{u}{v}\left|\nabla\hat{\otimes}b\right|^2 + \Omega \frac{u}{v}\hat{\underline{\chi}}\cdot\nabla\hat{\otimes}b - \Omega^{-1}\frac{u}{n v}\left({\rm div}b\right)^2+ \frac{u}{v}b^A\partial_A\left({\rm tr}\underline{\chi}\right)
\end{align*}
\end{proof}

Next, let's compute the $\nabla_4$ equation for $\underline{\omega}$.
\begin{proposition}\label{anequationforundomega}\begin{align*}
\nabla_4\underline{\omega} &= \frac{1}{4}\omega\left(\Omega^{-1}u^{-1}n + \frac{v}{u}{\rm tr}\chi - {\rm tr}\underline{\chi}\right) + \frac{1}{4}\underline{\omega}\left(\Omega^{-1}v^{-1}n + \frac{u}{v}{\rm tr}\underline{\chi} - {\rm tr}\chi\right) + \frac{1}{2}\slashed{\Delta}\log\Omega + \frac{1}{8}\Omega^{-2}v^{-1}{\rm div}b
\\ \nonumber &\qquad    + \frac{1}{32}\Omega^{-2}\frac{u}{v}\left|\nabla\hat{\otimes}b\right|^2 + \frac{1}{8n}\Omega^{-2}\frac{u}{v}\left({\rm div}b\right)^2 + \frac{1}{8}\Omega^{-1}\frac{u}{v}b^A\nabla_A\left(\Omega^{-1}{\rm div}b\right)
\\ \nonumber &\qquad + 2\underline{\omega}\omega + 3\left|\zeta\right|^2 - \left|\nabla\log\Omega\right|^2-\frac{1}{2}\left|\eta\right|^2.
\end{align*}
\end{proposition}
\begin{proof}We start with
\begin{align*}
\nabla_4{\rm tr}\underline{\chi} + \frac{1}{n}{\rm tr}\chi{\rm tr}\underline{\chi} + \hat{\chi}\hat{\underline{\chi}} &= \nabla_4\left(-\frac{v}{u}{\rm tr}\chi + \Omega^{-1}u^{-1}n + \Omega^{-1}{\rm div}b\right)
\\ \nonumber &\qquad + \frac{1}{n}{\rm tr}\chi\left(-\frac{v}{u}{\rm tr}\chi + \Omega^{-1}u^{-1}n + \Omega^{-1}{\rm div}b\right)
\\ \nonumber &\qquad + \hat{\chi}\left(-\frac{v}{u}\hat{\chi} + \frac{1}{2}\Omega^{-1}\left(\nabla\hat{\otimes}b\right)\right)
\\ \nonumber &= -\Omega^{-1}u^{-1}{\rm tr}\chi + 2\Omega^{-1}\omega u^{-1}n + \nabla_4\left(\Omega^{-1}{\rm div}b\right)
\\ \nonumber &\qquad +\frac{1}{n}\frac{v}{u}\left({\rm tr}\chi\right)^2 + 2\frac{v}{u}\omega{\rm tr}\chi + \frac{v}{u}\left|\hat{\chi}\right|^2
\\ \nonumber &\qquad -\frac{1}{n}\frac{v}{u}\left({\rm tr}\chi\right)^2 + \Omega^{-1}u^{-1}{\rm tr}\chi + \frac{1}{n}\Omega^{-1}{\rm div}b{\rm tr}\chi 
\\\ \nonumber &\qquad -\frac{v}{u}\left|\hat{\chi}\right|^2 + \frac{1}{2}\Omega^{-1}\hat{\chi}\left(\nabla\hat{\otimes}b\right)
\\ \nonumber &= 2\Omega^{-1}\omega u^{-1}n + \nabla_4\left(\Omega^{-1}{\rm div}b\right)+ 2\frac{v}{u}\omega{\rm tr}\chi
\\ \nonumber &\qquad + \frac{1}{n}\Omega^{-1}{\rm tr}\chi {\rm div}b + \frac{1}{2}\Omega^{-1}\hat{\chi}\left(\nabla\hat{\otimes}b\right).
\end{align*}

Similarly,
\begin{align*}
\nabla_3{\rm tr}\chi + \frac{1}{n}{\rm tr}\underline{\chi}{\rm tr}\chi + \hat{\underline{\chi}}\hat{\chi} &= \nabla_3\left(-\frac{u}{v}{\rm tr}\underline{\chi} + \Omega^{-1}v^{-1}n + \frac{u}{v}\Omega^{-1}{\rm div}b\right)
\\ \nonumber &\qquad + \frac{1}{n}{\rm tr}\underline{\chi}\left(-\frac{u}{v}{\rm tr}\underline{\chi} + \Omega^{-1}v^{-1}n + \frac{u}{v}\Omega^{-1}{\rm div}b\right)
\\ \nonumber &\qquad + \hat{\underline{\chi}}\left(-\frac{u}{v}\hat{\underline{\chi}} + \frac{1}{2}\frac{u}{v}\Omega^{-1}\left(\nabla\hat{\otimes}b\right)\right)
\\ \nonumber &= -v^{-1}\Omega^{-1}{\rm tr}\underline{\chi} + 2\Omega^{-1}\underline{\omega}v^{-1}n + \nabla_3\left(\frac{u}{v}\Omega^{-1}{\rm div}b\right)
\\ \nonumber &\qquad +\frac{1}{n}\frac{u}{v}\left({\rm tr}\underline{\chi}\right)^2 + 2\underline{\omega}\frac{u}{v}{\rm tr}\underline{\chi} + \frac{u}{v}\left|\hat{\underline{\chi}}\right|^2 
\\ \nonumber &\qquad -\frac{1}{n}\frac{u}{v}\left({\rm tr}\underline{\chi}\right)^2 + \Omega^{-1}v^{-1}{\rm tr}\underline{\chi} + \frac{1}{n}\frac{u}{v}\Omega^{-1}{\rm tr}\underline{\chi}{\rm div}b
\\ \nonumber &\qquad -\frac{u}{v}\left|\hat{\underline{\chi}}\right|^2 +\frac{1}{2}\frac{u}{v}\Omega^{-1}\hat{\underline{\chi}}\cdot\left(\nabla\hat{\otimes}b\right)
\\ \nonumber &=   2\Omega^{-1}\underline{\omega}v^{-1}n + \nabla_3\left(\frac{u}{v}\Omega^{-1}{\rm div}b\right)+ 2\underline{\omega}\frac{u}{v}{\rm tr}\underline{\chi}
\\ \nonumber &\qquad + \frac{1}{n}\frac{u}{v}\Omega^{-1}{\rm tr}\underline{\chi}{\rm div}b+\frac{1}{2}\frac{u}{v}\Omega^{-1}\hat{\underline{\chi}}\cdot\left(\nabla\hat{\otimes}b\right).
\end{align*}

Now we use the null structure equations for $\nabla_3{\rm tr}\chi$ and $\nabla_4{\rm tr}\underline{\chi}$. We obtain
\begin{align*}
&4\rho +2\underline{\omega}{\rm tr}\chi + 2\omega{\rm tr}\underline{\chi} + 4\slashed{\Delta}\log\Omega + 2\left|\eta\right|^2 + 2\left|\underline{\eta}\right|^2
\\ \nonumber &= 2\omega\left(\Omega^{-1} u^{-1}n + \frac{v}{u}{\rm tr}\chi\right) + 2\underline{\omega}\left(\Omega^{-1}v^{-1}n + \frac{u}{v}{\rm tr}\underline{\chi}\right) + \frac{1}{2}\Omega^{-1}\hat{\chi}\left(\nabla\hat{\otimes}b\right) + \frac{1}{2}\frac{u}{v}\Omega^{-1}\hat{\underline{\chi}}\left(\nabla\hat{\otimes}b\right)
\\ \nonumber &\qquad + \nabla_4\left(\Omega^{-1}{\rm div}b\right) + \nabla_3\left(\frac{u}{v}\Omega^{-1}{\rm div}b\right) + \frac{1}{n}\Omega^{-1}{\rm div b}\left({\rm tr}\chi + \frac{u}{v}{\rm tr}\underline{\chi}\right)
\\ \nonumber &=2\omega\left(\Omega^{-1} u^{-1}n + \frac{v}{u}{\rm tr}\chi\right) + 2\underline{\omega}\left(\Omega^{-1}v^{-1}n + \frac{u}{v}{\rm tr}\underline{\chi}\right) + \frac{1}{2}\Omega^{-1}\hat{\chi}\left(\nabla\hat{\otimes}b\right) 
\\ \nonumber &\qquad + \frac{1}{2}\frac{u}{v}\Omega^{-1}\hat{\underline{\chi}}\left(\nabla\hat{\otimes}b\right) + \frac{1}{n}\Omega^{-1}{\rm div b}\left(\Omega^{-1}v^{-1}n + \frac{u}{v}\Omega^{-1}{\rm div}b\right) +\Omega^{-1}b^A\nabla_A\left(\frac{u}{v}\Omega^{-1}{\rm div}b\right).
\end{align*}

In the final equality we used that
\[\nabla_4\left(\Omega^{-1}{\rm div}b\right) + \nabla_3\left(\frac{u}{v}\Omega^{-1}{\rm div}b\right) = \Omega^{-1}b^A\nabla_A\left(\frac{u}{v}\Omega^{-1}{\rm div}b\right).\]
To see why this is true, we simply note that
\[\Omega^{-1}{\rm div}b = u^{-1}F\left(\frac{v}{u},\theta\right)\]
for some function $F$. Finally, we use the $\nabla_4\underline\omega$ null structure equation to obtain:
\begin{align*}
\nabla_4\underline{\omega} &= \frac{1}{2}\rho + \frac{1}{4}\left|\underline{\eta}\right|^2 - \frac{1}{4}\left|\eta\right|^2 + 2\underline{\omega}\omega + 3\left|\zeta\right|^2 - \left|\nabla\log\Omega\right|^2
\\ \nonumber &= \frac{1}{4}\omega\left(\Omega^{-1}u^{-1}n + \frac{v}{u}{\rm tr}\chi - {\rm tr}\underline{\chi}\right) + \frac{1}{4}\underline{\omega}\left(\Omega^{-1}v^{-1}n + \frac{u}{v}{\rm tr}\underline{\chi} - {\rm tr}\chi\right) + \frac{1}{2}\slashed{\Delta}\log\Omega + \frac{1}{8}\Omega^{-2}v^{-1}{\rm div}b
\\ \nonumber &\qquad    + \frac{1}{32}\Omega^{-2}\frac{u}{v}\left|\nabla\hat{\otimes}b\right|^2 + \frac{1}{8n}\Omega^{-2}\frac{u}{v}\left({\rm div}b\right)^2 + \frac{1}{8}\Omega^{-1}\frac{u}{v}b^A\nabla_A\left(\Omega^{-1}{\rm div}b\right)
\\ \nonumber &\qquad + 2\underline{\omega}\omega + 3\left|\zeta\right|^2 - \left|\nabla\log\Omega\right|^2-\frac{1}{2}\left|\eta\right|^2.
\end{align*}
\end{proof}

\begin{proposition}\label{anequationforzeta}We have 
\begin{align}\label{alltogethernow}& \nabla_4\zeta_A + \frac{1}{2}\zeta^B\left(\chi_{AB} + {\rm tr}\chi\slashed{g}_{AB} - \frac{u}{v}\underline{\chi}_{AB} - \frac{u}{v}{\rm tr}\underline{\chi}\slashed{g}_{AB}+\Omega^{-1}v^{-1}\slashed{g}_{AB}\right) -\frac{1}{4}\Omega^{-1}\frac{u}{v}\slashed{\Delta}b_A
\\ \nonumber &=  \frac{1}{4}\nabla^B\left(\Omega^{-1}\right)\frac{u}{v}\left(\mathcal{L}_b\slashed{g}\right)_{AB} + \frac{1}{4}\Omega^{-1}\frac{u}{v}\nabla^B\nabla_Ab_B
\\ \nonumber &\qquad -\frac{1}{2}\Omega^{-1}\frac{u}{v}b^B\nabla_B\underline{\eta}_A -\frac{1}{2} \Omega^{-1}\frac{u}{v}\underline{\eta}_B\nabla_Ab^B+ \frac{1}{4}\Omega^{-1}\frac{u}{v}\left(\mathcal{L}_b\slashed{g}\right)_A^B\underline{\eta}_B
\\ \nonumber &\qquad +\frac{1}{2}\left[-\nabla_A\left(\Omega^{-1}\right)v^{-1}n - \nabla_A\left(\Omega^{-1}\right)\frac{u}{v}{\rm div}b - \Omega^{-1}\frac{u}{v}\nabla_A{\rm div}b\right].
\end{align}
\end{proposition}
Remember that $\zeta$ is essentially the \emph{negative} $v$-derivative of $b$! 
\begin{proof}

Next we turn to the wave equation for the shift. We start by deriving a relation between $\nabla_4\underline{\eta}$ and $\nabla_3\underline{\eta}$. In the Lie-Propagated frame we have that
\[\underline{\eta}_A\left(u,v,\theta\right) = f_A\left(\frac{v}{u},\theta\right),\]
for some tensor $f_A$. We denote the first variable of $f_A$ by $\rho$. We then compute in a Lie-propagated frame $\{E_A\}$ which we can also assume to be geodesic (with respect to $\slashed{g}$) at the give point of interest:
\begin{align}\label{selfsimforundeta}
\nabla_4\underline{\eta}_A &= \Omega^{-1}\partial_v\left[f_A\left(\frac{v}{u},\theta\right)\right]- \chi_A^{\ \ B}\underline{\eta}_B
\\ \nonumber &= \Omega^{-1}u^{-1}\left(\partial_{\rho}f_A\right)\left(\frac{v}{u},\theta\right) - \chi_A^{\ \ B}\underline{\eta}_B
\\ \nonumber &= -\Omega^{-1}\frac{u}{v}\partial_u\left[f_A\left(\frac{v}{u},\theta\right)\right] - \chi_A^{\ \ B}\underline{\eta}_B
\\ \nonumber &= -\frac{u}{v}\nabla_3\underline{\eta}_A + \Omega^{-1}\frac{u}{v}b^B\partial_B\left(\underline{\eta}_A\right)+ \frac{u}{v}\Omega^{-1}\left(\partial_Ab^B\right)\underline{\eta}_A-\frac{u}{v}\underline{\chi}_A^{\ \ B}\underline{\eta}_B - \chi_A^{\ B}\underline{\eta}_B
\\ \nonumber &= -\frac{u}{v}\nabla_3\underline{\eta}_A+ \Omega^{-1}\frac{u}{v}b^B\nabla_B\underline{\eta}_A+ \Omega^{-1}\frac{u}{v}\underline{\eta}_B\nabla_Ab^B-\Omega^{-1}v^{-1}\underline{\eta}_A - \frac{1}{2}\Omega^{-1}\frac{u}{v}\left(\mathcal{L}_b\slashed{g}\right)_A^{\ \ B}\underline{\eta}_B.
\end{align}
In the last line we used the self-similarity relation for $\frac{u}{v}\underline{\chi} + \chi$.

The $\nabla_4$ equation for $\eta_A$ is 
\[\nabla_4\eta_A = -2\chi_A^{\ \ B}\zeta_B - \beta_A \Rightarrow \]
\begin{equation}\label{someequation}
2\nabla_4\zeta_A + \nabla_4\underline{\eta}_A = -2\chi_A^{\ \ B}\zeta_B - \beta_A.
\end{equation}

The $\nabla_3$ equation for $\underline{\eta}_A$ is
\begin{equation}\label{someequation2}
\nabla_3\underline{\eta}_A = 2\underline{\chi}_A^{\ \ B}\zeta_B + \underline{\beta}_A.
\end{equation}

Now we take~\eqref{someequation}, add it to $\frac{u}{v}$ times~\eqref{someequation2}, and use~\eqref{selfsimforundeta} to simplify. We obtain:
\begin{align}
\nonumber 2\nabla_4\zeta_A + \nabla_4\underline{\eta}_A + \frac{u}{v}\nabla_3\underline{\eta}_A &= -2\chi_{A}^{\ \ B}\zeta_B - \beta_A + 2\frac{u}{v}\underline{\chi}_A^{\ \ B}\zeta_B + \frac{u}{v}\underline{\beta}_A \Rightarrow 
\end{align}
\begin{align}
&2\nabla_4\zeta_A +\Omega^{-1}\frac{u}{v}b^B\nabla_B\underline{\eta}_A+ \Omega^{-1}\frac{u}{v}\underline{\eta}_B\nabla_Ab^B-\Omega^{-1}v^{-1}\underline{\eta}_A - \frac{1}{2}\Omega^{-1}\frac{u}{v}\left(\mathcal{L}_b\slashed{g}\right)_A^{\ \ B}\underline{\eta}_B = 
\\ \nonumber &\qquad -2\chi_{A}^{\ \ B}\zeta_B - \beta_A + 2\frac{u}{v}\underline{\chi}_A^{\ \ B}\zeta_B + \frac{u}{v}\underline{\beta}_A \Rightarrow
\end{align}
\begin{align}
\label{niceequationfortorsion}\nabla_4\zeta_A &+ \zeta_B\left(\chi_A^{\ B}-\frac{u}{v}\underline{\chi}_A^{\ B}\right) - \frac{1}{2}\Omega^{-1}v^{-1}\underline{\eta}_A = -\frac{1}{2}\beta_A + \frac{1}{2}\frac{u}{v}\underline{\beta}_A
\\ \nonumber &\qquad  -\frac{1}{2}\Omega^{-1}\frac{u}{v}b^B\nabla_B\underline{\eta}_A - \frac{1}{2}\Omega^{-1}\frac{u}{v}\underline{\eta}_B\nabla_Ab^B+ \frac{1}{4}\Omega^{-1}\frac{u}{v}\left(\mathcal{L}_b\slashed{g}\right)_A^{\ \ B}\underline{\eta}.
\end{align}

Next, we want to use the constraints to eliminate curvature from~\eqref{niceequationfortorsion}. The relevant equations are the following:
\begin{align*}
 \beta_B &=   -\nabla^A\chi_{AB}+\nabla_B{\rm tr}\chi + {\rm tr}\chi \zeta_B - \zeta^A\chi_{AB},
\\ \nonumber \underline{\beta}_B &= \nabla^A\underline\chi_{AB}-\nabla_B{\rm tr}\underline\chi + {\rm tr}\underline\chi \zeta_B -\zeta^B\underline{\chi}_{AB}
\end{align*}

These imply
\begin{align}\label{killbetas}
-\beta_A + \frac{u}{v}\underline{\beta}_A &= \nabla^B\chi_{BA}-\nabla_A{\rm tr}\chi - {\rm tr}\chi \zeta_A + \zeta^B\chi_{BA}+\frac{u}{v}\nabla^B\underline\chi_{BA}-\frac{u}{v}\nabla_A{\rm tr}\underline\chi +\frac{u}{v} {\rm tr}\underline\chi \zeta_A -\frac{u}{v} \zeta^B\underline\chi_{BA}
\\ \nonumber &= v^{-1}\nabla_A\left(\Omega^{-1}\right) + \frac{1}{2}\nabla^B\left(\Omega^{-1}\right)\frac{u}{v}\left(\mathcal{L}_b\slashed{g}\right)_{AB} + \frac{1}{2}\Omega^{-1}\frac{u}{v}\nabla^B\left(\mathcal{L}_b\slashed{g}\right)_{AB}
\\ \nonumber &\qquad +\left[-\nabla_A\left(\Omega^{-1}\right)v^{-1}n - \nabla_A\left(\Omega^{-1}\right)\frac{u}{v}{\rm div}b - \Omega^{-1}\frac{u}{v}\nabla_A{\rm div}b\right]
\\ \nonumber &\qquad +\zeta^B\left(\chi_{AB} - {\rm tr}\chi\slashed{g}_{AB} -\frac{u}{v}\underline{\chi}_{AB} + \frac{u}{v}{\rm tr}\underline{\chi}\slashed{g}_{AB}\right).
\end{align}

All together we obtain
\begin{align}\label{alltogethernow}& \nabla_4\zeta_A + \frac{1}{2}\zeta^B\left(\chi_{AB} + {\rm tr}\chi\slashed{g}_{AB} - \frac{u}{v}\underline{\chi}_{AB} - \frac{u}{v}{\rm tr}\underline{\chi}\slashed{g}_{AB}+\Omega^{-1}v^{-1}\slashed{g}_{AB}\right) -\frac{1}{4}\Omega^{-1}\frac{u}{v}\slashed{\Delta}b_A
\\ \nonumber &=  \frac{1}{2}v^{-1}\nabla_A\left(\Omega^{-1}\right) + \frac{1}{4}\nabla^B\left(\Omega^{-1}\right)\frac{u}{v}\left(\mathcal{L}_b\slashed{g}\right)_{AB} + \frac{1}{4}\Omega^{-1}\frac{u}{v}\nabla^B\nabla_Ab_B+\frac{1}{2}\Omega^{-1}v^{-1}\nabla_A\left(\log\Omega\right)
\\ \nonumber &\qquad +\frac{1}{2}\left[-\nabla_A\left(\Omega^{-1}\right)v^{-1}n - \nabla_A\left(\Omega^{-1}\right)\frac{u}{v}{\rm div}b - \Omega^{-1}\frac{u}{v}\nabla_A{\rm div}b\right]
\\ \nonumber &\qquad -\frac{1}{2}\Omega^{-1}\frac{u}{v}b^B\nabla_B\underline{\eta}_A + \frac{1}{4}\Omega^{-1}\frac{u}{v}\left(\mathcal{L}_b\slashed{g}\right)_A^B\underline{\eta}_B- \frac{1}{2}\Omega^{-1}\frac{u}{v}\underline{\eta}_B\nabla_Ab^B
\\ \nonumber &=  \frac{1}{4}\nabla^B\left(\Omega^{-1}\right)\frac{u}{v}\left(\mathcal{L}_b\slashed{g}\right)_{AB} + \frac{1}{4}\Omega^{-1}\frac{u}{v}\nabla^B\nabla_Ab_B-\frac{1}{2}\Omega^{-1}\frac{u}{v}b^B\nabla_B\underline{\eta}_A -\frac{1}{2} \Omega^{-1}\frac{u}{v}\underline{\eta}_B\nabla_Ab^B+ \frac{1}{4}\Omega^{-1}\frac{u}{v}\left(\mathcal{L}_b\slashed{g}\right)_A^B\underline{\eta}_B
\\ \nonumber &\qquad +\frac{1}{2}\left[-\nabla_A\left(\Omega^{-1}\right)v^{-1}n - \nabla_A\left(\Omega^{-1}\right)\frac{u}{v}{\rm div}b - \Omega^{-1}\frac{u}{v}\nabla_A{\rm div}b\right].
\end{align}

\end{proof}

We have
\begin{proposition}\label{initialvanishingstr}For each fixed $u \in [-1,0)$ every self-similar solution satisfies 
\begin{equation}\label{vanishomzeta}
\left|\nabla^i\zeta\right| + \left|\nabla^i\omega\right| = O\left(v^{\frac{n}{2}}\right)\text{ as }v\to 0,
\end{equation}
for every $i \geq 0$, and where the implied constant is allowed to depend on $i$.
\end{proposition}
\begin{proof}Throughout the proof, we fix an arbitrary value of $u \in [-1,0)$. We also allow all constants to depend on $i$ without explicitly saying so.

We first note that Propositions~\ref{anequationfordivb}, ~\ref{anequationforzeta}, and Lemma~\ref{someselfsimilarrelationsyay} imply that:
\begin{equation}\label{divbvanstr}
4\omega n u^{-1} = -\partial_v\left(\Omega^{-1}{\rm div}b\right) + \Omega^{-1}{\rm div b}\left(v^{-1} + O\left(1\right)\right)  + v^{-1}O\left(\left|\slashed{\nabla}b\right|^2\right) + O\left(\left|\slashed{\nabla}b\right| + |b|\right),
\end{equation}
\begin{equation}\label{zetavanstr}
\partial_v\zeta_A - \frac{n}{2}\zeta_A\left(v^{-1} + O\left(1\right)\right) = v^{-1}O\left(\left|\nabla^2b\right| +\left|\nabla\Omega\right|+ \left|\nabla\Omega\right|^2 + \left|\nabla b\right|^2 + \left|\nabla\zeta\right|^2\right),
\end{equation}
where $\zeta_A$ is expressed in the coordinate frame and the big $O$ expansions all hold as $v\to 0$. 

Furthermore, it is straightforward to check that the above equations continue to hold with $\omega$, $\zeta$, $\Omega$, and $b$ replaced everywhere by $\nabla^i$ of $\omega$, $\zeta$, $\Omega$, and $b$ respectively.

We now will prove by induction on that the following statement holds for each $k = 1,\cdots,\lfloor \frac{n}{2}\rfloor$:
\begin{equation}\label{inducthyp}
\nabla^i\zeta = O\left(v^k\right),\qquad \nabla^i\omega = O\left(v^k\right),\qquad \nabla^ib = O_i\left(v^{k+1}\right),\qquad \nabla^i\left(\Omega-1\right) = O\left(v^{k+1}\right).
\end{equation}

We start with the base case $k = 1$. It follows immediately from the initial data analysis of Propositions~\ref{incomingdatan2},~\ref{incomingdatanodd}, and~\ref{incomingdataneven} that we have
\[\nabla^i\zeta|_{v = 0} = 0,\qquad \nabla^i\omega|_{v=0} = 0\qquad \forall i \geq 0.\]
Next, it follows from Propositions~\ref{ambientreg2},~\ref{ambientregodd}, and~\ref{ambientregeven} that $\nabla^i\zeta$ and $\nabla^i\omega$ are $C^1$ in $v$. Thus we have
\[\nabla^i\zeta = O(v),\qquad \nabla^i\omega = O(v)\qquad \forall i \geq 0.\]

Similarly, we have that $\Omega-1$ and $b$ vanish when $\{v = 0\}$, and we know that
\begin{equation}\label{shiftomegabetter}
\partial_v\partial_{B_I}b^A = -4\partial_{B_I}\left(\Omega^2\zeta^A\right),\qquad \partial_v\left(\partial_{B_I}\Omega^{-1}\right) = 2\partial_{B_I}\omega,
\end{equation}
where $I = (i_1,\cdots,i_l)$ and $\partial_{B_I} = \partial_{B_{i_1}}\cdots\partial_{B_{i_l}}$ denotes an arbitrary product of coordinate derivatives along $\mathcal{S}$. We immediately obtain that~\eqref{inducthyp} holds for $k=1$. This completes the proof of the proposition when $n = 2$.

Next, we assume that~\eqref{inducthyp} holds for some $k \leq \lfloor \frac{n}{2}\rfloor - 1$ and we will show that~\eqref{inducthyp} holds for $k+1$: 

First of all, using the $\nabla_4$ equation for $\eta$ and $\omega$ and Propositions~\ref{ambientreg2},~\ref{ambientregodd}, and~\ref{ambientregeven}, it is straightforward to show that there exists tensors $\zeta^{(1)}$, $\cdots$, $\zeta^{\left(\lfloor\frac{n}{2}\rfloor\right)}$ and $\omega^{(1)}$, $\cdots$, $\omega^{\left(\lfloor\frac{n}{2}\rfloor\right)}$ on $\mathcal{S}$ so that
\[\zeta_A = \zeta^{(1)}v + \cdots + \zeta^{\left(\lfloor \frac{n}{2}\rfloor\right)}v^{\lfloor \frac{n}{2}\rfloor} + O\left(v^{\lfloor \frac{n}{2}\rfloor + \frac{1}{2}}\right)\text{ as }v\to 0,\]
\[\omega = \omega^{(1)}v + \cdots + \omega^{\left(\lfloor \frac{n}{2}\rfloor\right)}v^{\lfloor \frac{n}{2}\rfloor} + O\left(v^{\lfloor \frac{n}{2}\rfloor+\frac{1}{2}}\right)\text{ as }v\to 0.\]

 The induction hypothesis implies that $\zeta^{(i)}$ and $\omega^{(i)}$ vanish for $i = 1$, $\cdots$, $k - 1$. Plugging in the expansion for $\zeta$ into~\eqref{zetavanstr} and using the induction hypothesis yields 
 \[\left(k+1-\frac{n}{2}\right)\zeta^{(k)}v^{k-1} = o\left(v^{k-1}\right) \Rightarrow \zeta^{(k)} = 0.\] 
 Similarly, we obtain that $\nabla^i\zeta^{(k)} = 0$ for all $i \geq 0$. In particular, for all $i \geq 0$,
 \[\nabla^i\zeta = O\left(v^{k+1}\right).\]
 
 Using~\eqref{shiftomegabetter} we also obtain that
 \[\left|\nabla^ib\right| = O\left(v^{k+2}\right).\]
 Then~\eqref{divbvanstr} immediately implies that 
  \[\nabla^i\omega = O\left(v^{k+1}\right).\]
  Finally,~\eqref{shiftomegabetter} then gives that
  \[\nabla^i\left(\Omega-1\right) = O\left(v^{k+2}\right).\]
  
  This finishes the induction argument and the proof when $n$ is even. When $n$ is odd, the argument is finished by one final iteration of the induction argument.  
 
\end{proof}

It is clear that the argument from Proposition~\ref{initialvanishingstr} cannot be used to show that $\zeta$ vanishes faster than $v^{\frac{n}{2}}$. However, the next proposition shows that if we happen to know that $\zeta$ vanishes faster than $v^{\frac{n}{2}}$, then $\zeta$,  $\omega$,  $b$, and $\Omega^{-1}-1$ all vanish to infinite order as $v\to 0$.

\begin{proposition}\label{fromzetatoeverything}Let $n\geq 3$. Suppose that we know that for every $i \geq 0$
\[\lim_{v\to 0}v^{-\frac{n}{2}}\left|\nabla^i\zeta\right| = 0.\]
Then $\omega$, $\zeta$, $b$, and $\Omega^{-1}-1$ all vanish to infinite order as $v\to 0$.
\end{proposition}
\begin{proof}By induction on $k$ we will show that for all $i,k \geq 0$
\begin{equation}\label{infvanishinduct}
\nabla^i\zeta = O\left(v^{\frac{n}{2} + k}\right),\qquad \nabla^ib = O\left(v^{\frac{n}{2}+k+1}\right),\qquad \nabla^i\omega = O\left(v^{\frac{n}{2}+k}\right),\qquad \nabla^i\left(\Omega-1\right) = O\left(v^{\frac{n}{2}+k+1}\right),
\end{equation}
where the implied constant may depend on $k$ and $i$.

We start with the base case $k = 0$. By assumption the desired estimate for $\zeta$ already holds. The desired estimate for $b$ then follows from~\eqref{shiftomegabetter}. In turn the desired estimate for $\omega$ follows from~\eqref{divbvanstr}. Finally, the equation for $\Omega$ from~\eqref{shiftomegabetter} finishes the base case.

Now we assume that~\eqref{infvanishinduct} holds for some fixed $k$ and we need to show it then holds for $k+1$. It follows from~\eqref{zetavanstr} and the induction hypothesis that, for any multi-index $B_I$  we have
\[\partial_v\left(\partial_{B_I}\zeta_A\right) - \frac{n}{2}v^{-1}\left(\partial_{B_I}\zeta_A\right) = O\left(v^{\frac{n}{2}+k}\right) \Rightarrow \]
\[\partial_v\left(v^{-\frac{n}{2}}\partial_{B_I}\zeta_A\right) = O\left(v^k\right) \Rightarrow \]
\[\partial_{B_I}\zeta = O\left(v^{\frac{n}{2} + k + 1}\right).\]
This establishes the desired estimate for $\zeta$. The estimates for the other quantities follow just as in the base case.
\end{proof}

Next we show that the ``straightness'' assumption on initial data implies the hypothesis of Proposition~\ref{fromzetatoeverything} and thus the infinite order vanishing of $\zeta$, $\omega$, $b$, and $\Omega- 1$.
\begin{proposition}\label{infvanish}Let $n \geq 3$. Suppose that we have a self-similar solution which satisfies the following, depending on the parity of the dimension:
\begin{enumerate}
	\item When $n$ is odd then then 
	\[{\rm div}\left({\rm tf}\mathcal{L}_v^{\frac{n}{2}-2}\alpha|_{\mathcal{S}_{-1,0}}\right) = 0.\]
	\item When $n$ is even then 
		\[{\rm div}\left({\rm tf}\mathcal{L}_v^{\frac{n}{2}-2}\alpha|_{\mathcal{S}_{-1,0}}\right) = D,\]
		where $D$ is an explicit $1$-form depending on the $\slashed{g}_{\mathcal{S}_{-1,0}}$, and the proof below will indicate how to calculate $D$ in principle.
\end{enumerate}
Then $\omega$, $\zeta$, $b$, and $\Omega^{-1}-1$ all vanish to infinite order as $v\to 0$.
\end{proposition}
\begin{proof}By Proposition~\ref{fromzetatoeverything} we just need to show that
\[\lim_{v\to 0}v^{-\frac{n}{2}}\left|\nabla^i\zeta\right| = 0,\qquad \forall i\geq 0.\]

To see why this holds, we first consider the case of $n$-odd and $i = 0$. First of all, it follows immediately from Proposition~\ref{initialvanishingstr} and the null structure equation for $\nabla_4\eta$ that $\beta = O\left(v^{\frac{n}{2}-1}\right)$ as $v\to 0$. Next, since $\eta = \zeta + \nabla_A\log\Omega$ and $\nabla\log\Omega = O\left(v^{\frac{n}{2}+1}\right)$, it is clear that $\lim_{v\to 0}v^{-\frac{n}{2}}\left|\zeta\right| = 0$ if and only if 
\begin{equation}\label{stuffinbetatobe0}
\lim_{v\to 0}v^{-\frac{n}{2}+1}\left|\beta\right| = 0.
\end{equation}

However, Proposition~\ref{ambientregodd} and signature considerations imply that the only term on the right hand side of the Bianchi equation for $\nabla_4\beta$ which can generate the fractional power $v^{\frac{n}{2}-2}$  and thus invalidate~\eqref{stuffinbetatobe0} is the $\nabla^A\alpha_{AB}$.  Then, it follows immediately from self-similarity and the hypothesis of the proposition that $\mathcal{L}_v^{\frac{n}{2}-2}\alpha|_{v=0} = 0$ and~\eqref{stuffinbetatobe0} holds. An analogous argument works for $\nabla^i\zeta$. 

When $n$ is even, essentially the same argument goes through except that we now require $\mathcal{L}_v^{\frac{n}{2}-2}\alpha|_{v=0}$ is such that the right hand side of the $\nabla_4$ Bianchi equation for $\beta$ is $o\left(v^{\frac{n}{2}-2}\right)$. Finally, it follows easily from our analysis of admissible conjugate data, see Proposition~\ref{incomingdataneven}, that this can be arranged if and only if ${\rm div}\left({\rm tf}\mathcal{L}_v^{\frac{n}{2}-2}\alpha|_{\mathcal{S}_{-1,0}}\right) = D$ for a specific $1$-form $D$ which only depends on $\slashed{g}|_{\mathcal{S}_{-1,0}}$. 
\end{proof}

Finally, a straightforward unique continuation argument implies that under the hypothesis of the previous proposition, we in fact have that $\zeta$ and $\omega$ vanish identically.
\begin{proposition}Under the same hypothesis as Proposition~\ref{infvanish}, we have that $\zeta$ and $\omega$ vanish identically.
\end{proposition}
\begin{proof}Propositions~\ref{anequationfordivb}, ~\ref{anequationforundomega},~\ref{anequationforzeta}, and Lemma~\ref{someselfsimilarrelationsyay} imply that
\begin{equation}\label{divbvanstr2}
\partial_v\left(\Omega^{-1}{\rm div}b\right) - \Omega^{-1}{\rm div b}\left(v^{-1} + O\left(1\right)\right) = -4\omega nu^{-1} + v^{-1}O\left(\left|\slashed{\nabla}b\right|^2\right) + O\left(\left|\slashed{\nabla}b\right| + |b|\right),
\end{equation}
\begin{equation}\label{zetavanstr2}
\nabla_4\zeta - \frac{n}{2}\zeta\left(v^{-1} + O\left(1\right)\right) - \frac{1}{4}\frac{u}{v}\slashed{\Delta}b +\frac{1}{4}\Omega^{-1}\frac{u}{v}\nabla_A{\rm div}b= v^{-1}O\left(|b| + \left|\nabla\Omega\right|+ \left|\nabla\Omega\right|^2 + \left|\nabla b\right|^2 + \left|\nabla\zeta\right|^2\right).
\end{equation}
\begin{equation}\label{omegavanstr2}
\nabla_4\omega - \frac{n}{2}\omega\left(v^{-1} + O\left(1\right)\right) +\frac{1}{2}\frac{u}{v}\slashed{\Delta}\log\Omega =  v^{-1}O\left( \left|\zeta\right|^2 + \left|\nabla\Omega\right|^2\right) + v^{-2}\left(\left|\nabla b\right|^2 + \left|b\right|^2 + \left|{\rm div}b\right|\right) + O\left(|\omega|^2\right).
\end{equation}
Since we have already established the infinite order vanishing of the quantities as $v\to 0$ this is a straightforward unique continuation argument for which we will sketch one possible approach. Neglecting the nonlinear terms, the linear part of the equations generate a good energy estimate by letting $A_1$ and $A_2$ be large constants with $A_2 \gg A_1$, multiplying~\eqref{divbvanstr2} by $A_2v^{-A-1}{\rm div}b$, multiplying~\eqref{zetavanstr2} by $v^{-A}\zeta$, multiplying~\eqref{omegavanstr2} by $v^{-A}\omega$, integrating over a region $[0,v_0] \times \mathcal{S}$, integrating by parts, adding the resulting estimates together and noting that there is no contribution from $v = 0$ due to Proposition~\ref{infvanish}. The nonlinear terms can be easily controlled by the linear terms after using Proposition~\ref{infvanish} to control one of the terms in each quadratic combination.
\end{proof}

Lastly, we note that when $n = 2$ there is a simple argument directly at the level of the Bianchi equations.
\begin{proposition}Let $n=2$ and suppose that $\nabla^A\hat{\chi}_{AB}|_{(u,v) = (-1,0)} = \nabla k$, where $k$ is the Gaussian curvature of $\slashed{g}_0$. Then $\zeta$ and $\omega$ vanish identically, and furthermore, all null curvature components vanish identically.
\end{proposition}
\begin{proof}Under the hypothesis of the proposition, Proposition~\ref{incomingdatan2} implies that all curvature components vanish along $\{v = 0\}$. In particular, the solution will obey the initial data estimates of Proposition~\ref{initsupern2} without overlining the double null quantities. It immediately follows that we can apply the energy estimate scheme from Section~\ref{superduperestiamtes} to establish supercritical estimates. It immediately follows that all null curvature components must vanish from the same scaling considerations as in the proof of Theorem~\ref{selfsimilarextract}.

Given that $\beta$ and $\rho$ vanish identically, it is straightforward to use the null structure equations for $\eta$ and $\omega$ as well as~\eqref{shiftomegabetter} to conclude that $\omega$, $\zeta$, $b$, and $\Omega - 1$ all vanish identically. 
\end{proof}
\section{Some Coordinates}\label{somecoordinates}

\subsection{Fefferman--Graham coordinates}\label{fgcoord}
In~\cite{FG1,FG2} Fefferman and Graham did not work in a double null foliation; however, their coordinates can be easily understood in the context of a double null foliation after defining
\[t \doteq u,\qquad \rho \doteq \frac{v}{u}.\]

Using
\[du = dt,\qquad dv = \rho dt + td\rho,\]
we find 
\begin{align*}
g &= -2\Omega^2 \left(dt\otimes \left(\rho dt + td\rho\right) + \left(\rho dt + t d\rho\right) \otimes dt\right)+ \slashed{g}_{AB}\left(d\theta^A - b^Adt\right)\otimes\left(d\theta^B - b^Bdt\right)
\\ \nonumber &= \left(-4\Omega^2\rho + |b|^2\right)dt^2 -4t\Omega^2dtd\rho - 2b_Adtd\theta^A  + \slashed{g}_{AB}d\theta^Ad\theta^B.
\end{align*}

The reader may easily check that we have
\[\partial_u = \partial_t - \frac{\rho}{t}\partial_{\rho},\qquad \partial_v = \frac{1}{t}\partial_{\rho}, \qquad K = t\partial_t.\]

Finally, we observe that if we set 
\[\mathring{\slashed{g}}_{AB}\left(\rho,\theta\right) \doteq \slashed{g}_{AB}\left(1,\rho,\theta\right),\]
then self-similar metrics will satisfy
\begin{equation}\label{slashgselfsim}
\slashed{g}_{AB}\left(t,\rho,\theta\right) = t^2\mathring{\slashed{g}}_{AB}\left(\rho,\theta\right).
\end{equation}
\subsection{Straight ambient metrics and asymptotically de Sitter spacetimes}
There exist another set of coordinates which, in certain cases allows us to quotient out the dilation symmetry and produce an $n+1$ dimensional asymptotically de Sitter cosmological spacetime. We define these ``cosmological coordinates'' $\left(r,s,\theta\right)$ by
\[\rho \doteq -\frac{1}{4}r^2,\qquad t \doteq \frac{s}{r}.\]
Using that
\[d\rho = -\frac{1}{2}rdr,\qquad dt = \frac{1}{r}ds - \frac{s}{r^2}dr,\]
we find
\begin{align*}
g &= \left(\Omega^2 + r^{-2}\left|b\right|^2\right)ds^2 - \frac{2s}{r^3}\left|b\right|^2dsdr - \frac{2}{r}b_Adsd\theta^A
\\ \nonumber &\qquad + \frac{2s}{r^2}b_Adrd\theta^A + \left(r^{-2}\left|b\right|^2 - \Omega^2\right)\frac{s^2}{r^2}dr^2 + \slashed{g}_{AB}d\theta^Ad\theta^B.
\end{align*}

We say the ambient metric is \emph{straight} if 
\[\Omega^2 = 1,\qquad b = 0.\]
In this case, the metric will take the form 
\begin{align*}
g &= ds^2  + \frac{s^2}{r^2}\left(-dr^2 + \frac{r^2}{s^2}\slashed{g}_{AB}d\theta^Ad\theta^B\right)
\\ \nonumber &= ds^2  + \frac{s^2}{r^2}\left(-dr^2 + \mathring{\slashed{g}}_{AB}\left(r,\theta\right)d\theta^Ad\theta^B\right)
\end{align*}

A straightforward calculation (see~\cite{FG1,FG2}) shows that the Ricci flatness of $g$ implies that the metric
\[\tilde g \doteq r^{-2}\left(-dr^2 + \mathring{\slashed{g}}_{AB}d\theta^Ad\theta^B\right)\]
is a solution to the Einstein equations with a cosmological constant:
\[Ric\left(\tilde g\right) - n\tilde g = 0.\]

Also note that
\[\partial_t = r\partial_s,\qquad \partial_{\rho} = -\frac{2}{r}\partial_r - \frac{2t}{r}\partial_s,\qquad K = s\partial_s.\]

\end{document}